\documentclass{article}

\usepackage[margin=1in]{geometry}
\usepackage[numbers]{natbib}

\usepackage[utf8]{inputenc} 
\usepackage[T1]{fontenc}    
\usepackage{hyperref}       
\usepackage{url}            
\usepackage{booktabs}       
\usepackage{multirow}

\usepackage{amsfonts}       
\usepackage[fleqn]{amsmath}
\usepackage{amssymb}
\usepackage{amsthm}
\usepackage{nicefrac}       

\usepackage{graphicx}
\usepackage[section]{algorithm}
\usepackage{algorithmic}
\usepackage[colorinlistoftodos]{todonotes}

\usepackage{tikz}
\usepackage{pgfplots}

\usepackage{xcolor}
\newcommand{\blue}{\color{black}}

\usepackage{caption}
\usepackage{subcaption}
\usepackage{microtype}      
\usepackage{afterpage}

\newtheorem{thm}{Theorem}

\newcommand{\argmin}{\text{argmin}}

\begin{document}
\title{Sparse regression: Scalable algorithms and empirical performance}
\author{Dimitris Bertsimas, Jean Pauphilet, Bart Van Parys \\ 
\texttt{\{dbertsim,jpauph,vanparys\}@mit.edu} \\
Operations Research Center, MIT, Cambridge, MA, USA}
\date{February, 2019}
\maketitle

\begin{abstract}
In this paper, we review state-of-the-art methods for feature selection in statistics with an application-oriented eye. Indeed, sparsity is a valuable property and the profusion of research on the topic might have provided little guidance to practitioners. We demonstrate empirically how noise and correlation impact both the accuracy - the number of correct features selected - and the false detection - the number of incorrect features selected - for five methods: the  cardinality-constrained formulation, its Boolean relaxation, $\ell_1$ regularization and two methods with non-convex penalties. A cogent feature selection method is expected to exhibit a two-fold convergence, namely the accuracy and false detection rate should converge to $1$ and $0$ respectively, as the sample size increases. As a result, proper method should recover all and nothing but true features. Empirically, the integer optimization formulation and its Boolean relaxation { are the closest to} exhibit this two properties consistently in various regimes of noise and correlation. In addition, apart from the discrete optimization approach which requires a substantial, yet often affordable, computational time, all methods terminate in times comparable with the \verb|glmnet| package for Lasso. We released code for methods that were not publicly implemented. Jointly considered, accuracy, false detection and computational time provide a comprehensive assessment of each feature selection method and shed light on alternatives to the Lasso-{ regularization} which are not as popular in practice yet. 
\end{abstract}

\section{Introduction}
The identification of important variables in regression is valuable to practitioners and decision makers in settings with large data sets of high dimensionality. Correspondingly, the notion of sparsity, i.e., the ability to make predictions based on a limited number of covariates, has become cardinal in statistics. The so-called { cardinality-penalized} estimators for instance minimize the trade-off between prediction accuracy and number of input variables. Though computationally expensive, they have been considered as a relevant benchmark in high-dimensional statistics. Indeed, { these} estimators are characterized as the solution of the NP-hard problem
\begin{equation}
\label{eqn:bic}
\min_{w \in \mathbb{R}^p} ~\sum_{i=1}^n \ell(y_i, w^\top   x_i) + \lambda \Vert w \Vert_0,
\end{equation}
where $\ell$ is an appropriate convex loss function, such as the ones reported in Table \ref{tab:loss_functions} (p. \pageref{tab:loss_functions}). The covariates are denoted by the matrix $X \in \mathbb{R}^{n \times p}$, whose rows are the $x_i^\top  $'s, and the response data by $Y = (y_1,...,y_n) \in \mathbb{R}^{n}$. Here,  $\|w\|_0 := |\{ j: w_j \neq 0 \}|$ denotes the 0-pseudo norm, i.e., the number of non-zero coefficients of $w$. Alternatively, one can explicitly constrain the number of features used for prediction and solve
\begin{equation}
\label{eqn:sparse}
\min_{w \in \mathbb{R}^p} ~\sum_{i=1}^n \ell(y_i, w^\top   x_i) \mbox{ s.t. } \Vert w \Vert_0 \leqslant k,
\end{equation}
which is likewise an NP-hard optimization problem \citep{natarajan1995sparse}. For decades, such problems have thus been solved using greedy heuristics, such as step-wise regression, matching pursuits \citep{mallat1993matching}, or recursive feature elimination (RFE) \citep{guyon2002gene}.

\begin{table}[h]
\centering
\caption{Relevant loss functions $\ell$ and their corresponding Fenchel conjugates $\hat\ell$ as defined in Theorem \ref{cio_formulation}. The observed data is continuous, $y\in \mathbb{R}$, for regression and categorical, $y \in \{ -1,1\}$, for classification. By convention, $\hat \ell$ equals $+ \infty$ outside of its domain. The binary entropy function is denoted $H(x) := - x \log x - (1-x) \log(1-x)$.}
\label{tab:loss_functions}
\begin{tabular}{lll}
\toprule
Method &Loss $\ell(y, u)$ &Fenchel conjugate $\hat \ell(y,\alpha)$ \\
\midrule
Ordinary Least Square & $\tfrac{1}{2}(y-u)^2$ & $\tfrac{1}{2} \alpha^2 + y \alpha $\\
Logistic loss & $\log \left( 1 + e^{-y u}\right)$ & $-H(-y \alpha) \mbox{ for } y \alpha \in [-1,0]$\\
1-norm SVM - Hinge loss & $\max(0, 1-y u)$ & $y \alpha \mbox{ for } y \alpha \in [-1,0]$  \\
\bottomrule
\end{tabular}
\end{table}

Consequently, much attention has been directed to convex surrogate estimators which tend to be sparse, while {\blue requiring less computational effort}. The Lasso estimator, commonly defined as the solution of  
\begin{equation*}
\min_{w \in \mathbb{R}^p}~\sum_{i=1}^n \ell(y_i, w^\top   x_i) + \lambda \Vert w \Vert_1,
\end{equation*}
and initially proposed by \citet{tibshirani} is widely known and used. Its practical success can be explained by three concurrent ingredients: Efficient numerical algorithms exist \citep{efron2004least,friedman2010regularization,beck2009fast}, off-the-shelf implementations are publicly available \citep{friedman2013glmnet} and recovery of the true sparsity is theoretically guaranteed under admittedly strong assumptions on the data \citep{wainwright2009sharp}. However, recent works \citep{su2015false,fan2010sure} have pointed out several key deficiencies of the Lasso regressor in its ability to select the true features without including many irrelevant ones as well. In a parallel direction, theoretical work in statistics \citep{wainwright2009information,wang2010information,gamarnik17} has identified regimes where Lasso fails to recover the true support even though support recovery is possible from an information theoretic point of view. 

Therefore, new research in numerical algorithms for solving the exact formulation \eqref{eqn:sparse} directly has flourished. Leveraging recent advances in mixed-integer solvers \citep{bertsimas2014statistics,bertsimas2015logistic}, Lagrangian relaxation \citep{pilanci2015sparse}, cyclic coordinate descent \citep{hazimeh2018fast}, or cutting-plane methods \citep{bertsimas2016cio,2017sparseclass}, these works have demonstrated significant improvement over existing Lasso-based heuristics. To the best of our knowledge, the exact algorithm proposed by \citet{bertsimas2016cio,2017sparseclass} is the most scalable method providing provably optimal solutions to the optimization problem \eqref{eqn:sparse}, at the expense of potentially significant computational time and the use of a commercial integer optimization solver.

Another line of research has focused on replacing the $\ell_1$ norm in the Lasso formulation by other sparsity-inducing penalties which are less sensitive to noise or correlation between features. In particular, non-convex penalties such as smoothly clipped absolute deviation (SCAD) \citep{fan2001variable} and minimax concave penalty (MCP) \citep{zhang2010nearly} have been proposed. Both SCAD and MCP have the so-called oracle property, meaning that they do not require a priori knowledge of the sparsity pattern to achieve an optimal asymptotic convergence rate, which is theoretically appealing. From a computational point of view, coordinate descent algorithms \citep{ncvreg} have shown very effective, even though lack of convexity in the objective function hindered their wide adoption in practice.  

Convinced that sparsity is an extremely valuable property in high-impact applications where interpretability matters, and conscious that the profusion of research on the matter might have caused confusion and provided little guidance to practitioners, we propose with the present paper a comprehensive treatment of state-of-the-art methods for feature selection in ordinary least square and logistic regression. Our goal is not to provide a theoretical analysis. On the contrary, we selected and evaluated the methods with an eye towards practicality, taking into account both scalability to large data sets and availability of the implementations. In some cases where open-source implementation was not available, we released code on our website, in an attempt to bridge the gap between theoretical advances and practical adoption. Statistical performance of the methods is assessed in terms of Accuracy ($A$),
\begin{align*}
A(w) := \dfrac{|\{j: w_j \neq 0, \; w_{true,j} \neq 0\}|}{|\{j: w_{true,j} \neq 0\}|},
\end{align*} i.e., the proportion of true features which are selected, and False Discovery Rate ($FDR$), 
\begin{align*}
FDR(w) := \dfrac{|\{j: w_j \neq 0, \; w_{true,j} = 0\}|}{|\{j: w_{j} \neq 0\}|},
\end{align*} i.e., the proportion of selected features which are not in the true support. 

\subsection{Outline and contribution}
Our key contributions can be summarized as follows: 
\begin{itemize}
\item We provide a unified treatment of state-of-the-art methods for feature selection in statistics. More precisely, we cover the cardinality-constrained formulation \eqref{eqn:sparse}, its Boolean relaxation, the Lasso formulation and its derivatives, and the MCP and SCAD penalty. { We did not include step-wise regression methods, for they may require a high number of iterations in high dimension and exist in many variants. }
\item Encouraged by theoretical results obtained for the Boolean relaxation of \eqref{eqn:sparse} by \citet{pilanci2015sparse}, we propose an efficient sub-gradient algorithm to solve it and provide theoretical rate of convergence of our method.
\item We make our code freely available as a \verb|Julia|  package named \verb|SubsetSelection|. Our algorithm scales to problems with $n,p=100,000$ or $n=10,000$ and $p=1,000,000$ within minutes, while providing high-quality estimators.
\item We compare the performance of all methods on three metrics of crucial interest in practice: accuracy, false detection rate and computational tractability, and in various regimes of noise and correlation. 
\item More precisely, { under the mutual incoherence condition}, all methods exhibit a convergence in accuracy, that is the proportion of correct features selected converges to $1$ as the sample size $n$ increases, in all regimes of noise and correlation. Yet, on this matter, cardinality-constrained and MCP formulations are the most accurate. { As soon as mutual incoherence condition fails to hold, $\ell_1$-based estimators are inconsistent with $A < 1$, while non-convex penalties eventually perfectly recover the support.}
\item In addition, we also observe a convergence in false detection rate, namely the proportion of irrelevant features selected converging to $0$ as the sample size $n$ increases, for some but not all methods: The convex integer formulation and its Boolean relaxation are the only methods which demonstrate this behavior, { in low noise settings and make the fewest false discoveries in other regimes}. In our experiments, Lasso-based estimators return at least $80\%$ of non-significant features. MCP and SCAD have a low but strictly positive false detection rate (around $15-30\%$ in our experiments) as $n$ increases and in all regimes. 
\item In terms of computational time, the integer optimization approach is unsurprisingly the most expensive option. Nonetheless, the computational cost is only one or two orders of magnitude higher than other alternatives and remains affordable in many real-world problems, even high-dimensional ones. Otherwise, the four remaining codes terminate in time comparable with the \verb|glmnet| implementation of the Lasso, that is within seconds for $n=1,000$ and $p=20,000$. 
\end{itemize} 

In Section \ref{sec:formulations}, we present each method, its formulation, its theoretical underpinnings and the numerical algorithms proposed to compute it. In each case, we point the reader to appropriate references and open-source implementations. We propose and describe our sub-gradient algorithm for the Boolean relaxation of \eqref{eqn:sparse} also in Section \ref{sec:formulations}. Appendix \ref{sec:subsetselection} provides further details on our implementation of the algorithm, its scalability and its applicability to cardinality-penalized estimators \eqref{eqn:bic} as well. In Section \ref{sec:regression} (and Appendix \ref{sec:regression.supp}), we compare the methods on synthetic data sets for linear regression. In particular, we observe and discuss the behavior of each method in terms of accuracy, false detection rate and computational time for three families of design matrices and at least three levels noise. In Section \ref{sec:classification} (and Appendix \ref{sec:classification.supp}), we apply the methods to classification problems on similar synthetic problems. We also analyze the implications of the feature selection methods in terms of induced sparsity and prediction accuracy on a real data set from genomics.

\subsection{Notations}
In the rest of the paper, we denote with $\textbf{e}$ the vector whose components are equal to one. For { $q\geqslant 1$}, $\| \cdot \|_q$ denotes the $\ell_q$ norm defined as $$\| x \|_q = \left(\sum_i |x_i|^q\right)^{1/q}.$$ For any $d$-dimensional vector $x$, we denote with $x_{[j]}$ the $j$th largest component of $x$. Hence, we have $$x_{[1]} \geqslant \cdots \geqslant x_{[d]}.$$

\section{Sparse regression formulations}
\label{sec:formulations}
In this section, we introduce the different formulations and algorithms that have been proposed to solve the sparse regression problem. We focus on the cardinality-constrained formulation, its Boolean relaxation, the Lasso and Elastic-Net estimators, the MCP and SCAD penalty. 

\subsection{Integer optimization formulation}
As mentioned in introduction, a natural way to compute sparse regressors is to explicitly constrain the number of non-zero coefficients, i.e., solve
\begin{equation*}
\tag{$2$}
\min_{w \in \mathbb{R}^p} ~\sum_{i=1}^n \ell(y_i, w^\top   x_i) \mbox{ s.t. } \Vert w \Vert_0 \leqslant k,
\end{equation*}
where $\ell$ is an appropriate loss function, appropriate in the sense that $\ell(y, \cdot)$ is convex for any $y$ . In this paper, we focus on Ordinary Least Square (OLS), logistic regression and Hinge loss, as presented in Table \ref{tab:loss_functions} on page \pageref{tab:loss_functions}. {\blue Unfortunately}, such a problem is NP-hard \citep{natarajan1995sparse} and believed to be intractable in practice. The original attempt by \citet{furnival1974regressions} using "Leaps and Bounds" scaled to problems with $n,p$ in the $10s$. Thanks to both hardware improvement and advances in mixed-integer optimization solvers, \citet{bertsimas2014statistics,bertsimas2015logistic} successfully used discrete optimization techniques to solve instances with $n,p$ in the $1,000$s within minutes. More recently, \citet{bertsimas2016cio,2017sparseclass} proposed a cutting plane approach which scales to data sizes of with $n,p$ in the $100,000$s for ordinary least square and $n,p$ in the $10,000$s for logistic regression. To the best of our knowledge, our approach is the only method which scales to instances of such sizes, while provably solving such an NP-hard problem.  

\subsubsection{Convex integer formulation}
\citet{bertsimas2016cio,2017sparseclass} consider an $\ell_2$-regularized version of the initial formulation \eqref{eqn:sparse},
\begin{equation}
\label{eqn:reg_sparse}
\min_{w \in \mathbb{R}^p}~\sum_{i=1}^n \ell(y_i, w^\top   x_i) + \dfrac{1}{2 \gamma} \Vert w \Vert_2^2 \mbox{ s.t. } \Vert w \Vert_0 \leqslant k,
\end{equation} 
where $\gamma>0$ is a regularization coefficient. From a statistical point of view, this extra regularization, referred to as ridge or Tikhonov regularization, is needed to account for correlation between features \citep{hoerl1970ridge} and {\blue mitigate} the effect of noise. Indeed, regularization and robustness are two intimately connected properties, as illustrated by \citet{bertsimas2009equivalence,xu2009robustness}. In addition, \citet{breiman1996heuristics} proved that subset selection is a very unstable problem and highlighted the stabilizing effect of Ridge regularization. {\blue Introducing} a binary variable $s \in \{0, 1\}^p$ to encode the support of $w$ and using convex duality, problem \eqref{eqn:reg_sparse} can be shown equivalent to a convex integer optimization problem as stated in the following theorem.

\begin{thm}{\citep[Theorem 1]{2017sparseclass}} \label{cio_formulation}
For any convex loss function $\ell$, problem \eqref{eqn:reg_sparse} is equivalent to 
\begin{equation}
\label{eqn:minmax}
\min_{s \in \{0,1\}^p: s^\top  \textbf{e} \leqslant k} \: \max_{\alpha \in \mathbb{R}^n} \: f(\alpha,s) := \left( -\sum_{i=1}^n \hat \ell(y_i, \alpha_i) - \dfrac{\gamma}{2} \sum_{j=1}^p s_j \alpha^\top   X_j X_j^\top   \alpha \right),
\end{equation}
where $ \hat \ell(y, \alpha) := \max_{u \in \mathbb{R}} u \alpha - \ell(y,u) $ is the \emph{Fenchel conjugate} of the loss function $\ell$ \citep[see][chap. ~6.4]{bertsekas2016nonlinear}, as reported in Table \ref{tab:loss_functions}.
In particular, the function $f$ is continuous, linear in $s$ and concave in $\alpha$. 
\end{thm}

In the special case of OLS, the function $f$ is a quadratic function in $\alpha$
\begin{align*}
f(\alpha,s) = -\dfrac{1}{2} \| \alpha \|^2 - Y^\top   \alpha - \dfrac{\gamma}{2} \alpha X_s X_s^\top   \alpha,
\end{align*}
where $X_s X_s^\top  := \sum_{j=1}^p s_j X_j X_j^\top  $. As a result, the inner maximization problem can be solved in closed form: The maximum is attained at $\alpha^\star(s)= - (I_n+\gamma X_s X_s^\top  )^{-1} Y$ and the objective value is 
\begin{align*}
\max_\alpha f(\alpha, s) = \dfrac{1}{2} Y^\top   (I_n+\gamma X_s X_s^\top  )^{-1} Y.
\end{align*}
\subsubsection{Cutting-plane algorithm}
Denoting
\begin{align*}
c(s) := \max_{\alpha \in \mathbb{R}^n} f(\alpha,s), 
\end{align*}
which is a convex function in $s$, the cutting-plane algorithm solves the convex integer optimization problem 
\begin{align*}
\min_{s \in \{0,1\}^p} ~c(s) \mbox{  s.t.  } s^\top  \textbf{e} \leqslant k,
\end{align*} by iteratively tightening a piece-wise linear lower approximation of $c$. Pseudo-code is given in Algoritm \ref{OA} (p. \pageref{OA}). Proof of termination and details on implementation can be found in \citet{bertsimas2016cio} for regression and \citet{2017sparseclass} for classification. This outer-approximation scheme was originally proposed by \citet{duran1986outer} for general nonlinear mixed-integer optimization problems.
\begin{algorithm*}
\caption{Outer-approximation algorithm}
\label{OA}
\begin{algorithmic}
\REQUIRE $X \in \mathbb{R}^{n \times p}$, $Y \in \mathbb{R}^{n}$, $k \in \lbrace 1,...,p \rbrace$ 
\STATE $t \leftarrow 1 $
\REPEAT
\STATE $s^{t+1},\eta^{t+1} \leftarrow \text{argmin}_{s \in \{0, 1\}^p,\eta} \left\lbrace \eta \: : \sum_{j=1}^p s_j \leqslant k, \; \eta \geqslant c(s^i) + \nabla c(s^i)^\top  (s-s^i),~  \forall i =1,\dots , t \right\rbrace $ 
\STATE $t \leftarrow t+1 $
\UNTIL{$\eta^t < c(s^t) - \varepsilon$}
\RETURN $s^t$ 
\end{algorithmic}
\end{algorithm*}
\subsubsection{Implementation and publicly available code}
{ A naive implementation of Algorithm \ref{OA} would solve a mixed-integer linear optimization problem at each iteration, which can be as expensive as explicit enumeration of all feasible supports $s$. Fortunately, with modern solvers such as Gurobi \citep{gurobi2016gurobi} or CPLEX \citep{cplex2009v12}, this outer-approximation scheme  can be implemented using \textit{lazy constraints}, enabling the use of a single Branch-and-Bound tree for all subproblems.

The algorithm terminates when the incumbent solution is $\varepsilon$-optimal for some fixed tolerance level $\varepsilon$ (we chose $\varepsilon = 10^{-4}$ in our simulations). We also need to impose a time limit on the algorithm. Indeed, as often in discrete optimization, the algorithm can quickly find the optimal solution, but spends a lot of the time proving its optimality. In our experiment, we fixed a time limit of $60$ seconds for regression and $180$ seconds for classification. Such choices were guided by confidence in the quality of the initial solution $s_1$ we provide to the algorithm (which we will describe in the next section) as well as time needed to compute $c(s)$ and $\nabla c(s)$ for a given support $s$. 

The formulation \eqref{eqn:reg_sparse} contains two hyper-parameters, $k$ and $\gamma$, to control for the amount of sparsity and regularization respectively. In practice, those parameters need to be tuned using a cross-validation procedure. Since the function to minimize does not depend on $k$, any piece-wise linear lower approximation of $c(s)$ computed to solve \eqref{eqn:reg_sparse} for some value of $k$ can be reused to solve the problem at another sparsity level. {\blue In recent work, \citet{kenney2018efficient} proposed a combination of implementation recipes to optimize such search procedures.} As for $\gamma$, we apply the procedure described in \citet{chu2015warm}, starting with a low value $\gamma_0$ (typically scaling as $1 / {\max_i \|x_i\|^2}$) and inflate it by a factor $2$ at each iteration. }

To bridge the gap between theory and practice, we present a \verb|Julia| code which implements the described cutting-plane algorithm publicly available on GitHub\footnote{\url{https://github.com/jeanpauphilet/SubsetSelectionCIO.jl}}. The code requires a commercial solver like Gurobi { or CPLEX} and our open-source package \verb|SubsetSelection| for the Lagrangian relaxation, which we introduce in the next section. { We also call the open-source library \verb|LIBLINEAR| \citep{fan2008liblinear} to efficiently compute $c(s)$ in the case of Hinge and logistic loss. }

\subsection{Lagrangian relaxation}
As often in discrete optimization, it is natural to consider the Boolean relaxation of problem \eqref{eqn:minmax}
\begin{equation}
\label{eqn:relax}
\min_{s\in [ 0, 1 ]^p \: : \: \textbf{e}^\top   s \leqslant k} \max_{\alpha \in \mathbb{R}^n} f(\alpha, s),
\end{equation}
and study its tightness, as done by \citet{pilanci2015sparse}. 

\subsubsection{Tightness result}
The above problem is recognized as a convex/concave saddle point problem. According to Sion's minimax theorem \citep{sion1958general}, the minimization and maximization in \eqref{eqn:relax} can be interchanged. Hence, saddle point solutions $(\bar \alpha, \bar s)$ of \eqref{eqn:relax} should satisfy 
$$\bar \alpha \in \arg \max_{\alpha \in \mathbb{R}^n} f(\alpha, \bar s), ~~~
\bar s  \in \arg \min_{s \in [0,1]^p} \,f(\bar \alpha, s) \mbox{ s.t. } ~s^\top   \textbf{e} \leqslant k .$$
Since $f$ is a linear function of $s$, a minimizer of $f(\bar \alpha, s)$ can be constructed easily by selecting the $k$ smallest components of the vector $(-\tfrac{\gamma}{2} \bar \alpha^\top   X_j X_j^\top   \bar \alpha)_{j=1,...,p}$. If those $k$ smallest components are unique, the so constructed binary vector must be equal to $\bar s$ and hence the relaxation \eqref{eqn:relax} is tight. In fact, the previous condition is necessary and sufficient as proven by \citet{pilanci2015sparse}:
\begin{thm}{\citep[Proposition 1]{pilanci2015sparse}} \label{tightness}
The Boolean relaxation \eqref{eqn:relax} is tight if and only if there exists a saddle point $(\bar \alpha, \bar s)$ such that the vector $\bar \beta := (\bar \alpha^\top   X_j X_j^\top   \bar \alpha)_{j=1,...,p}$ has unambiguously defined $k$ largest components, i.e., there exists $\lambda \in \mathbb{R}$ such that $\bar \beta_{[1]} \geqslant \cdots \geqslant \bar \beta_{[k]} > \lambda > \bar \beta_{[k+1]} \geqslant \cdots \geqslant \bar \beta_{[p]}$.
\end{thm}
This uniqueness condition in Theorem \ref{tightness} seems often fulfilled in real-world applications. {\blue It is satisfied with high probability, for instance,  when the covariates $X_j$ are independent} \citep[see][Theorem~2]{pilanci2015sparse}. In other words, randomness breaks the complexity of the problem. Similar behavior has already been observed for semi-definite relaxations \citep{iguchi2015tightness,ye2016data}. Such results have had  impact in practice and propelled the advancement of convex proxy based heuristics such as Lasso. Efficient algorithms can be designed to solve the saddle point problem \eqref{eqn:relax} without involving sophisticated discrete optimization tools and provide a high-quality heuristic for approximately solving \eqref{eqn:minmax} that could be used as a good warm-start in exact approaches.

\subsubsection{Dual sub-gradient algorithm}
In this section, we propose and describe an algorithm for solving problem \eqref{eqn:relax} efficiently and make our code available as a \verb|Julia|  package. Our algorithm is fast and scales to data sets with $n,p$ in the $100,000$s, which is two orders of magnitude larger than the implementation proposed by \citet{pilanci2015sparse}.

For a given $s$, maximizing $f$ over $\alpha$ cannot be done analytically, with the noteworthy exception of ordinary least squares, whereas minimizing over $s$ for a fixed $\alpha$ reduces to sorting the components of $(-\alpha^\top   X_j X_j^\top   \alpha)_{j=1,...,p}$ and selecting the $k$ smallest. We take advantage of this asymmetry by proposing a dual projected sub-gradient algorithm with constant step-size, as described in pseudo-code in Algorithm \ref{subgradient}. $\delta$ denotes the step size in the gradient update and $\mathcal{P}$ the projection operator over the domain of $f$. { At each iteration, the algorithm updates the support $s$ by minimizing $f(\alpha, s)$ with respect to $s$, $\alpha$ being fixed. Then, the variable $\alpha$ is updated by performing one step of projected sub-gradient ascent with constant step size $\delta$.} The denomination "sub-gradient" comes from the fact that at each iteration $\nabla_\alpha f(\alpha^T, s^T)$ is a sub-gradient to the function $\alpha \mapsto \min_s f(\alpha, s)$ at $\alpha = \alpha^T$. 
 
{ In terms of computational cost, updating $\alpha$ requires $O\left(n \|s \|_0 \right)$ operations for computing the sub-gradient plus at most $O\left(n \right)$ operations for the projection on the feasible domain. The most time-consuming step in Algorithm \ref{subgradient} is updating $s$ which requires on average $O\left(n p + p \log p \right)$ operations.}
 
\begin{algorithm*}
\caption{Dual sub-gradient algorithm}
\label{subgradient}
\begin{algorithmic}
\STATE $s^0, \alpha^0 \leftarrow$ Initial solution
\STATE $T = 0 $
\REPEAT
\STATE $s^{T+1} \in \text{argmin}_{s} f(\alpha^{T}, s)$
\STATE $\alpha^{T+1} = \mathcal{P} \left( \alpha^T + \delta \nabla_\alpha f(\alpha^T, s^T) \right)$
\STATE $T = T+1 $
\UNTIL{Stop criterion}
\STATE $\hat \alpha_T = \tfrac{1}{T} \sum_t \alpha^t$
\RETURN $\hat s = \text{argmin}_{s} f(\hat \alpha_T, s)$ 
\end{algorithmic}
\end{algorithm*}

The final averaging step $\hat \alpha_T = \tfrac{1}{T} \sum_t \alpha_t$ is critical in sub-gradient methods to ensure convergence of the algorithm in terms of optimal value \citep[see][chap. ~7.5]{bertsekas2016nonlinear}.

\begin{thm}{\citep[][chap. 7.5]{bertsekas2016nonlinear}} \label{convergence}
Assume the sequence of sub-gradients $\{ \nabla_\alpha f(\alpha_T, s_T) \}$ is uniformly bounded by some constant $L>0$, and that the set of saddle point solutions $\bar A$ in \eqref{eqn:relax} is non-empty. Then,
\begin{align*}
f\left( \hat \alpha_T, \hat s\right) &\geqslant f(\bar \alpha, \bar s) - \dfrac{\delta L^2}{2} - \dfrac{\text{dist}^2(\alpha^0, \bar A)}{2 \delta T},
\end{align*}
where $\text{dist}(\alpha^0, \bar A)$ denotes the distance of the initial point $\alpha^0$ to the set of saddle point solutions $\bar A$.
\end{thm}
As for any sub-gradient strategy, with an optimal choice of step size $\delta$\footnote{$\delta = \dfrac{\text{dist}(\alpha^0, \bar A)}{L \sqrt{T}}$.}, Theorem \ref{convergence} proves a $O(1/\sqrt{T})$ convergence rate in terms of objective value, which is disappointingly slow. However, in practice, convergence towards the optimal primal solution $\bar s$ is more relevant. In that metric, our algorithm performs particularly well as numerical experiments in Sections \ref{sec:regression} and \ref{sec:classification} demonstrate. The key to our success is that the optimal primal solution is estimated using partial minimization $$\hat s = \arg \min_s f(\hat \alpha_T,s),$$ as opposed to averaging $$\hat s= \tfrac{1}{T} \sum_t s^t, $$ as studied by \citet{nedic2009approximate} and commonly implemented for sub-gradient methods. In addition, even though we are solving a relaxation, we are interested in binary vectors $s$, which can be interpreted as a set of features. To that extend, averaging would not have been a suitable option since the averaged solution is neither binary, nor $k$-sparse. { With the extra cost of computing $c(s^t)$ for all past iterates $s^t$ as well as $c(\hat s)$, one can also decide to return the support vector with the lowest value. This can only produce a better approximation of $\argmin_{s \in \{ 0,1\}^p : s^\top   \textbf{e} \leqslant k} \: c(s)$.  }

\subsubsection{Implementation and open-source package} 
{ The algorithm terminates after a fixed number of iterations $T_{max}$ which is standard for sub-gradient methods. In our case, however, the quality of the primal variable $s$ should be the key concern. By computing $c(s^t)$ at each iteration and keeping track of the best upper-bound $\min_{t=1,\dots,T} c(s^t)$, one can use the duality gap or the number of consecutive iterations without any improvement on the upper-bound as alternative stopping criteria. Computing  $c(s^t)$ increases the cost per iteration, with the hope of terminating the algorithm earlier. By default, our implementation stops after $T_{max}=200$ iterations or when the duality gap is $10^{-4}$. 

The constant step size rule is difficult to implement in practice. Indeed, as seen in Theorem \ref{convergence}, an optimal step size should depend on quantities that are hard do estimate \textit{a priori}, namely $L$ and $\text{dist}^2(\alpha^0, \bar A)$. In particular for logistic loss, $L$ can be arbitrarily large. Instead, one can use an adaptive stepsize rule such as
\begin{align*}
\delta^T = \dfrac{\min_{t=1,\dots,T} c(s^t) - \max_{t=1,\dots,T} f(\alpha^t, s^t)}{\| \nabla_\alpha f(\alpha^T, s^T) \|^2}.
\end{align*}
We implemented such a rule and refer to \citet[][chap. ~7.5]{bertsekas2016nonlinear} for proofs of convergence and alternative choice.

We apply the same grid-search procedures as for the cutting-plane algorithm in order to  cross-validate the hyperparameters $k$ and $\gamma$.}

We make our code publicly available as a \verb|Julia|  package named \verb|SubsetSelection| and source repository can be found on GitHub\footnote{\url{https://github.com/jeanpauphilet/SubsetSelection.jl}}. Our code implements Algorithm \ref{subgradient} for six loss functions including those presented in Table \ref{tab:loss_functions}. The package consists of one main function, \verb|subsetSelection|, which solves problem \eqref{eqn:relax} for a given value of $k$.  { The algorithm can be extended to more loss functions, and cardinality-penalized estimators as well, as described in Appendix \ref{sec:subsetselection}.}

\subsection{Lasso - $\ell_1$ relaxation}
Instead of solving the NP-hard problem \eqref{eqn:bic}, \citet{tibshirani} proposed replacing the non-convex $\ell_0$-pseudo norm by the convex $\ell_1$-norm which is sparsity-inducing. Indeed, extreme points of the unit $\ell_1$ ball $\{x: \|x\|_1 \leqslant 1\}$ are $1$-sparse vectors. The resulting formulation 
\begin{equation}
\label{eqn:lasso}
\min_{w \in \mathbb{R}^p}~\sum_{i=1}^n \ell(y_i, w^\top   x_i) + \lambda \Vert w \Vert_1,
\end{equation}
is referred to as the Lasso estimator. More broadly, $\ell_1$-regularization is now a widely used technique to induce sparsity in a variety of statistical settings \citep[see][for an overview]{hastie2015statistical}. Its popularity has been supported by an extensive corpus of theoretical results from signal processing and high-dimensional statistics. Since the seminal work of \citet{donoho2001uncertainty}, assumptions needed for the Lasso estimator to accurately approximate the true sparse signal are pretty well understood. We refer to reader to \citet{candes2006stable,meinshausen2006high,zhao2006model,wainwright2009sharp} for some of these results. In practice however, those assumptions, namely mutual incoherence and restricted eigenvalues conditions, are quite stringent and hard to verify. In addition, even { when} the Lasso regressor provably converges to the true sparse regressor in terms of $\ell_2$ distance and identifies all the correct features, it also systematically incorporates irrelevant ones, a behavior observed and partially explained by \citet{su2015false} and of crucial practical impact.
\subsubsection{Elastic-Net formulation}
The Lasso formulation in its original form \eqref{eqn:lasso} involves one hyper-parameter $\lambda$ only, which controls regularization, i.e., robustness of the estimator against uncertainty in the data \citep[see][]{xu2009robustness,bertsimas2009equivalence}. At the same time, the $\ell_1$ norm in the Lasso formulation \eqref{eqn:lasso} is also used for its fortunate but {\blue collateral} sparsity-inducing property. Robustness and sparsity, though related, are two very distinct properties demanding a separate hyper parameter each. The ElasticNet (ENet) formulation proposed by \citet{zou2005regularization}
\begin{equation}
\label{eqn:enet}
\min_{w \in \mathbb{R}^p}~\sum_{i=1}^n \ell(y_i, w^\top   x_i) + \lambda \left[ \alpha \Vert w \Vert_1 + \tfrac{1-\alpha}{2} \|w\|_2^2 \right],
\end{equation}
addresses the issue by adding an $\ell_2$ regularization to the objective. For $\alpha=1$, problem \eqref{eqn:enet} is equivalent to the Lasso formulation \eqref{eqn:lasso}, while $\alpha=0$ corresponds to Ridge regression. {\blue Although this extra regularization reduces bias and improves prediction error, it} does not significantly improve feature selection, as we will see on numerical experiments in Section \ref{sec:regression}. { In our view, this is due to the fact that $\ell_1$-regularization primarily induces robustness of the estimator, through shrinkage of the coefficients \citep{bertsimas2009equivalence,xu2009robustness}. In that perspective, it leads to first-rate out-of-sample predictive performance, even in high-noise regimes \citep[see][for extensive experiments]{hastie2017extended}. Nonetheless, the feature selection ability of $\ell_1$-regularization ought to be challenged.}

\subsubsection{Algorithms and implementation}
For $\ell_1$-regularized regression, Least Angle Regression (LARS) \citep{efron2004least} is { an efficient} method for computing an entire path of solutions for various values of the $\lambda$ parameter, exploiting the fact that the regularization path is piecewise linear. More recently, coordinate descent methods \citep{friedman2007pathwise,wu2008coordinate,friedman2010regularization} have successfully competed with { and surpassed } the LARS algorithm, especially in high dimension. Their implementation through the \verb|glmnet| package \citep{friedman2013glmnet} is publicly available in \verb|R| and many other programming languages. {\blue In a different direction, proximal gradient descent methods have also been proposed}, and especially the Fast Iterative Shrinkage Thresholding Algorithm (FISTA) proposed by \citet{beck2009fast}.
\subsection{Non-convex penalties}
Recently, other formulations have been proposed, of the form
\begin{equation}
\label{eqn:noconv}
\min_{w \in \mathbb{R}^p}~\sum_{i=1}^n \ell(y_i, w^\top   x_i) +\sum_{j=1}^p p_{\lambda,\gamma}(|w_j|),
\end{equation}
where $p_{\lambda,\gamma}(\cdot)$ is a function parametrized by $\lambda$ and $\gamma$, which control respectively the tradeoff between empirical loss and regularization, and the shape of the function. { We will consider two popular choice of penalty functions $p_{\lambda,\gamma}(\cdot)$, which are non-convex and are proved to recover the true support even when mutual incoherence condition fails to hold \citep{loh2017support}.}

\subsubsection{Minimax Concave Penalty (MCP)}
The minimax concave penalty of \citet{zhang2010nearly} is  defined on $[0,\infty)$ by
\begin{align*}
p_{\lambda,\gamma}(u) = \lambda \int_0^u \left(1-\dfrac{t}{\gamma \lambda}\right)_+ \text{d}t = \begin{cases} \lambda u - \dfrac{u^2}{2 \gamma} & \mbox{ if } u \leq \gamma \lambda, \\
\dfrac{\gamma \lambda^2}{2} & \mbox{ if } u > \gamma \lambda,
\end{cases}
\end{align*}
for some $\lambda \geqslant 0$ and $\gamma > 1$. The rationale behind the MCP {\blue can be explained } in the univariate OLS case: In the univariate case, MCP and $\ell_1$-regularization lead to the same solution as $\gamma \rightarrow \infty$, while the MCP is indeed equivalent to hard-thresholding when $\gamma=1$. In other words, in one dimension or under the orthogonal design assumption, the MCP produces the so-called firm-shrinkage estimator introduced by \citet{gao1997waveshrink}, which {\blue should} be understood as a continuous tradeoff between hard- and soft-thresholding. 

\subsubsection{Smoothly Clipped Absolute Deviation (SCAD)}
\citet{fan2001variable} orginally proposed the smoothly clipped absolute deviation penalty, defined on $[0,\infty)$ by
\begin{align*}
p_{\lambda,\gamma}(u) = \begin{cases} \lambda u &  \mbox{ if } u \leq \lambda \\
\dfrac{\gamma \lambda u - (u^2+\lambda^2)/2}{\gamma-1} & \mbox{ if } \lambda < u \leq \gamma \lambda, \\
\dfrac{\lambda^2(\gamma^2-1)}{2(\gamma-1)} & \mbox{ if } u > \gamma \lambda,
\end{cases}
\end{align*}
for some $\lambda \geqslant 0$ and $\gamma >2$. The rationale behing the SCAD penalty is similar to the MCP but less straightforward. We refer to \citet{fan2001variable} for a comparison of SCAD penalty with hard-thresholding and $\ell_1$-penalty and to \citet{ncvreg} for a comparison of SCAD and MCP.

\subsubsection{Algorithms and implementation}
For such non-convex penalties, \citet{zou2008one} designed a local linear approximation (LLA) approach where, at each iteration, the penalty function is linearized around the current iterate and the next iterate is obtained by solving the resulting convex optimization problem with linear penalty. Another, more computationally efficient, approach has been proposed by \citet{ncvreg} and implemented in the open-source \verb|R| package, \verb|ncvreg|. Their algorithm relies on coordinate descent and the fact that the objective function in \eqref{eqn:noconv} with OLS loss is convex in any $w_j$, the other $w_{j'}, j'\neq j$ being fixed. For logistic loss, they locally approximate the loss function by a second-order Taylor expansion at each iteration and use coordinate descent to compute the next iterate.   

\section{Linear regression on synthetic data}
\label{sec:regression}
In this section, we compare the aforementioned methods on synthetic linear regression data where the ground truth is known to be sparse. The convex integer optimization algorithm of \citet{bertsimas2016cio,2017sparseclass} (CIO in short) was implemented in \verb|Julia| \citep{LubinDunningIJOC} using the commercial solver Gurobi 7.5.0 \citep{gurobi2016gurobi}, interfaced using the optimization package \verb|JuMP| \citep{dunning2015jump}. The Lagrangian relaxation is also implemented in \verb|Julia| and openly available as the \verb|SubsetSelection| package (SS in short). We used the implementation of Lasso/Enet provided by the \verb|glmnet| package \citep{friedman2013glmnet}, available both in \verb|R| and \verb|Julia|. We also compared to MCP and SCAD penalty formulations implemented in the \verb|R| package \verb|ncvreg| \citep{ncvreg}. The computational tests were performed on a computer with Xeon @2.3GhZ processors, { 1} CPUs, { 16}GB RAM per CPU. 

\subsection{Data generation methodology}
The synthetic data was generated according to the following methodology: We draw $x_i \sim \mathcal{N}(0_p, \Sigma), i=1,..,n$ independent realizations from a $p$-dimensional normal distribution with mean $0_p$ and covariance matrix $\Sigma$. We randomly sample a weight vector $w_{true} \in \lbrace -1, 0, 1 \rbrace$ with exactly $k_{true}$ non-zero coefficient. We draw $\varepsilon_i, i=1,...,n,$ i.i.d. noise components from a normal distribution scaled according to a chosen signal-to-noise ratio $\sqrt{SNR} = \| X w_{true} \|_2 / \| \varepsilon \|_2$. Finally, we compute $Y = X w_{true} + \varepsilon$. With this methodology, we are able to generate data sets of arbitrary size $(n,p)$, sparsity $k_{true}$, correlation structure $\Sigma$ and level of noise $\sqrt{SNR}$. { The signal-to-noise ratio relates to the percentage of variance explained ($PVE$). Indeed, \citet{hastie2017extended} showed that 
\begin{align*}
PVE = \dfrac{SNR}{1+SNR}. 
\end{align*} 
Accordingly, we will consider $SNR$ values ranging from $6$ ($PVE = 85.7\%$) to $0.05$ ($PVE = 4.8\%$). }


\subsection{Metrics and benchmarks}
Statistical performance of the methods is assessed in terms of Accuracy ($A$),
\begin{align*}
A(w) := \dfrac{|\{j: w_j \neq 0, \; w_{true,j} \neq 0\}|}{|\{j: w_{true,j} \neq 0\}|},
\end{align*} i.e., the proportion of true features which are selected, and False Discovery Rate ($FDR$), 
\begin{align*}
FDR(w) := \dfrac{|\{j: w_j \neq 0, \; w_{true,j} = 0\}|}{|\{j: w_{j} \neq 0\}|},
\end{align*} i.e., the proportion of selected features which are not in the true support. We refer to the quantities in the numerators as the number of true features ($TF$) and false features ($FF$) respectively. One might argue that accuracy as we defined it here is a purely theoretical metric, since on real-world data, the ground truth is unknown and there is no such thing as \emph{true} features. Still, accuracy is the only metric which assesses feature selection only, while derivative measures such as predictive power depend on more factors than the features selected alone. Moreover, accuracy has some practical implications in terms of interpretability and also in terms of predictive power: Common sense and empirical results suggest that better selected features should yield diminished prediction error. To that end, we also compare the performance of the methods in terms of out-of-sample Mean Square Error 
{\blue \begin{align*}
MSE(w) := \dfrac{1}{n} \sum_{i=1}^n (y_i - x_i^\top   w)^2,
\end{align*} which will be the metric of interest on real data. Note that the sum can be taken over the observations in the training (in-sample) or test set (out-of-sample).}

Practical scalability of the algorithms is assessed in terms of computational time. In order to provide a fair comparison between methods that are not implemented in the same programming language, we report computational time for each algorithm relative to the time needed to compute a Lasso estimator with \verb|glmnet| in the same language and on the same data. For these experiments, we fixed a time limit of $60$ seconds for the cutting-plane algorithm and considered $150$ iterations of the sub-gradient algorithm for the Boolean relaxation.

\subsection{Synthetic data satisfying mutual incoherence condition}
{ We first consider Toeplitz covariance matrix $\Sigma = \left( \rho^{|i-j|} \right)_{i,j}$. Such matrices satisfy the mutual incoherence condition, required by $\ell_1$-regularized estimators to be statistically consistent. We compare the performance of the methods in six different regimes of noise and correlation described in Table \ref{tab:reg.mic.regimes} (p. \pageref{tab:reg.mic.regimes}).
\begin{table}[h]
\centering
\caption{Regimes of noise ($SNR$) and correlation ($\rho$) considered in our experiments on regression with Toeplitz covariance matrix}
\label{tab:reg.mic.regimes}
\begin{tabular}{lc|c}
\toprule
 &Low correlation &High correlation \\
\midrule
Low noise &  {$\begin{aligned} &\rho = 0.2 \\ &{SNR}=6  \\ & p =20,000 \\ &k=100\end{aligned} $}& {$\begin{aligned} &\rho = 0.7 \\ &{SNR}=6 \\ & p =20,000 \\ &k=100 \end{aligned} $}\\
\midrule
Medium noise & {$\begin{aligned} &\rho = 0.2 \\ &{SNR}=1 \\ & p =10,000 \\ &k=50 \end{aligned} $} & {$\begin{aligned} &\rho = 0.7 \\ &{SNR}=1 \\ & p =10,000 \\ &k=50 \end{aligned} $}\\
\midrule
High noise & {$\begin{aligned} &\rho = 0.2 \\ &{SNR}=0.05 \\ & p =2,000 \\ &k=10 \end{aligned} $} & {$\begin{aligned} &\rho = 0.7 \\ &{SNR}=0.05 \\ & p =2,000 \\ &k=10 \end{aligned} $}\\
\bottomrule
\end{tabular}
\end{table}
}

\subsubsection{Feature selection with a given support size} 
We first consider the case when the cardinality $k$ of the support to be returned is given and equal to the true sparsity $k_{true}$ for all methods, {\blue while all other hyper-parameters are cross-validated on a separate validation set.} In this case, accuracy and false discovery rate are complementary. Indeed, in this case 
\begin{align*}
 |\{j: w_{true,j} \neq 0\}| &= |\{j: w_{j} \neq 0\}| = k_{true}, 
\end{align*} which leads to $A = 1-FDR$ so that we may consider accuracy by itself. 

As shown on Figure \ref{fig:RegFixTF} (p. \pageref{fig:RegFixTF}), all methods converge in terms of accuracy. That is their ability to select correct features as measured by $A$ smoothly converges to $1$ with an increasing number of observations $n\rightarrow \infty$. Noise in the data has an equalizing effect on all methods, meaning that noise reduces the gap in performance. Indeed, in high-noise regimes, all methods are comparable. On the contrary, correlation is discriminating: High correlation strongly hinders the performance of Lasso/ENet, moderately those of SCAD and very slightly CIO, SS and MCP methods. Among all methods, $\ell_1$-regularization is the less accurate, selects fewer correct features than the four other methods and is sensitive to correlation between features. { SCAD provides modest improvement over ENet in terms of accuracy, in comparison with CIO, SS and MCP. Unsurprisingly, we observe a gap between the solutions returned by the cutting-plane method and its Boolean relaxation, gap which decreases as noise increases. All things considered, CIO and the MCP penalization are the best performing method in all six regimes, with a fine advantage for CIO.} 

\begin{figure*}
\centering
\begin{subfigure}[t]{.45\linewidth}
	\centering
	\includegraphics[width=\linewidth]{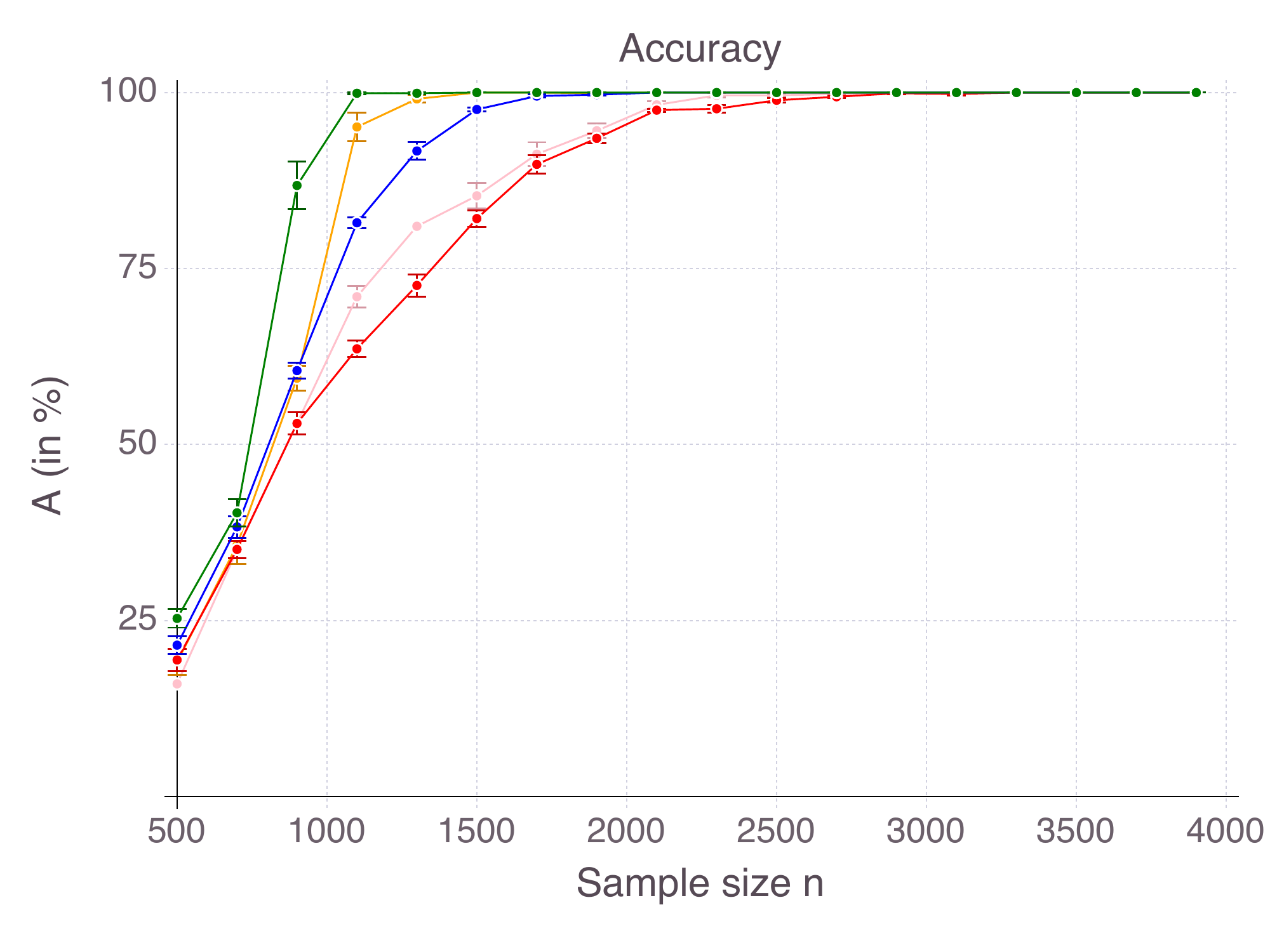}
	\caption{Low noise, low correlation}
\end{subfigure} %
~
\begin{subfigure}[t]{.45\linewidth}
	\centering
	\includegraphics[width=\linewidth]{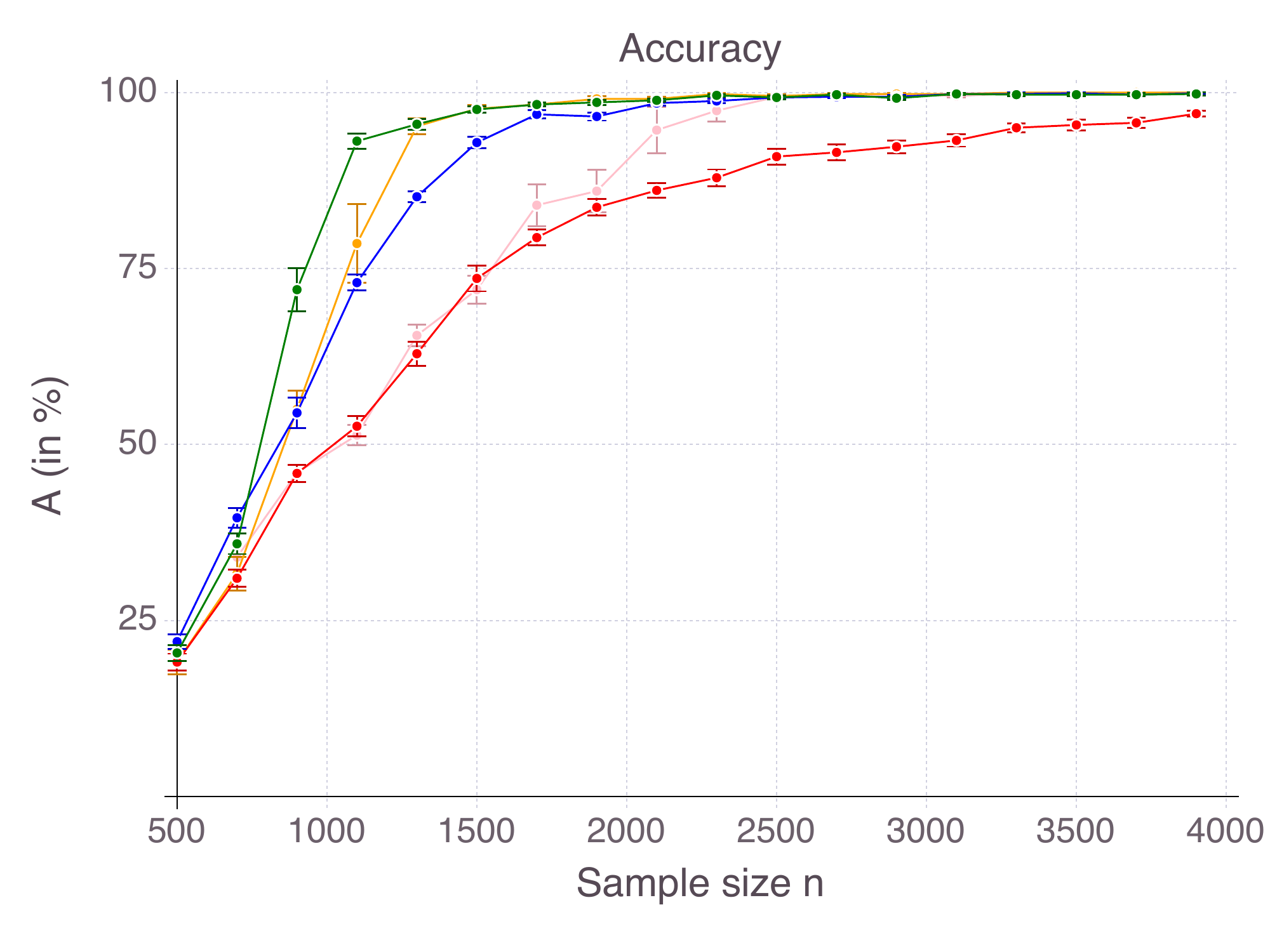}
	\caption{Low noise, high correlation}
\end{subfigure}

\begin{subfigure}[t]{.45\linewidth}
	\centering
	\includegraphics[width=\linewidth]{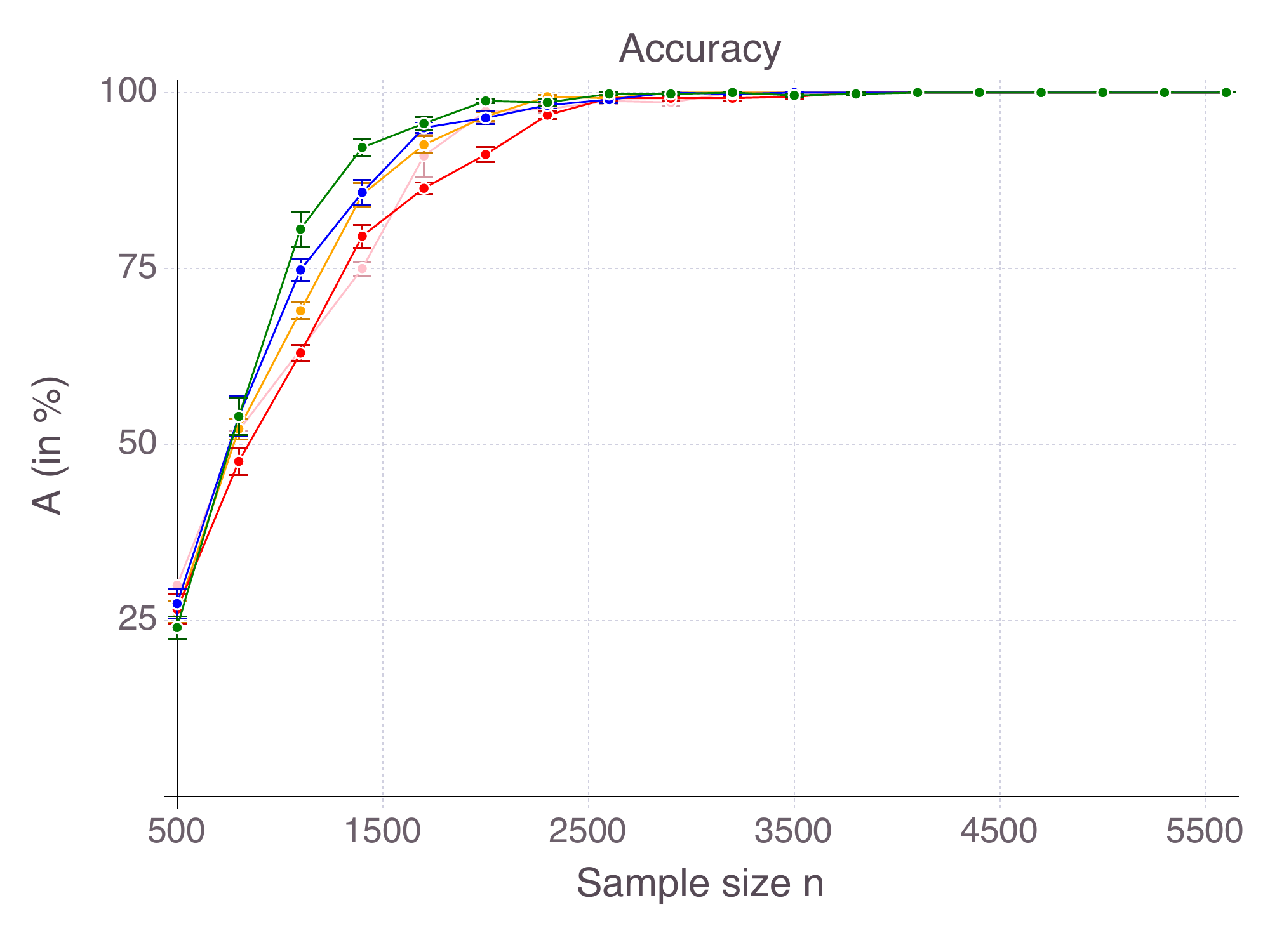}
	\caption{Medium noise, low correlation}
\end{subfigure} %
~
\begin{subfigure}[t]{.45\linewidth}
	\centering
	\includegraphics[width=\linewidth]{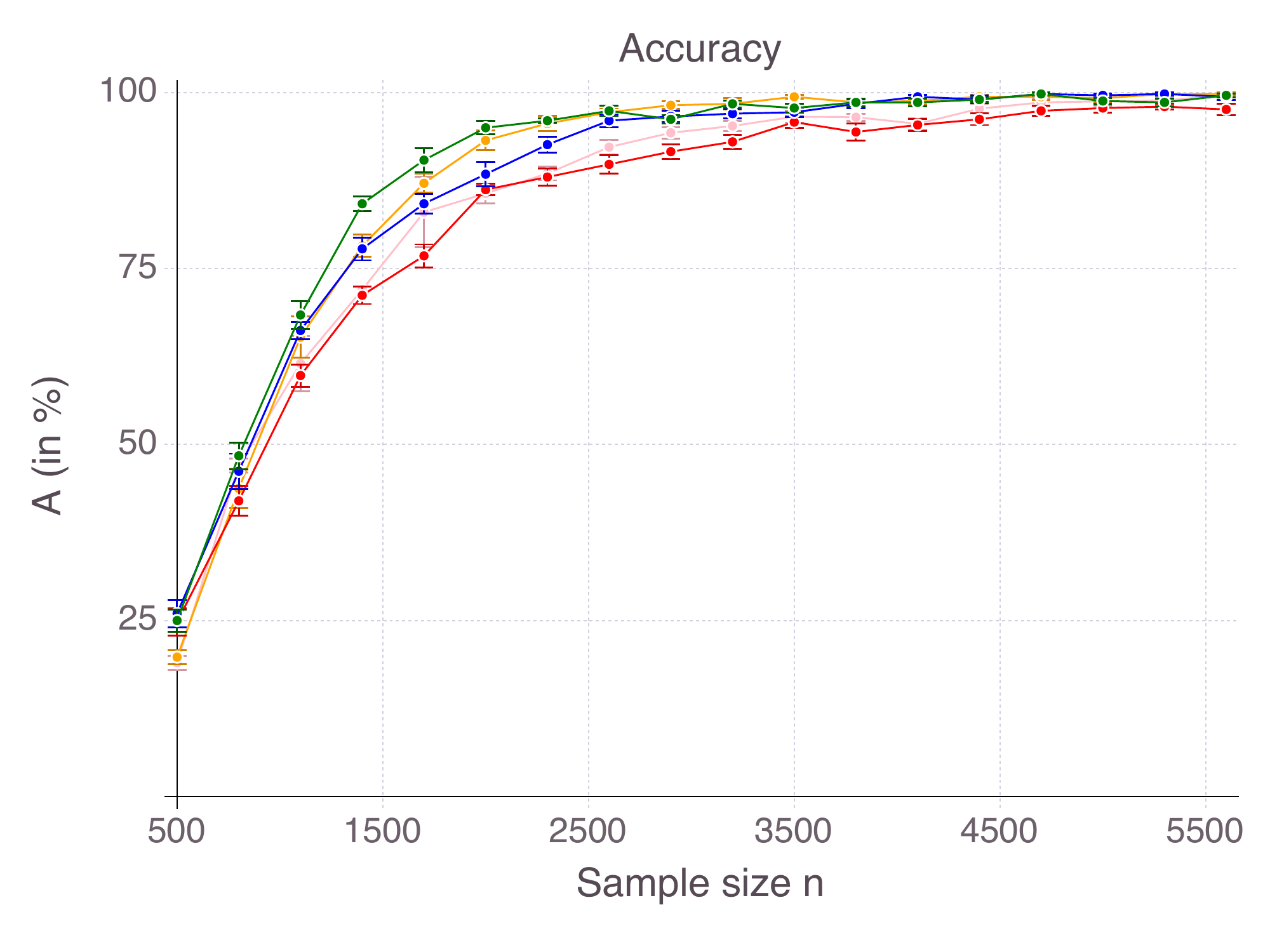}
	\caption{Medium noise, high correlation}
\end{subfigure}

\begin{subfigure}[t]{.45\linewidth}
	\centering
	\includegraphics[width=\linewidth]{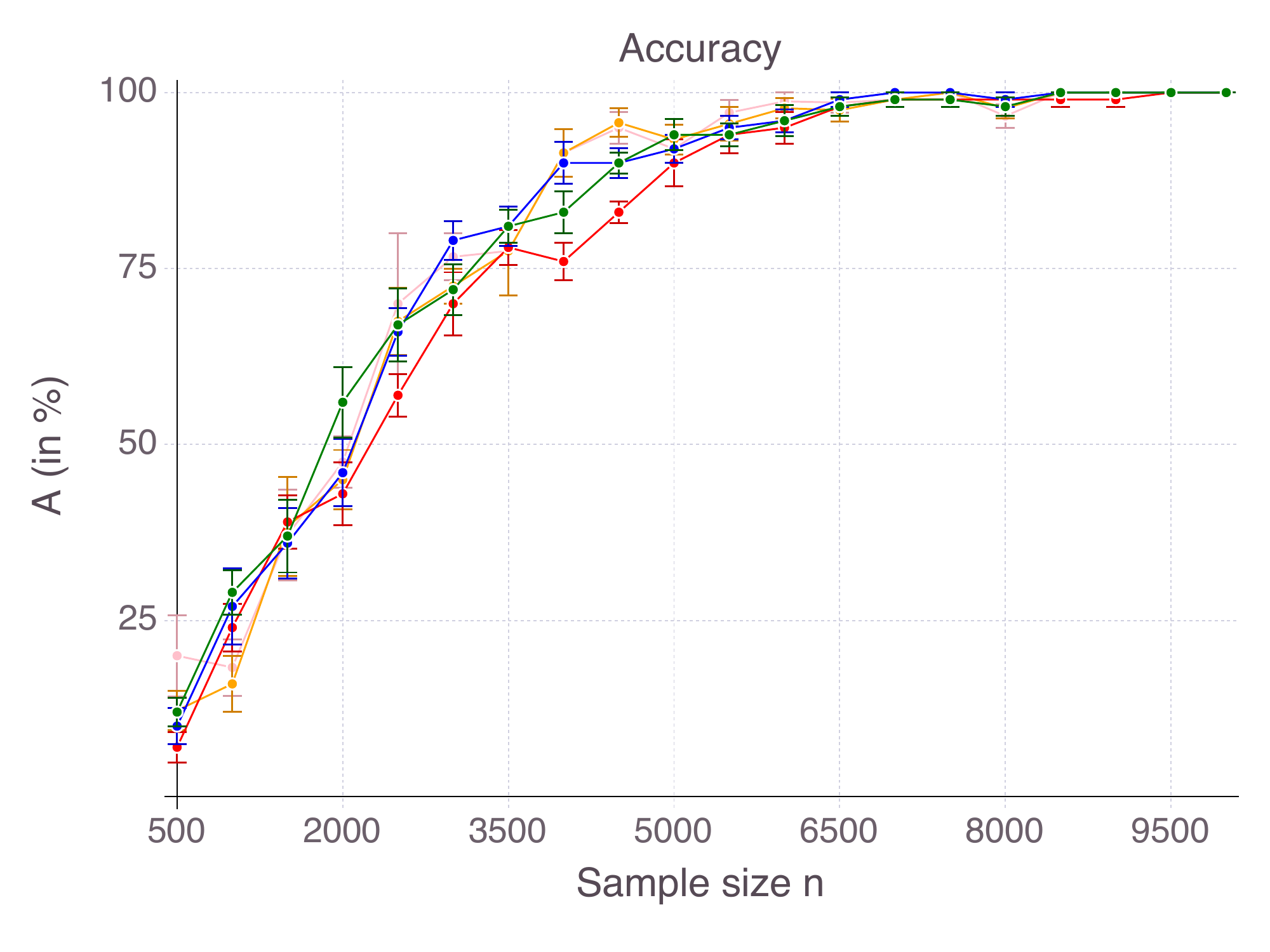}
	\caption{High noise, low correlation}
\end{subfigure} %
~
\begin{subfigure}[t]{.45\linewidth}
	\centering
	\includegraphics[width=\linewidth]{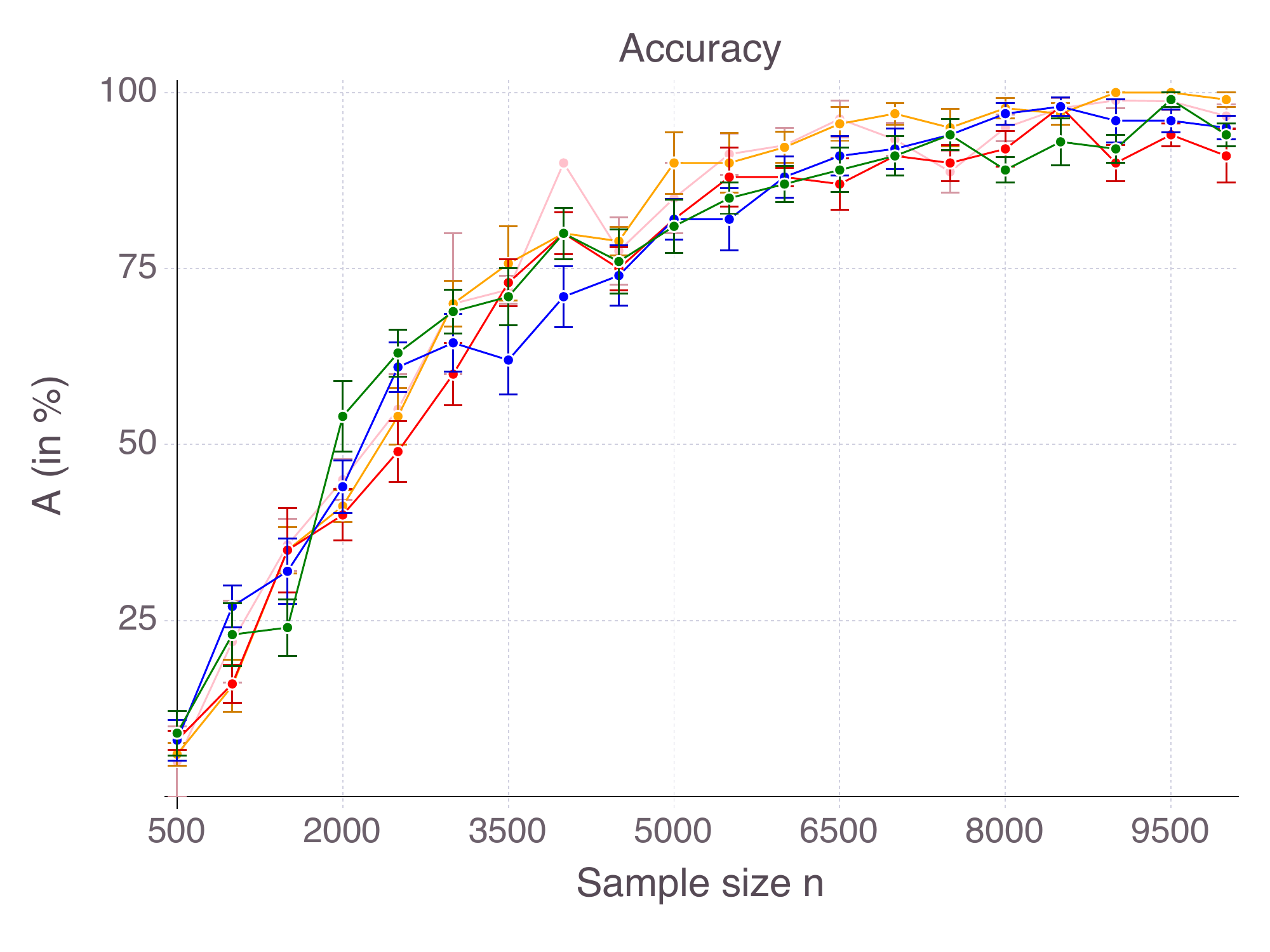}
	\caption{High noise, high correlation}
\end{subfigure}
\caption{Accuracy as $n$ increases, for the CIO (in green), SS (in blue with $T_{max}=200$), ENet (in red), MCP (in orange), SCAD (in pink) with OLS loss. We average results over $10$ data sets.}
\label{fig:RegFixTF}
\end{figure*}

Figure \ref{fig:RegFixTime} on page \pageref{fig:RegFixTime} reports relative computational time compared to \verb|glmnet| { in log scale}. It should be kept in mind that we restricted the cutting-plane algorithm to a $60$-second time limit and the sub-gradient algorithm to $T_{max}=200$ iterations. { All methods terminate in times of one to two orders of magnitude larger than \verb|glmnet| (seconds for the problem size at hand), contradicting the common belief that $\ell_1$-penalization is the \emph{only} tractable alternative to exact subset selection. Computational time for the discrete optimization algorithm CIO and sub-gradient algorithm SS highly depends on the regularization parameter $\gamma$. For low $\gamma$, which are suited in high noise regimes, the algorithm is extremely fast, while it can take as long as a minute in low noise regimes. This phenomenon explains the relative comparison of SS with \verb|glmnet| in Figure \ref{fig:RegFixTime}. For this is an important practical aspect, we provide detailed experiments regarding computational time in Appendix \ref{sec:regression.supp.time}. As previously mentioned, stopping the algorithm SS after a consecutive number of non-improvements can drastically reduce computational time. Empirically, this strategy did not hinder the quality of the solution in regression settings, but was not as successful in classification setting, so we did not reported its performance.} As CIO has a fixed time limit independent of $n$, the relative gap in terms of computational time with \verb|glmnet| narrows as sample size increases.

\begin{figure*}
\centering
\begin{subfigure}[t]{.45\linewidth}
	\centering
	\includegraphics[width=\linewidth]{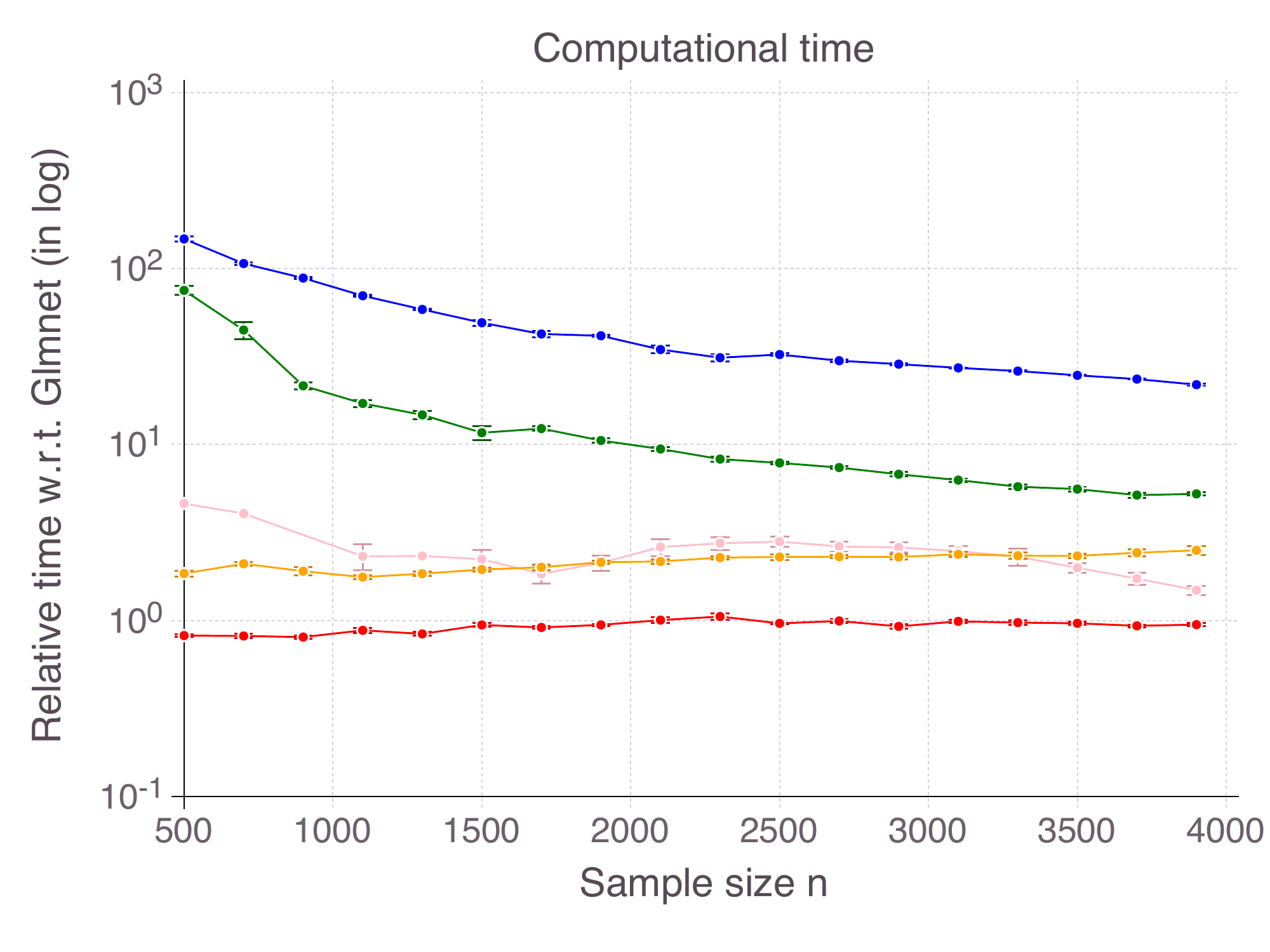}
	\caption{Low noise, low correlation}
\end{subfigure} %
~
\begin{subfigure}[t]{.45\linewidth}
	\centering
	\includegraphics[width=\linewidth]{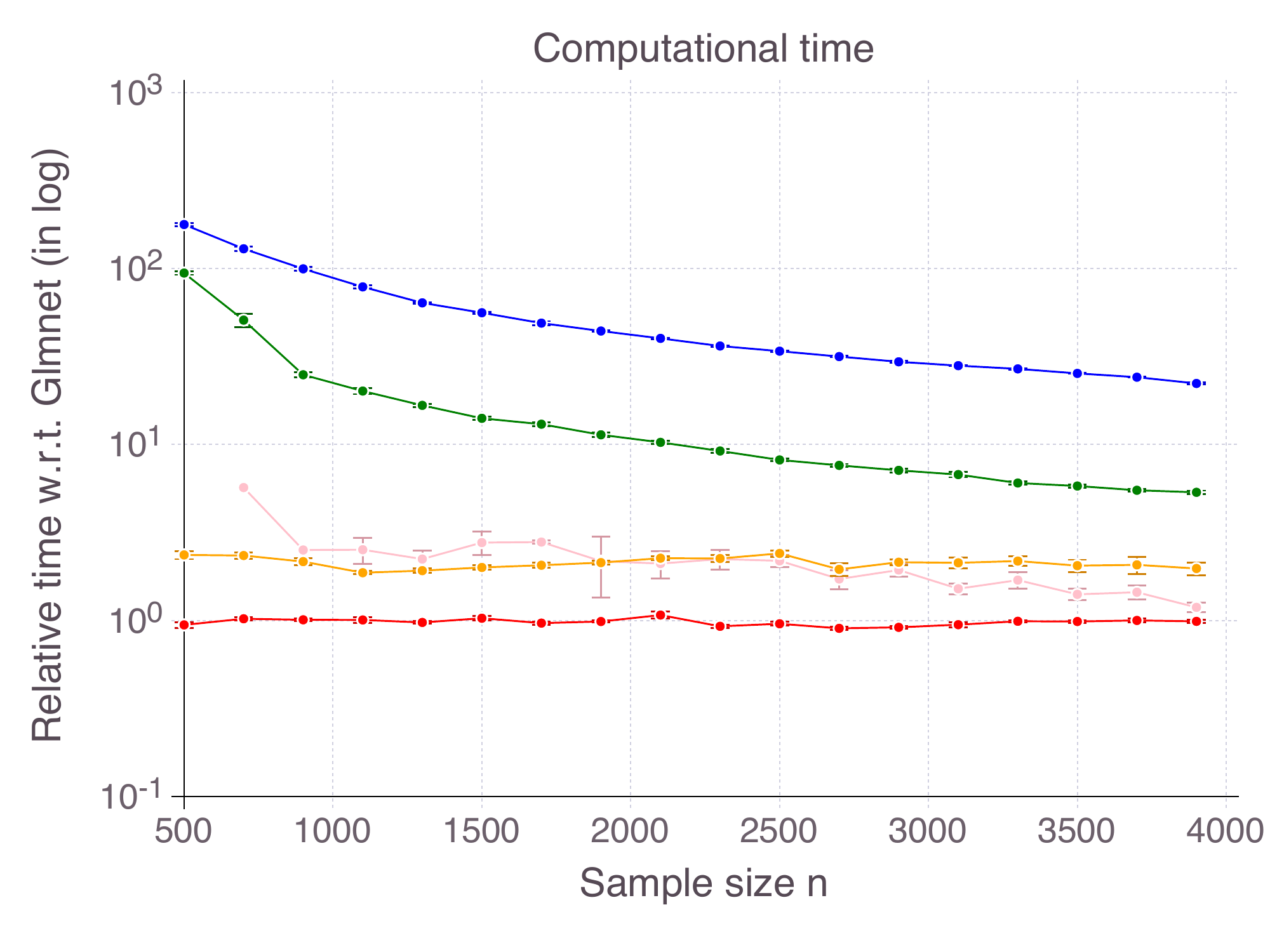}
	\caption{Low noise, high correlation}
\end{subfigure}

\begin{subfigure}[t]{.45\linewidth}
	\centering
	\includegraphics[width=\linewidth]{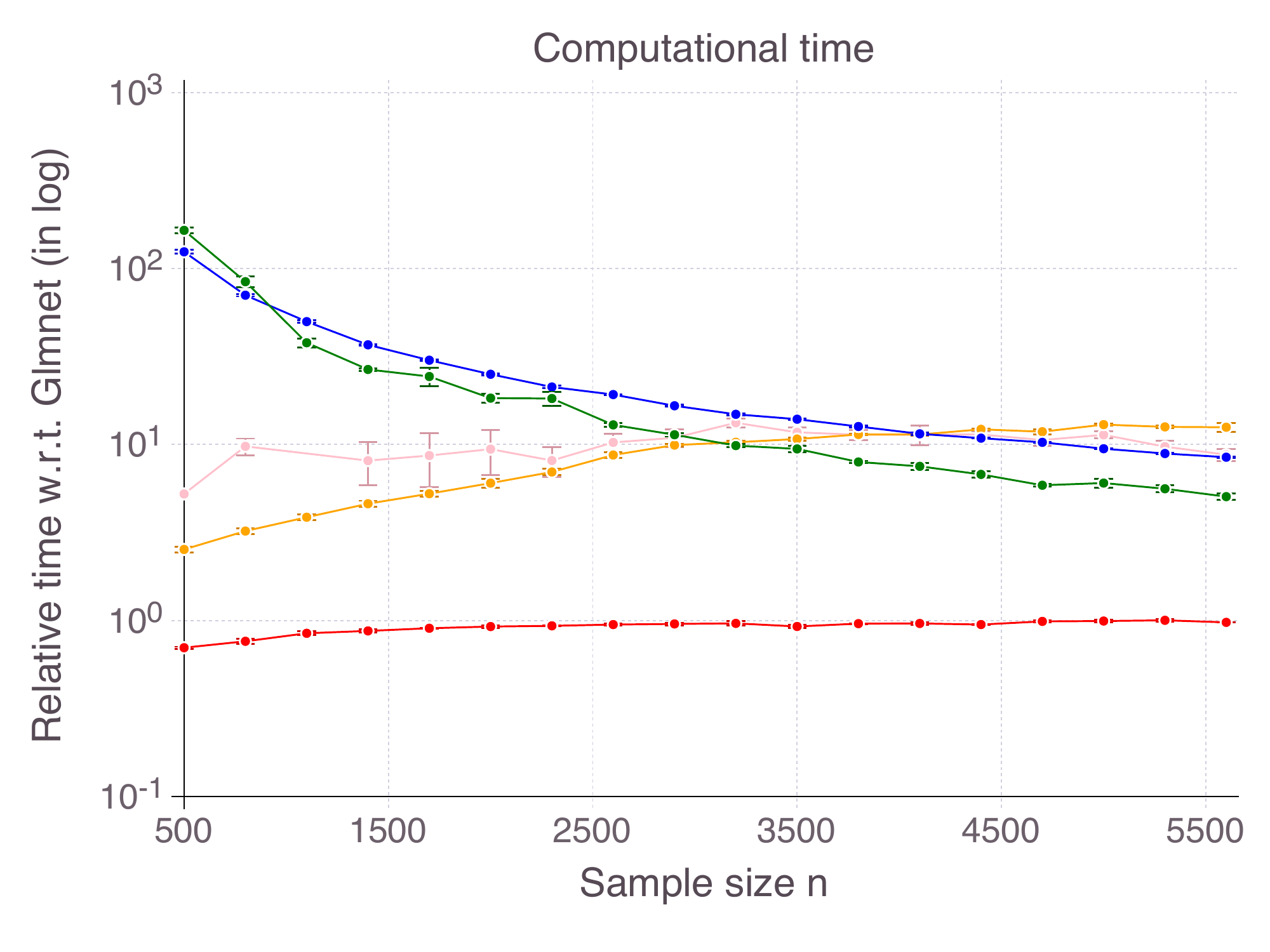}
	\caption{Medium noise, low correlation}
\end{subfigure} %
~
\begin{subfigure}[t]{.45\linewidth}
	\centering
	\includegraphics[width=\linewidth]{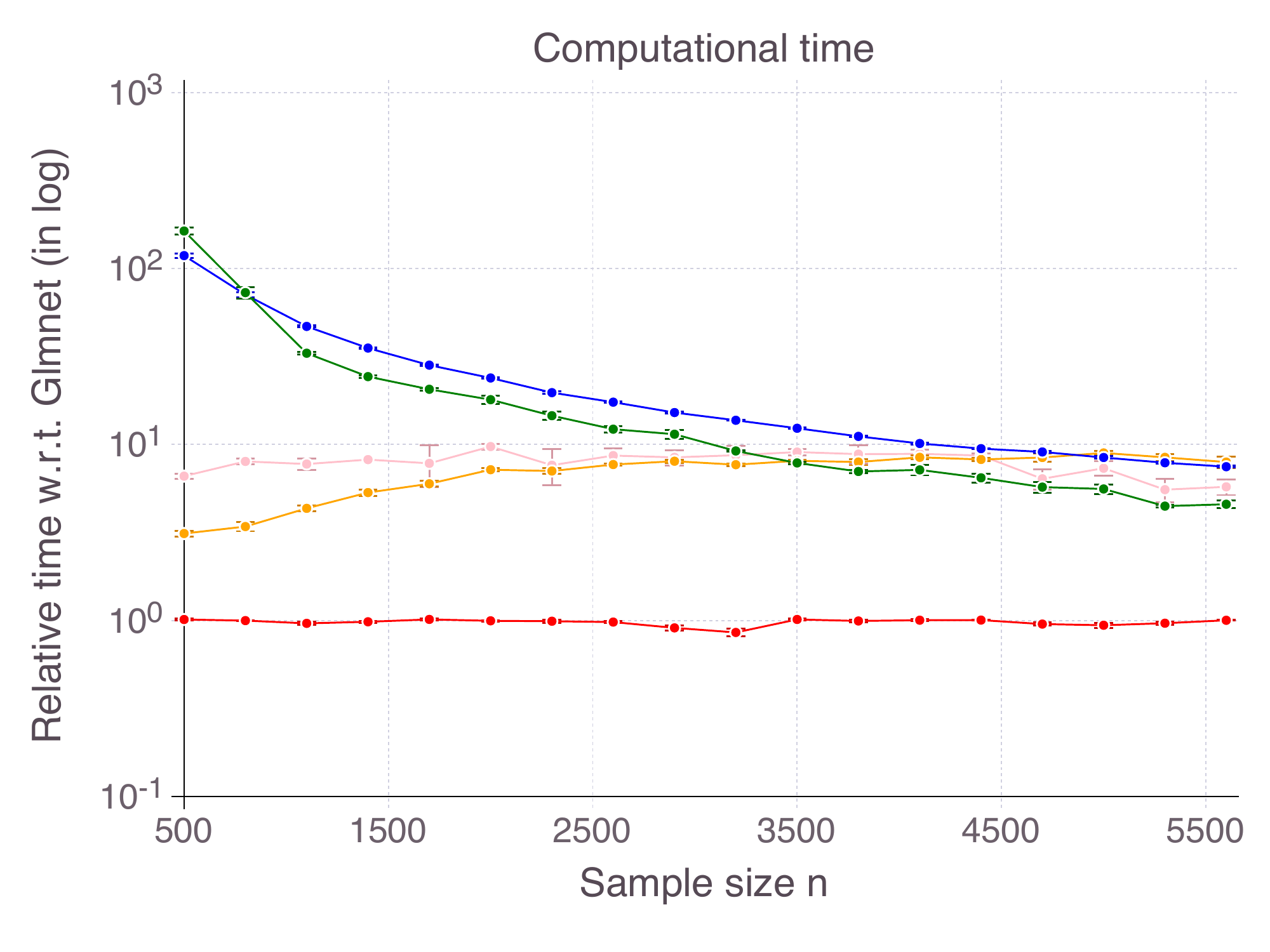}
	\caption{Medium noise, high correlation}
\end{subfigure}

\begin{subfigure}[t]{.45\linewidth}
	\centering
	\includegraphics[width=\linewidth]{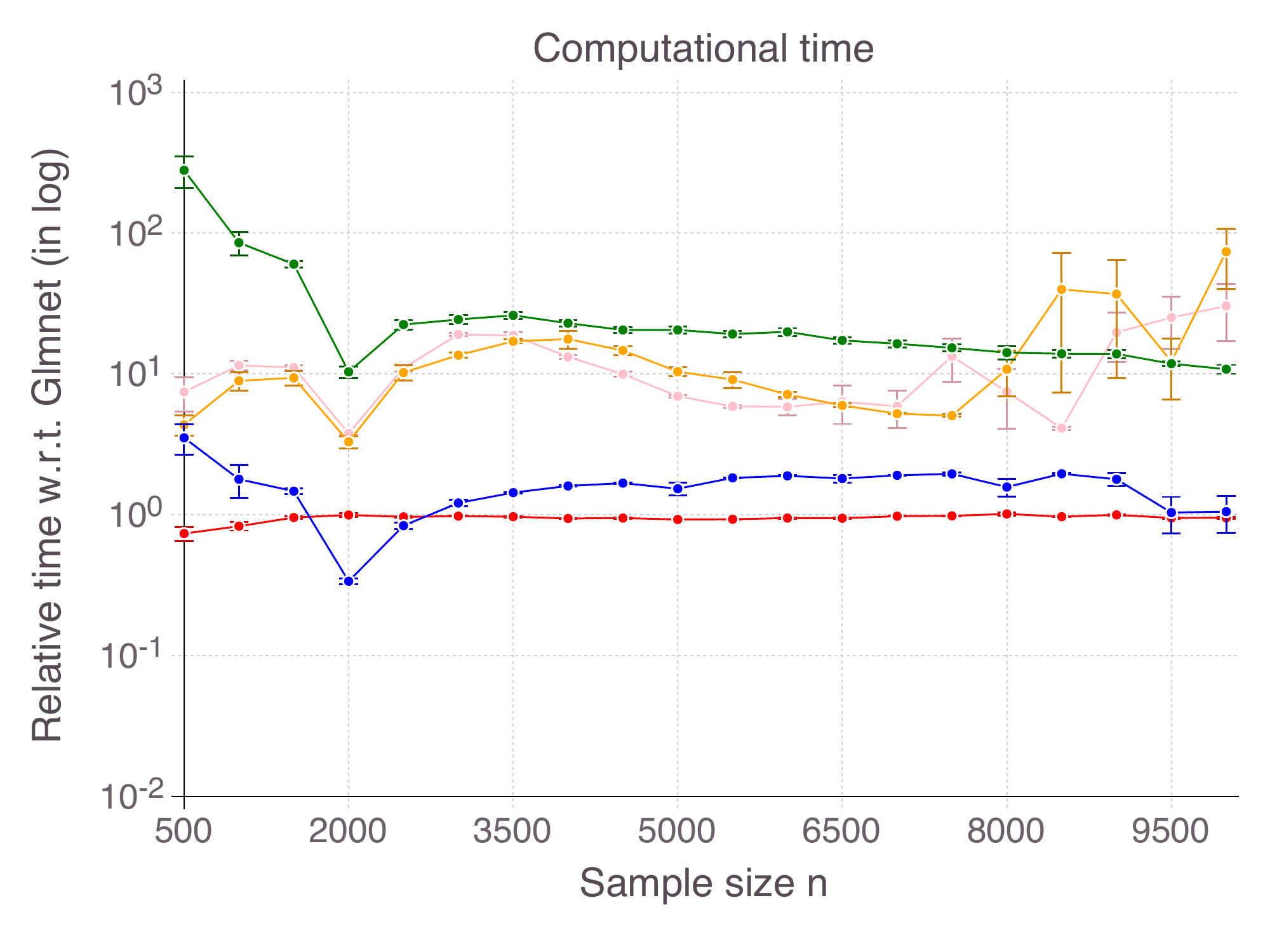}
	\caption{High noise, low correlation}
\end{subfigure} %
~
\begin{subfigure}[t]{.45\linewidth}
	\centering
	\includegraphics[width=\linewidth]{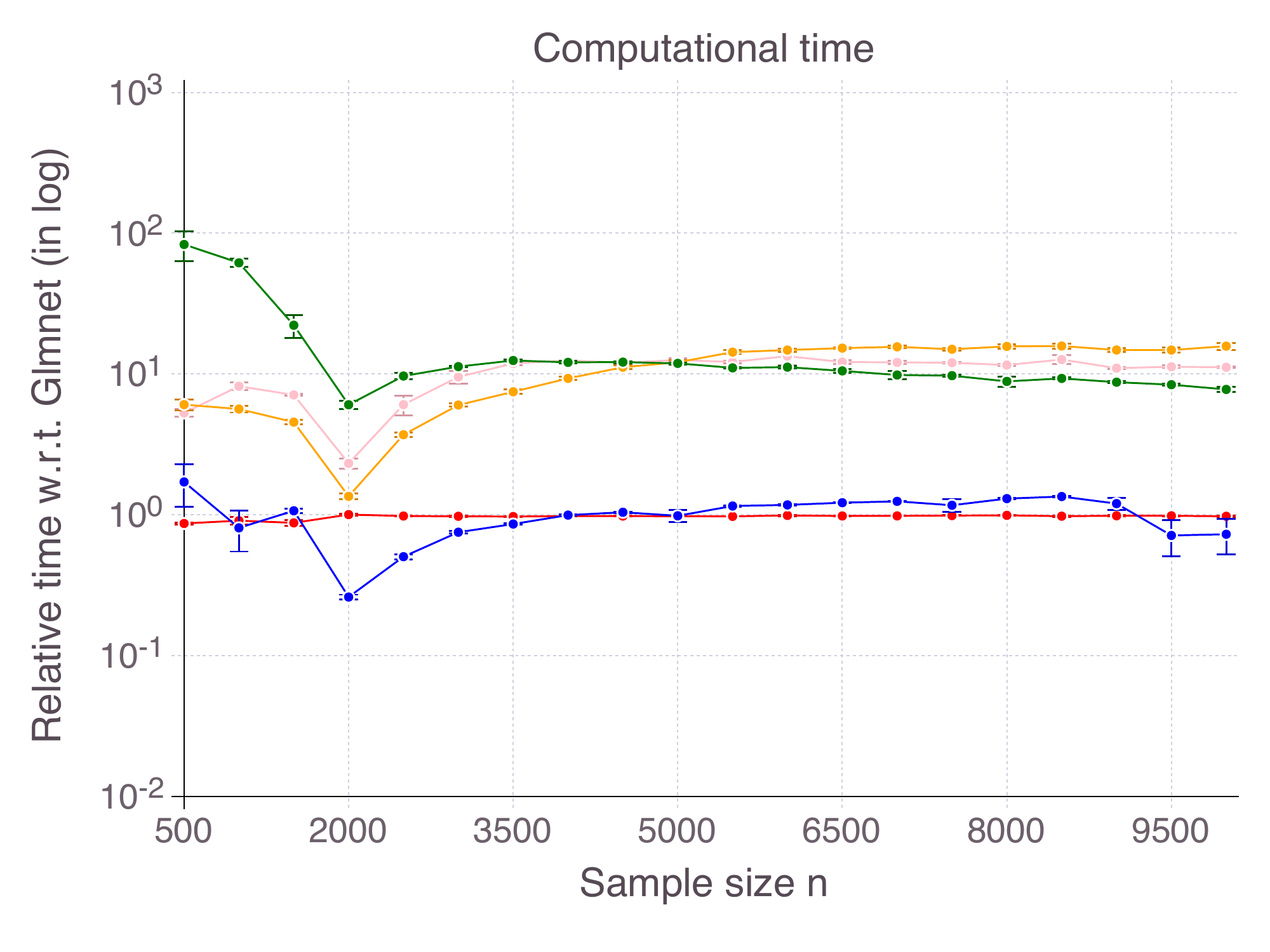}
	\caption{High noise, high correlation}
\end{subfigure}
\caption{Computational time relative to Lasso with glmnet as $n$ increases, for CIO (in green), SS (in blue with $T_{max}=200$), ENet (in red), MCP (in orange), SCAD (in pink) with OLS loss. We average results over $10$ data sets.}
\label{fig:RegFixTime}
\end{figure*}

\begin{figure*}
\centering
\begin{subfigure}[t]{.45\linewidth}
	\centering
	\includegraphics[width=\linewidth]{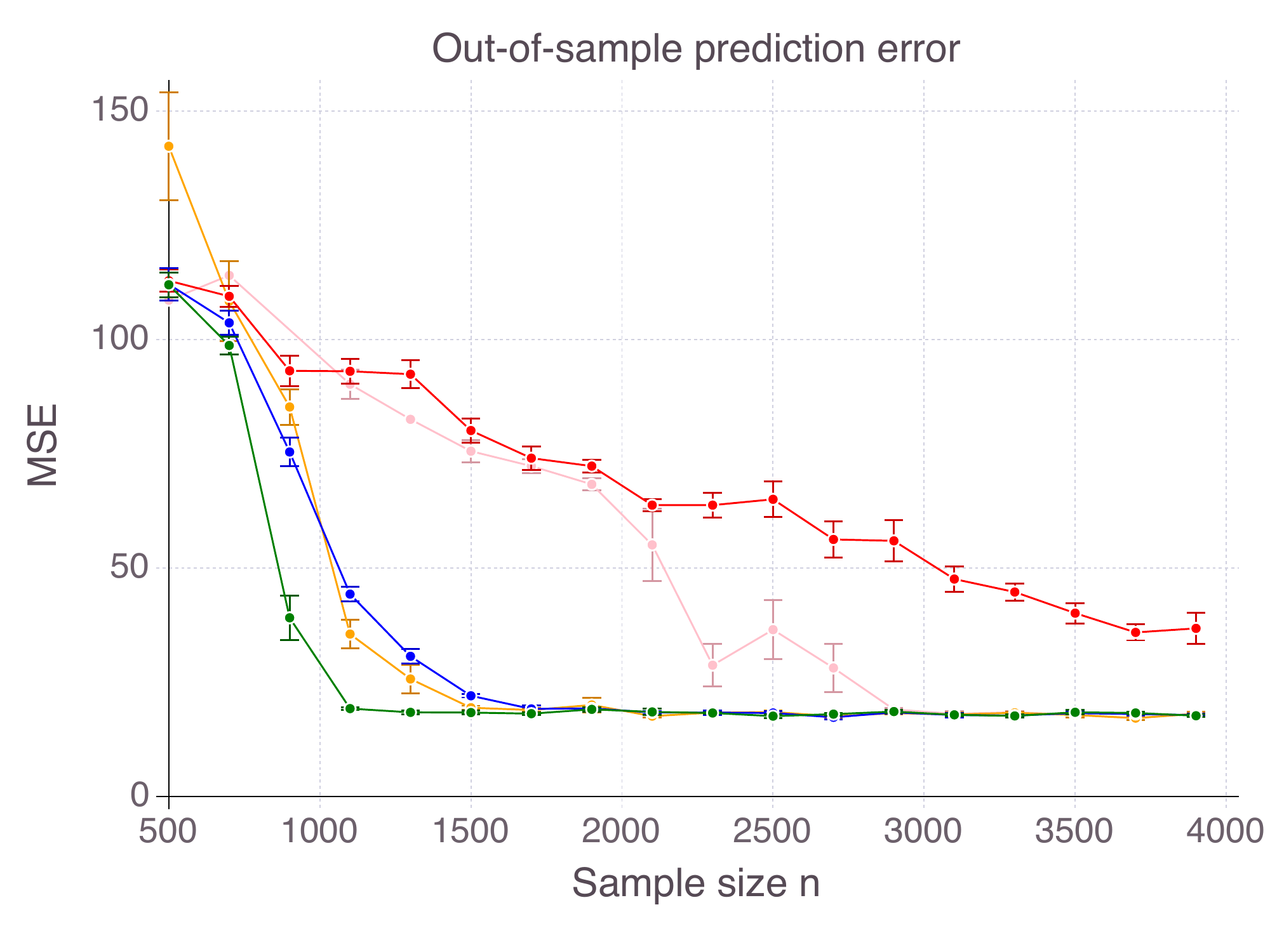}
	\caption{Low noise, low correlation}
\end{subfigure} %
~
\begin{subfigure}[t]{.45\linewidth}
	\centering
	\includegraphics[width=\linewidth]{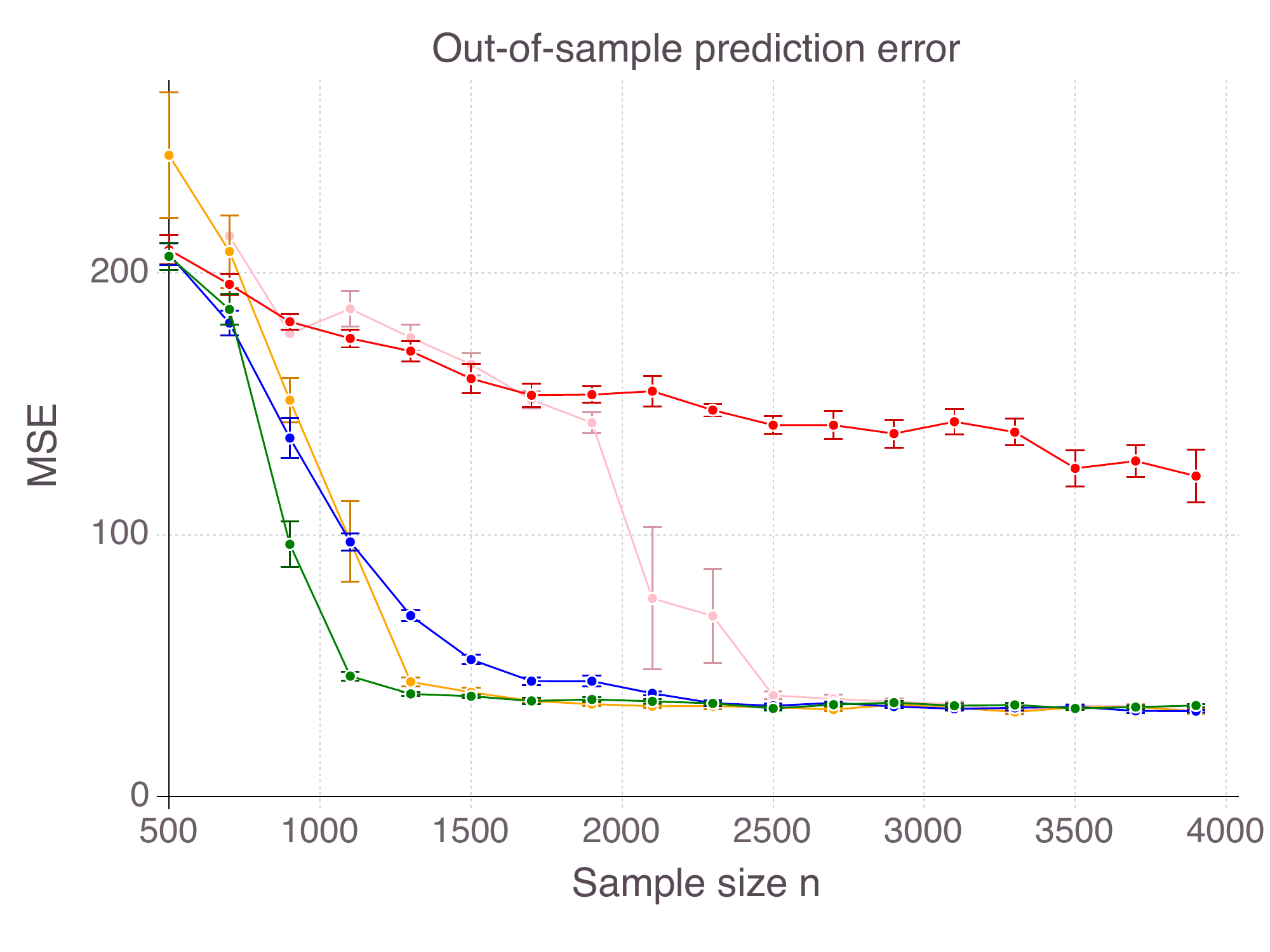}
	\caption{Low noise, high correlation}
\end{subfigure}

\begin{subfigure}[t]{.45\linewidth}
	\centering
	\includegraphics[width=\linewidth]{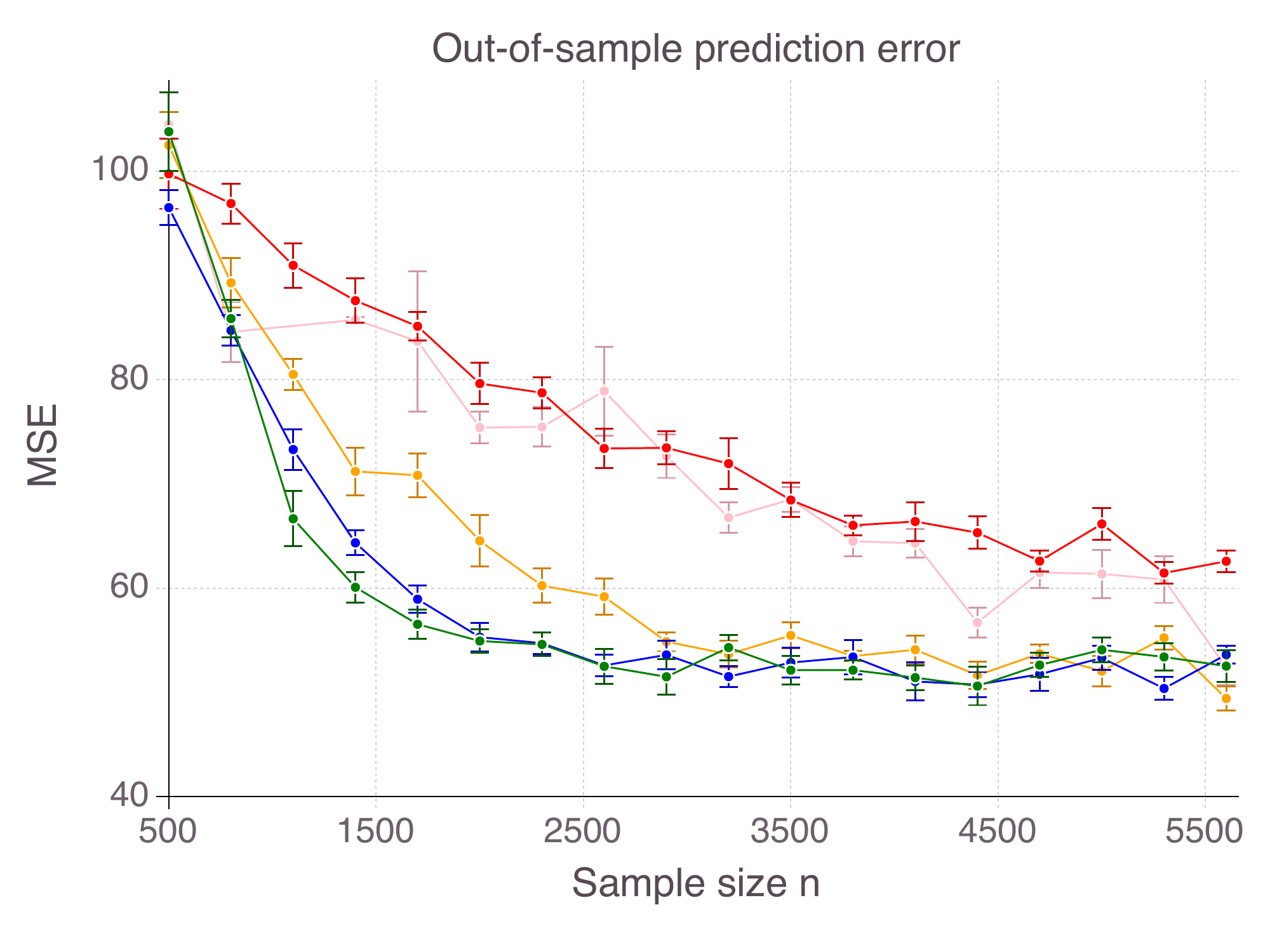}
	\caption{Medium noise, low correlation}
\end{subfigure} %
~
\begin{subfigure}[t]{.45\linewidth}
	\centering
	\includegraphics[width=\linewidth]{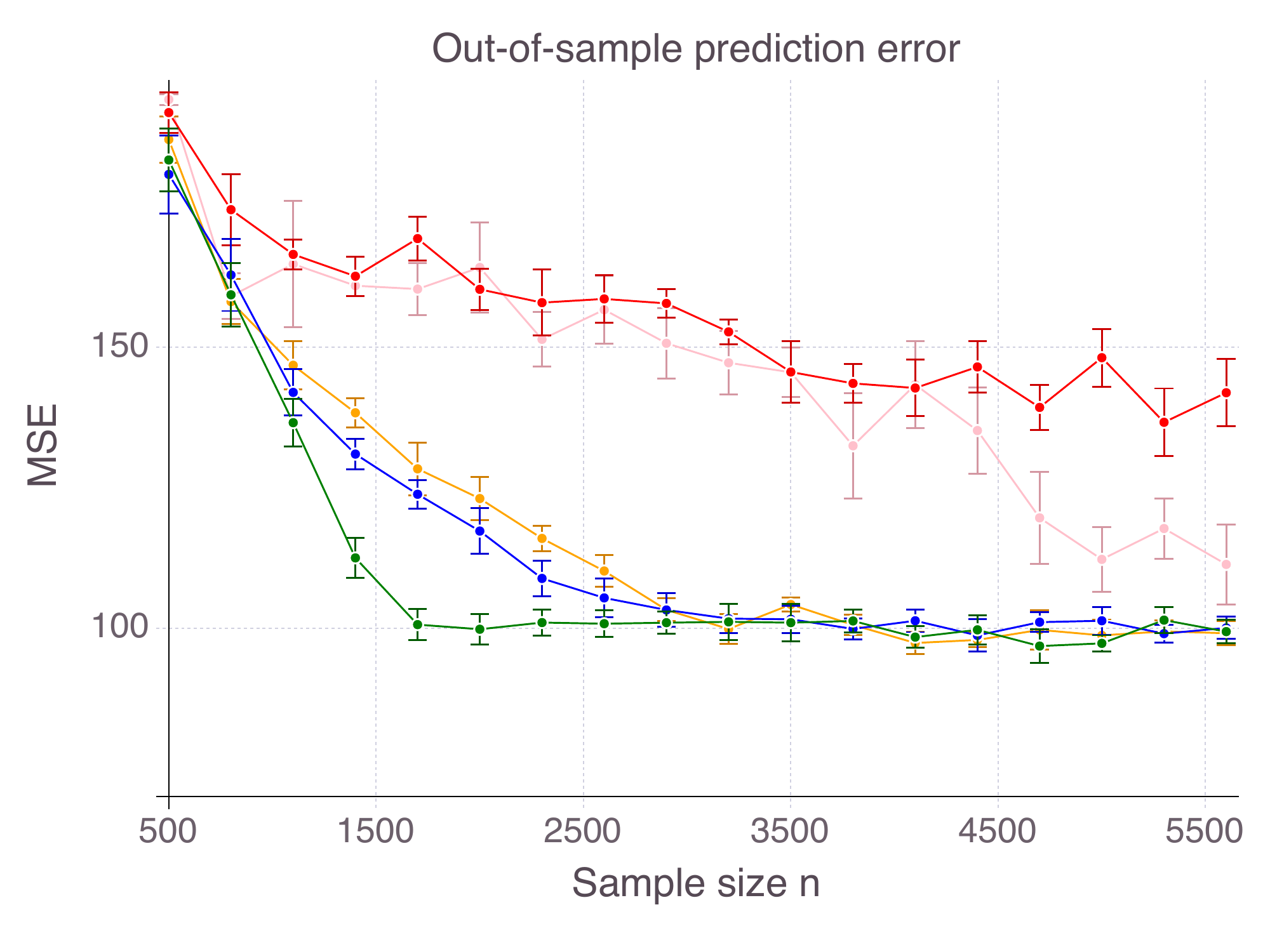}
	\caption{Medium noise, high correlation}
\end{subfigure}

\begin{subfigure}[t]{.45\linewidth}
	\centering
	\includegraphics[width=\linewidth]{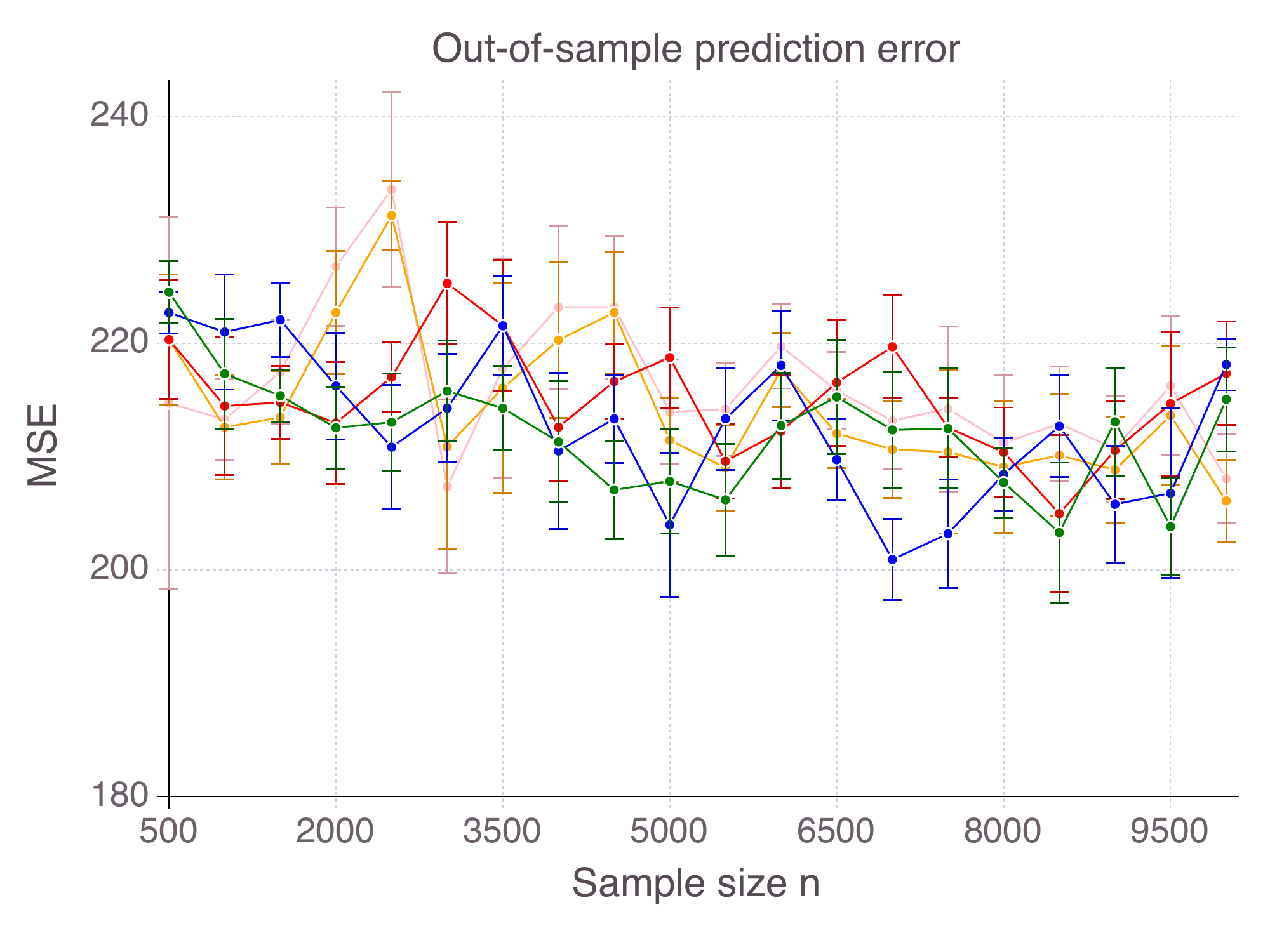}
	\caption{High noise, low correlation}
\end{subfigure} %
~
\begin{subfigure}[t]{.45\linewidth}
	\centering
	\includegraphics[width=\linewidth]{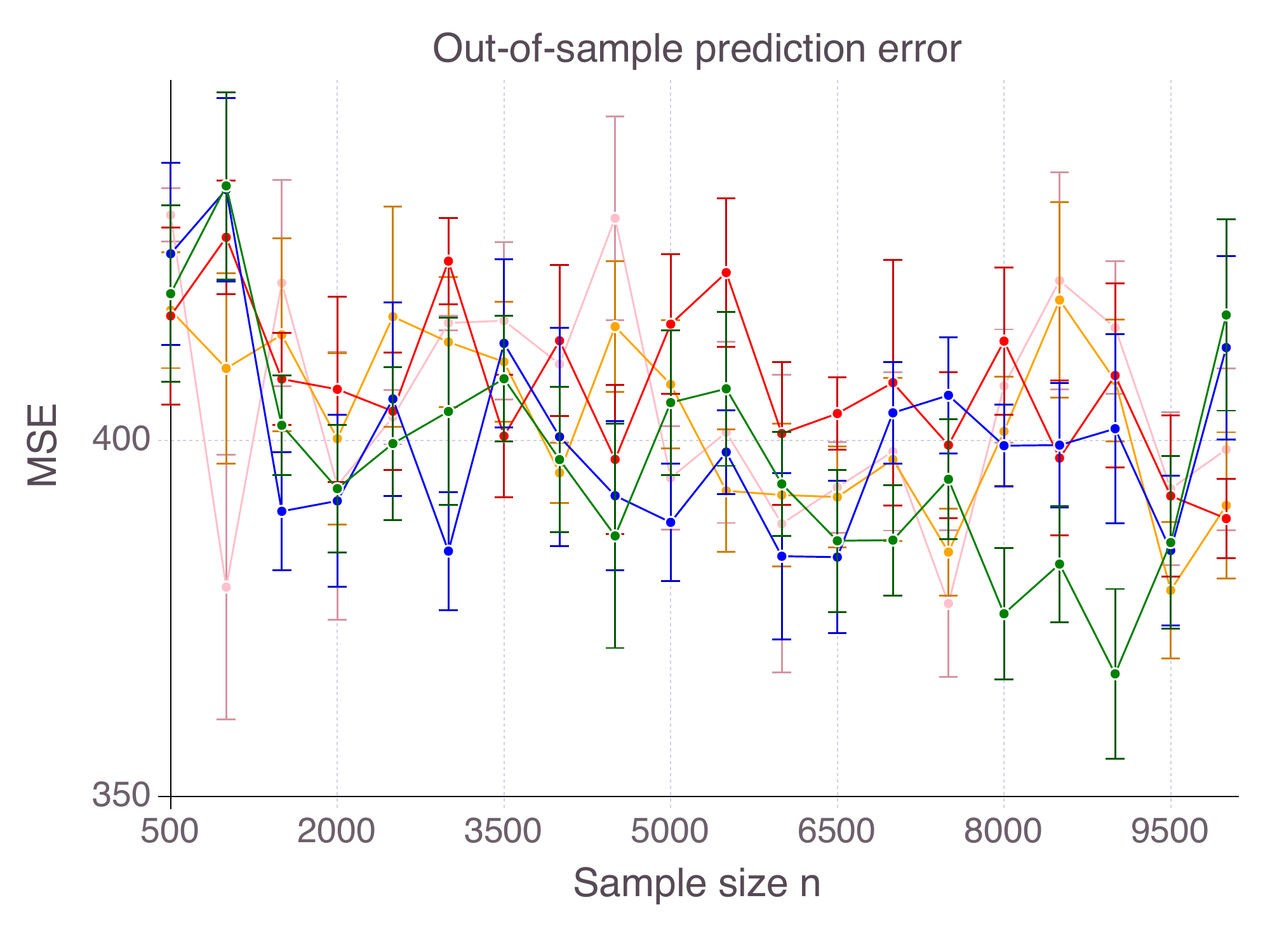}
	\caption{High noise, high correlation}
\end{subfigure}
\caption{Out-of-sample mean square error as $n$ increases, for the CIO (in green), SS (in blue with $T_{max}=200$), ENet (in red), MCP (in orange), SCAD (in pink) with OLS loss. We average results over $10$ data sets.}
\label{fig:RegFixMSE}
\end{figure*}

Finally, though a purely theoretic metric, accuracy has some intuitive and practical implications in terms of out-of-sample prediction. To support our claim, Figure \ref{fig:RegFixMSE} (p. \pageref{fig:RegFixMSE}) represents the out-of-sample $MSE$ for all five methods, as $n$ increases, for the { six} noise/correlation settings of interest. There is a clear connection between performance in terms of accuracy and in terms of predictive power, with CIO performing the best. Still, good predictive power does not necessarily imply that the features selected are mostly correct.{  SCAD, for instance, seems to provide a larger improvement over ENet in terms of predictive power than in accuracy. Similarly, SS dominates MCP in terms of out-of-sample MSE, while this is not the case in terms of accuracy.} 

\subsubsection{Feature selection with cross-validated support size} 
We now compare all methods when $k_{true}$ is no longer given and needs to be cross-validated from the data itself. 

For each value of $n$, each method fits a model on a training set for various levels of sparsity $k$, either explicitly or by adjusting the penalization parameter. { For each sparsity level $k$, the resulting classifier incorporates some true and false features. Figure \ref{fig:RegCVROC} (p. \pageref{fig:RegCVROC}) represents the number of true features against the number of false features for all five methods, for a range of sparsity levels $k$, all other hyper-parameters being tuned so as to minimize $MSE$ on a validation set. To obtain a fair comparison, we used the same range of sparsity levels for all methods. Some methods only indirectly control the sparsity $k$ through a regularization parameter $\lambda$ and do not guarantee to return \emph{exactly} $k$ features. In these cases, we calibrated  $\lambda$ as precisely as possible and used linear interpolation when we were unable to get the exact value of $k$ we were interested in. From Figure \ref{fig:RegCVROC}, we observe that in low correlation settings, CIO and MCP strictly dominate ENet, SCAD and SS. There is no clear winner between CIO and MCP. When noise is low, CIO tends to make less false discoveries, while the latter is generally more accurate, but the difference between all methods diminishes as noise increases. In high correlation settings, no method clearly dominates. CIO, SS and MCP are better for small support size $k$, while ENet and SCAD dominate for larger supports. In high noise and high correlation regimes though, Enet and SCAD seem to clearly dominate their competitors.} 
{ In practice however, one does not have access to "true" features and cannot decide on the value of $k$ based on such ROC curves.} As often, we select the value $k^\star$ which minimizes out-of-sample error on a validation set. To this end, Figure \ref{fig:RegCVcoupe} (p. \pageref{fig:RegCVcoupe}) visually represents validation $MSE$ as a function of $k$ for all five methods.The vertical black line corresponds to $k = k_{true}$. For each method, $k^\star$ is identified as the minimum of the out-of-sample $MSE$ curve. From Figure \ref{fig:RegCVcoupe}, we can expect the Lasso/ENet and SCAD formulations to select many irrelevant features, while CIO, SS and MCP are relatively close to the true sparsity pattern.
\begin{figure*}
\centering
\begin{subfigure}[t]{.45\linewidth}
	\centering
	\includegraphics[width=\linewidth]{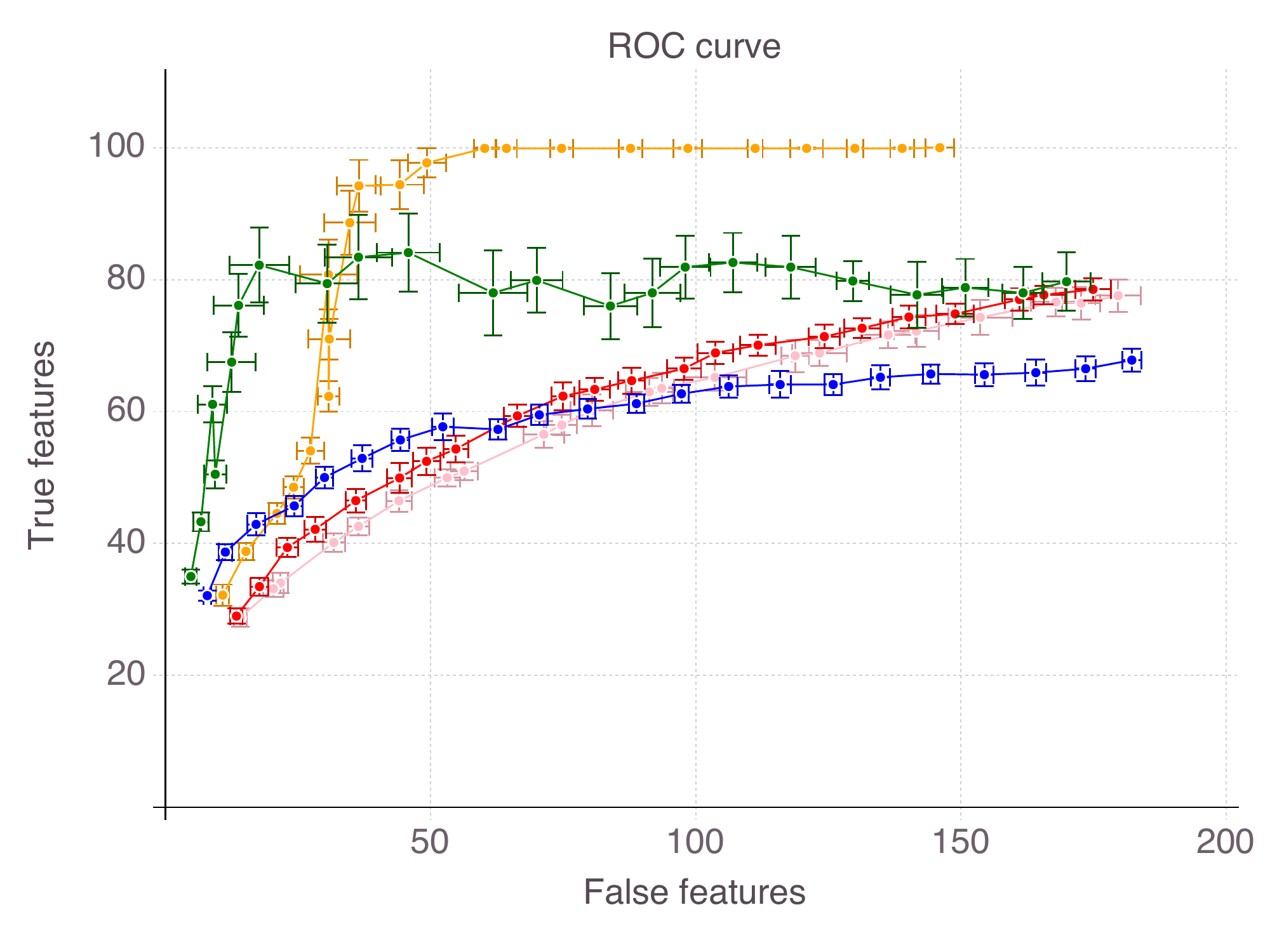}
	\caption{Low noise, low $\rho$, $n=900$}
\end{subfigure} %
~
\begin{subfigure}[t]{.45\linewidth}
	\centering
	\includegraphics[width=\linewidth]{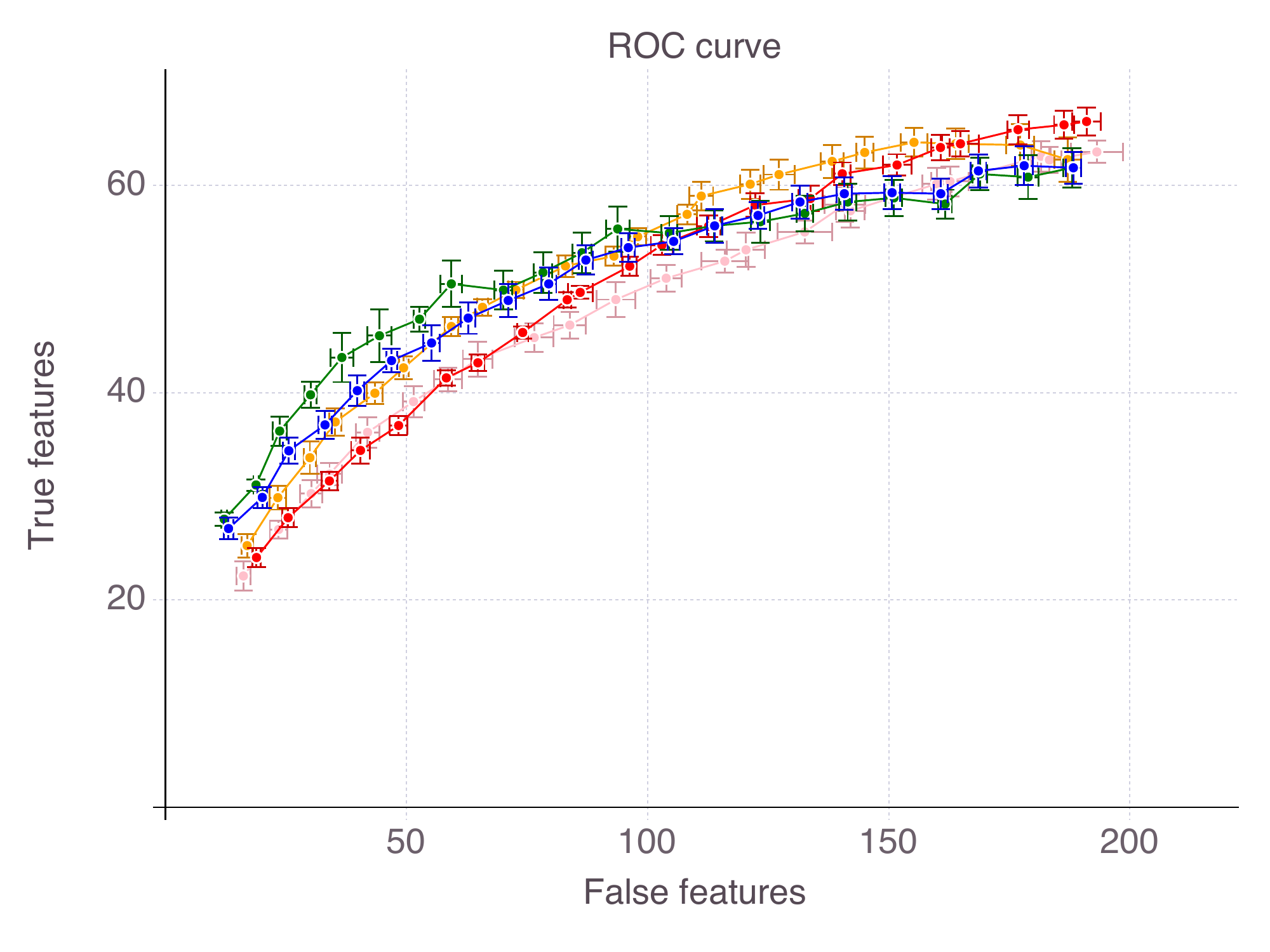}
	\caption{Low noise, high $\rho$, $n=900$}
\end{subfigure}

\begin{subfigure}[t]{.45\linewidth}
	\centering
	\includegraphics[width=\linewidth]{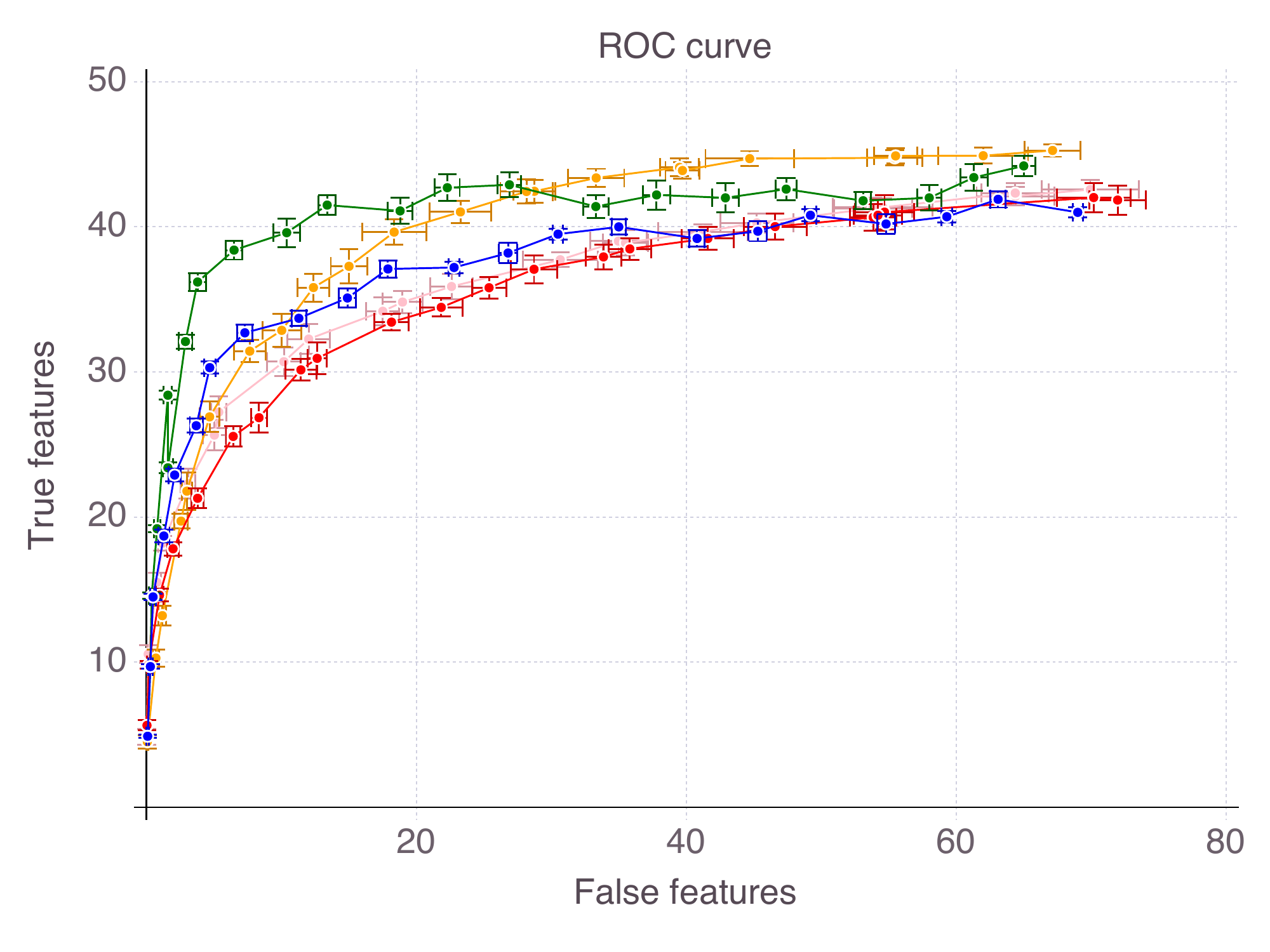}
	\caption{Medium noise, low $\rho$, $n=1,100$}
\end{subfigure} %
~
\begin{subfigure}[t]{.45\linewidth}
	\centering
	\includegraphics[width=\linewidth]{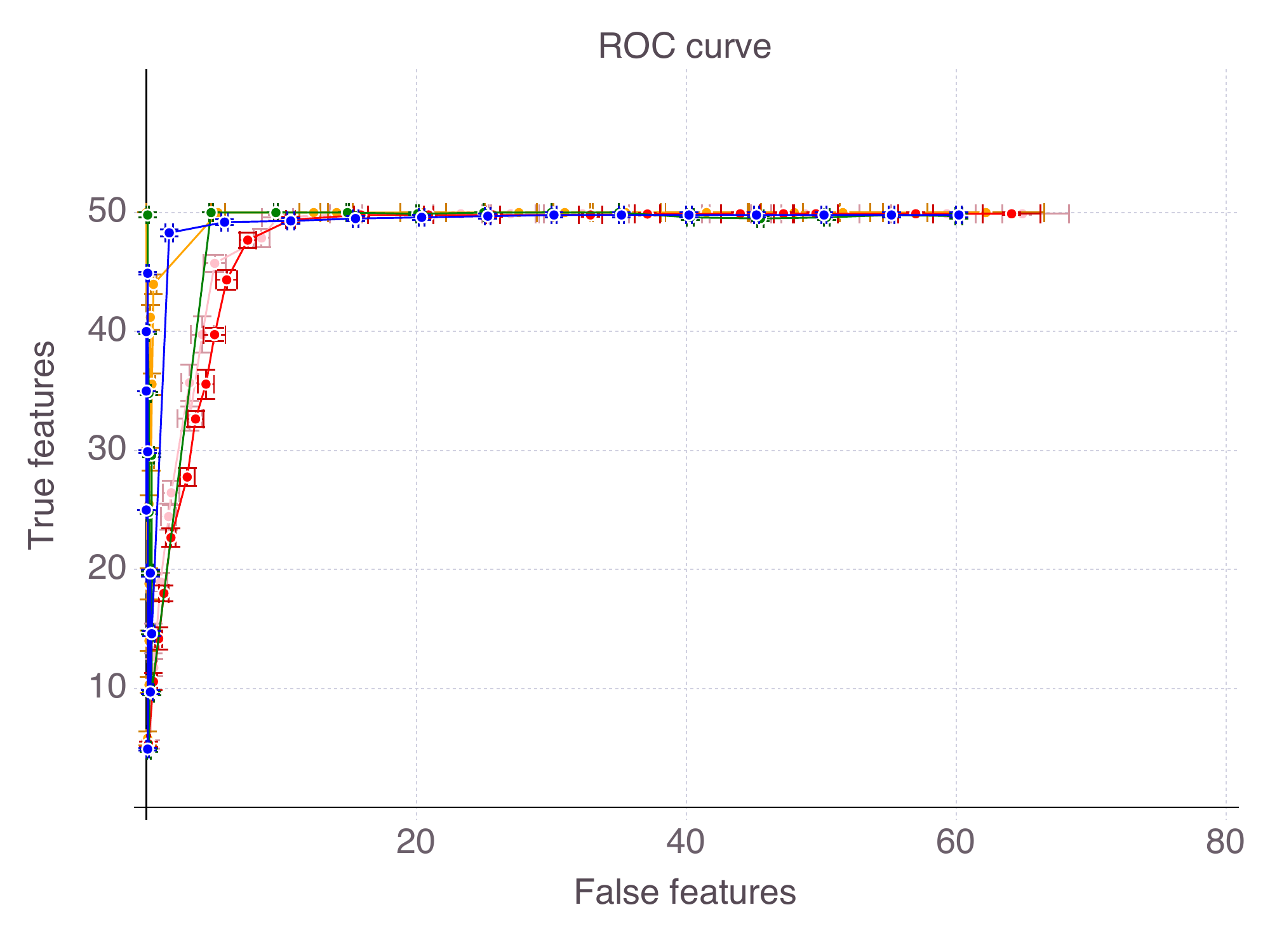}
	\caption{Medium noise, high $\rho$, $n=1,100$}
\end{subfigure}

\begin{subfigure}[t]{.45\linewidth}
	\centering
	\includegraphics[width=\linewidth]{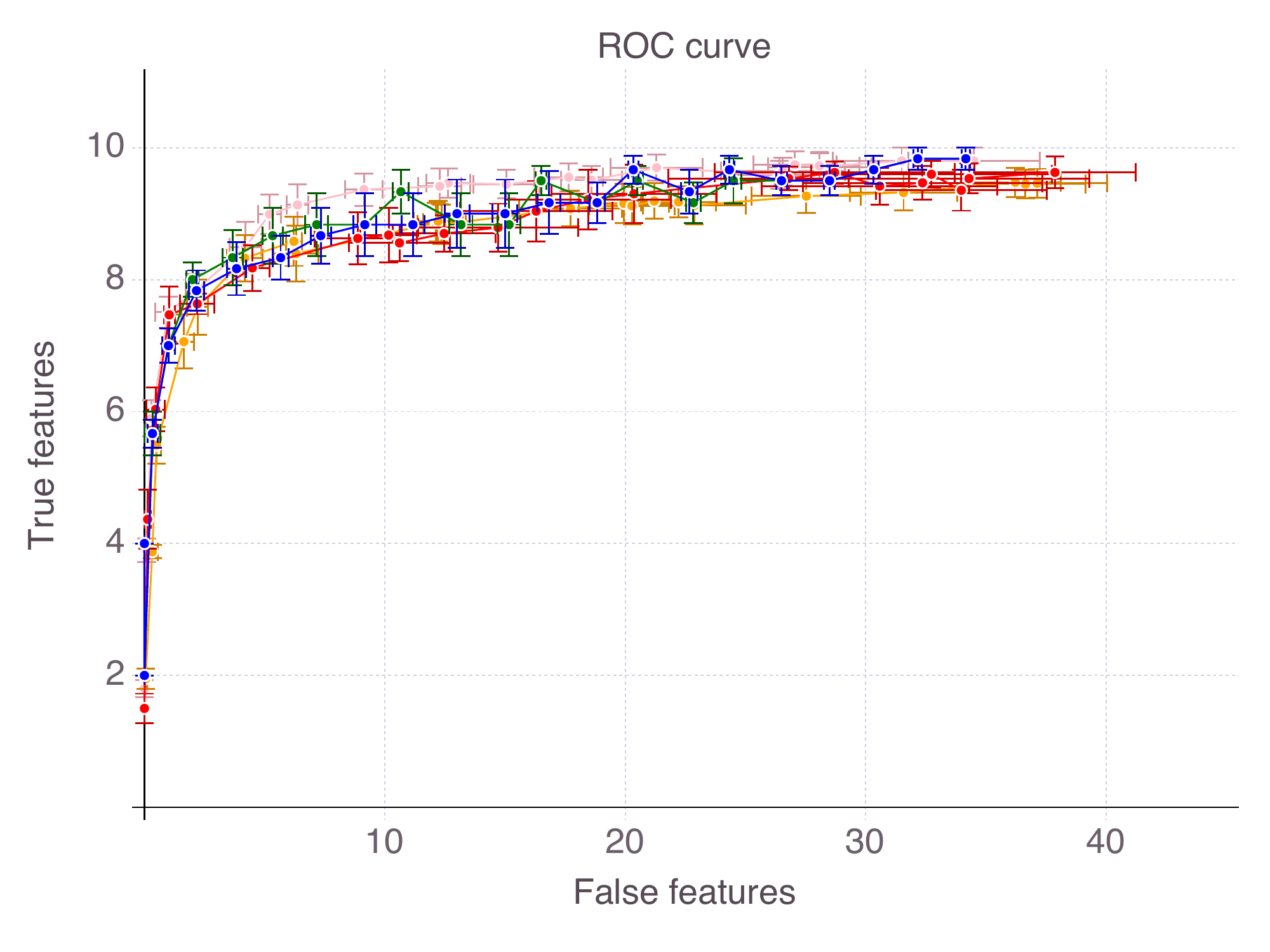}
	\caption{High noise, low $\rho$, $n=3,500$}
\end{subfigure} %
~
\begin{subfigure}[t]{.45\linewidth}
	\centering
	\includegraphics[width=\linewidth]{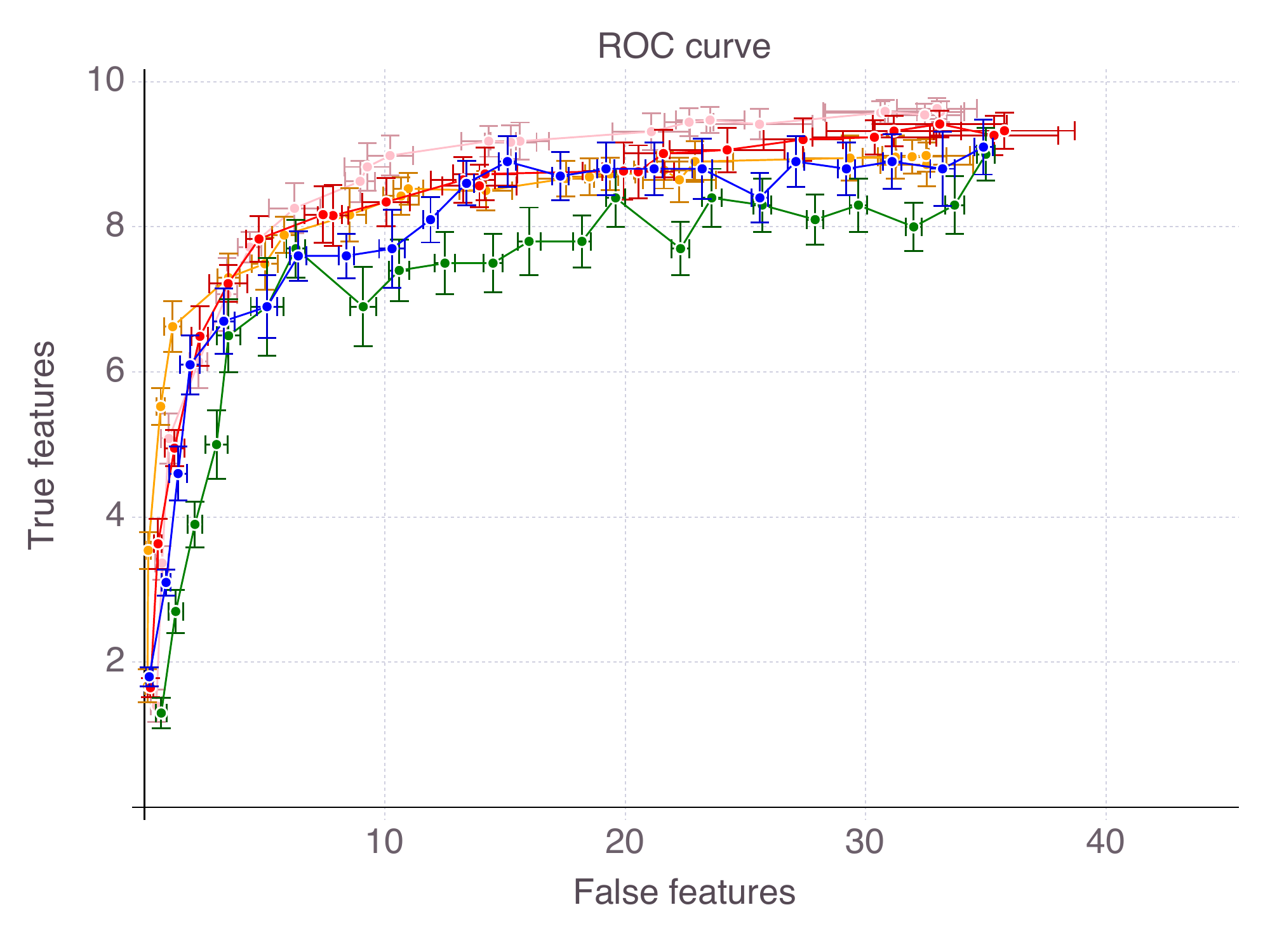}
	\caption{High noise, high $\rho$, $n=3,500$}
\end{subfigure}

\caption{Number of true features $TF$ vs. number of false features $FF$ for the CIO (in green), SS (in blue with $T_{max}=200$), ENet (in red), MCP (in orange), SCAD (in pink) with OLS loss. We average results over $10$ data sets with a fixed $n$.}
\label{fig:RegCVROC}
\end{figure*}
\begin{figure*}
\centering
\begin{subfigure}[t]{.45\linewidth}
	\centering
	\includegraphics[width=\linewidth]{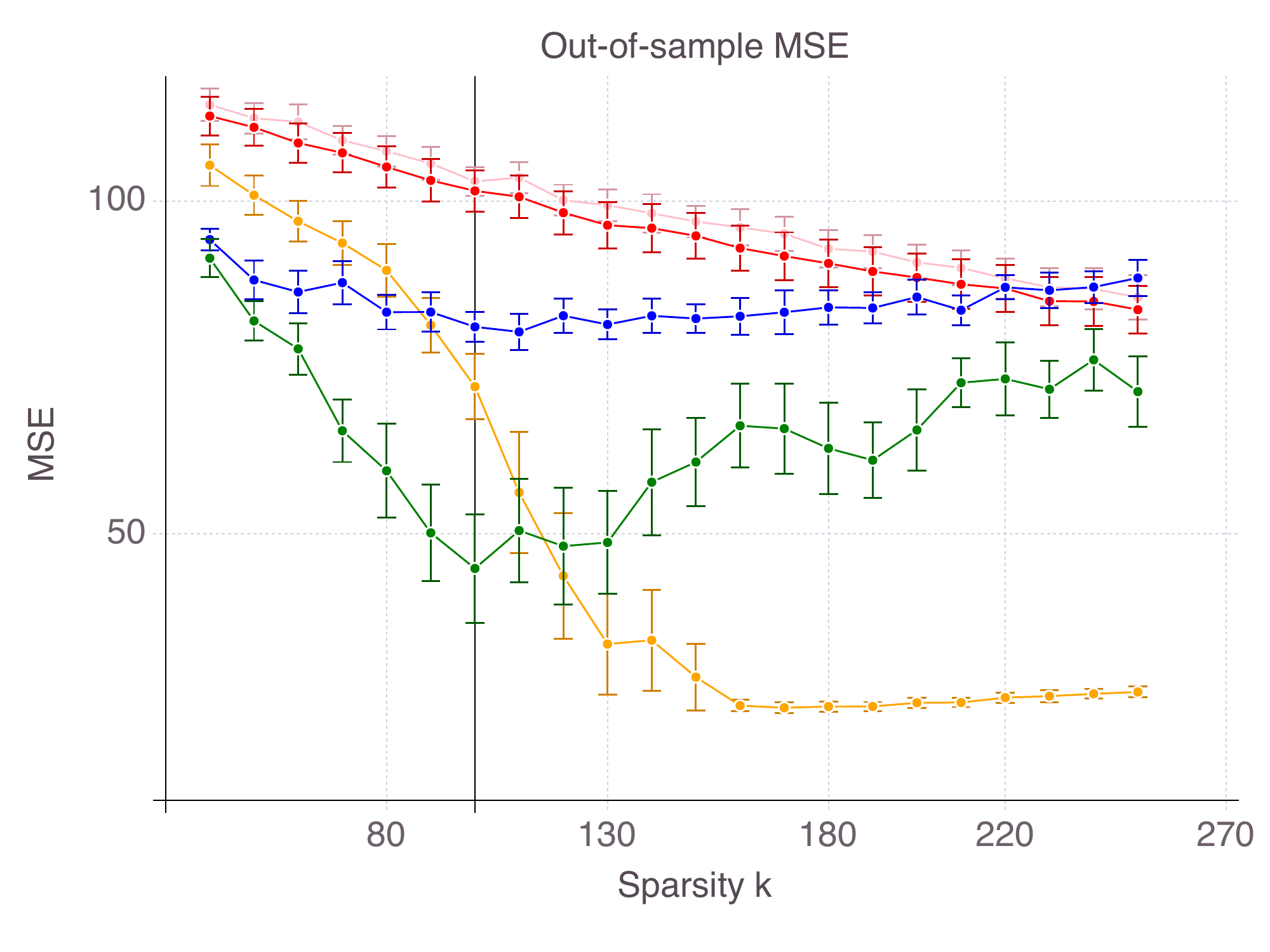}
	\caption{Low noise, low $\rho$, $n=900$}
\end{subfigure} %
~
\begin{subfigure}[t]{.45\linewidth}
	\centering
	\includegraphics[width=\linewidth]{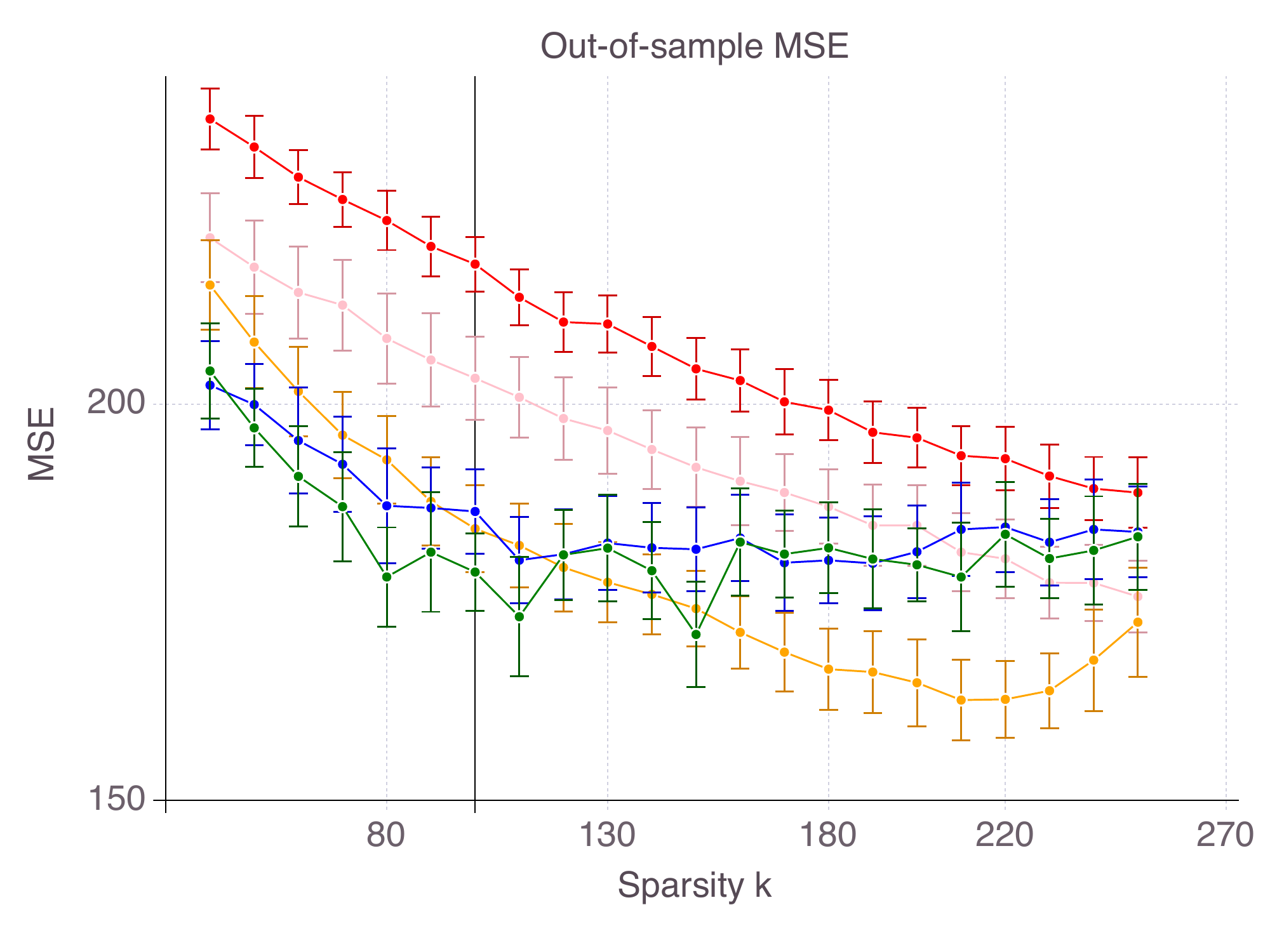}
	\caption{Low noise, high $\rho$, $n=900$}
\end{subfigure}

\begin{subfigure}[t]{.45\linewidth}
	\centering
	\includegraphics[width=\linewidth]{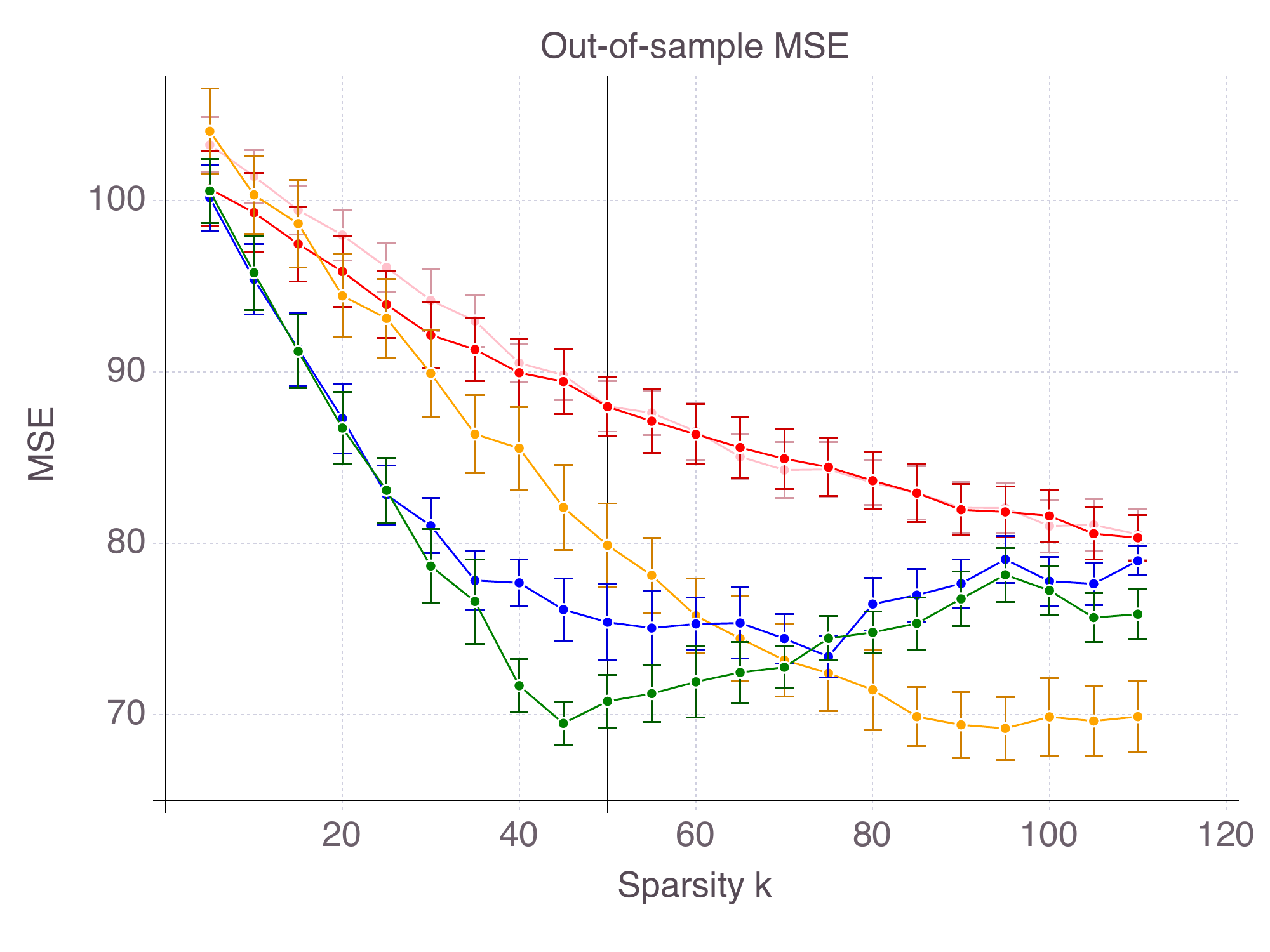}
	\caption{Medium noise, low $\rho$, $n=1,100$}
\end{subfigure} %
~
\begin{subfigure}[t]{.45\linewidth}
	\centering
	\includegraphics[width=\linewidth]{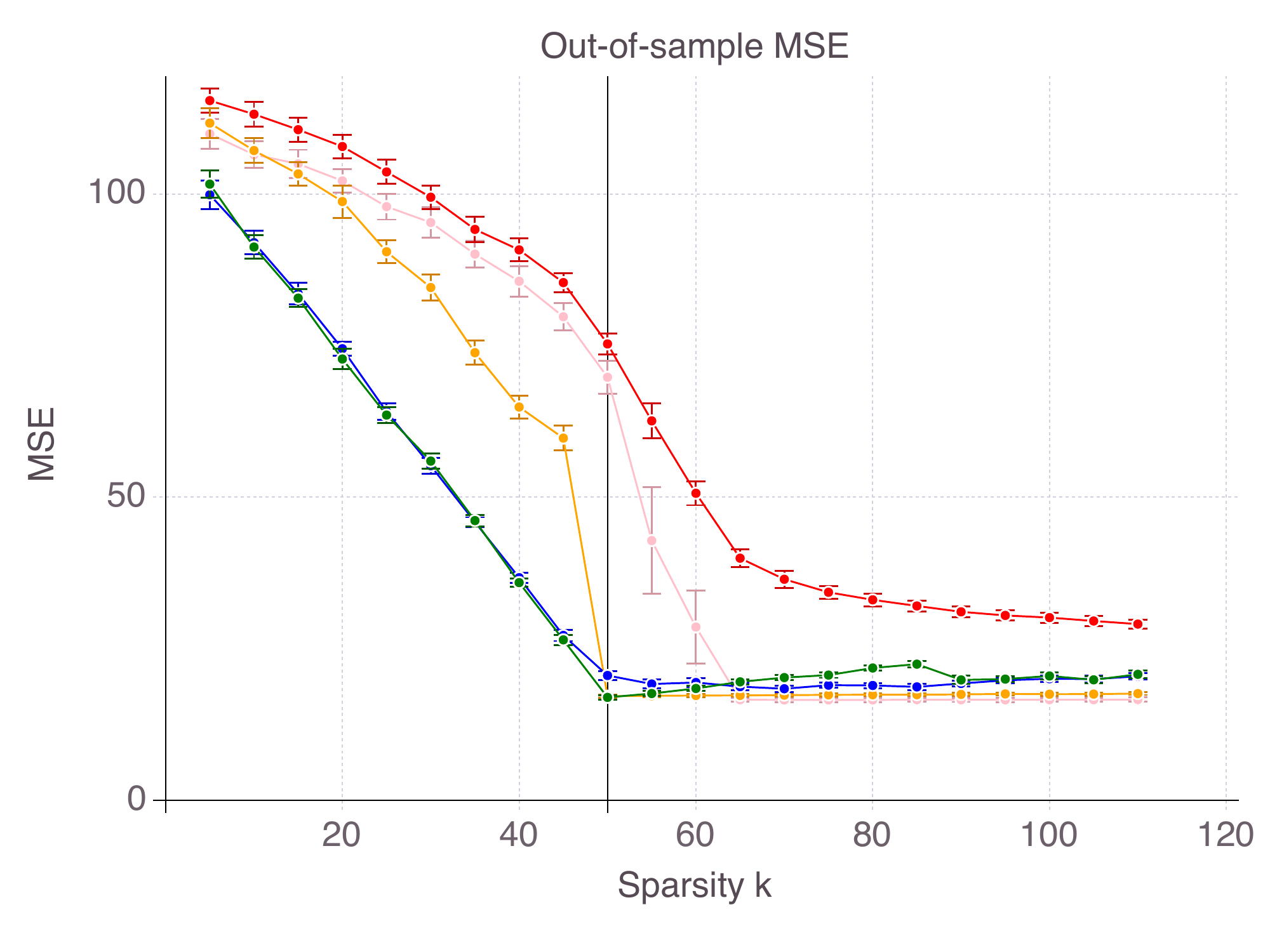}
	\caption{Medium noise, high $\rho$, $n=1,100$}
\end{subfigure}
\begin{subfigure}[t]{.45\linewidth}
	\centering
	\includegraphics[width=\linewidth]{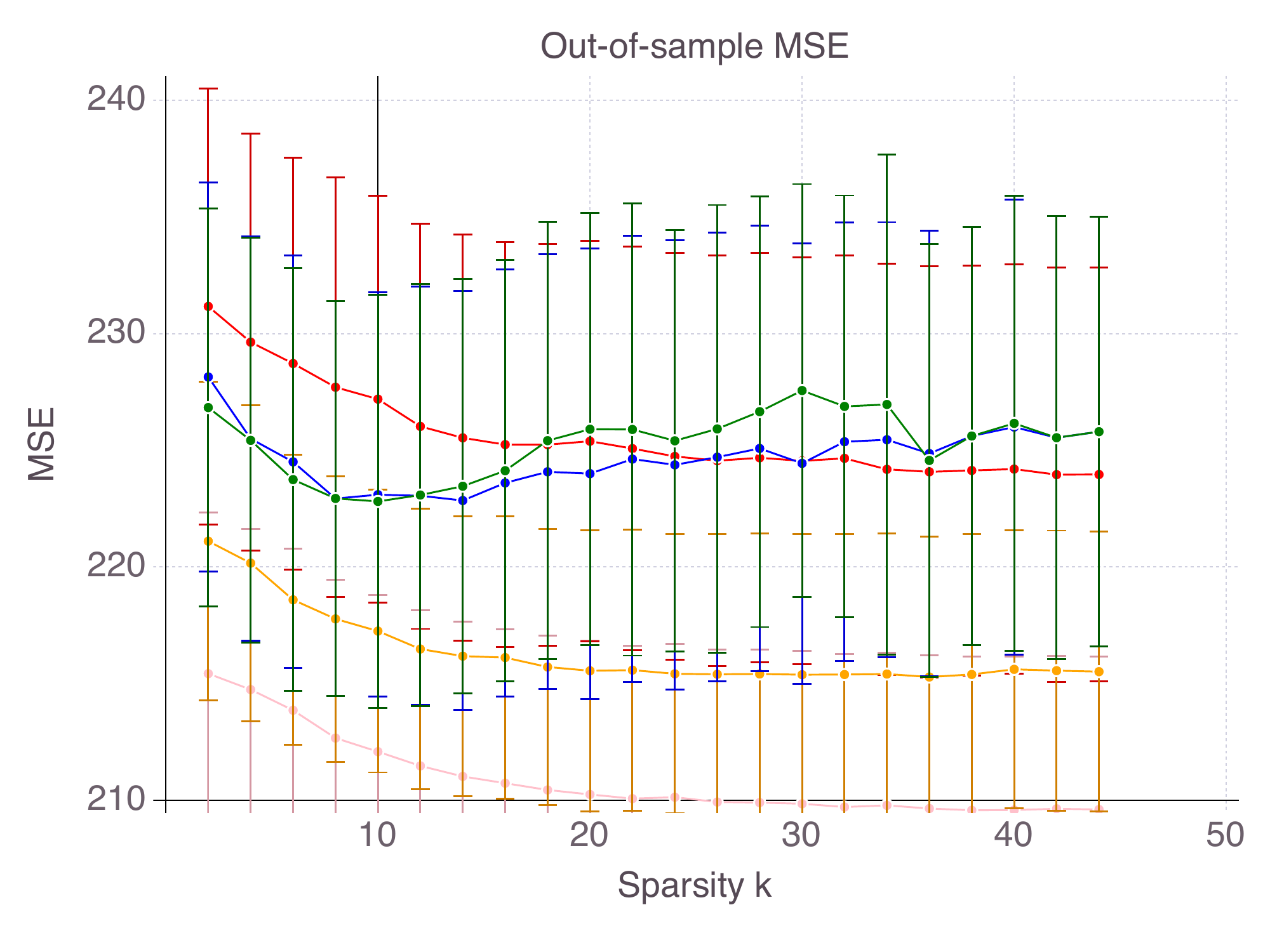}
	\caption{High noise, low $\rho$, $n=3,500$}
\end{subfigure} %
~
\begin{subfigure}[t]{.45\linewidth}
	\centering
	\includegraphics[width=\linewidth]{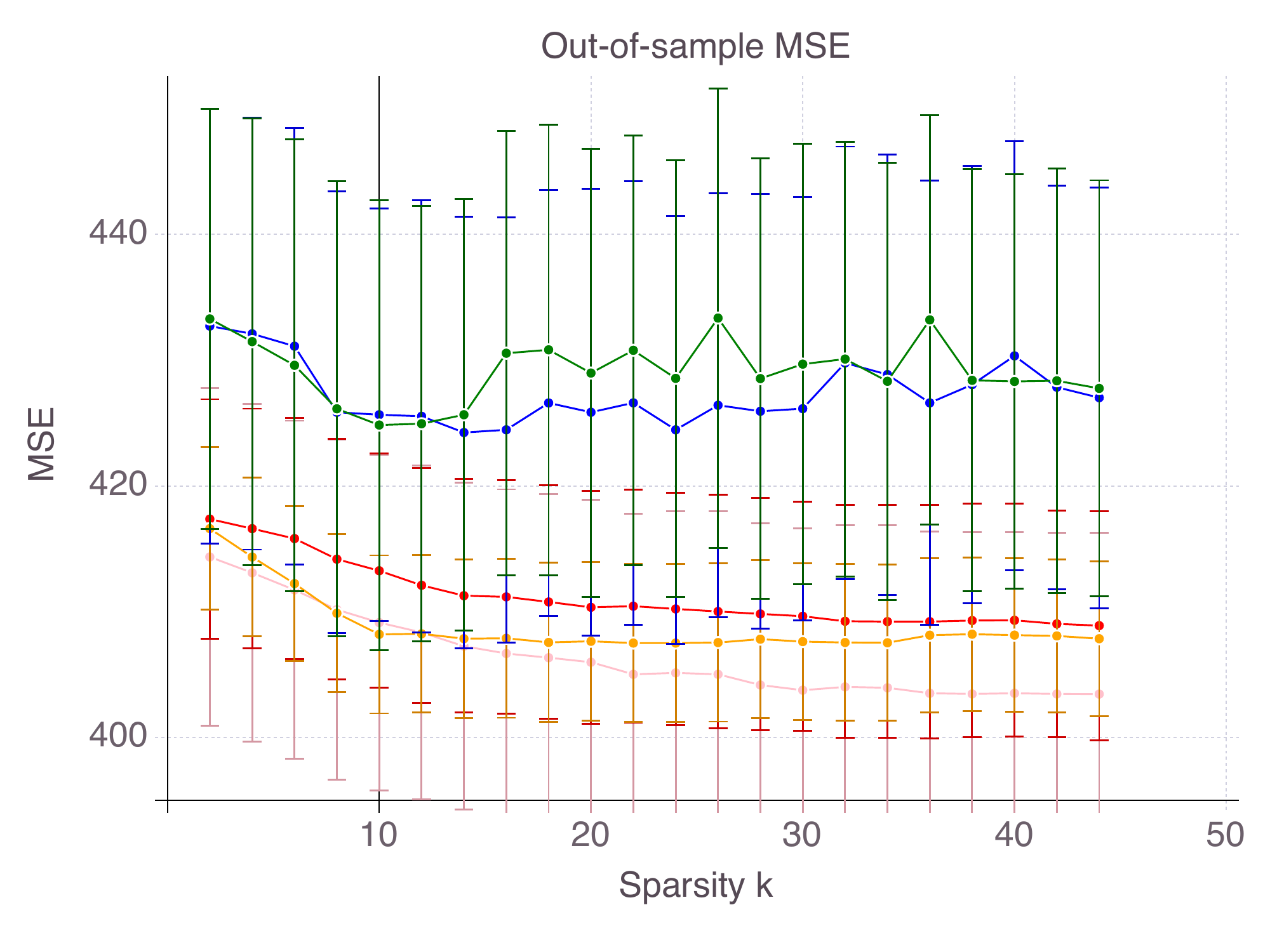}
	\caption{High noise, high $\rho$, $n=3,500$}
\end{subfigure}

\caption{Out-of-sample mean square error as $k$ increases, for the CIO (in green), SS (in blue with $T_{max}=200$), ENet (in red), MCP (in orange), SCAD (in pink) with OLS loss. We average results over $10$ data sets with a fixed $n$.}
\label{fig:RegCVcoupe}
\end{figure*}

As a result, for every $n$, each method selects $k^\star$ features, some of which are in the true support, others being irrelevant, as measured by accuracy and false detection rate respectively. Figures \ref{fig:RegCV.A} (p.  \pageref{fig:RegCV.A}) and \ref{fig:RegCV.FDR} (p. \pageref{fig:RegCV.FDR}) report the results of the cross-validation procedure for { increasing} $n$. In terms of accuracy (Figure \ref{fig:RegCV.A}), all five methods are relatively equivalent and demonstrate a clear convergence: $A\rightarrow 1$ as $n \rightarrow \infty$. { The first to achieve perfect accuracy is ENet, followed by SCAD, MCP, CIO and then SS. On false detection rate however (Figure \ref{fig:RegCV.FDR}), the methods rank in the opposite order. Among all five, CIO and SS achieve the lowest $FDR$ with $FDR$ as low as $0\%$ in low noise settings and around $30\%$ when noise is high. On the contrary, ENet persistently returns around $80\%$ of incorrect features in all regimes of noise and correlation. Concerning MCP and SCAD, in low noise regimes, false detection rate quickly drops as sample size increases. Yet, for large values of $n$, we observe a strictly positive $FDR$ on average (around $15\%$ in our experiments) and high variance, suggesting that feature selection with these regularizations is pretty unstable. As noise increases, $FDR$ for those methods remains significant (around $50\%$), with a fine advantage of MCP over SCAD. } In our opinion, this is due to the fact that MCP and SCAD, just like Lasso/ENet, do not enforce sparsity explicitly, like CIO or SS do, but rely on regularization to \emph{induce} it.
\begin{figure*}[p]
\centering
\begin{subfigure}[h]{.45\linewidth}
	\centering
	\includegraphics[width=\linewidth]{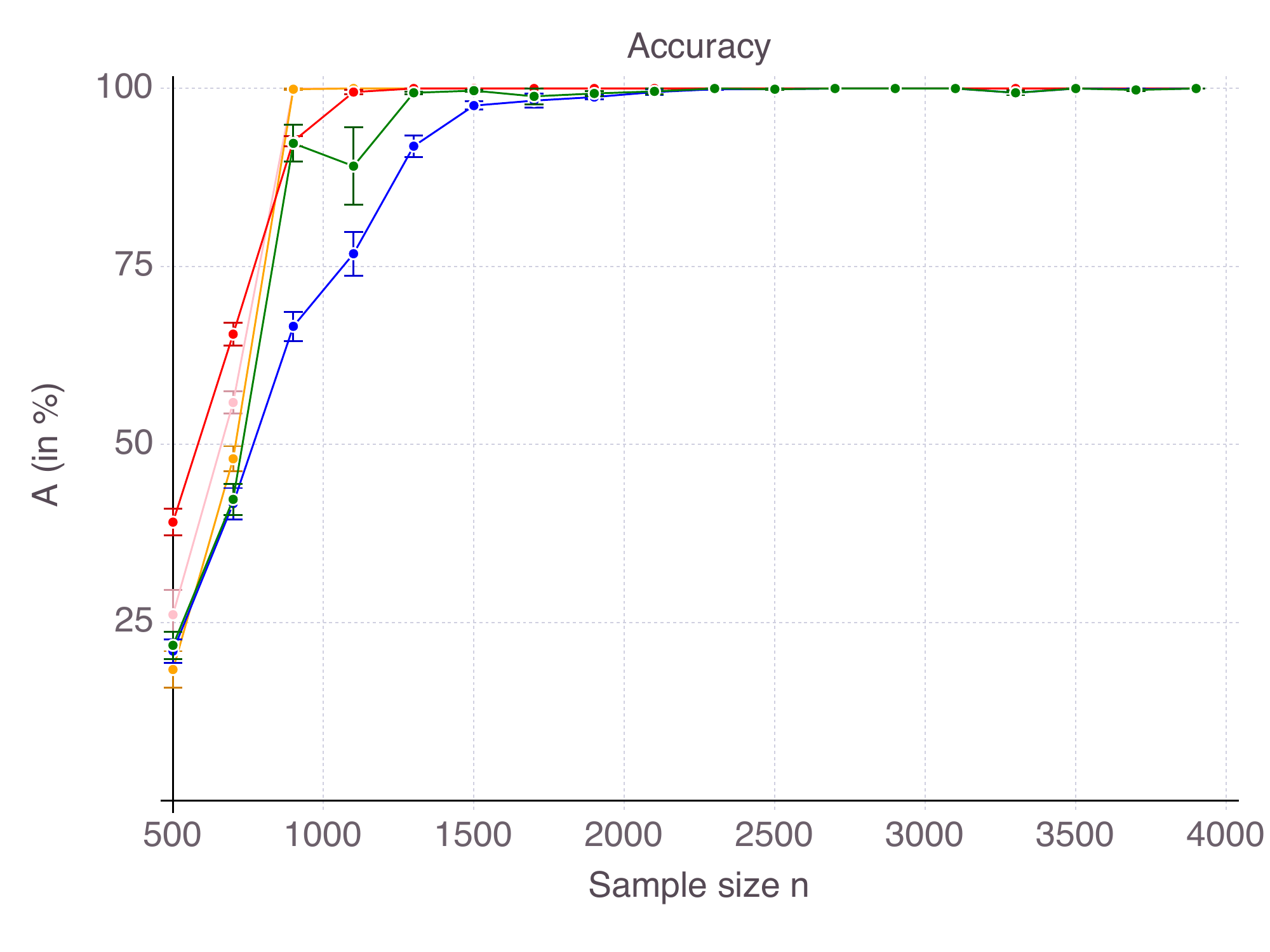}
	\caption{Low noise, low correlation}
\end{subfigure} %
~
\begin{subfigure}[h]{.45\linewidth}
	\centering
	\includegraphics[width=\linewidth]{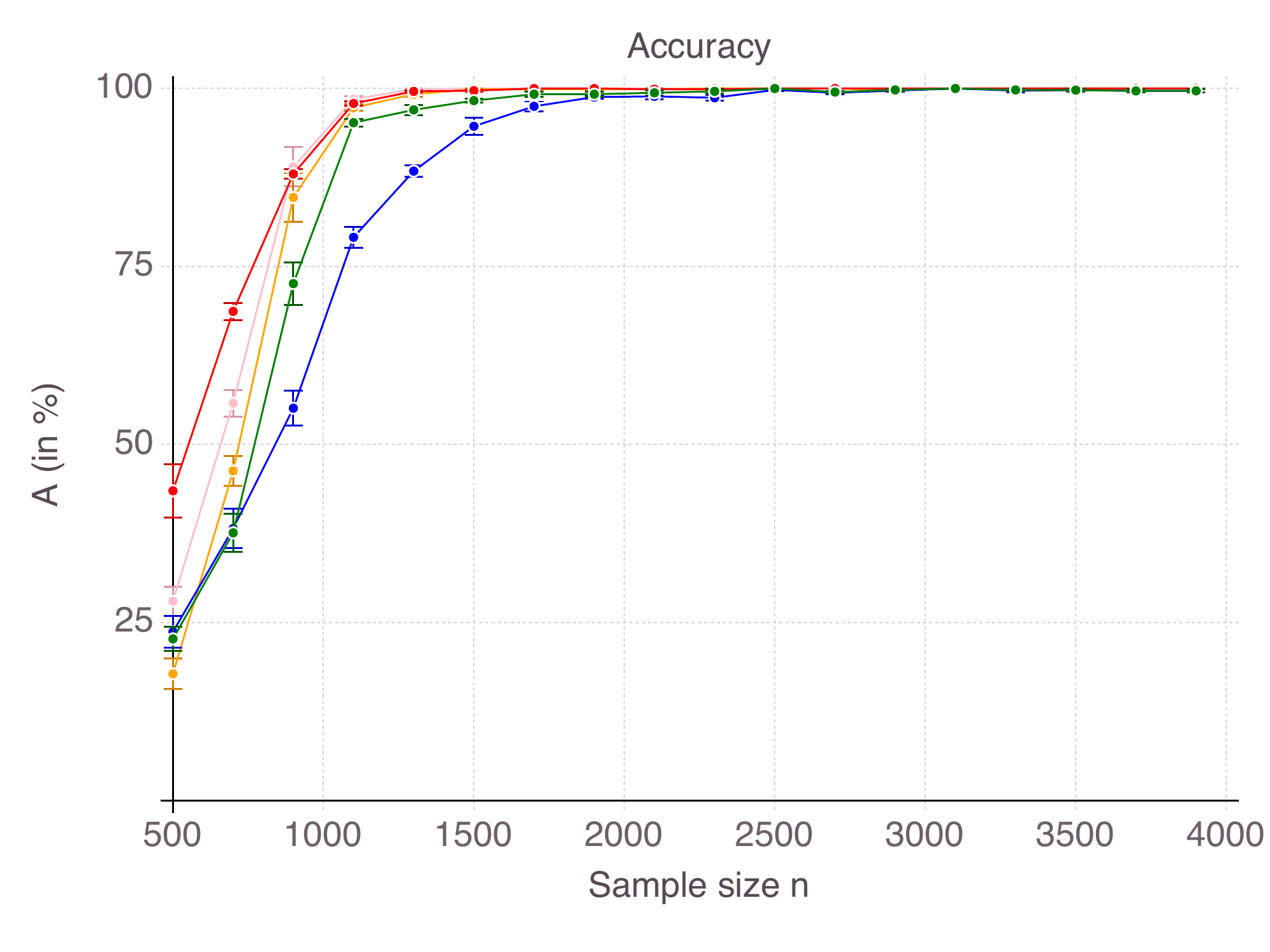}
	\caption{Low noise, high correlation}
\end{subfigure}

\begin{subfigure}[h]{.45\linewidth}
	\centering
	\includegraphics[width=\linewidth]{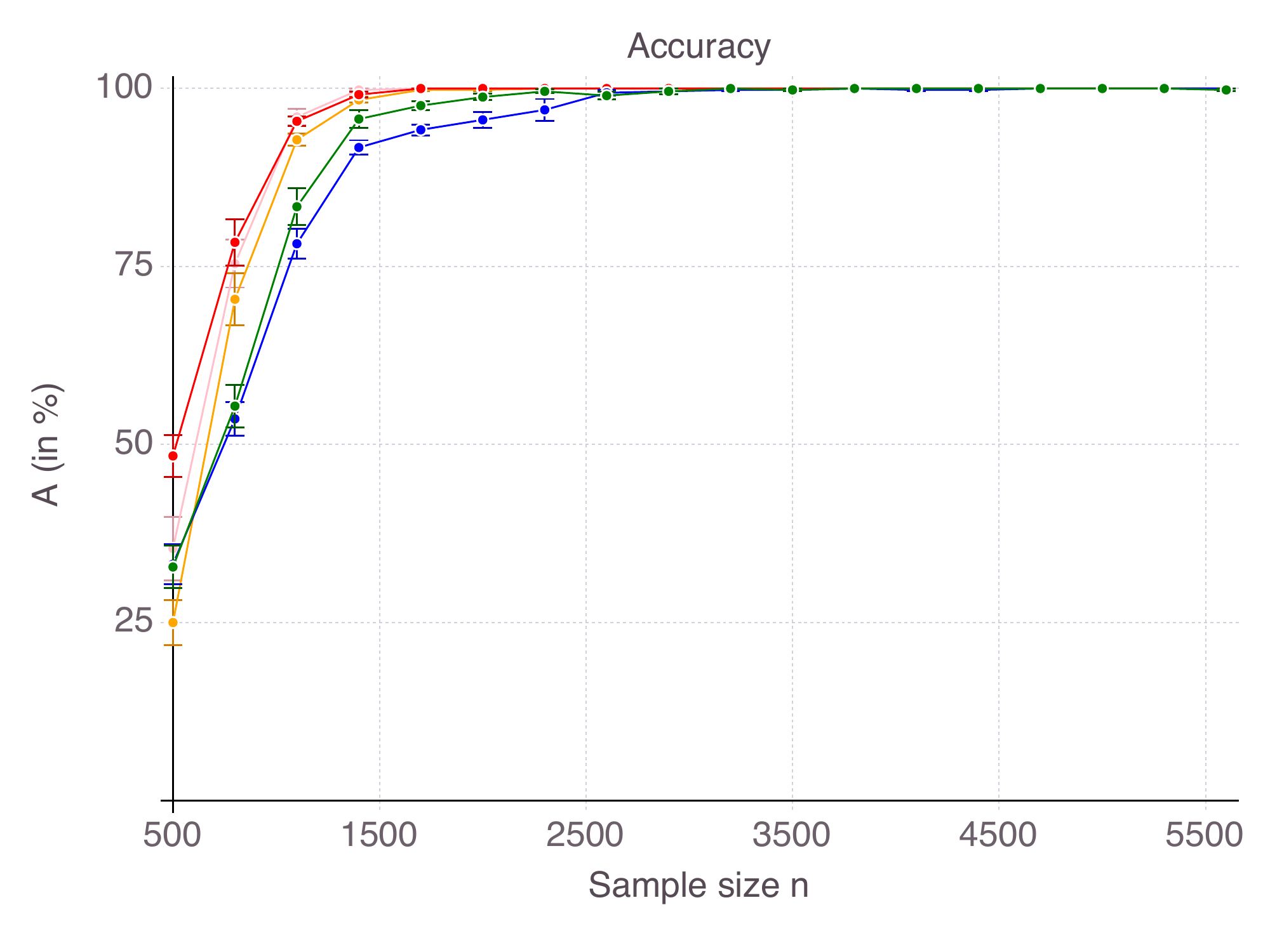}
	\caption{Medium noise, low correlation}
\end{subfigure} %
~
\begin{subfigure}[h]{.45\linewidth}
	\centering
	\includegraphics[width=\linewidth]{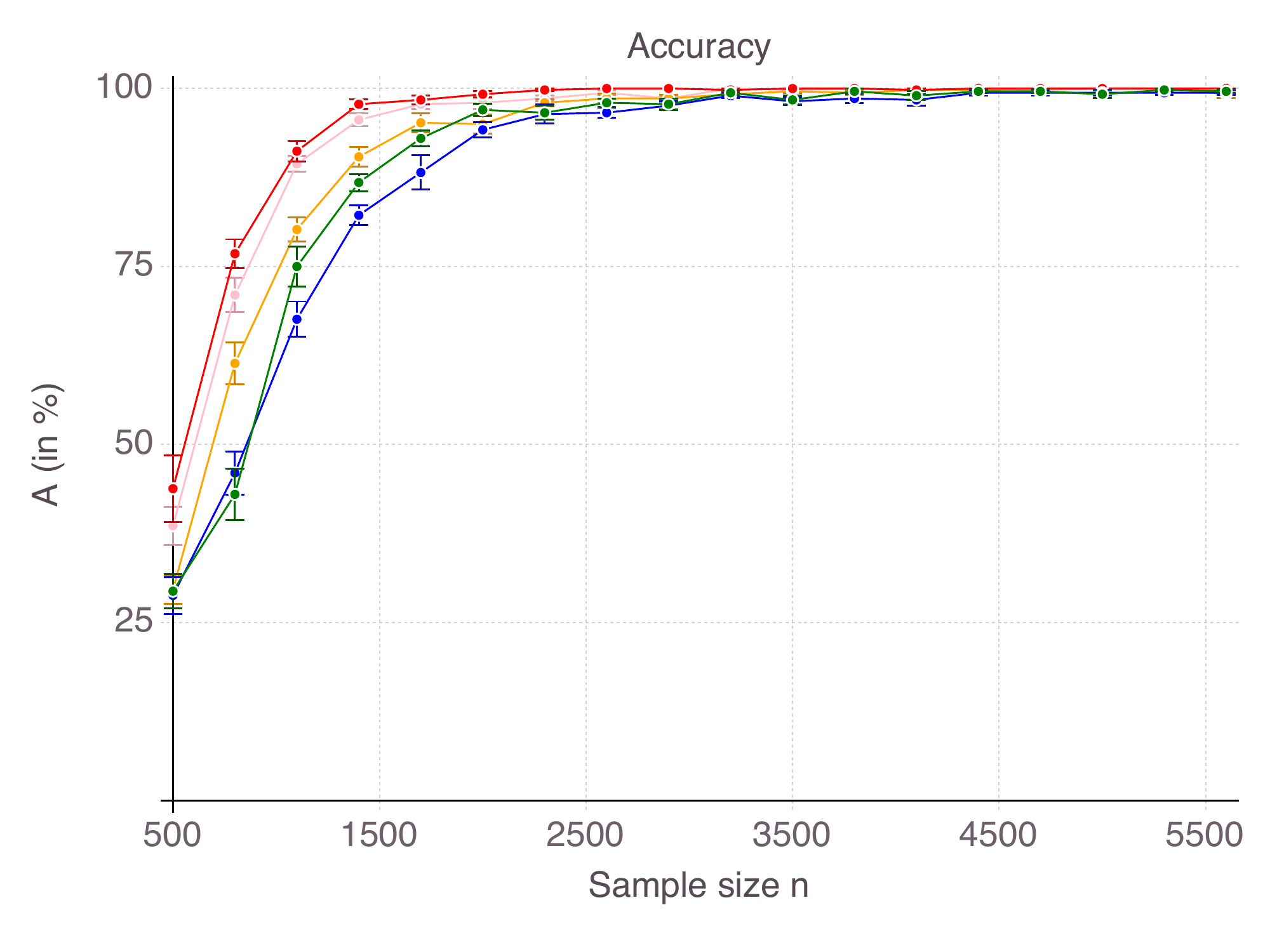}
	\caption{Medium noise, high correlation}
\end{subfigure}

\begin{subfigure}[h]{.45\linewidth}
	\centering
	\includegraphics[width=\linewidth]{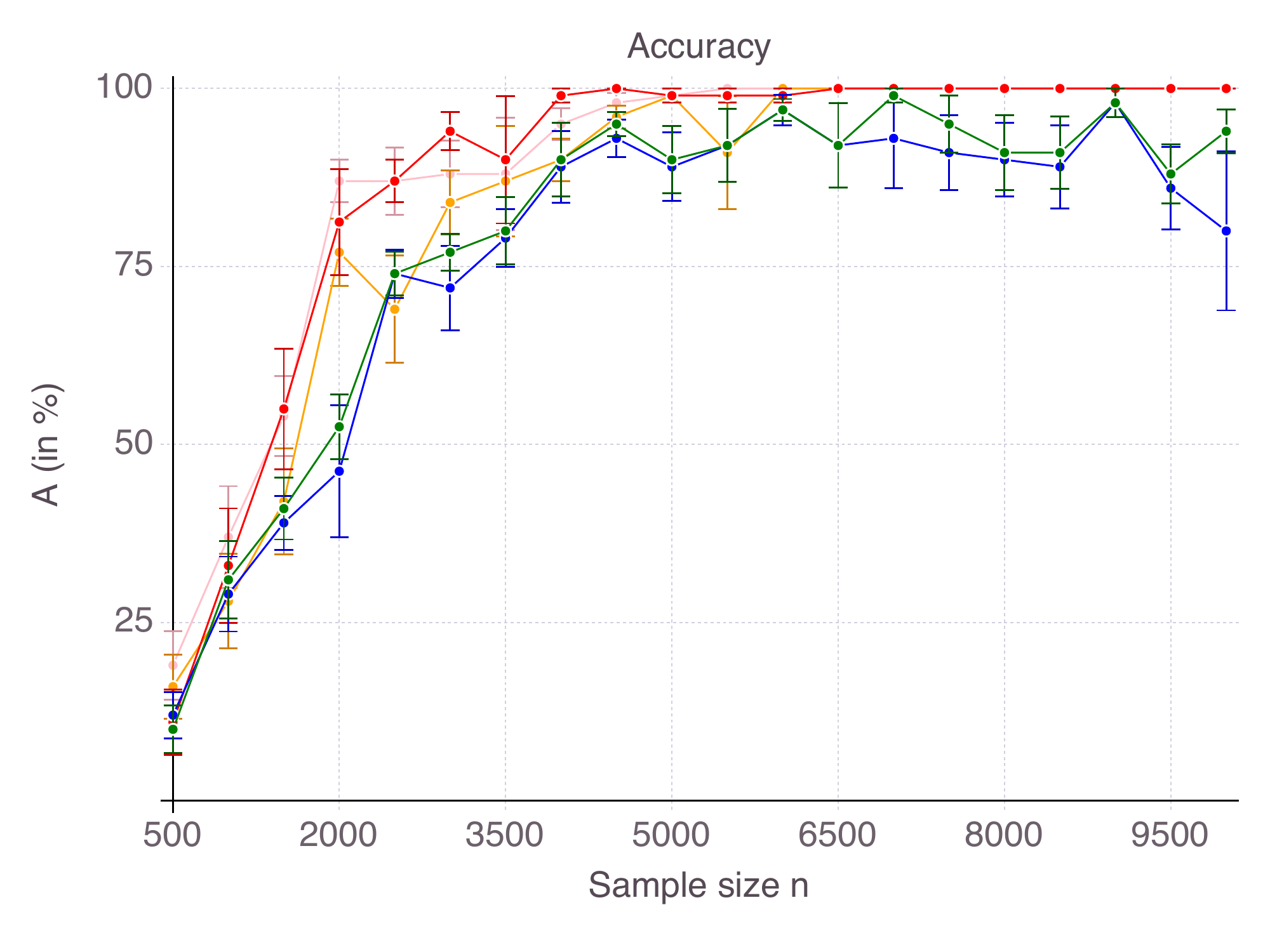}
	\caption{High noise, low correlation}
\end{subfigure} %
~
\begin{subfigure}[h]{.45\linewidth}
	\centering
	\includegraphics[width=\linewidth]{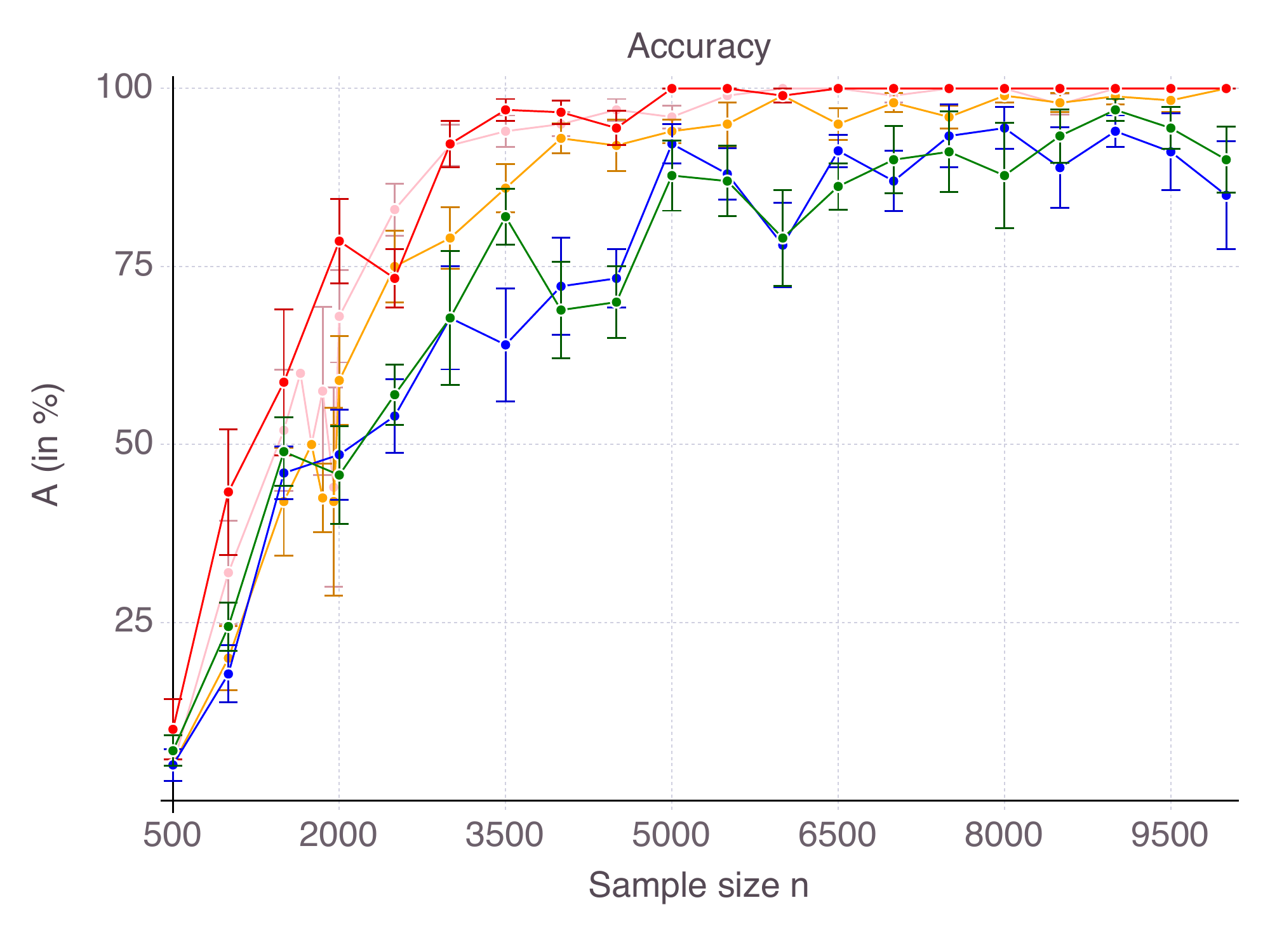}
	\caption{High noise, high correlation}
\end{subfigure}
\caption{Accuracy $A$ as $n$ increases, for the CIO (in green), SS (in blue with $T_{max}=200$), ENet (in red), MCP (in orange), SCAD (in pink) with OLS loss. We average results over $10$ data sets.}
\label{fig:RegCV.A}
\end{figure*}
\begin{figure*}[p]
\centering
\begin{subfigure}[h]{.45\linewidth}
	\centering
	\includegraphics[width=\linewidth]{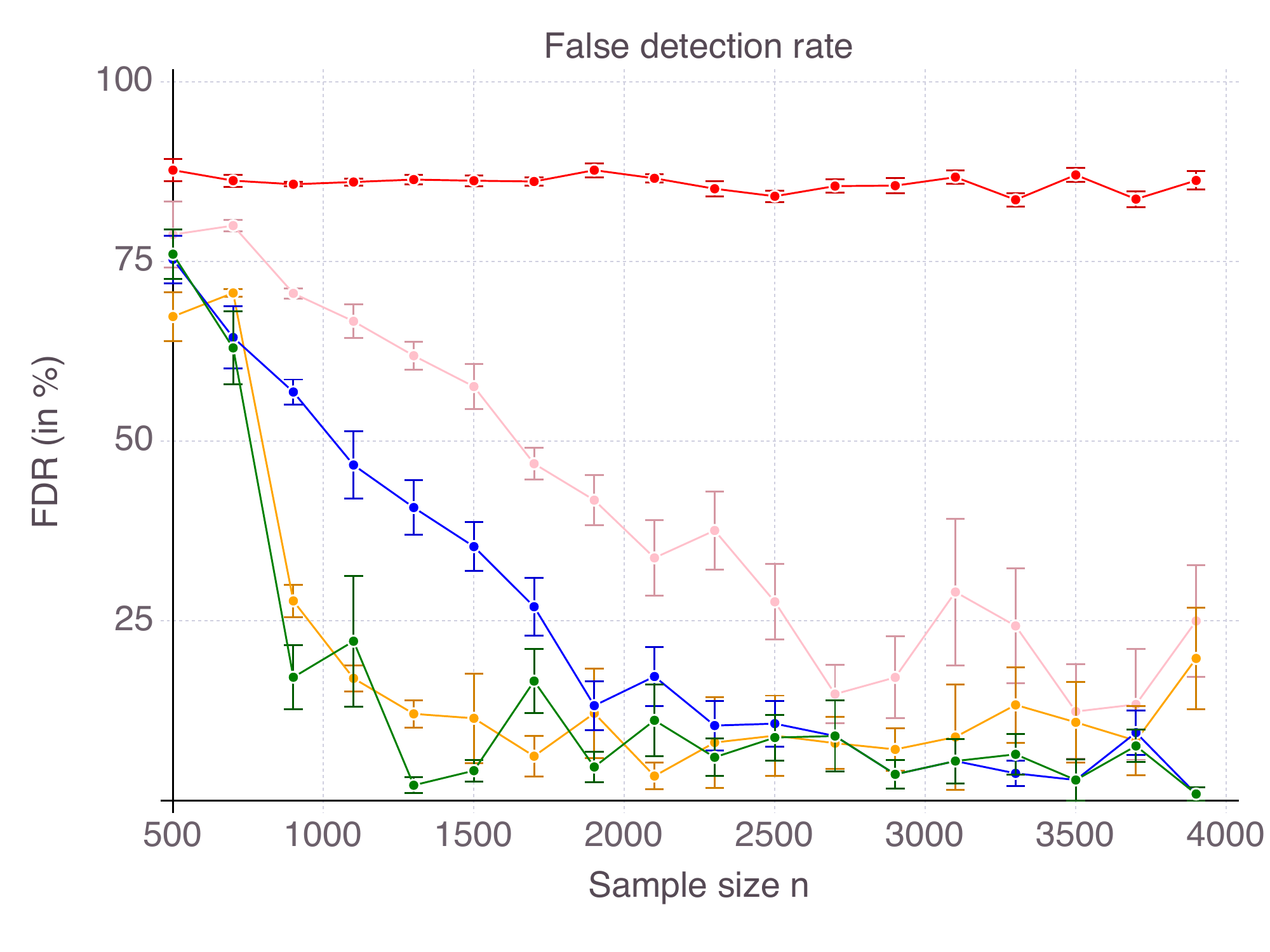}
	\caption{Low noise, low correlation}
\end{subfigure} %
~
\begin{subfigure}[h]{.45\linewidth}
	\centering
	\includegraphics[width=\linewidth]{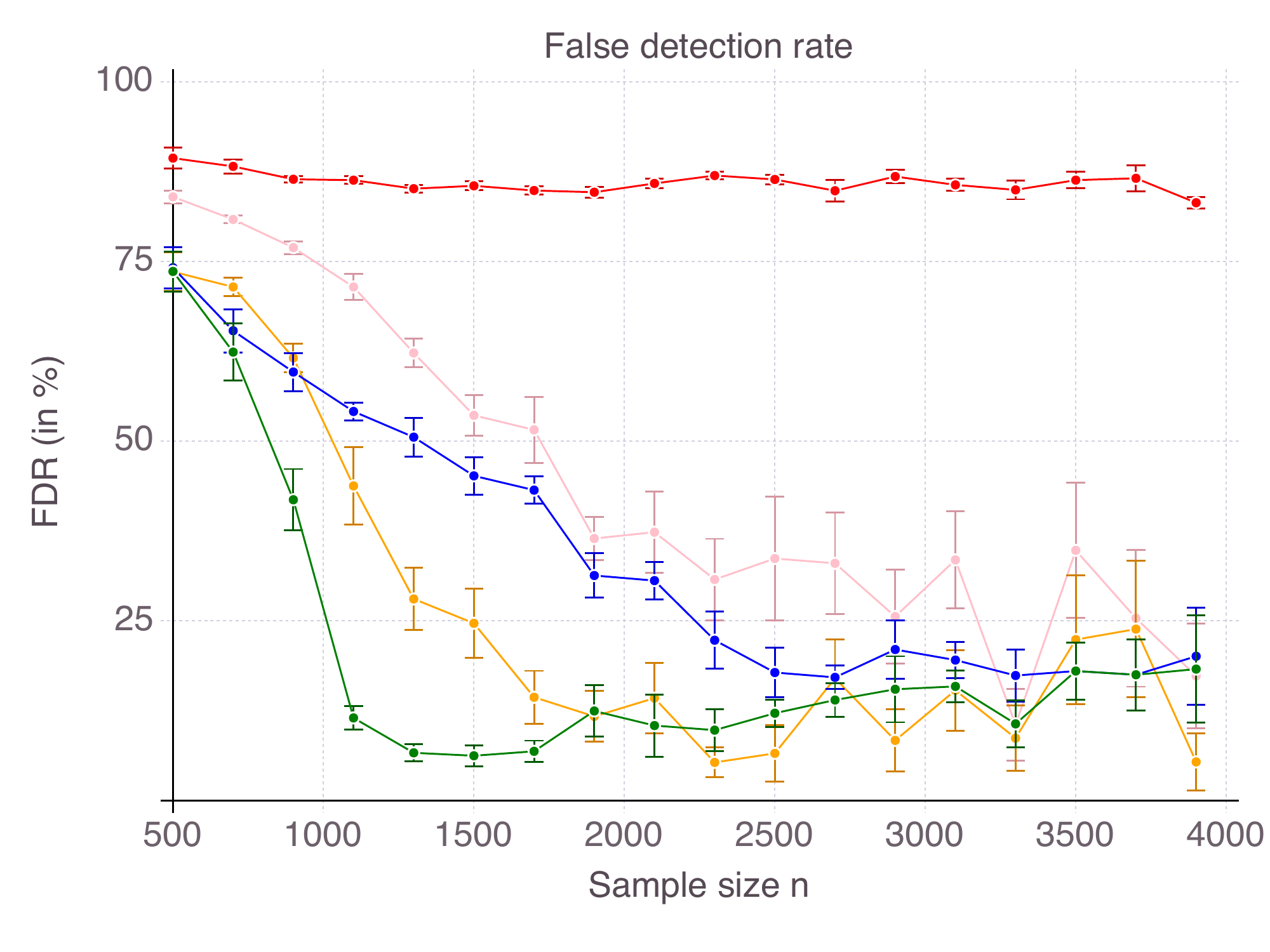}
	\caption{Low noise, high correlation}
\end{subfigure}

\begin{subfigure}[h]{.45\linewidth}
	\centering
	\includegraphics[width=\linewidth]{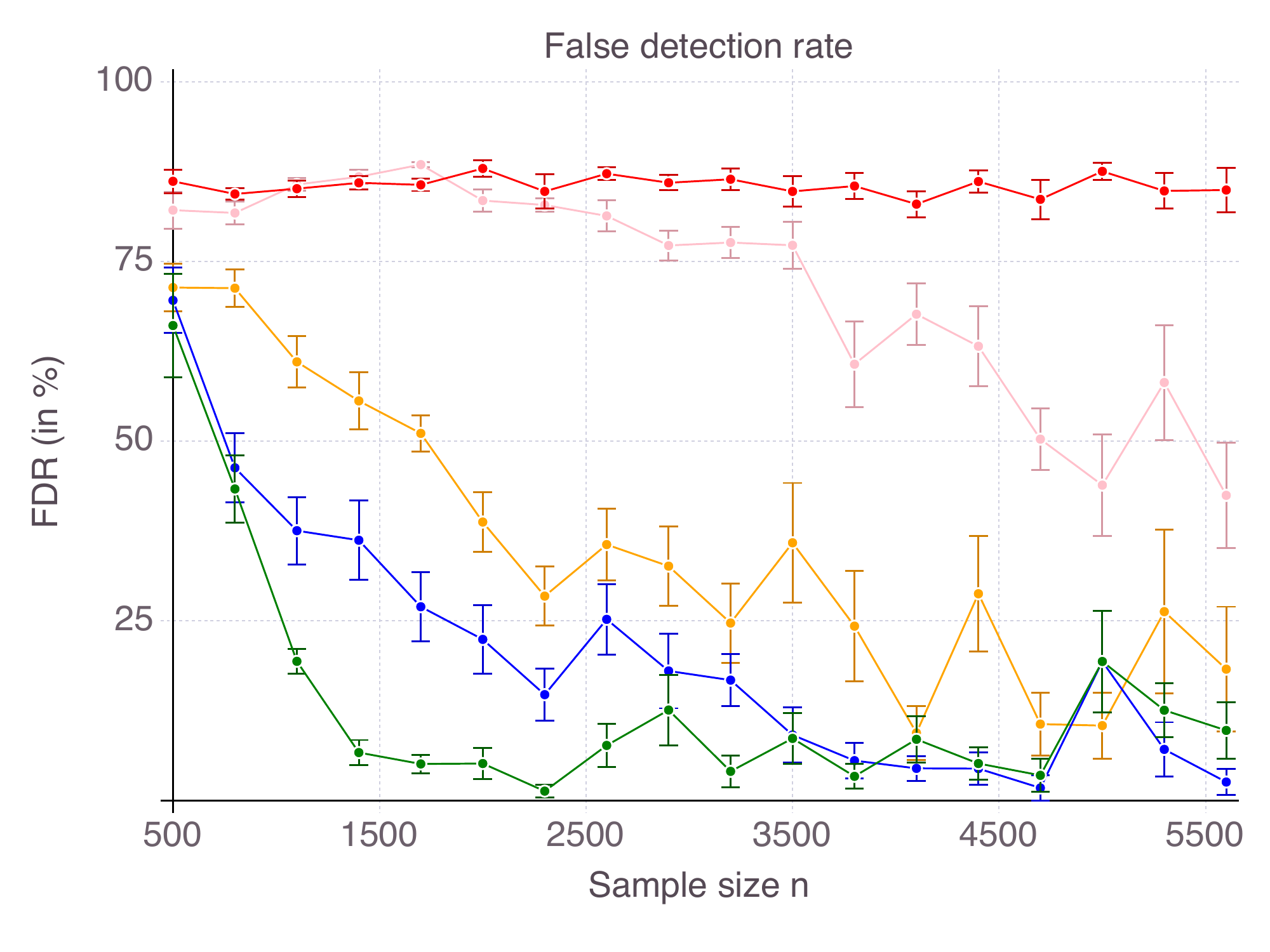}
	\caption{Medium noise, low correlation}
\end{subfigure} %
~
\begin{subfigure}[h]{.45\linewidth}
	\centering
	\includegraphics[width=\linewidth]{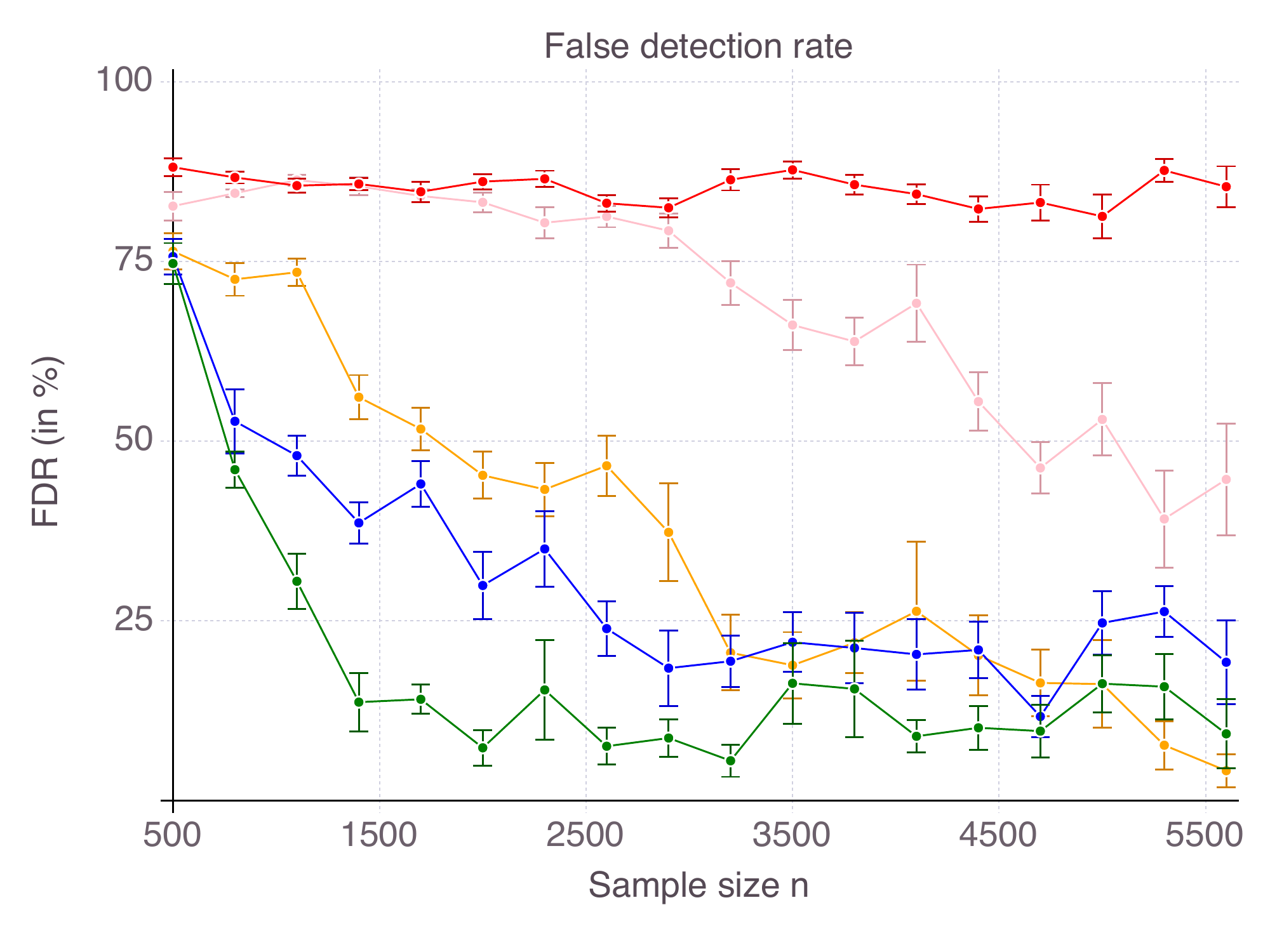}
	\caption{Medium noise, high correlation}
\end{subfigure}

\begin{subfigure}[h]{.45\linewidth}
	\centering
	\includegraphics[width=\linewidth]{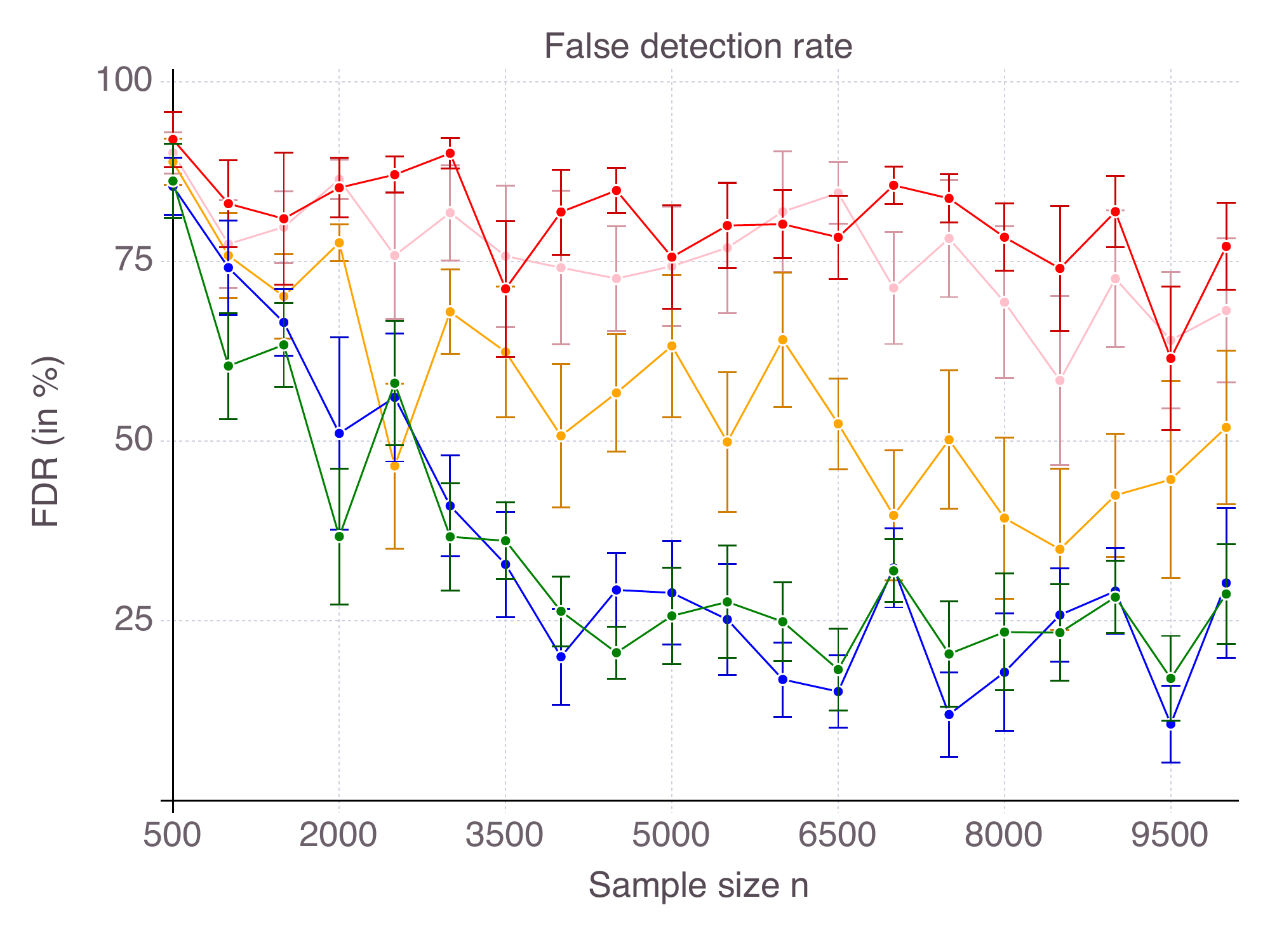}
	\caption{High noise, low correlation}
\end{subfigure} %
~
\begin{subfigure}[h]{.45\linewidth}
	\centering
	\includegraphics[width=\linewidth]{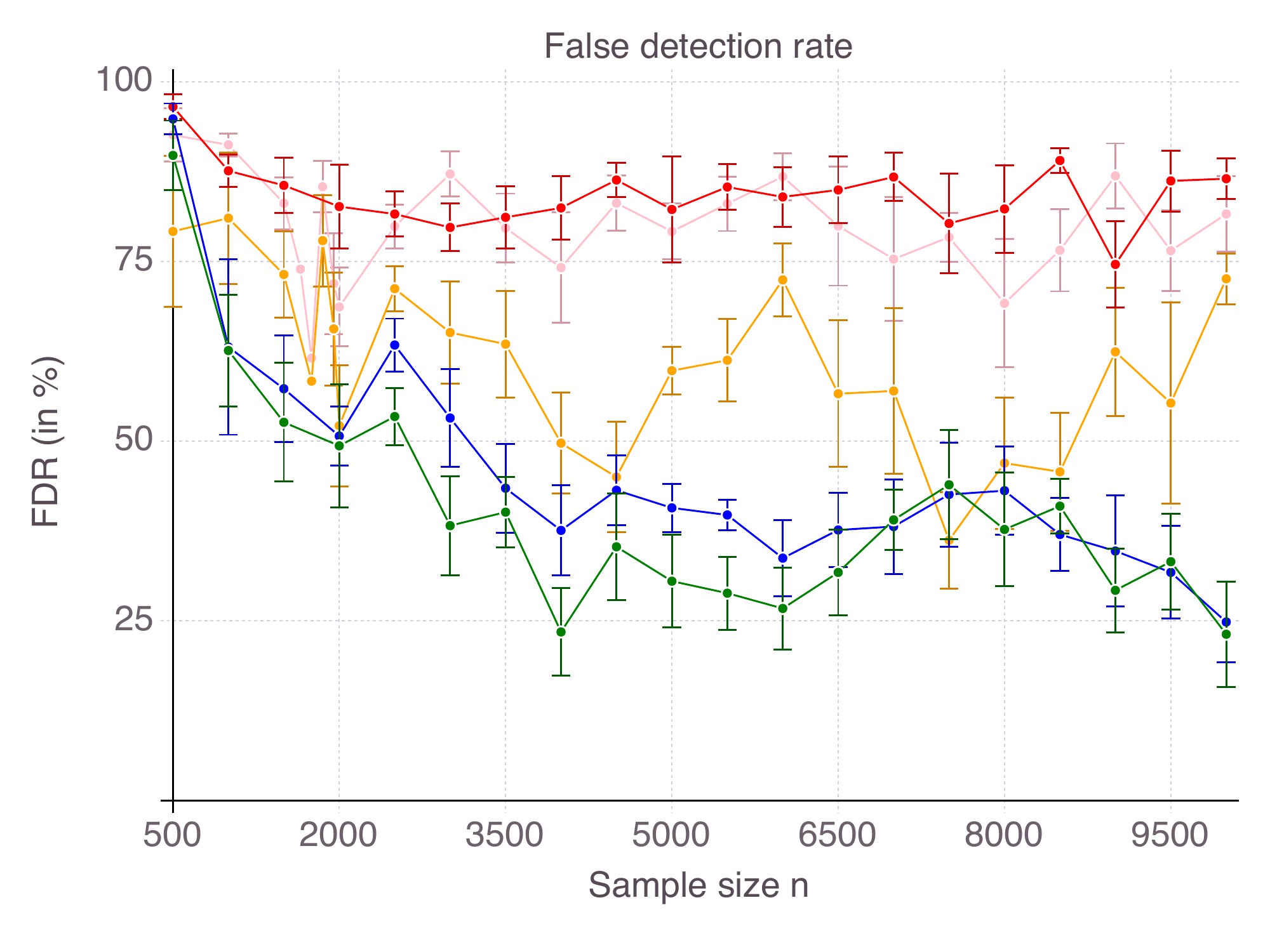}
	\caption{High noise, high correlation}
\end{subfigure}
\caption{False detection rate $FDR$ as $n$ increases, for the CIO (in green), SS (in blue with $T_{max}=200$), ENet (in red), MCP (in orange), SCAD (in pink) with OLS loss. We average results over $10$ data sets.}
\label{fig:RegCV.FDR}
\end{figure*}

{ 
\subsection{Synthetic data \emph{not} satisfying mutual incoherence condition}
We now consider a "hard" correlation structure, i.e., a setting where the standard Lasso estimator is inconsistent. Fix $p$, $k_{true}$ and a scalar $\theta \in \left( \tfrac{1}{k_{true}}, \tfrac{1}{\sqrt{k_{true}}} \right)$\footnote{in our experiment we take $\theta =\tfrac{1}{2 k_{true}} + \tfrac{1}{2 \sqrt{k_{true}}}$.} Define $\Sigma$ as a matrix with $1$’s on the diagonal, $\theta$’s in the first $k_{true}$ positions of the $(k_{true}+1)$th row and column, and $0$’s everywhere else. Such a matrix does not satisfy mutual incoherence \citep[see][Appendix F.2. for a proof]{loh2017support}. As opposed to the previous setting, we fix $w_{true} = \left(\tfrac{1}{\sqrt{k_{true}}}, \dots, \tfrac{1}{\sqrt{k_{true}}}, 0, \dots, 0\right)$, and compute noisy signals, for increasing noise levels (see Table \ref{tab:m1.regimes} p. \pageref{tab:m1.regimes}). In this setting, the $\ell_1$-penalty result in an estimator that puts nonzero weight on the $(k+1)$th coordinate, while MCP and SCAD penalties eventually recover the true support \citep{loh2017support}. 

\begin{table}[h]
\centering
\caption{Regimes of noise ($SNR$) considered in our experiments on regression}
\label{tab:m1.regimes}
\begin{tabular}{lc}
\toprule
Low noise &  {$\begin{aligned}  &{SNR}=6  \\ & p =20,000 \\ &k=100\end{aligned} $} \\
\midrule
Medium noise & {$\begin{aligned} &{SNR}=1 \\ & p =10,000 \\ &k=50 \end{aligned} $} \\
\midrule
High noise & {$\begin{aligned} &{SNR}=0.05 \\ & p =2,000 \\ &k=10 \end{aligned} $} \\
\bottomrule
\end{tabular}
\end{table}
}
{
\subsubsection{Feature selection with a given support size} 
We first consider the case when the cardinality $k$ of the support to be returned is given and equal to the true sparsity $k_{true}$. In this setting, $\ell_1$-estimators are expected to always return at least $1$ incorrect feature, while MCP and SCAD will provably recover the entire support \citep{loh2017support}.

As shown on Figure \ref{fig:RegHardFixMSE} (p. \pageref{fig:RegHardFixMSE}), we observe empirically what theory dictates: The accuracy of ENet reaches a threshold strictly lower than $1$. Non-convex penalties MCP and SCAD, on the other hand, see their accuracy converging to $1$ as $n$ increases. Cardinality-constrained estimators CIO and SS, which are also non-convex, behave similarly, although no theory like \citet{loh2017support} exists, to the best of our knowledge. As far as accuracy is concerned (left panel), CIO dominates all other methods. Interestingly, while ENet is the least accurate in the limit $n \rightarrow +\infty$, it is sometimes more accurate than non-convex penalties for smaller values of $n$. 

We report computational time in Appendix \ref{sec:regression.supp.nomic.fix}.

\begin{figure*}
\centering
\begin{subfigure}[t]{\linewidth}
	\centering
	\includegraphics[width=.45\linewidth]{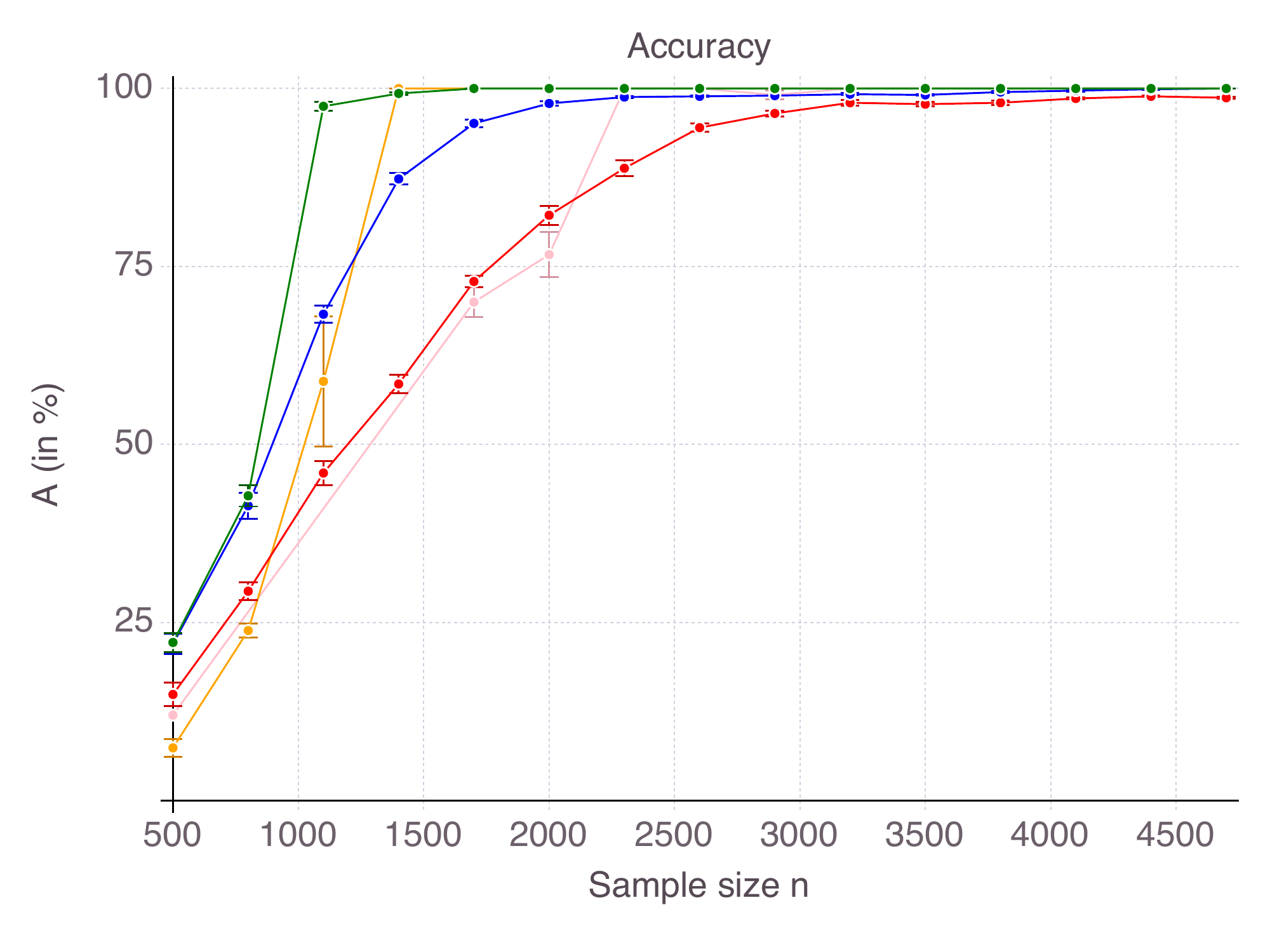}
	\includegraphics[width=.45\linewidth]{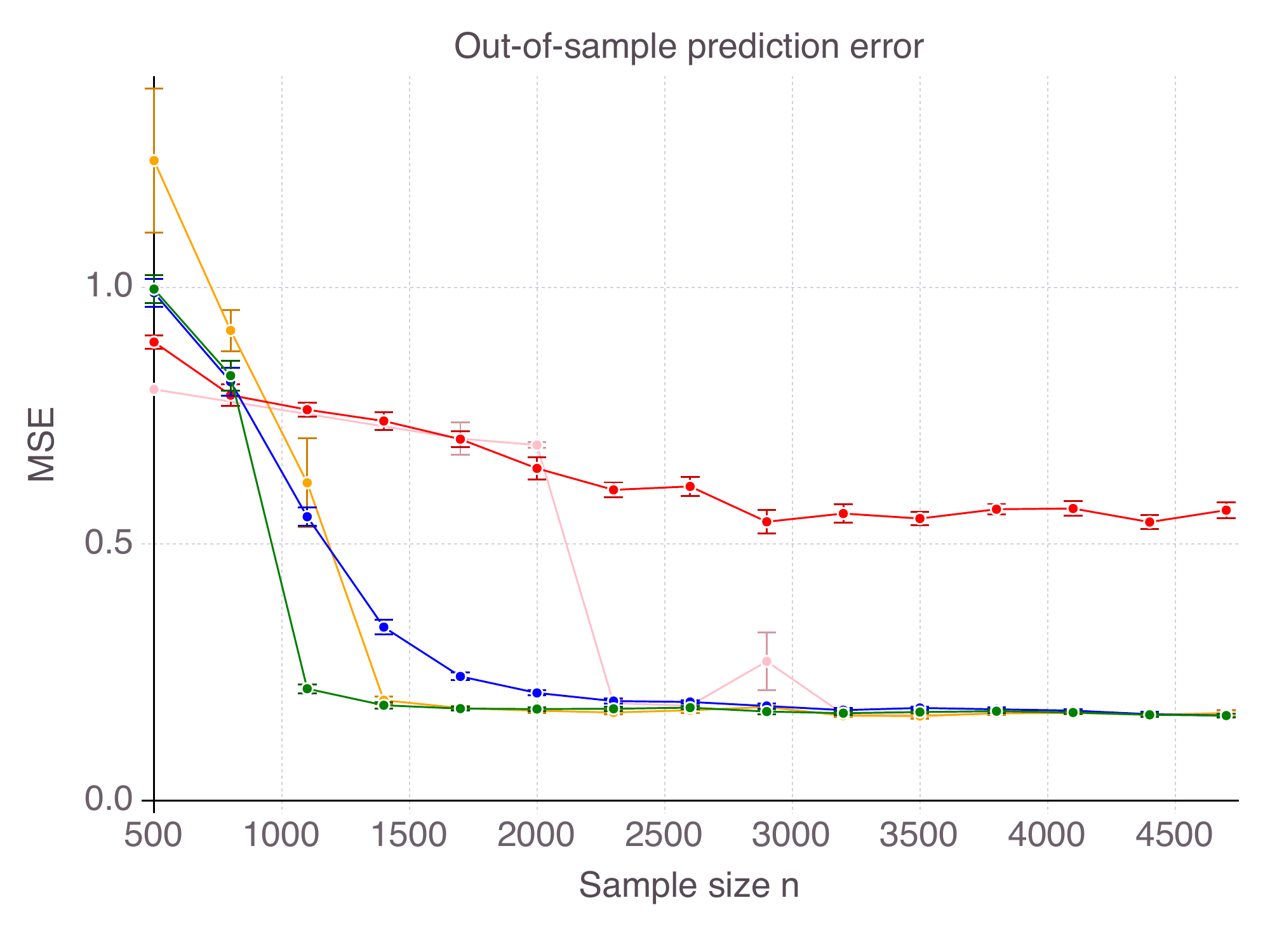}
	\caption{Low noise}
\end{subfigure} %

\begin{subfigure}[t]{\linewidth}
	\centering
	\includegraphics[width=.45\linewidth]{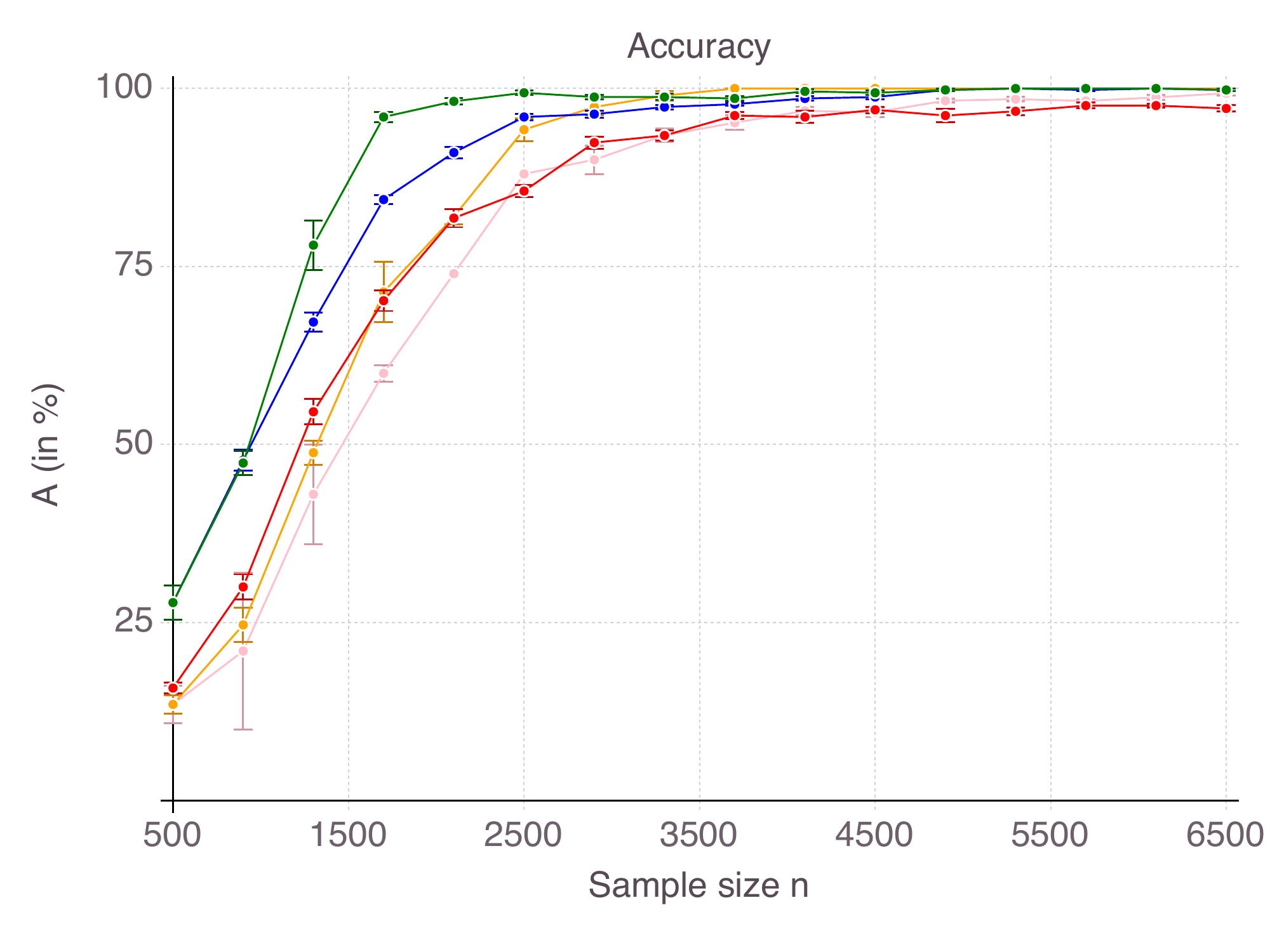}
	\includegraphics[width=.45\linewidth]{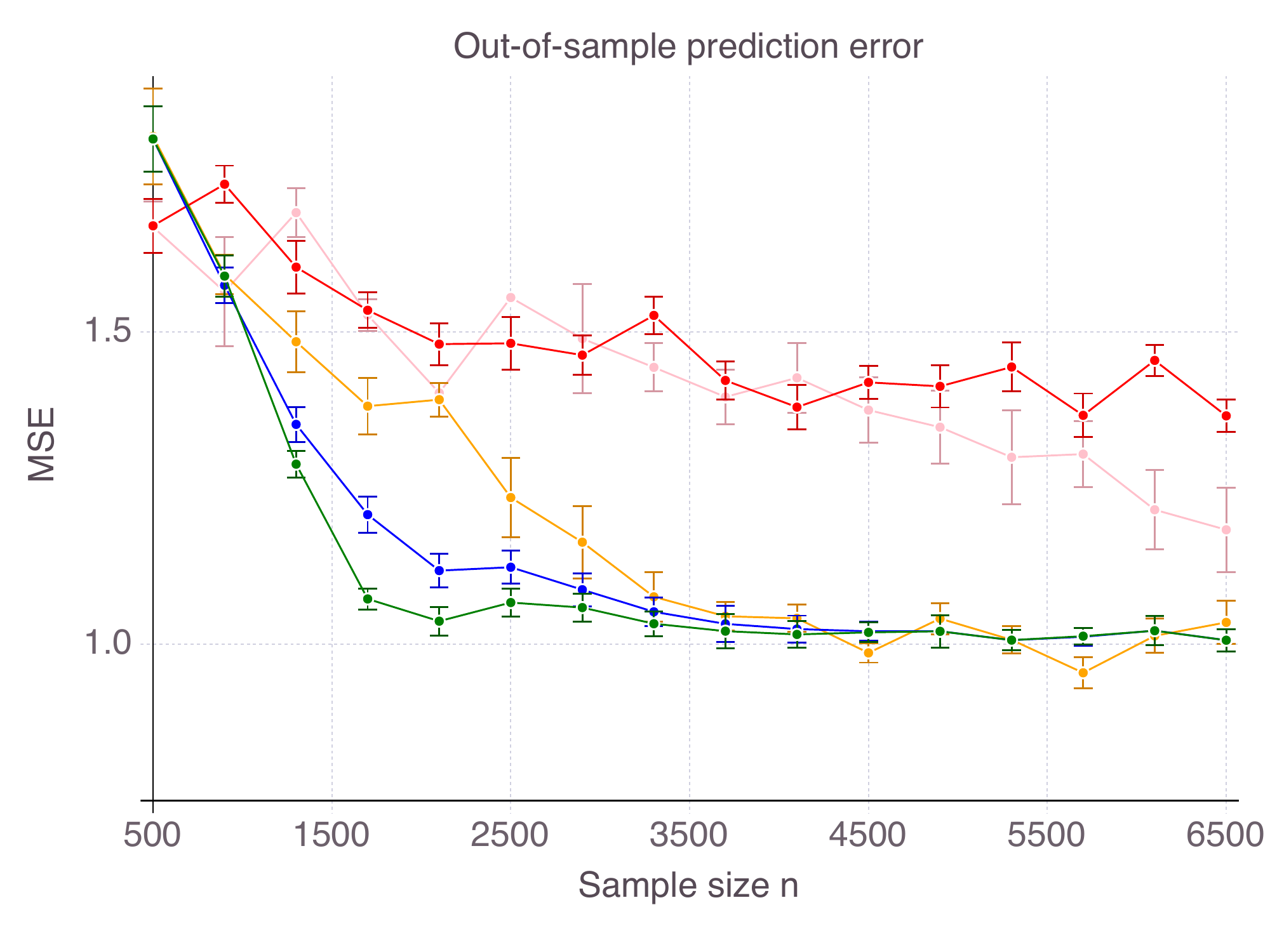}
	\caption{Medium noise}
\end{subfigure} %

\begin{subfigure}[t]{\linewidth}
	\centering
	\includegraphics[width=.45\linewidth]{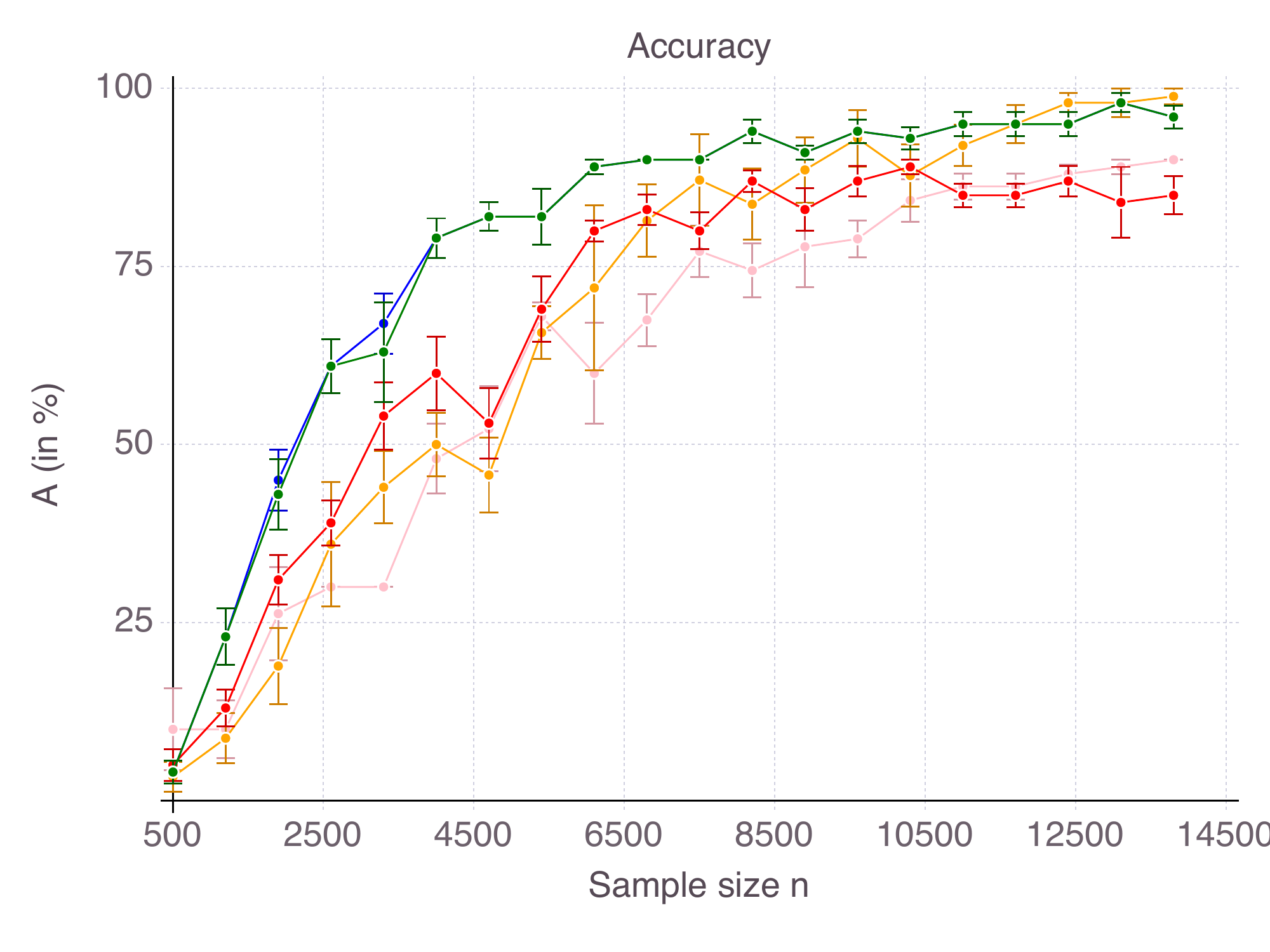}
	\includegraphics[width=.45\linewidth]{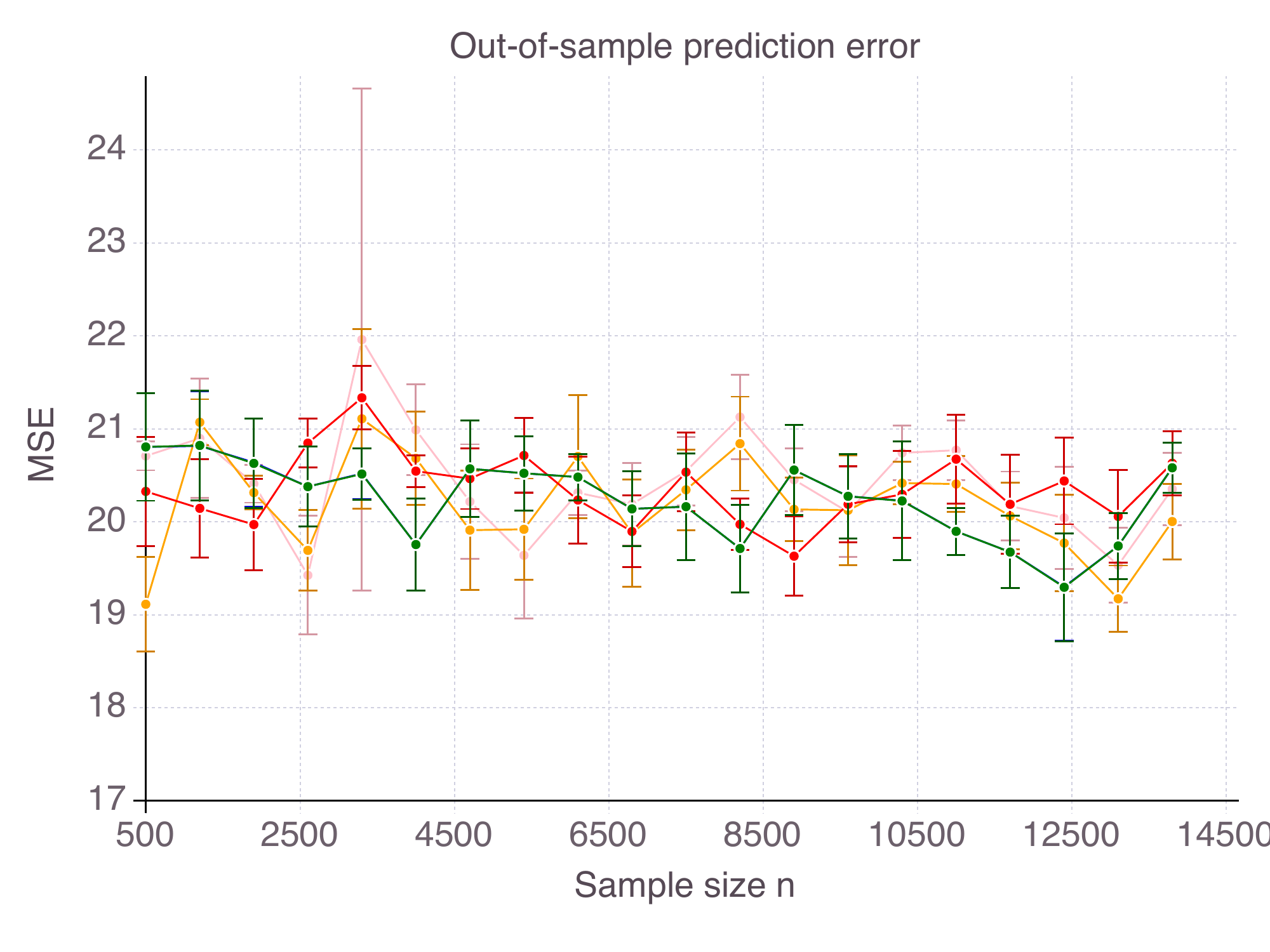}
	\caption{High noise}
\end{subfigure} %

\caption{Accuracy (left panel) and out-of-sample mean square error (right panel) as $n$ increases, for the CIO (in green), SS (in blue with $T_{max}=150$), ENet (in red), MCP (in orange), SCAD (in pink) with OLS loss. We average results over $10$ data sets with $n=500,...,3700$, $p=20,000$, $k_{true}=100$.}
\label{fig:RegHardFixMSE}
\end{figure*}
}

\subsubsection{Feature selection with cross-validated support size} 
Behavior of the methods when $k_{true}$ is no longer given and needs to be cross-validated from the data itself is very similar to the case where $\Sigma$ satisfies the mutual incoherence condition. To avoid redundancies, we report those results in Appendix \ref{sec:regression.supp.nomic.cv}. 
{
\subsection{Real-world design matrix $X$}

To illustrate the implications of feature selection on real-world applications, we consider an example from genomics. We collected data from The Cancer Genome Atlas Research Network\footnote{\url{http://cancergenome.nih.gov}} on $n = 1,145$ lung cancer patients. 
The data set consists of $p=14,858$ gene expression data for each patient. We discarded genes for which information was only partially recorded so there is no missing data. We used this data as our design matrix $X\in \mathbb{R}^{n \times p}$ and generated synthetic noisy outputs $Y$, for 10 uniformely log-spaced values of $SNR$, as in \citet{hastie2017extended} (Table \ref{tab:cancer.regimes} p. \pageref{tab:cancer.regimes}).
\begin{table}[h]
\centering
\caption{Regimes of noise ($SNR$) considered in our regression experiments on the Cancer data set}
\label{tab:cancer.regimes}
\begin{tabular}{r|cccccccccc}
\toprule
$SNR$ & $0.05$ & $0.09$  & $0.14$  & $0.25$  & $0.42$  & $0.71$  & $1.22$  & $2.07$  & $3.52$  & $6$  \\
\midrule
$PVE$ & $0.05$ & $0.08$  & $0.12$  & $0.20$  & $0.30$  & $0.42$  & $0.55$  & $0.67$  & $0.78$  & $0.86$  \\
\bottomrule
\end{tabular}
\end{table}
We held $15\%$ of patients in a test set ($171$ patients). We used the remaining $974$ patients as a training and validation set. For each algorithm, we computed models with various degrees of sparsity and regularization on the training set, evaluated them on the validation set and took the most accurate model. Figure \ref{fig:RegCancer} (p. \pageref{fig:RegCancer}) represents the accuracy and false detection rate of the resulting regressor, for all methods, as $SNR$ increases. ENet ranks the highest both in terms of number of true and false features. On the contrary, MCP is the least accurate, while making fewer incorrect guesses. CIO, SS, SCAD demonstrate in-between performance. Nevertheless, differences in feature selection does not translate into significant differences into predictive power in this case (see Figure \ref{fig:RegCancer2} in Appendix \ref{sec:regression.supp.real} page \pageref{fig:RegCancer2}).
\begin{figure*}[h]
\centering
\begin{subfigure}[t]{.48\linewidth}
	\centering
	\includegraphics[width=\linewidth]{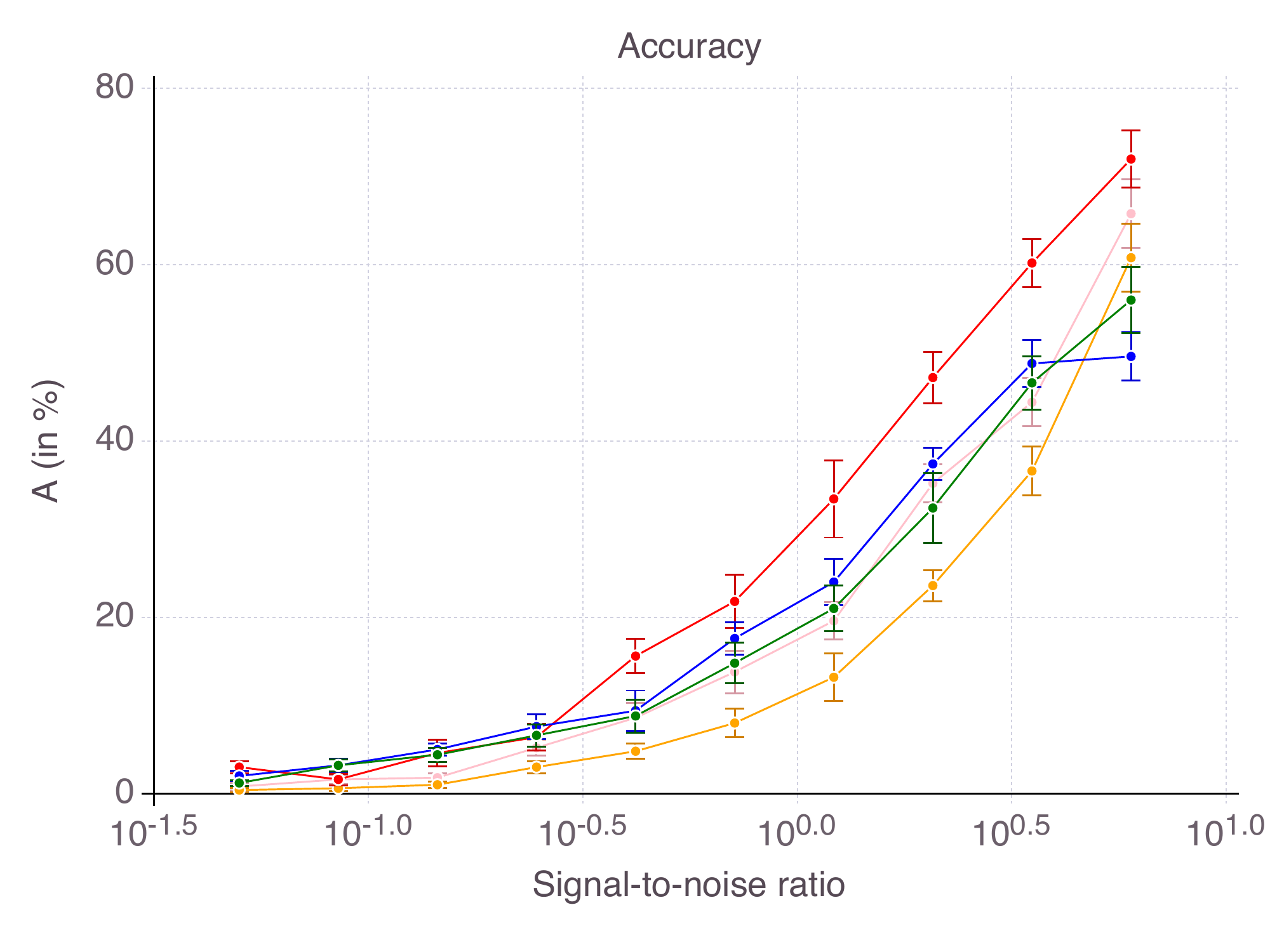}
	\caption{Accuracy}
\end{subfigure}
~
\begin{subfigure}[t]{.48\linewidth}
	\centering
	\includegraphics[width=\linewidth]{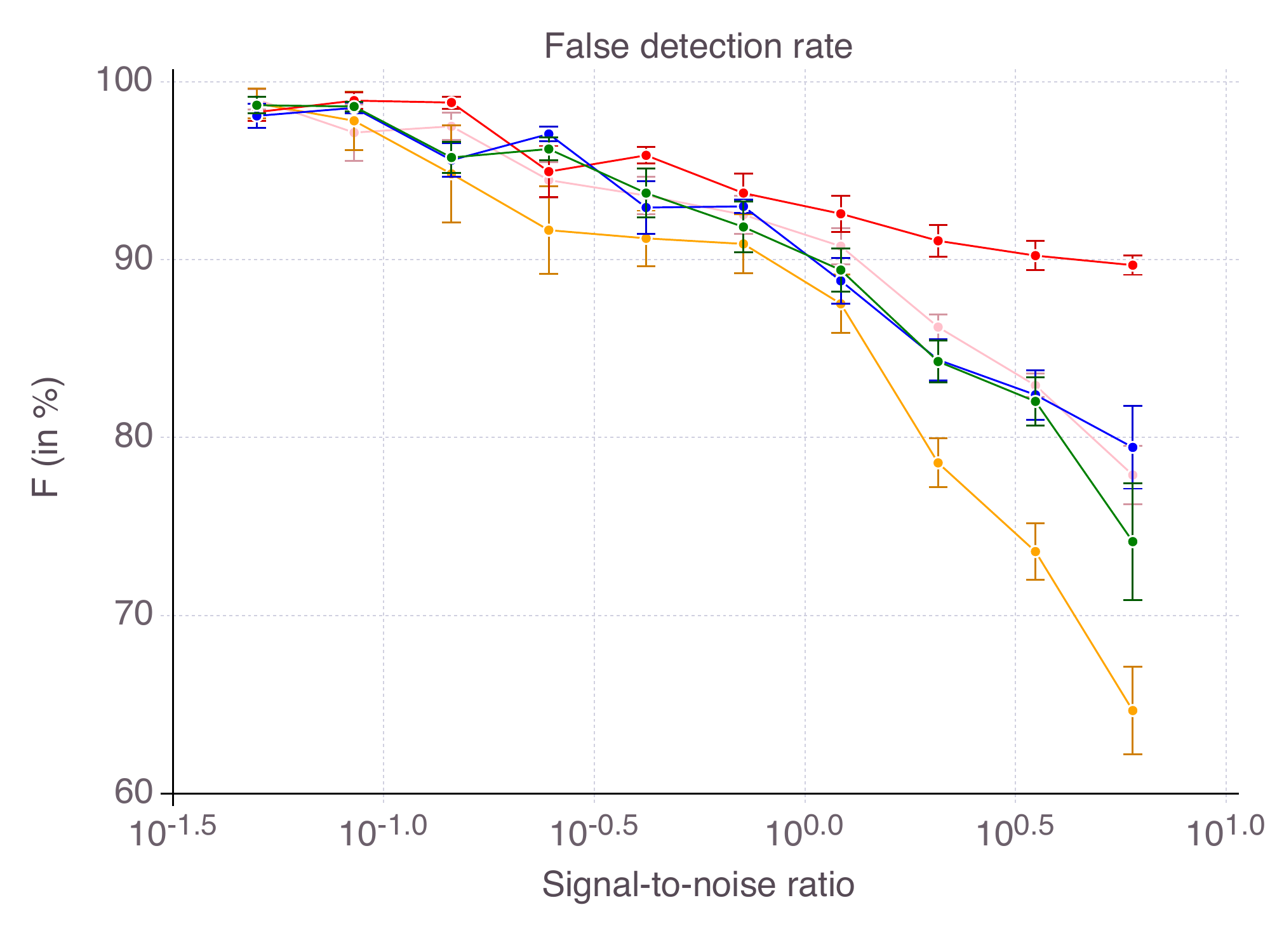}
	\caption{False detection rate}
\end{subfigure} %

\caption{Accuracy and false detection  as $SNR$ increases, for the CIO (in green), SS (in blue with $T_{max}=150$), ENet (in red), MCP (in orange), SCAD (in pink) with OLS loss. We average results over $10$ data sets with $SNR=0.05,...,6$, $k_{true}=50$.}
\label{fig:RegCancer}
\end{figure*}

As previously mentioned, these results are the conclusion of a cross-validation procedure to find the right value of $k$. In Figure \ref{fig:RegCancerROC} (p. \pageref{fig:RegCancerROC}), we represent the ROC curve corresponding to four of the ten regimes of noise. For low noise, ENet is dominated by SCAD, SS, MCP and CIO. As noise increases however, ENet gradually improves and even dominates all methods in very noisy regimes. These ROC curves are of little interest in practice, where true features are unknown - and potentially do not even exist. They raise, in our view, interesting research questions about the cross-validation procedure and its ability to efficiently select the "best" model.  

\begin{figure*}
\centering
\begin{subfigure}[t]{.45\linewidth}
	\centering
	\includegraphics[width=\linewidth]{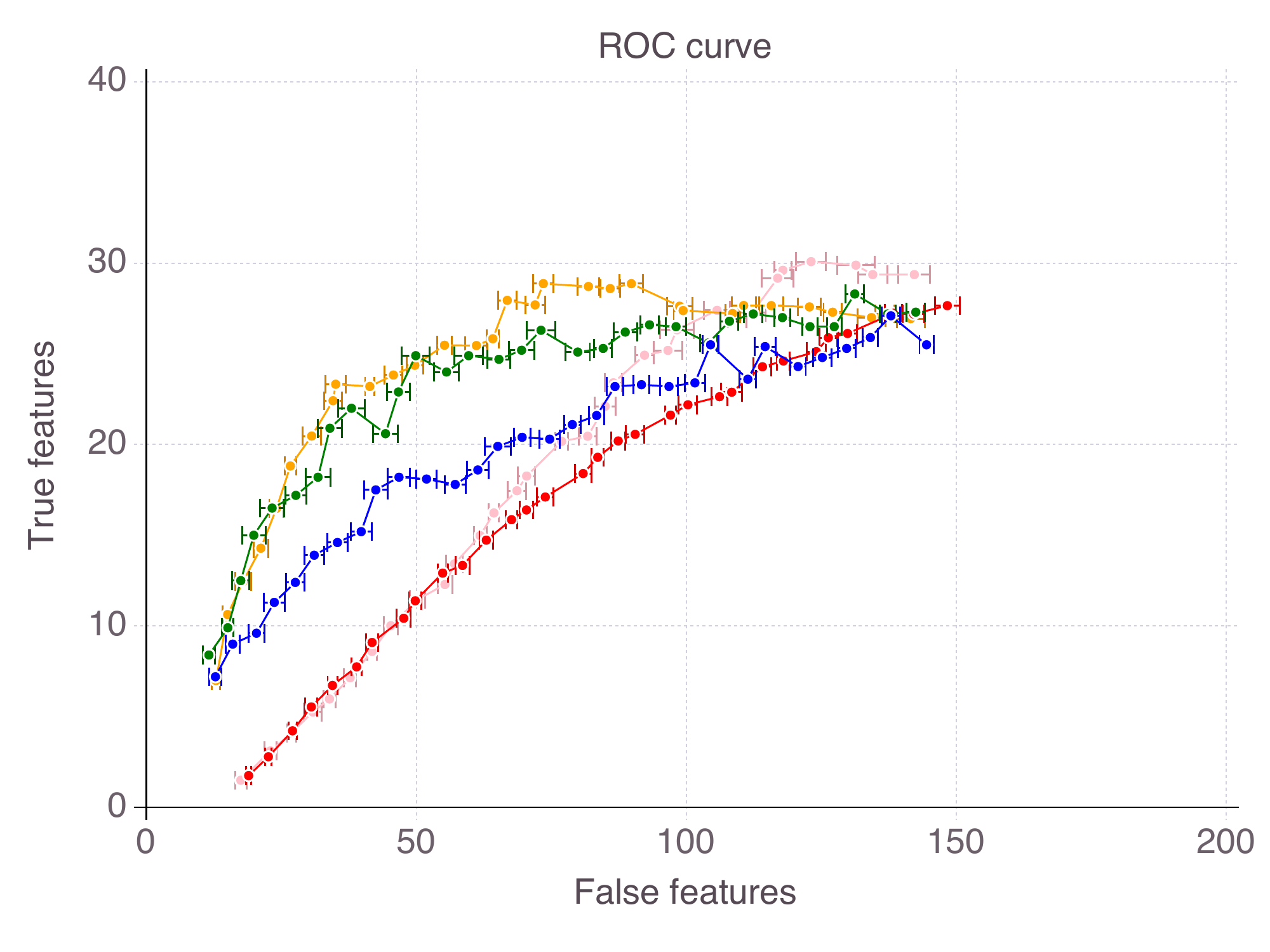}
	\caption{$SNR = 6$}
\end{subfigure} %
~
\begin{subfigure}[t]{.45\linewidth}
	\centering
	\includegraphics[width=\linewidth]{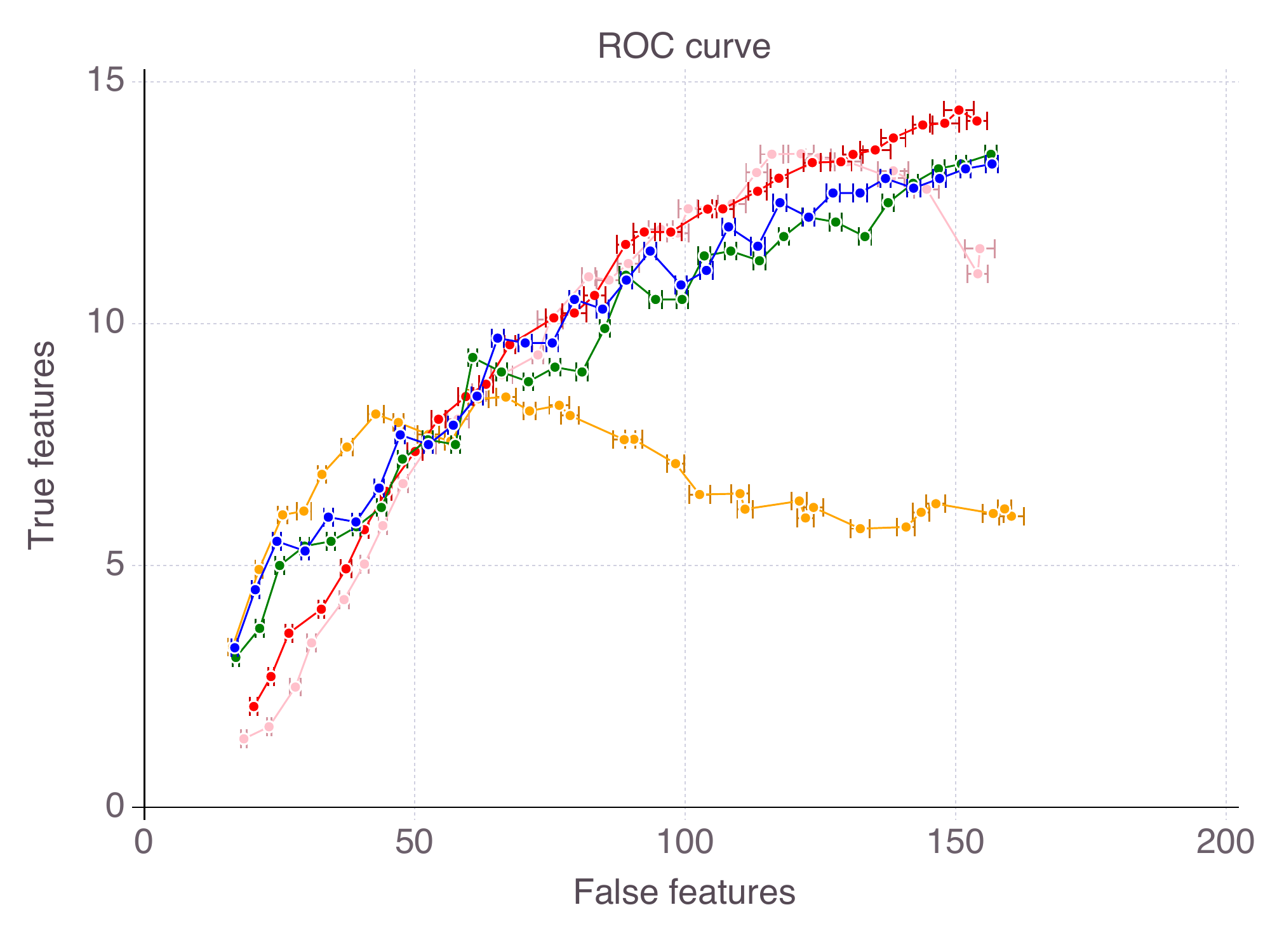}
	\caption{$SNR=1.22$}
\end{subfigure}

\begin{subfigure}[t]{.45\linewidth}
	\centering
	\includegraphics[width=\linewidth]{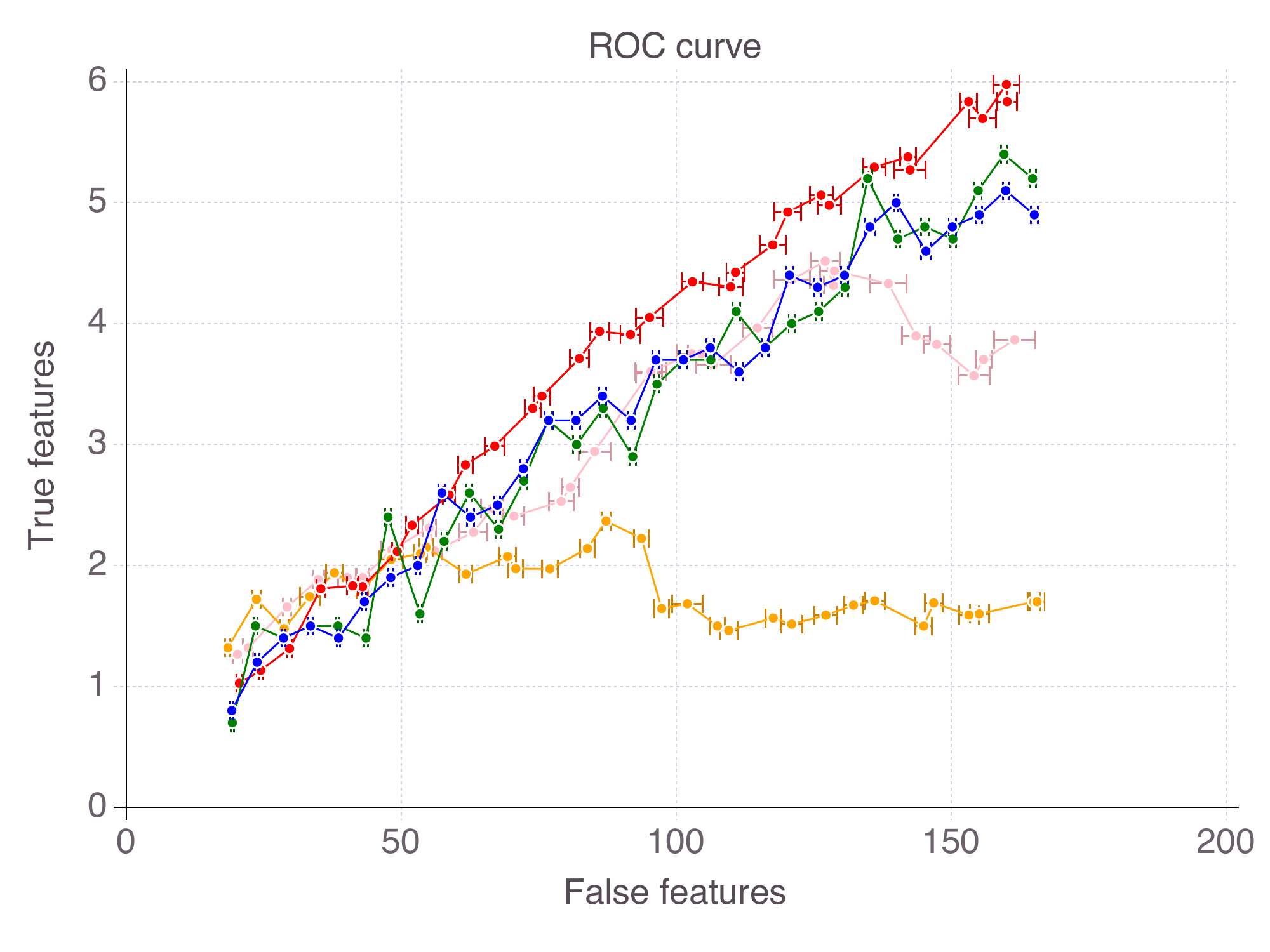}
	\caption{$SNR=0.25$}
\end{subfigure} %
~
\begin{subfigure}[t]{.45\linewidth}
	\centering
	\includegraphics[width=\linewidth]{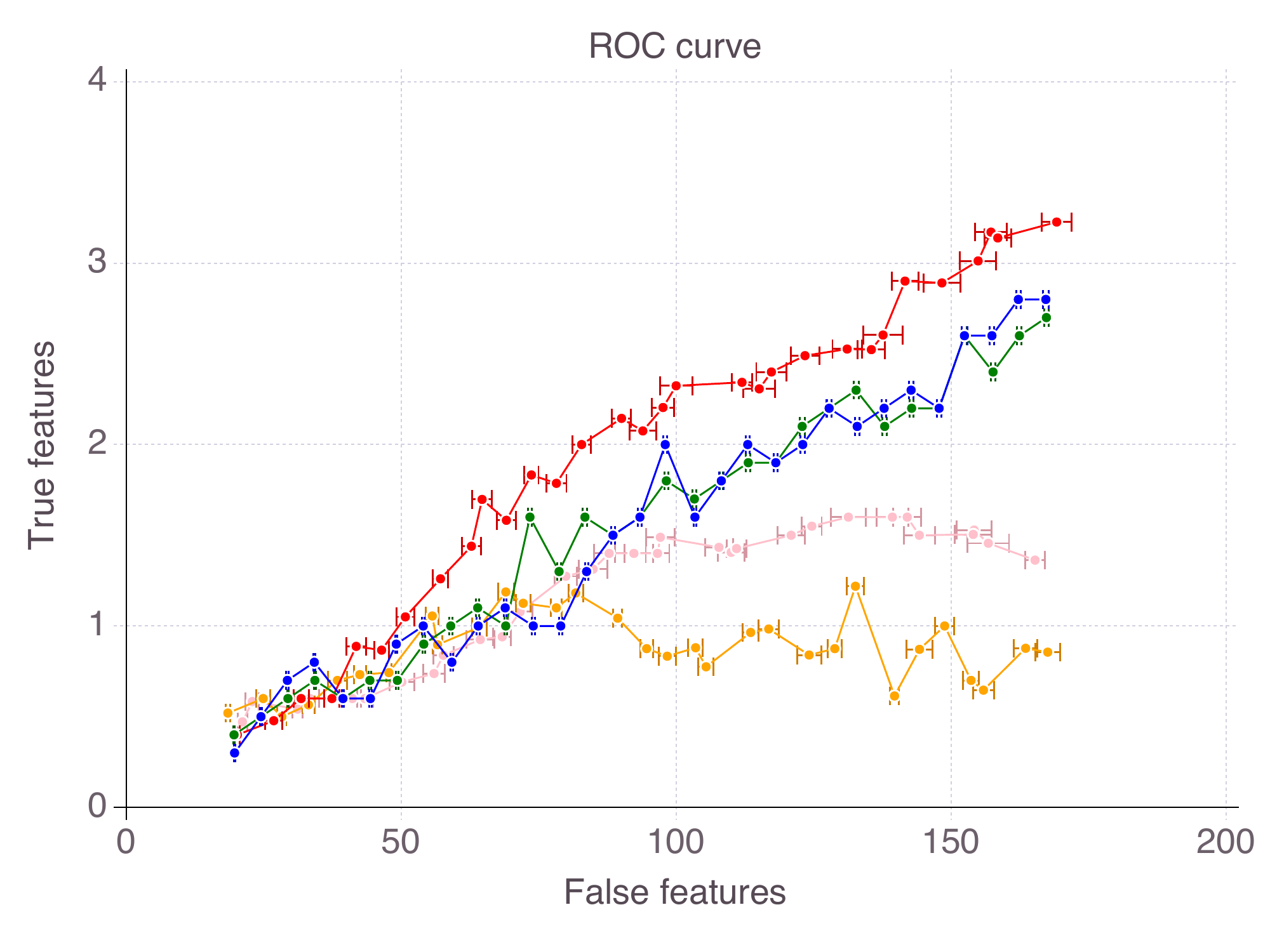}
	\caption{$SNR=0.05$}
\end{subfigure}
\caption{True features against false features, for the CIO (in green), SS (in blue with $T_{max}=150$), ENet (in red), MCP (in orange), SCAD (in pink) with OLS loss. We average results over $10$ data sets.}
\label{fig:RegCancerROC}
\end{figure*}

}

{ 
\subsection{Summary and guidelines}
In this section, we compared five feature selection methods in regression, in various regimes of noise and correlation, and under different design matrices. Based on those extensive experiments, we can make the following observations:
\begin{itemize}
\item As far as accuracy is concerned, non-convex methods should be preferred over $\ell_1$-regularization for they provide better feature selection, even in the absence of the mutual incoherence condition. In particular, MCP, the cutting-plane algorithm for the cardinality-constrained formulation, and its Boolean relaxation have been particularly effective in our experiments.
\item In terms of false detection rate, cardinality-constrained formulations improve substantially over ENet and SCAD{\blue, and  moderately over MCP}. 
\item Computational time might still be the limiting factor in the use of such methods in practice. To that matter, publicly available software, such as the \verb|ncvreg| package for SCAD and MCP estimators and our package \verb|SubsetSelection| for the Boolean relaxation, should be advertised to practitioners since they compete with \verb|glmnet|, which remains the gold standard for tractability. For time can be a crucial bottleneck in practice, we provide detailed experiments regarding computational time and scalability of the algorithms in Appendix \ref{sec:regression.supp.time}. 
\item In practice, we should recommend using a combination of all these methods: Lasso or ENet can be used as first feature screening/dimension reduction step, to be followed by a more computationally expensive non-convex feature selection method if time permits. 
\item While Lasso/ENet performs poorly in low noise settings, its competes and sometimes dominates other methods as noise increases. This observation supports the view that $\ell_1$-regularization is, first and foremost, a robustness story \citep{bertsimas2009equivalence,xu2009robustness}: Through shrinkage of the coefficients, the $\ell_1$ penalty reduces variance in the estimator and improves out-of-sample accuracy, especially in presence of noise. Experiments by \citet{hastie2017extended} even suggested that Lasso outperforms cardinality-constrained estimators in high noise regimes. Our experiments suggest that their observations are still valid but less obvious as soon as the best subset selection estimator is regularized as well (with an $\ell_2$ penalty in our case).
\end{itemize} }

\section{Synthetic and real-world classification problems}
\label{sec:classification}
In this section, we compare the five methods included in our study on classification problems. For implementation considerations, we use CIO and SS with the Hinge loss and ENet, MCP and SCAD with the logistic loss.

\subsection{Methodology and metrics}
Synthetic data is generated according to the same methodology as for regression, except that we now compute the signal $Y$ according to $$Y = \text{sign}(X w_{true} + \varepsilon),$$ instead of $Y = X w_{true} + \varepsilon$ previously. 

On synthetic data, feature selection is assessed in terms of accuracy $A$ and false detection rate $FDR$ as in the previous section. Prediction accuracy, on the other hand, is assessed in terms of Area Under the Curve ($AUC$). {
The $AUC$ corresponds to the area under the receiver operating characteristic curve, which represents true positive rate against false positive rate. The $AUC$ ranges from $0.5$ (for a completely random classifier) to $1$. This area also corresponds to the probability that a randomly chosen positive example is correctly ranked with higher suspicion than a randomly chosen negative example.} Correspondingly, $1-AUC$ is a common measure of prediction error for real-world data. 

\subsection{Synthetic data satisfying mutual incoherence condition}
In this section, we consider consider Toeplitz covariance matrix $\Sigma = \left( \rho^{|i-j|} \right)_{i,j}$, which satisfy mutual incoherence condition. We compare the performance of the methods in six different regimes of noise and correlation described in Table \ref{tab:class.mic.regimes} (p. \pageref{tab:class.mic.regimes}).
\begin{table}[h]
\centering
\caption{Regimes of noise ($SNR$) and correlation ($\rho$) considered in our experiments on regression with Toeplitz covariance matrix}
\label{tab:class.mic.regimes}
\begin{tabular}{lc|c}
\toprule
 &Low correlation &High correlation \\
\midrule
Low noise &  {$\begin{aligned} &\rho = 0.2 \\ &{SNR}=6  \\ & p =10,000 \\ &k=100\end{aligned} $}& {$\begin{aligned} &\rho = 0.7 \\ &{SNR}=6 \\ & p =10,000 \\ &k=100 \end{aligned} $}\\
\midrule
Medium noise & {$\begin{aligned} &\rho = 0.2 \\ &{SNR}=1 \\ & p =5,000 \\ &k=50 \end{aligned} $} & {$\begin{aligned} &\rho = 0.7 \\ &{SNR}=1 \\ & p =5,000 \\ &k=50 \end{aligned} $}\\
\midrule
High noise & {$\begin{aligned} &\rho = 0.2 \\ &{SNR}=0.05 \\ & p =1,000 \\ &k=10 \end{aligned} $} & {$\begin{aligned} &\rho = 0.7 \\ &{SNR}=0.05 \\ & p =1,000 \\ &k=10 \end{aligned} $}\\
\bottomrule
\end{tabular}
\end{table}

\subsubsection{Feature selection with a given support size}
We first conducted experiments where the cardinality $k$ of the support to be returned is given and equal to the true sparsity $k_{true}$ for all methods. We report the results in Appendix \ref{sec:classification.supp.mic.fix}.

\subsubsection{Feature selection with cross-validated support size}
We now compare the methods on cases where the support size needs to be cross-validated from data. 

For every $n$, each method selects $k^\star$ features, some of which are in the true support, others being irrelevant, as measured by accuracy and false detection rate respectively. Figures \ref{fig:ClassCV.A} (p.  \pageref{fig:ClassCV.A}) and \ref{fig:ClassCV.FDR} (p. \pageref{fig:ClassCV.FDR}) report the results of the cross-validation procedure for increasing $n$. In terms of accuracy (Figure \ref{fig:ClassCV.A}), all methods increase in accuracy as $n$ increases, although CIO and SS converge significantly slower than ENet, MCP and SCAD. However, this lower accuracy comes with the benefit of a strictly lower false detection rate (Figure \ref{fig:ClassCV.FDR}). 

\begin{figure*}[p]
\centering
\begin{subfigure}[h]{.45\linewidth}
	\centering
	\includegraphics[width=\linewidth]{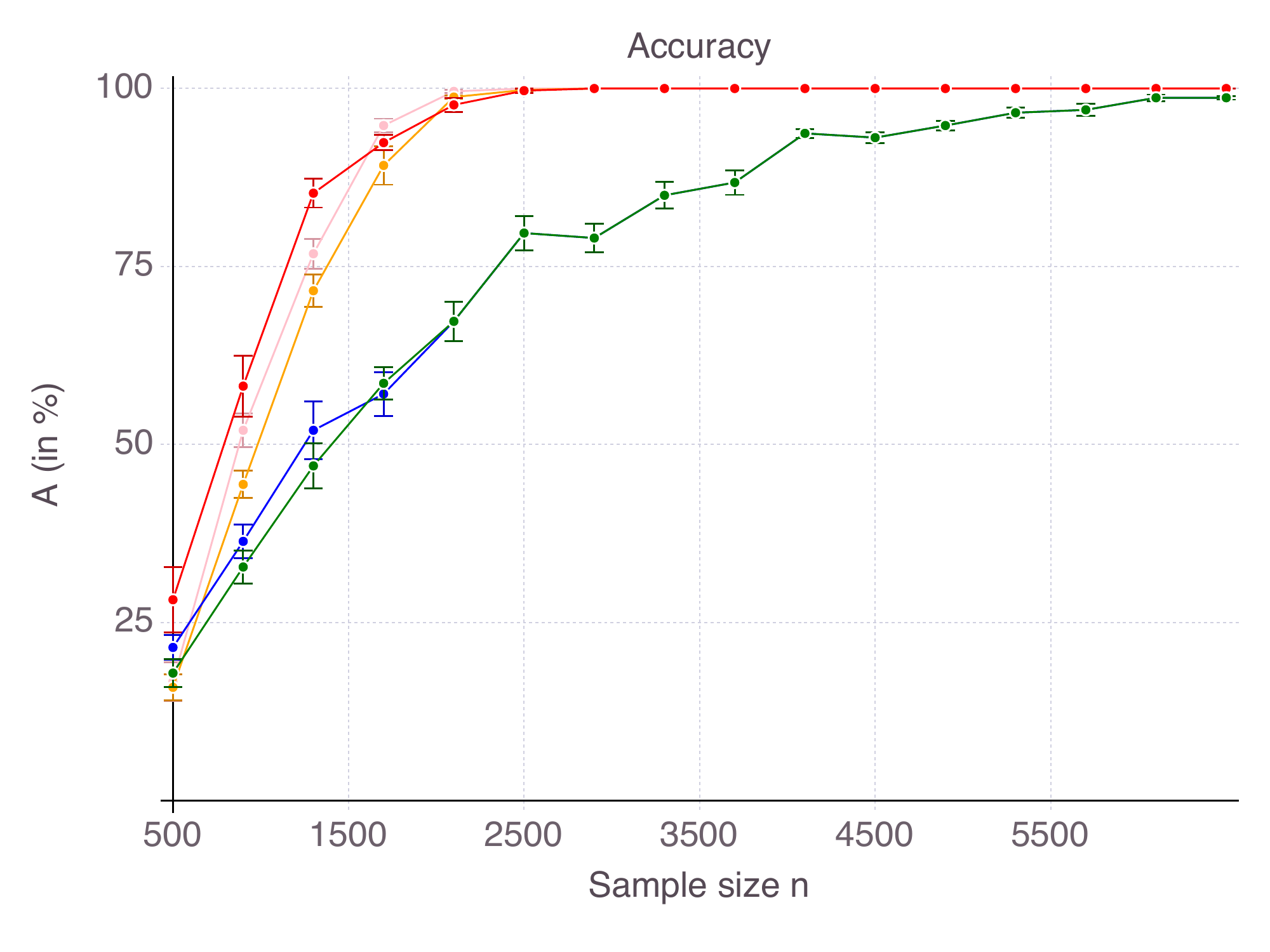}
	\caption{Low noise, low correlation}
\end{subfigure} %
~
\begin{subfigure}[h]{.45\linewidth}
	\centering
	\includegraphics[width=\linewidth]{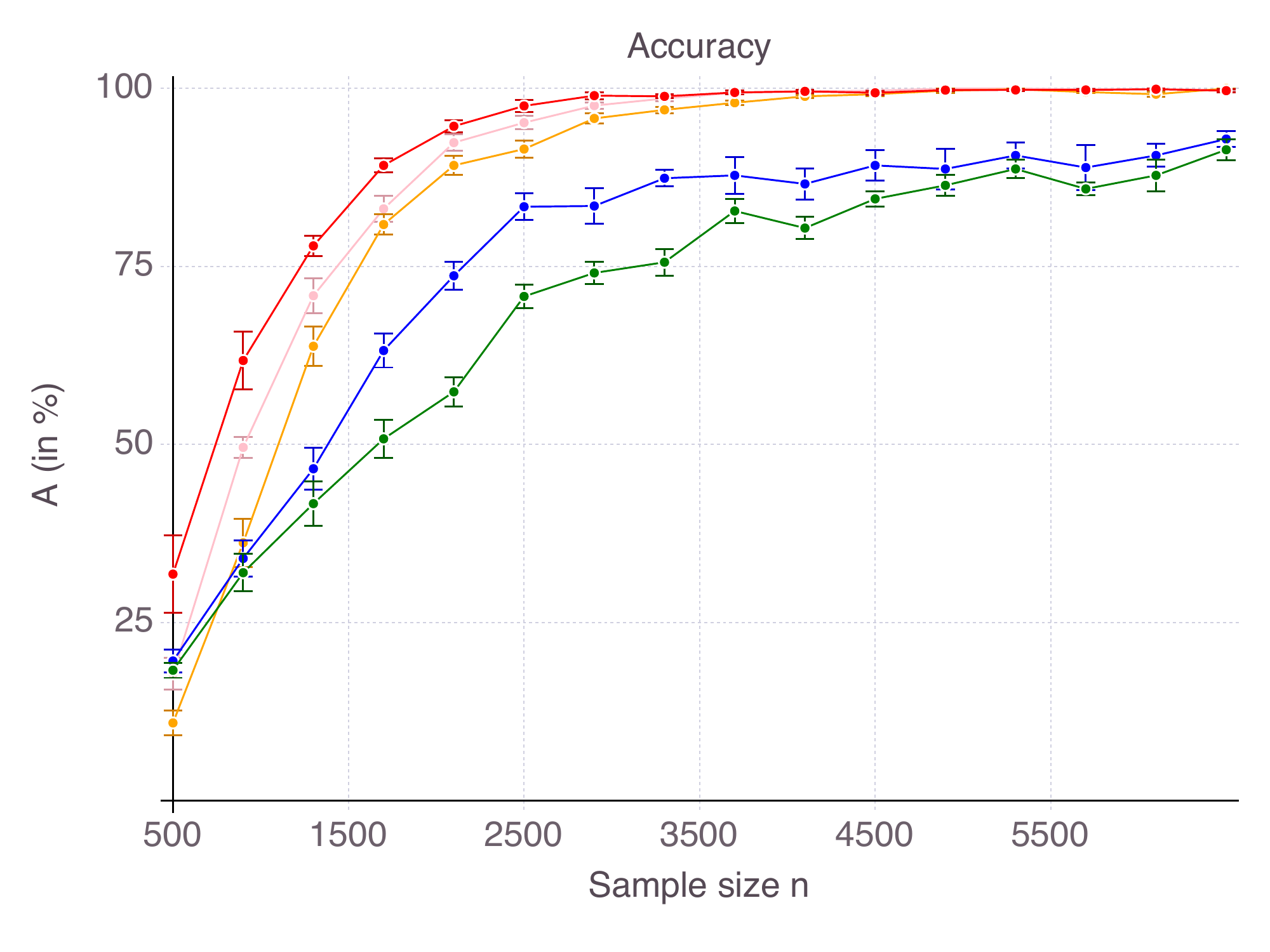}
	\caption{Low noise, high correlation}
\end{subfigure}

\begin{subfigure}[h]{.45\linewidth}
	\centering
	\includegraphics[width=\linewidth]{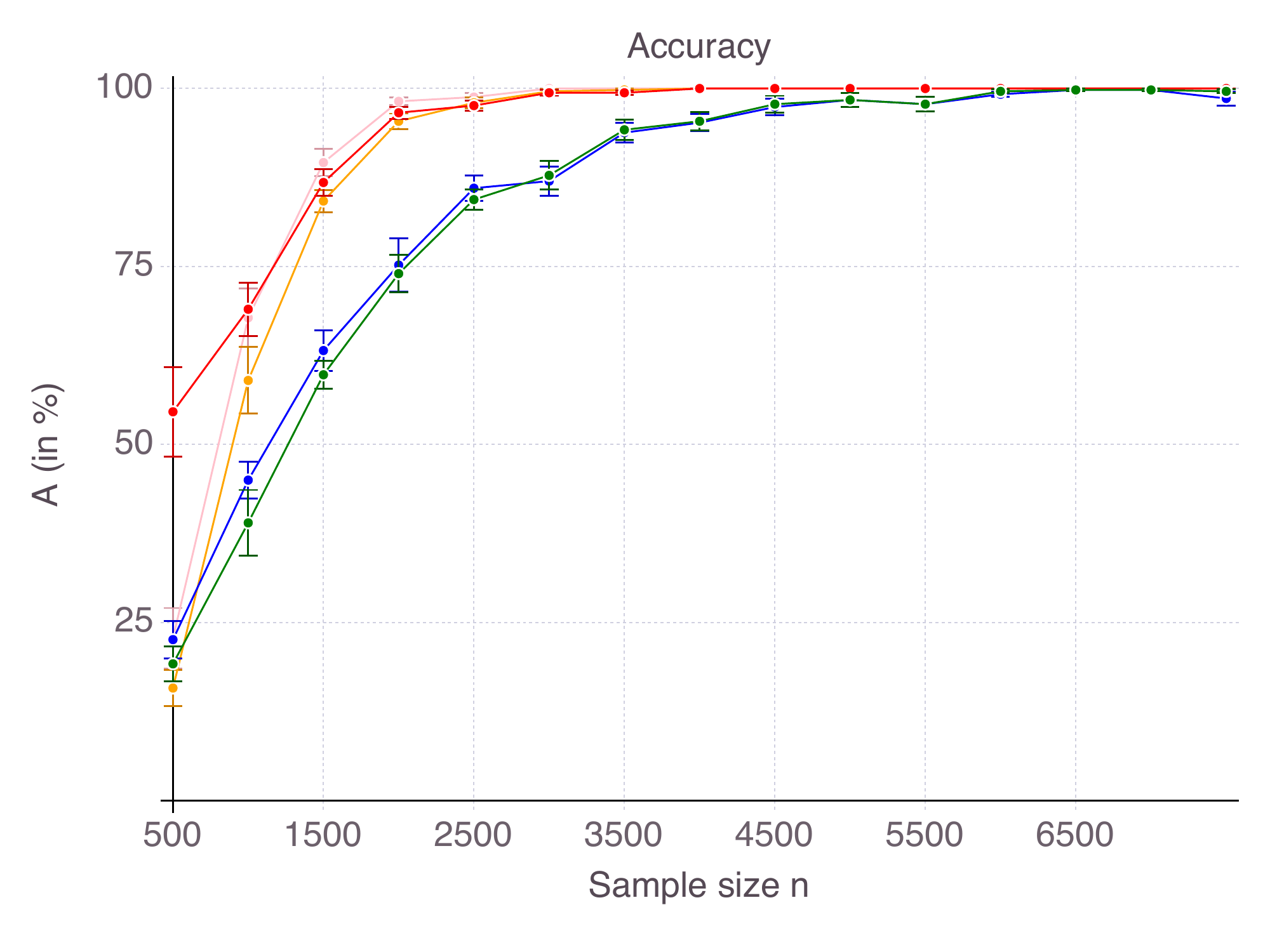}
	\caption{Medium noise, low correlation}
\end{subfigure} %
~
\begin{subfigure}[h]{.45\linewidth}
	\centering
	\includegraphics[width=\linewidth]{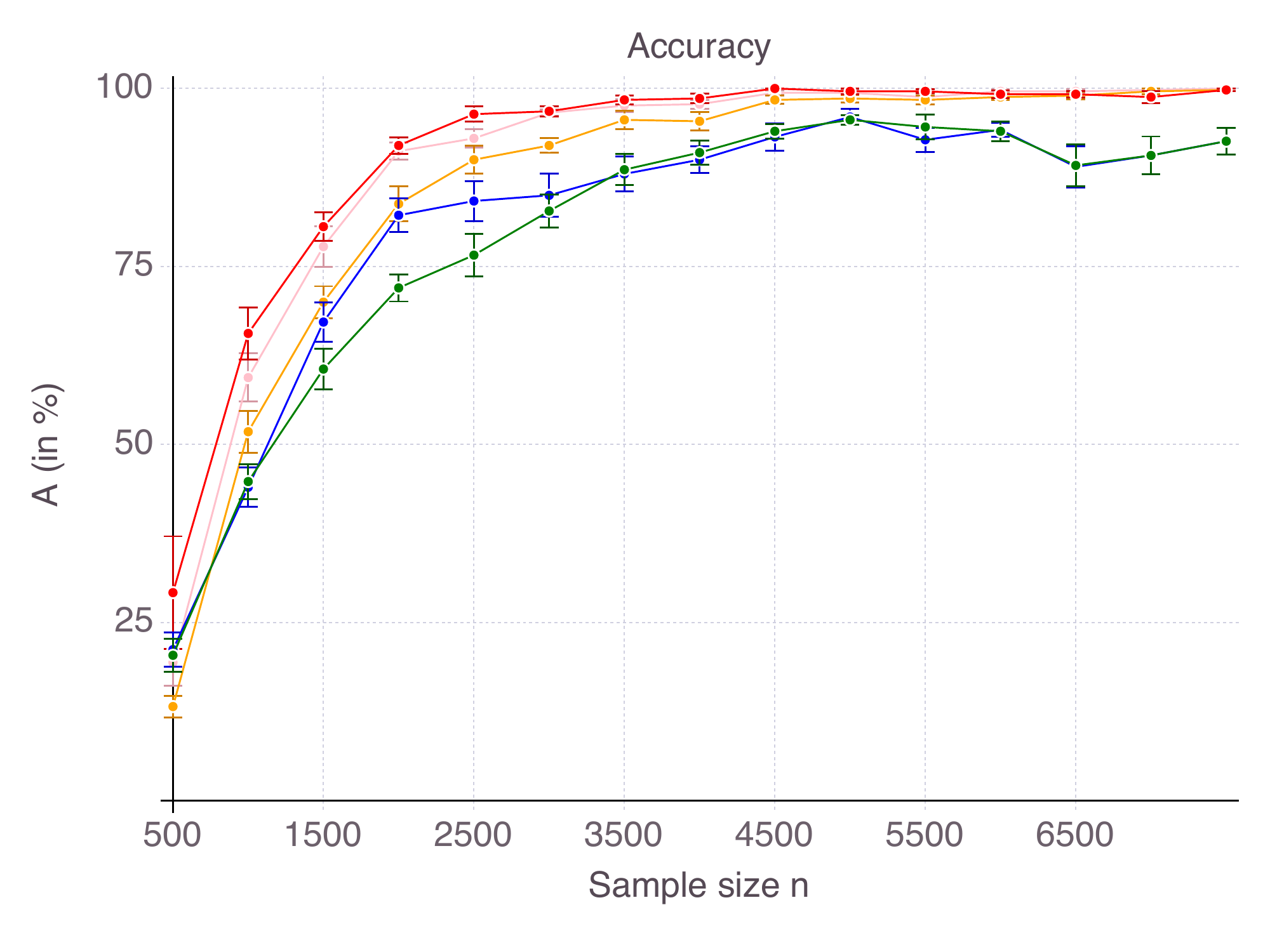}
	\caption{Medium noise, high correlation}
\end{subfigure}

\begin{subfigure}[h]{.45\linewidth}
	\centering
	\includegraphics[width=\linewidth]{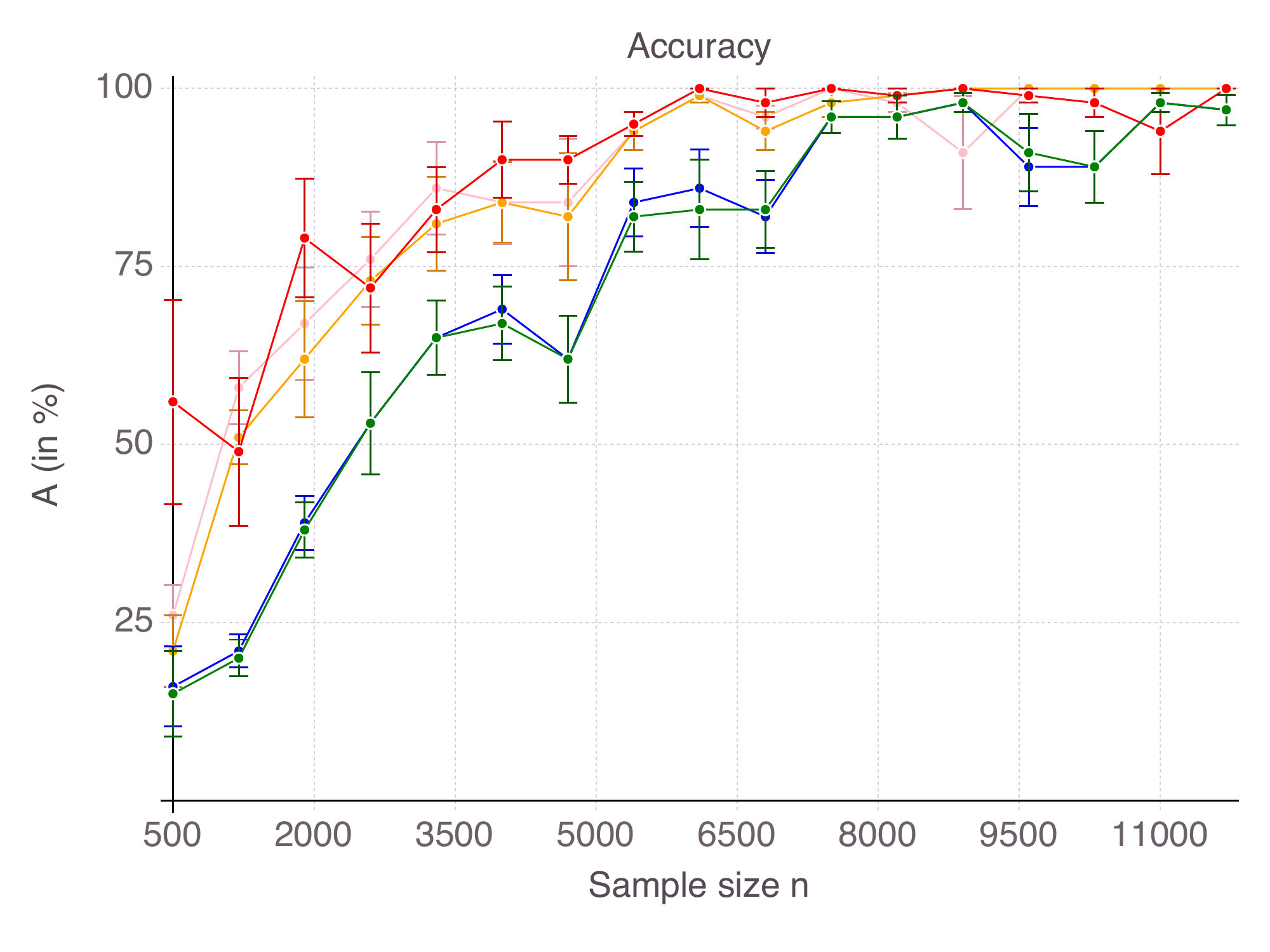}
	\caption{High noise, low correlation}
\end{subfigure} %
~
\begin{subfigure}[h]{.45\linewidth}
	\centering
	\includegraphics[width=\linewidth]{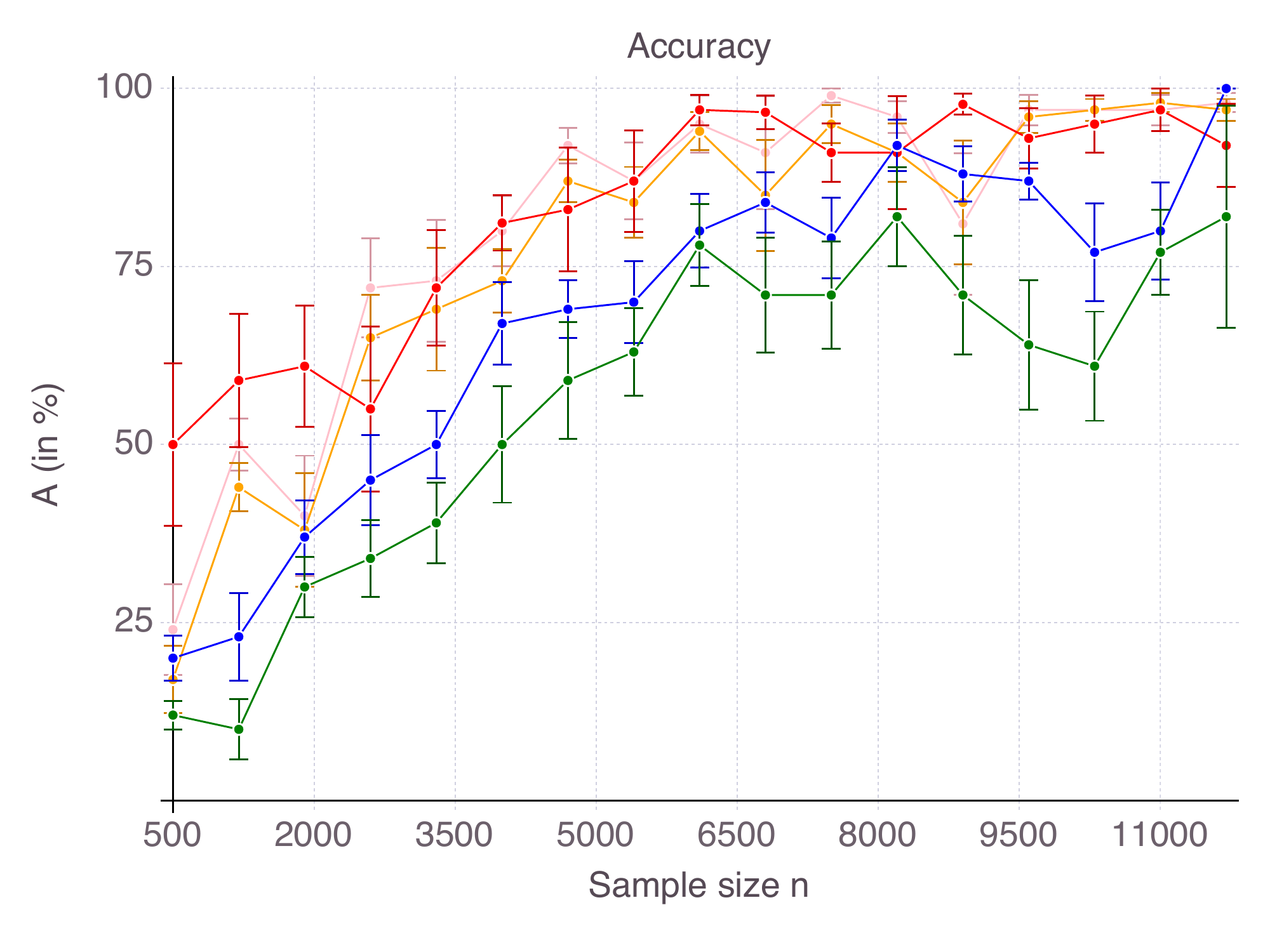}
	\caption{High noise, high correlation}
\end{subfigure}
\caption{Accuracy $A$ as $n$ increases, for the CIO (in green), SS (in blue with $T_{max}=200$) with Hinge loss, ENet (in red), MCP (in orange), SCAD (in pink) with logistic loss. We average results over $10$ data sets.}
\label{fig:ClassCV.A}
\end{figure*}
\begin{figure*}[p]
\centering
\begin{subfigure}[h]{.45\linewidth}
	\centering
	\includegraphics[width=\linewidth]{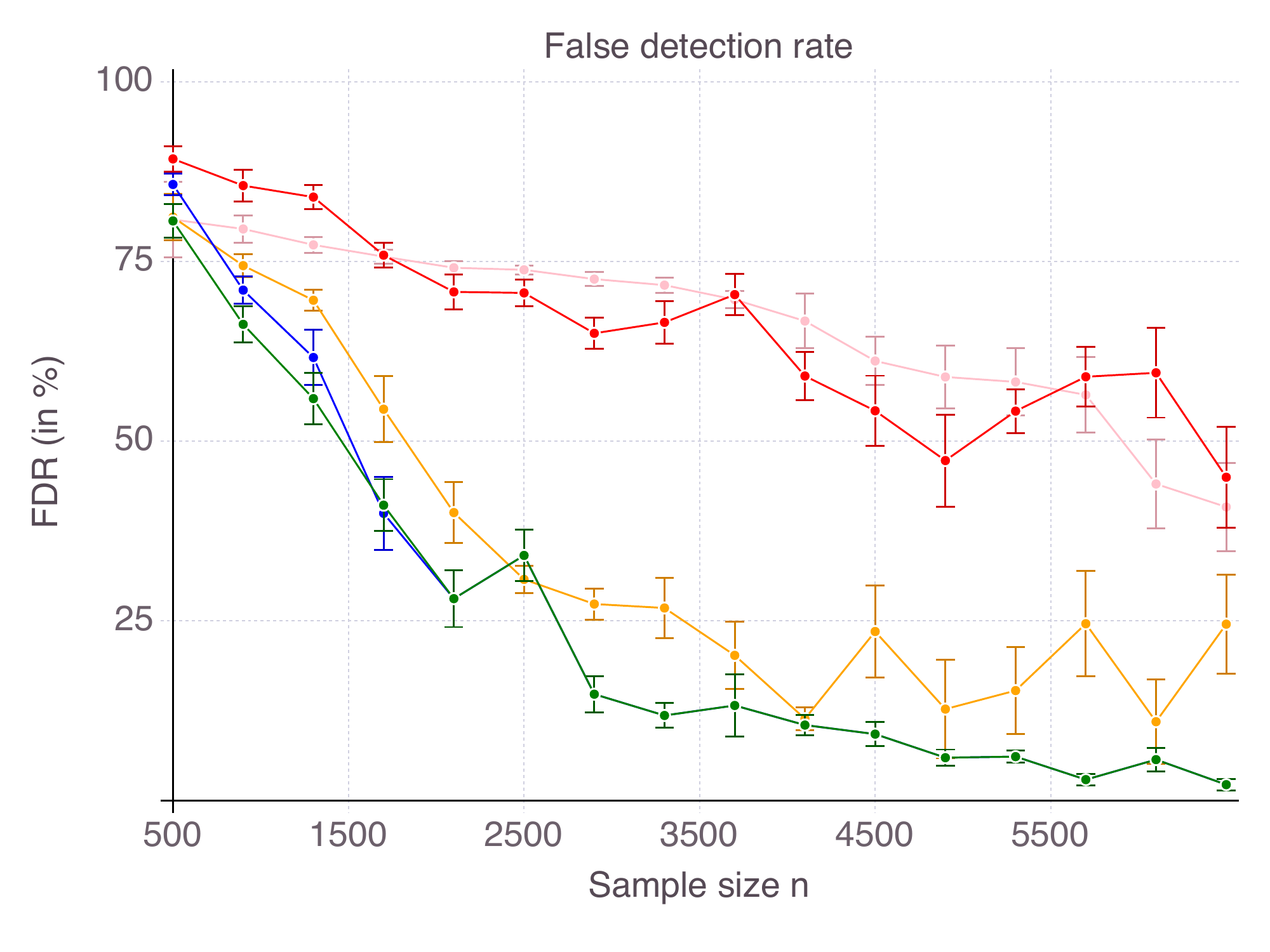}
	\caption{Low noise, low correlation}
\end{subfigure} %
~
\begin{subfigure}[h]{.45\linewidth}
	\centering
	\includegraphics[width=\linewidth]{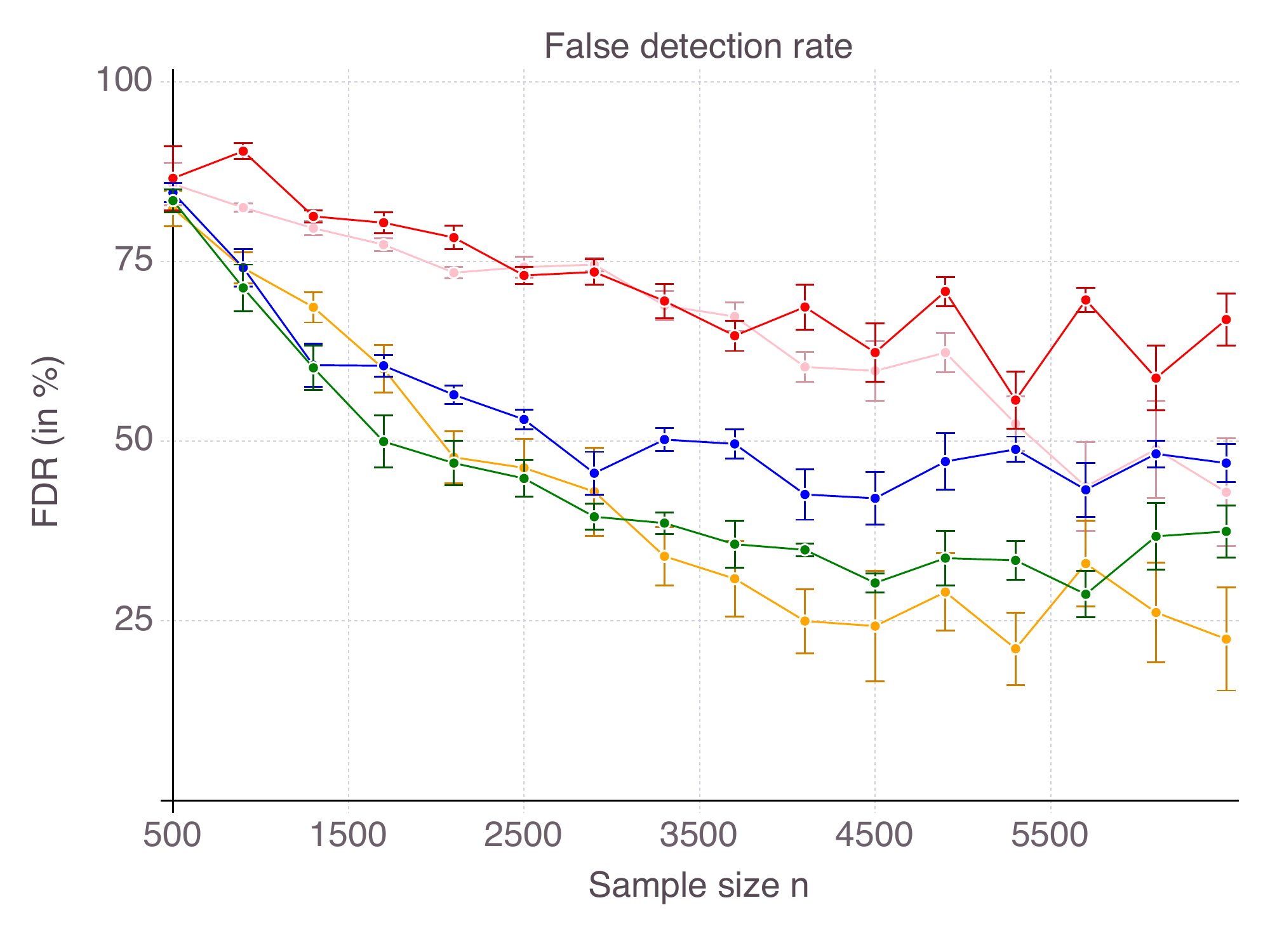}
	\caption{Low noise, high correlation}
\end{subfigure}

\begin{subfigure}[h]{.45\linewidth}
	\centering
	\includegraphics[width=\linewidth]{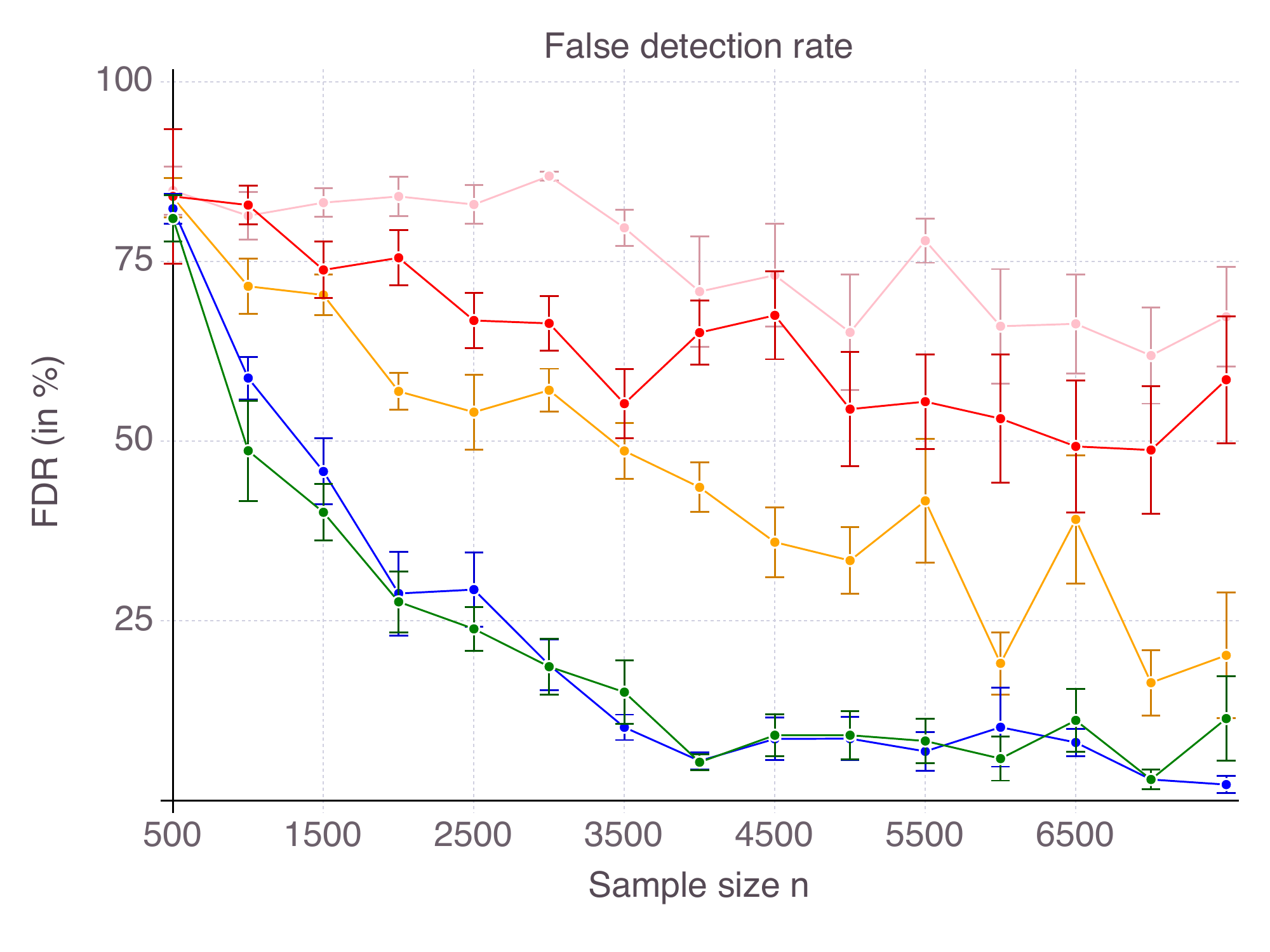}
	\caption{Medium noise, low correlation}
\end{subfigure} %
~
\begin{subfigure}[h]{.45\linewidth}
	\centering
	\includegraphics[width=\linewidth]{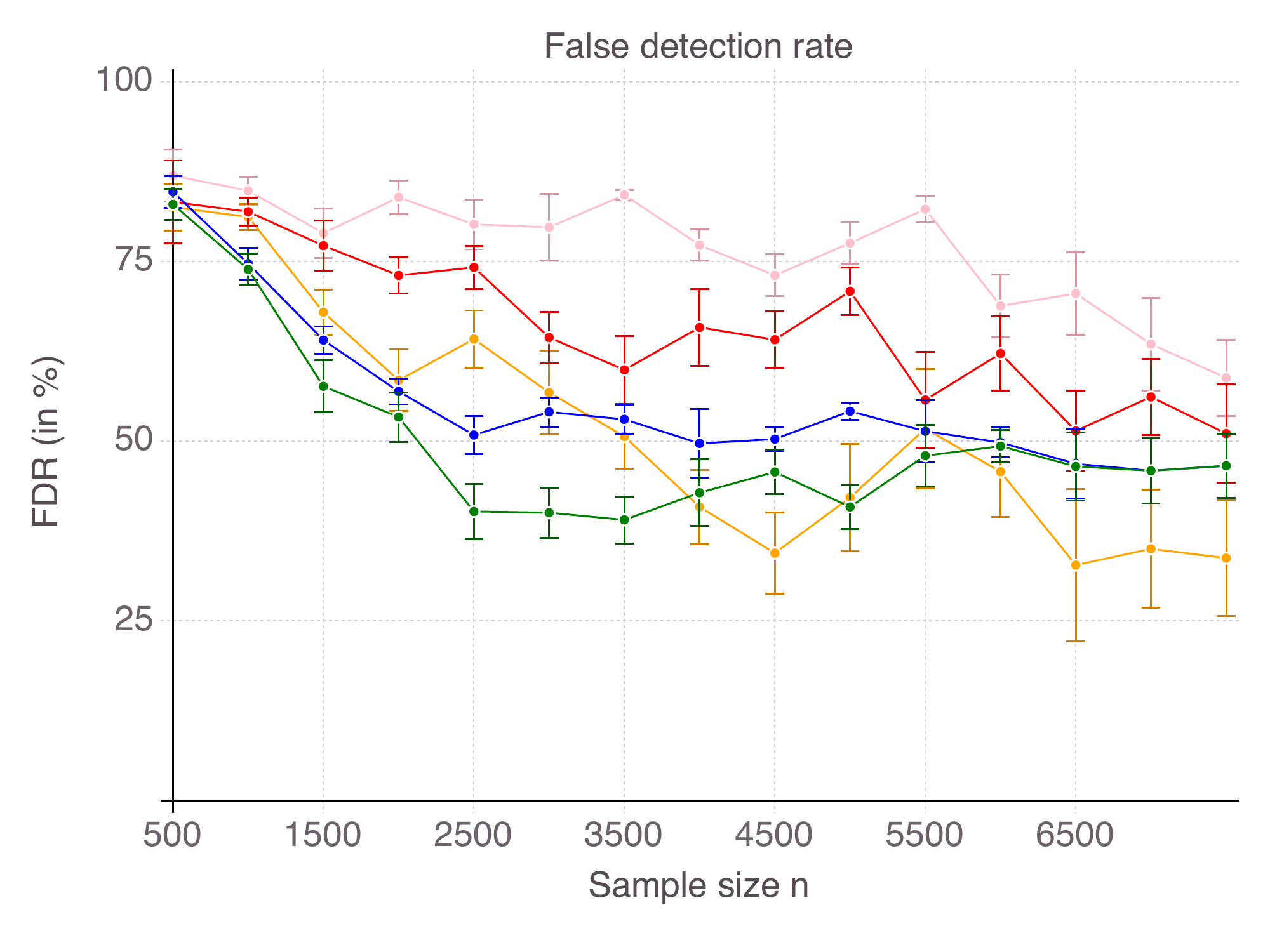}
	\caption{Medium noise, high correlation}
\end{subfigure}

\begin{subfigure}[h]{.45\linewidth}
	\centering
	\includegraphics[width=\linewidth]{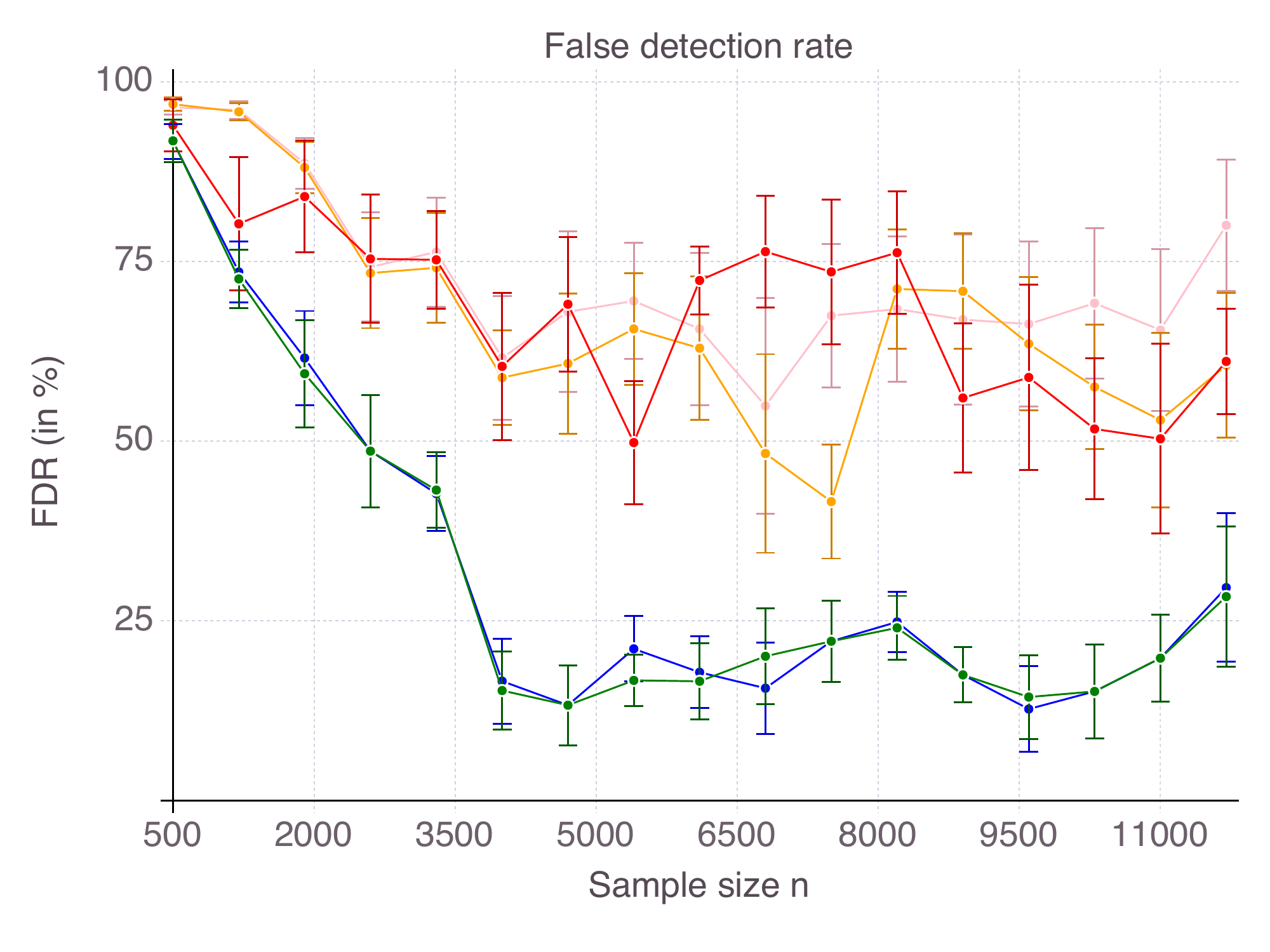}
	\caption{High noise, low correlation}
\end{subfigure} %
~
\begin{subfigure}[h]{.45\linewidth}
	\centering
	\includegraphics[width=\linewidth]{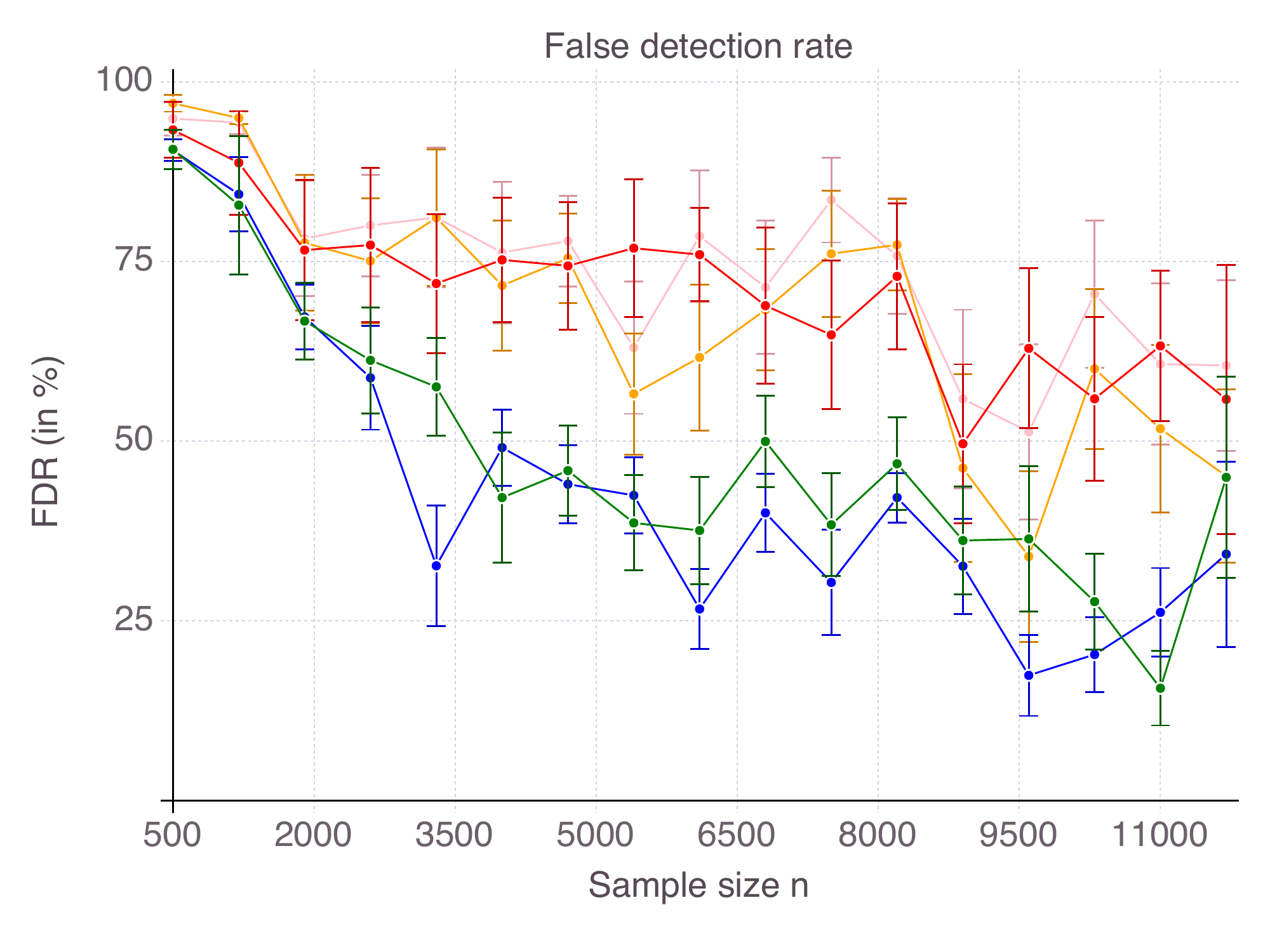}
	\caption{High noise, high correlation}
\end{subfigure}
\caption{False detection rate $FDR$ as $n$ increases, for the CIO (in green), SS (in blue with $T_{max}=200$) with Hinge loss, ENet (in red), MCP (in orange), SCAD (in pink) with logistic loss. We average results over $10$ data sets.}
\label{fig:ClassCV.FDR}
\end{figure*}

{ 
\subsection{Synthetic data \emph{not} satisfying mutual incoherence condition}
As for regression, we now consider the covariance matrix that does not satisfy mutual incoherence \citep{loh2017support}, in three regimes of noise (see Table \ref{tab:class.m1.regimes} p. \pageref{tab:class.m1.regimes}). 
\begin{table}[h]
\centering
\caption{Regimes of noise ($SNR$) considered in our experiments on regression}
\label{tab:class.m1.regimes}
\begin{tabular}{lc}
\toprule
Low noise &  {$\begin{aligned}  &{SNR}=6  \\ & p =10,000 \\ &k=100\end{aligned} $} \\
\midrule
Medium noise & {$\begin{aligned} &{SNR}=1 \\ & p =5,000 \\ &k=50 \end{aligned} $} \\
\midrule
High noise & {$\begin{aligned} &{SNR}=0.05 \\ & p =1,000 \\ &k=10 \end{aligned} $} \\
\bottomrule
\end{tabular}
\end{table}
}
{
We consider the case when the cardinality $k$ of the support to be returned is given and equal to the true sparsity $k_{true}$. 

Results are shown on Figure \ref{fig:ClassHardFix} (p. \pageref{fig:ClassHardFix}) and corroborate our previous observations in the case of regression: the accuracy of ENet reaches a threshold strictly lower than $1$. Non-convex penalties MCP and SCAD, on the other hand, will see their accuracy converging to $1$, yet for a fixed $n$ there are not necessarily more accurate than ENet. Cardinality-constrained estimators CIO and SS dominate all other methods, with a clear edge for CIO. }

\begin{figure*}
\centering
\begin{subfigure}[t]{\linewidth}
	\centering
	\includegraphics[width=.45\linewidth]{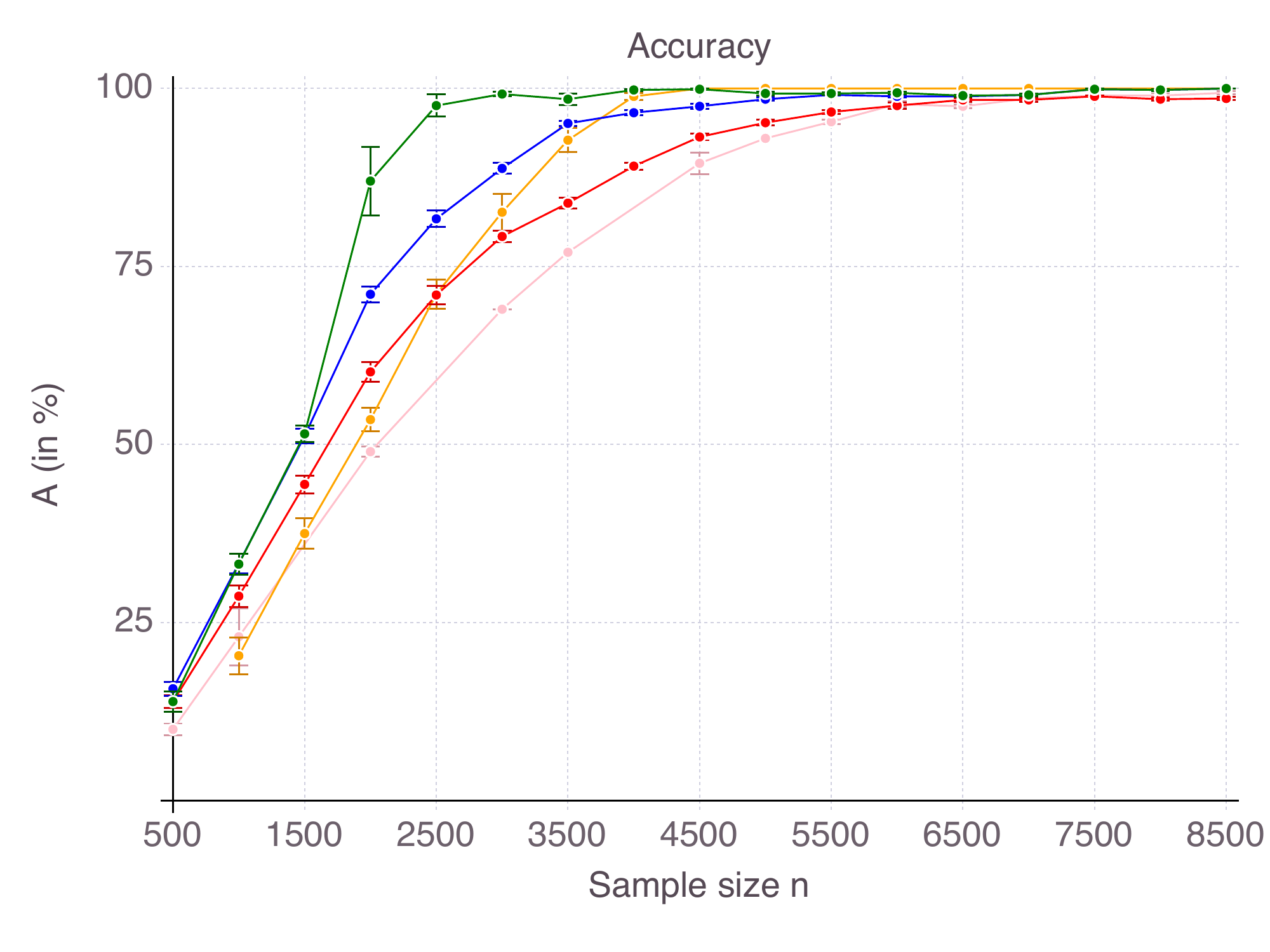}
	\includegraphics[width=.45\linewidth]{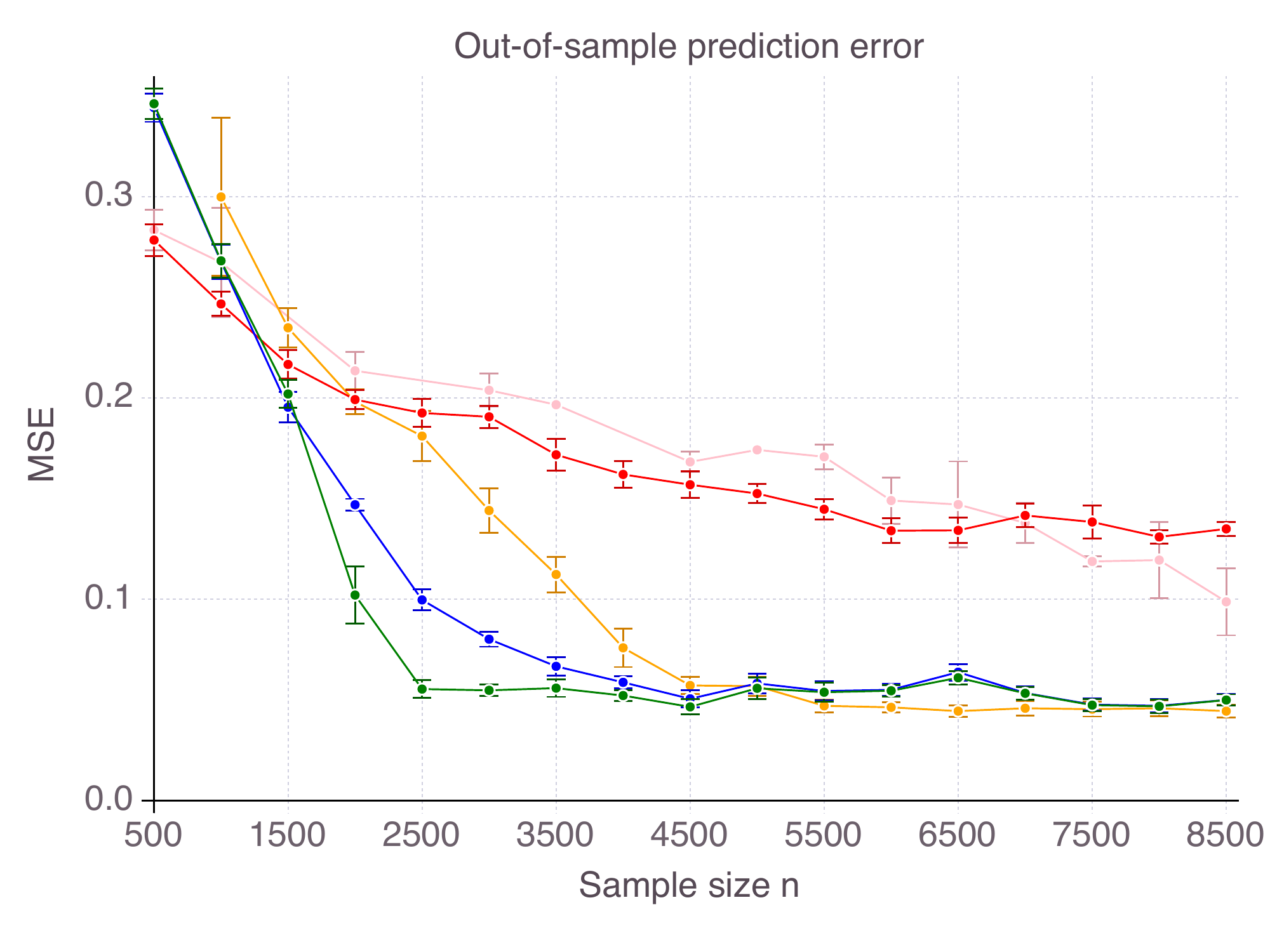}
	\caption{Low noise}
\end{subfigure} %

\begin{subfigure}[t]{\linewidth}
	\centering
	\includegraphics[width=.45\linewidth]{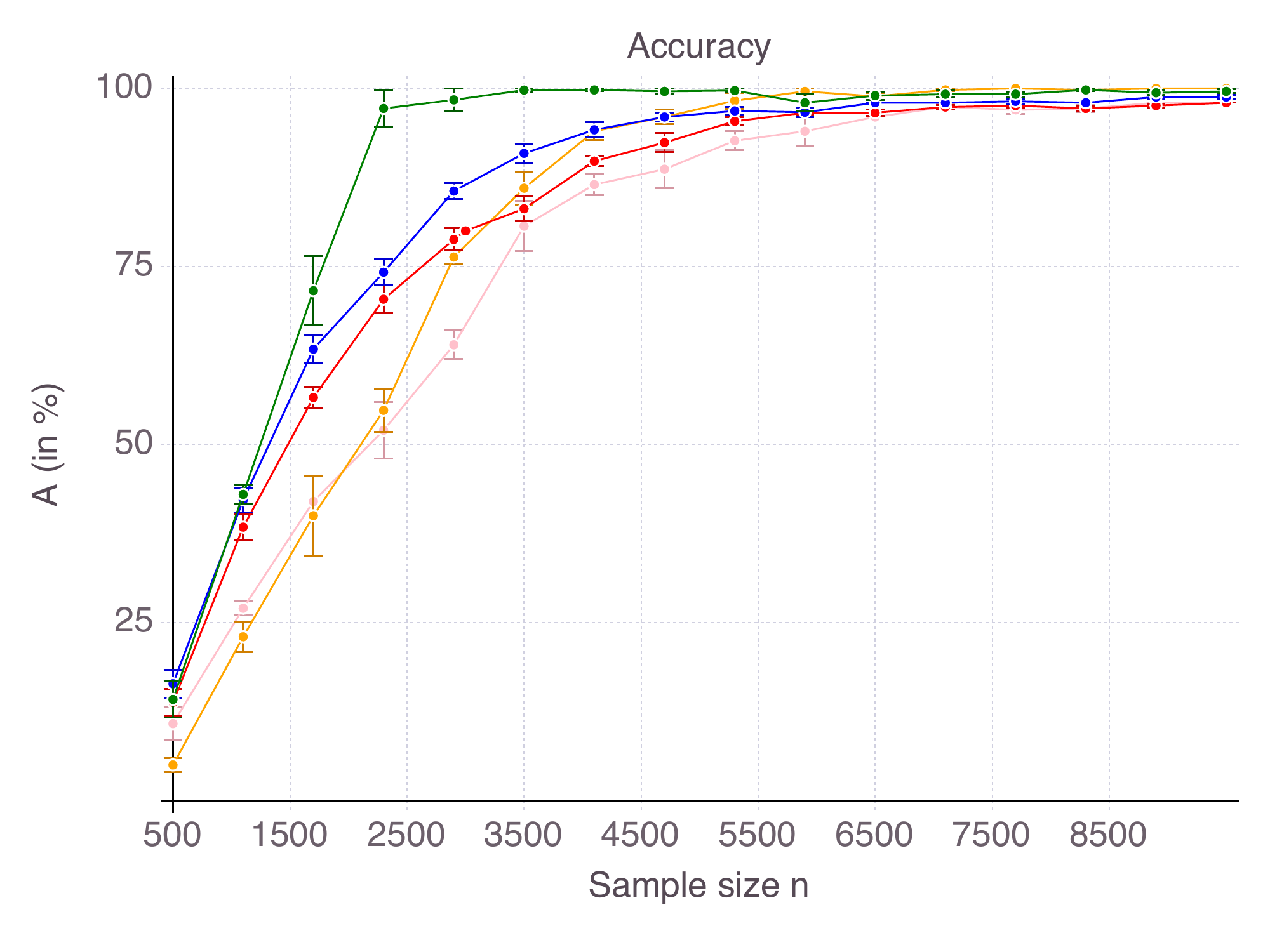}
	\includegraphics[width=.45\linewidth]{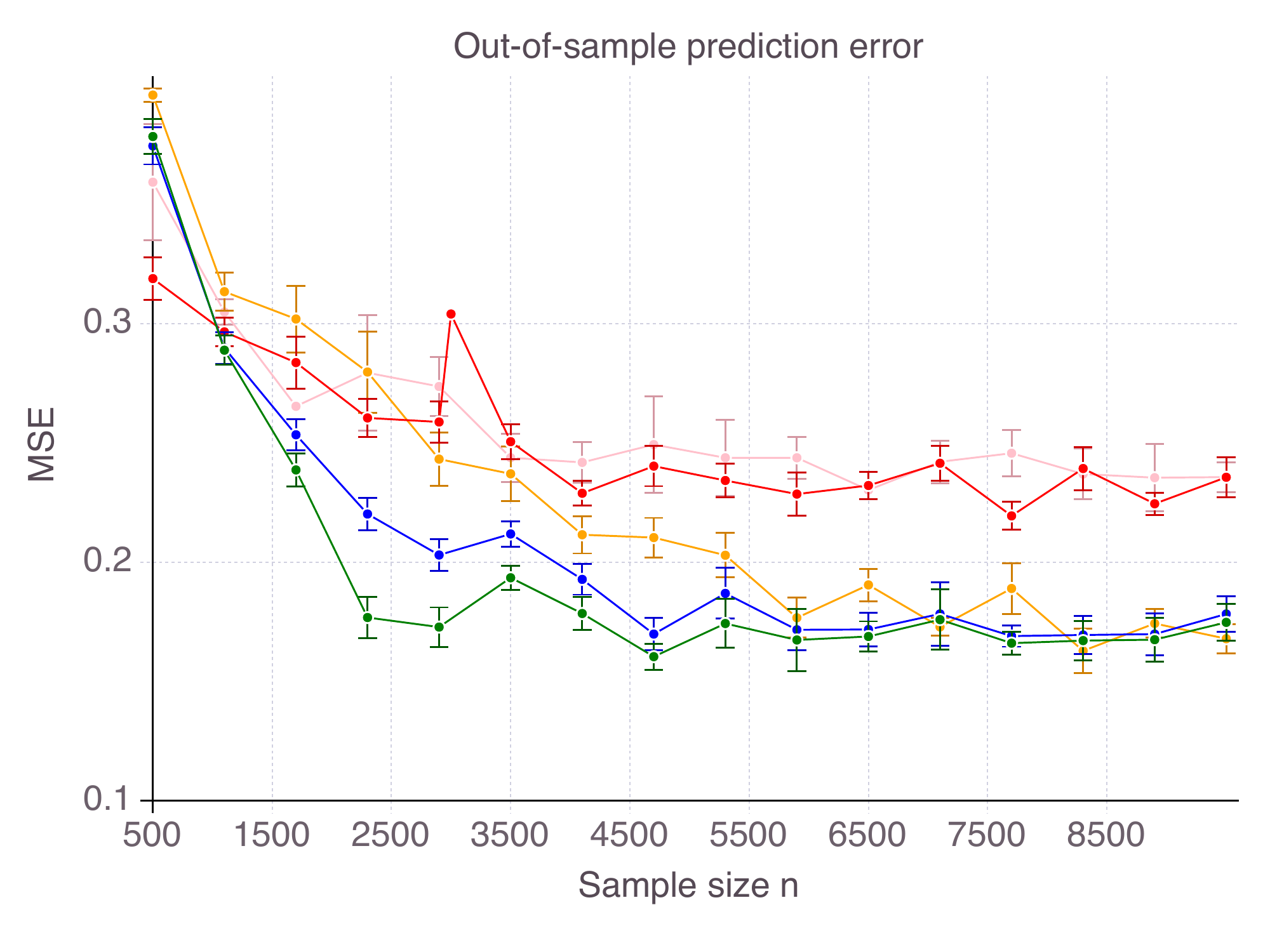}
	\caption{Medium noise}
\end{subfigure} %

\begin{subfigure}[t]{\linewidth}
	\centering
	\includegraphics[width=.45\linewidth]{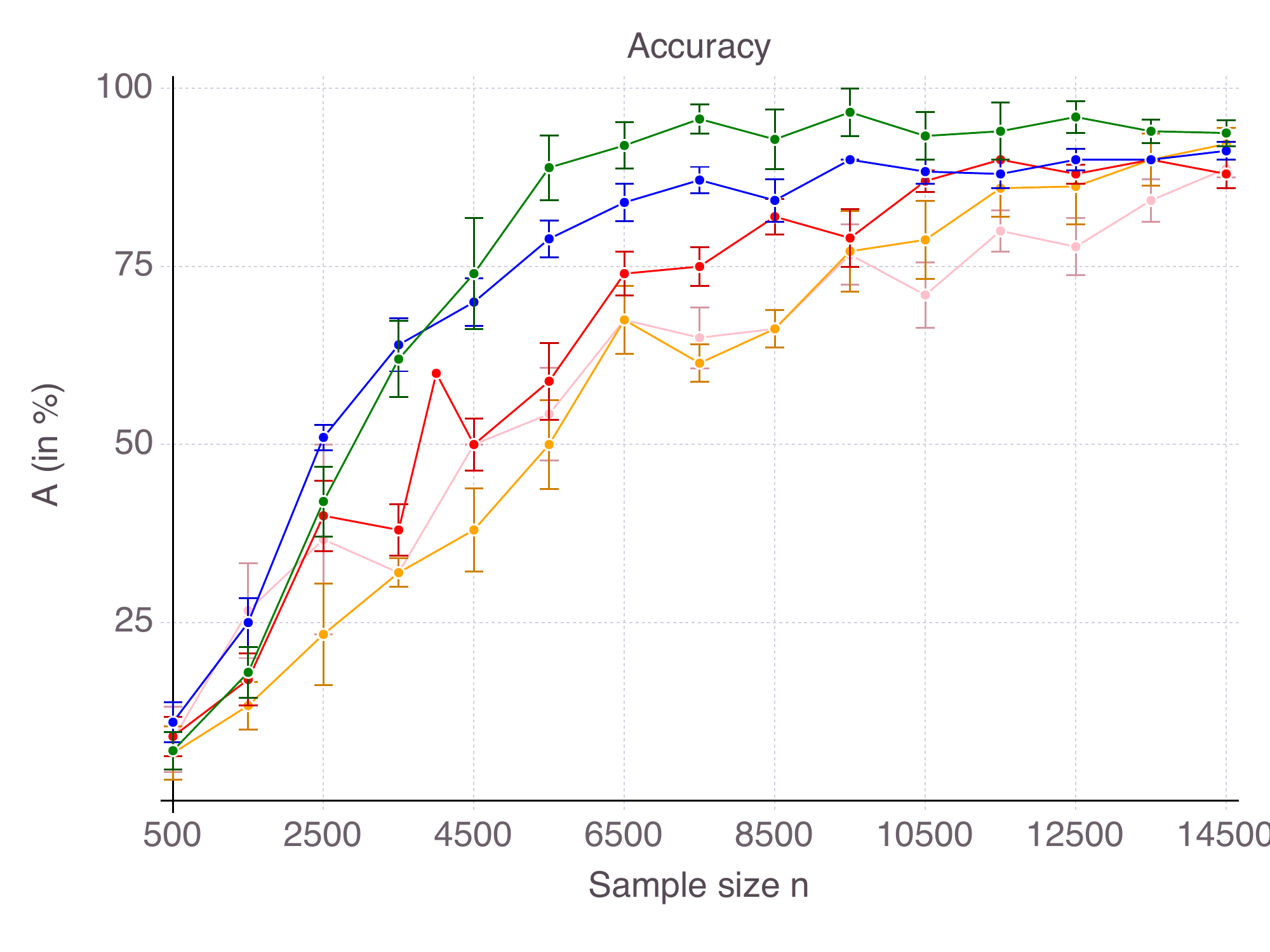}
	\includegraphics[width=.45\linewidth]{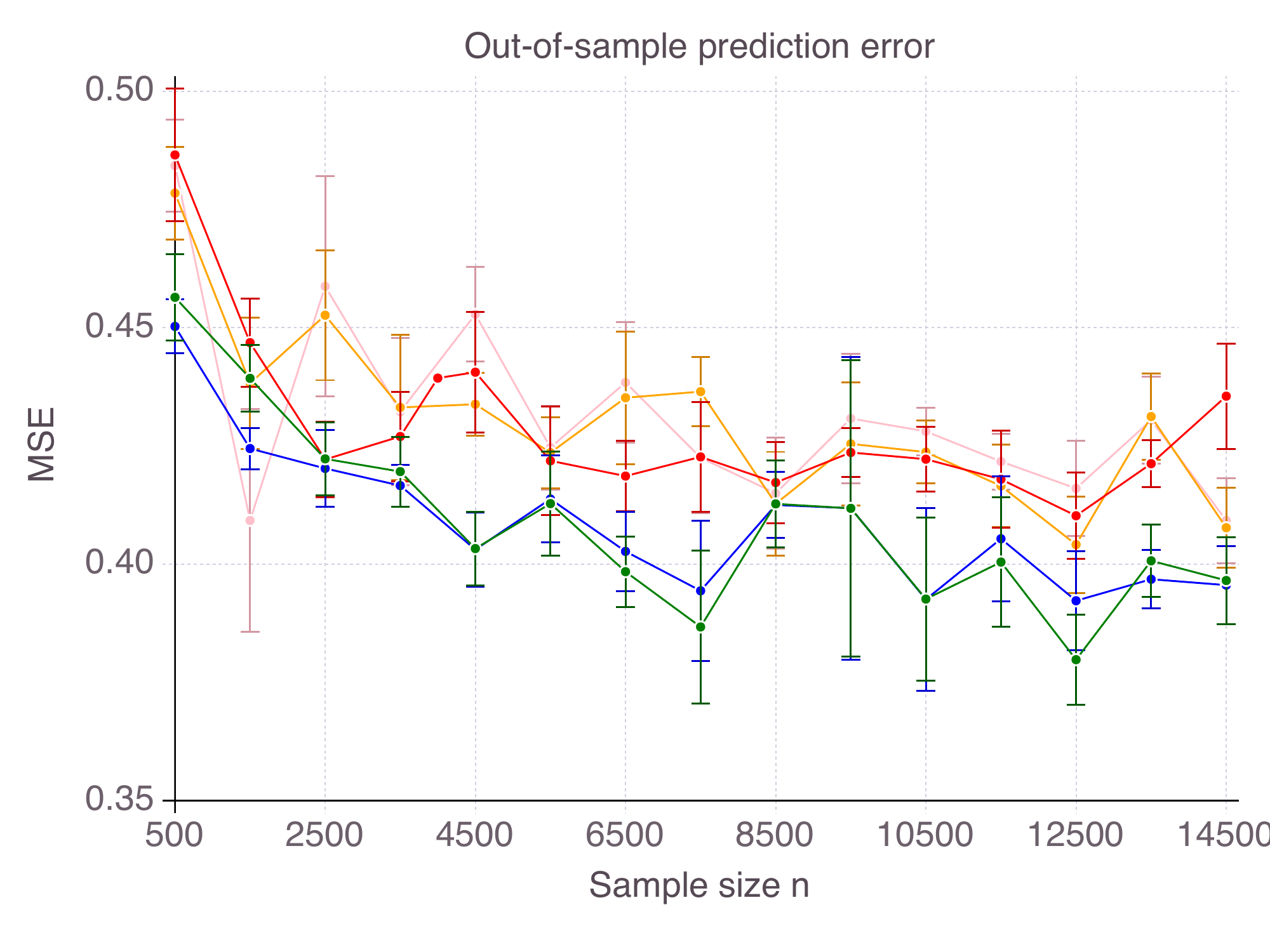}
	\caption{High noise}
\end{subfigure} %

\caption{Accuracy (left panel) and out-of-sample mean square error (right panel) as $n$ increases, for the CIO (in green), SS (in blue with $T_{max}=150$) with Hinge loss, ENet (in red), MCP (in orange), SCAD (in pink) with logistic loss. We average results over $10$ data sets with $n=500,...,3700$, $p=20,000$, $k_{true}=100$.}
\label{fig:ClassHardFix}
\end{figure*}

\subsection{Real-world design matrix $X$}
To illustrate the implications of feature selection on real-world applications, we re-consider the example from genomics we introduced in the previous section. Our $n = 1,145$ lung cancer patients naturally divide themselves into two groups, corresponding to different tumor types. In our sample for instance, $594$ patients ($51.9\%$) suffered from Adenocarcinoma while the remaining $551$ patients ($48.1\%$) suffered from Squamous Cell Carcinoma, making our data amenable to a binary classification task. Our goal is to identify a genetic signature for each tumor type, which only involves a limited number of genes, to better understand the disease or narrow the search for potential treatment for instance. We held $15\%$ of patients from both groups in a test set ($171$ patients). We used the remaining $974$ patients as a training and validation set. For each algorithm, we computed models with various degrees of sparsity on the training set, evaluated them on the validation set and took the most accurate model. Table \ref{tab:cancer} reports the induced sparsity $k^\star$ and out-of-sample accuracy in terms of $AUC$ on the test set for each models. Results correspond to the median values obtained over ten different training/validation splits. Compared to the regression cases, we now have a real-world design matrix $X$ and real world signals $Y$ as well. 

\begin{table}[h]
\centering
\caption{Median of the results on Lung Cancer data, over $10$ different training/validation set splits.}
\label{tab:cancer}
\begin{tabular}{lcc}
Method & Sparsity $k^\star$& Out-of-sample $AUC$\\
\midrule
Exact sparse & $87.5$ & $0.9798$ \\
Boolean relaxation & $65$ & $0.9821$ \\
Lasso & $114$ & $0.9814$ \\
ENet & $398.5$ & $0.9806$  \\
MCP & $39$ & $0.9741$ \\
SCAD & $79.5$ & $0.9752$ \\
\bottomrule
\end{tabular}
\end{table}

The first conclusion to be drawn from our results is that the all the feature selection methods considered in this paper, including the convex integer optimization formulation, scale to sizes encountered in real-world impactful applications. MCP and SCAD provide the sparsest classifiers with median sparsity of $39$ and $79.5$ respectively. At the same time, they achieve the lowest prediction accuracy, which questions the relevance of the genes selected by those methods. The $\ell_1$-based formulations, Lasso and ENet, reach an $AUC$ above $0.98$ with $114$ and $398.5$ genes respectively. In comparison, CIO and SS have similar accuracy while selecting respectively only $87$ and $65$ genes.
 
\section{Conclusion}
In this paper, we provided a unified treatment of methods for feature selection in statistics. We focused on five methods: the NP-hard cardinality-constrained formulation \eqref{eqn:reg_sparse}, its Boolean relaxation, $\ell_1$-regularized estimators (Lasso and Elastic-Net) and two non-convex penalties, namely the smoothly clipped absolute deviation (SCAD) and minimax concave penalty (MCP). 

In terms of statistical performance, we compared the methods based on two metrics: accuracy and false detection rate. A reasonable feature selection method should exhibit a two-fold convergence property: the accuracy and false detection rate should converge to $1$ and $0$ respectively, as the sample size increases. Jointly observed, these two properties ensure the method selects all true features and nothing but true features. 

Most of the literature on feature selection so far has focused solely on accuracy, from both a theoretical and empirical point of view. Indeed, on that matter, our observations match existing theoretical results. When mutual incoherence condition is satisfied, all five methods attain perfect accuracy, irrespective of the noise and correlation level. As soon as mutual incoherence fails to hold however, $\ell_1$-regularized estimators do not recover all true features, while non-convex formulations do. In all our experiments, Lasso-based formulations are the least accurate and sensitive to correlation between features, while cardinality-constrained formulation and the MCP non-convex estimator are the most accurate.  
{ As far as accuracy is concerned, we observe a clear distinction between convex and non-convex penalties, which echoes in our opinion the distinction between robustness and sparsity. Robustness is the property of an estimator to demonstrate good out-of-sample predictive performance, even in noisy settings, and convex regularization techniques are known to produce robust estimators \citep{bertsimas2009equivalence,xu2009robustness}. When it comes to sparsity however, non-convex penalties are theoretically more appealing, for they do not require stringent assumptions on the data \citep{loh2017support}. Because both properties should deserve attention, we believe - and observe - that the best approaches are those combining a convex and a non-convex component. The $\ell_1$-regularization on its own is not sufficient to produce reliably accurate feature selection.}

In real-world applications, false detection rate is at least as important as accuracy. We were able to observe a zero false detection rate for Lasso-based formulations only under the mutual incoherence condition and in low noise settings where $SNR \rightarrow \infty$. Otherwise, false detection rate remains strictly positive and stabilizes above $80\%$ (we observed this behavior as early as for $SNR\leqslant 25$).  False detection rate for non-convex penalties MCP and SCAD quickly drops as $n$ increases, but remains strictly positive (around $15-30\%$) and highly volatile, even for large sample sizes. The exact sparse formulation is the only method in our study which clearly outperforms all other methods, in all settings, and both for regression and classification, with the lowest false detection rate. Its Boolean relaxation demonstrates a similar behavior but less acute, especially in classification. In our opinion, such an observation speaks in favor of formulations that explicitly constrain the number of features instead of using regularization to induce sparsity. In practice, one could use Lasso or non-convex penalties as a good feature screening method, that is to discard irrelevant features and reduce the dimensionality of the problem. Nonetheless, in order to select relevant features only, we highly recommend the use of cardinality-constrained formulation or its relaxation, depending on available computing resources. Table \ref{tab:summary} (p. \pageref{tab:summary}) summarizes the advantages and disadvantages we observed for each method.

Those observations would be of little use if the best performing method were neither scalable nor available to practitioners. To that end, we released the code of a cutting-plane algorithm which solves the exact formulation \eqref{eqn:reg_sparse} in minutes for $n$ and $p$ in the $10,000$s. Though computationally expensive, this method requires only one to two orders of magnitude more time than other methods. We believe this additional computational cost is affordable in many applications and justified by the resulting improved statistical performance. For more time-sensitive applications, its Boolean relaxation provides a high-quality approximation. We proposed a scalable sub-gradient algorithm to solve it and released our code in the \verb|Julia| package \verb|SubsetSelection|, which can compete with the \verb|glmnet| implementation of Lasso in terms of computational time, while returning statistically more relevant features. With \verb|SubsetSelection|, we hope to bring to the community an easy-to-use and generic feature selection tool, which addresses deficiencies of $\ell_1$-penalization but scales to high-dimensional data sets.

\begin{table}[h]
\caption{Summary of the advantages and disadvantages of each method.}
\label{tab:summary}
\centering
\begin{tabular}{ll}
Method & Pros and Cons \\
\midrule
\multirow{2}{*}{Exact sparse} & $(+)$ Very good $A$/$FDR$. Convergence robust to noise/correlation. \\
& $(-)$ Commercial solver and extra computational time. \\
\midrule
\multirow{2}{*}{Boolean relaxation}  &$(+)$ Good $A$/$FDR$. Convergence robust to noise/correlation. \\
& $(-)$ Heuristic. \\
\midrule
\multirow{2}{*}{Lasso/ENet}  & $(+)$ Whole regularization path at no extra cost. \\
& $(-)$  $FDR$ very sensitive to noise. $A$ very sensitive to correlation. \\
\midrule
\multirow{2}{*}{MCP}  & $(+)$ Excellent $A$. \\
& $(-)$ Unstable $FDR$. \\
\midrule
\multirow{2}{*}{SCAD}  & $(+)$ Very good $A$. \\
& $(-)$ Unstable $FDR$. $A$ sensitive to correlation.   \\
\bottomrule
\end{tabular}
\end{table}

\appendix

\section{Extension of the sub-gradient algorithm and implementation}
\label{sec:subsetselection}
From a theoretical point of view, the boolean relaxation of the sparse learning problem is often tight, especially in the presence of randomness or noise in the data, so that feature selection can be expressed as a saddle point problem with continuous variables only. This min-max formulation is easier to solve numerically than the original mixed-integer optimization problem \eqref{eqn:reg_sparse} because the original combinatorial structure vanishes. Our proposed algorithm benefits from both tightness of the Boolean relaxation and integrality of optimal solutions. 
In this section, we present the \verb|Julia| package \verb|SubsetSelection|  which competes with the \verb|glmnet| implementation of Lasso in terms of computational time, while returning statistically more relevant features, in terms of accuracy but more significantly in terms of false discovery rate. With \verb|SubsetSelection|, we hope to bring to the community an easy-to-use and generic feature selection tool, which addresses deficiencies of $\ell_1$-penalization but scales just as well to high-dimensional data sets.

\subsection{Cardinality-constrained formulation}
In this paper, we have proposed a sub-gradient algorithm for solving the following boolean relaxation 
\begin{equation*}
\min_{s\in [ 0, 1 ]^p \: : \: \textbf{e}^\top   s \leqslant k} \: \max_{\alpha \in \mathbb{R}^n} \: f(\alpha, s),
\end{equation*}
where $$f(\alpha,s) := -\sum_{i=1}^n \hat \ell(y_i, \alpha_i) - \dfrac{\gamma}{2} \sum_{j=1}^p s_j \alpha^\top   X_j X_j^\top   \alpha $$ is a linear function in $s$ and concave function in $\alpha$. The function $f$ depends on the loss function $\ell$ through its Fenchel conjugate $\hat \ell$. In this paper, we mainly focused on OLS and logistic loss but the same methodology could be applied to any convex loss function. Indeed, the package \verb|SubsetSelection| supports all loss functions presented in Table \ref{tab:loss_extended}.
\begin{table}[h]
\centering
\caption{Supported loss functions $\ell$ and their corresponding Fenchel conjugates $\hat\ell$ as defined in Theorem \ref{cio_formulation}. The observed data $y\in \mathbb{R}$ for regression and $y \in \{ -1,1\}$ for classification. By convention, $\hat \ell$ equals $+ \infty$ outside of its domain. The binary entropy function is denoted as $H(x) := - x \log x - (1-x) \log(1-x)$.}
\label{tab:loss_extended}
\begin{tabular}{lll}
\toprule
Method &Loss $\ell(y, u)$ &Fenchel conjugate $\hat \ell(y,\alpha)$ \\
\midrule
Logistic loss & $\log \left( 1 + e^{-y u}\right)$ & $-H(-y \alpha) \mbox{ for } y \alpha \in [-1,0]$ \\
1-norm SVM - Hinge loss & $\max(0, 1-y u)$ & $y \alpha \mbox{ for } y \alpha \in [-1,0]$  \\
2-norm SVM & $\tfrac{1}{2} \max(0, 1-y u)^2$ & $\tfrac{1}{2} \alpha^2 + y \alpha \mbox{ for } y \alpha \leqslant 0$  \\ \midrule
Least Square Regression & $\tfrac{1}{2}(y-u)^2$ & $\tfrac{1}{2} \alpha^2 + y \alpha $\\
1-norm SVR & $(|y-u|-\varepsilon)_+$ & $ y \alpha + \varepsilon |\alpha|\mbox{ for } |\alpha| \leqslant 1 $ \\
2-norm SVR & $\tfrac{1}{2}(|y-u|-\varepsilon)^2_+$ & $\tfrac{1}{2} \alpha^2 + y \alpha +\varepsilon |\alpha|$ \\
\bottomrule
\end{tabular}
\end{table}

At each iteration, the algorithm updates the variable $\alpha$ by performing one step of projected gradient ascent with step size $\delta$, and updates the support $s$ by minimizing $f(\alpha, s)$ with respect to $s$, $\alpha$ being fixed. Since $s$ satisfies $s \in [0,1]^p, \; s^\top  \textbf{e} \leqslant k$, and $f$ is linear in $s$, this partial minimization boils down to sorting the components of $(-\alpha^\top   X_j X_j^\top   \alpha)_{j=1,...,p}$ and selecting the $k$ smallest. Pseudo-code is given in Algorithm \ref{subgradient}.

%

\subsection{Scalability}
As experiments in Sections \ref{sec:regression} and \ref{sec:classification} demonstrated, our proposed algorithm and implementation provides an excellent approximation for the solution of the discrete optimization problem \eqref{eqn:reg_sparse}, while terminating in times comparable with coordinate descent for Lasso estimators for low values of $\gamma$. Table \ref{tab:time} reports some computational time of \verb|SubsetSelection| for data sets with various values of $n$, $p$ and $k$. The algorithm scales to data sets with $(n, p)=(10^5,10^5)$s or  $(10^4, 10^6)$s within a few minutes. More comparison on computational time are given in Appendix \ref{sec:regression.supp.time}.
\begin{table}[h]
\centering
\caption{Computational time of SS with $T_{max} = 200$ for data sets with large values of $n$ and $p$, $\gamma = 2 \, p /k /\max_i \|x_i \|^2 /n$. Due to the dimensionality of the data, computations where performed on $1$ CPU with $250$GB of memory. We provide the average computational time (and the standard deviation) over $10$ experiments.}
\label{tab:time}
\begin{tabular}{lrrrr}
\toprule
Loss function $\ell$ & $n$ & $p$ & $k$ & time (in s) \\
\midrule
Least Squares & $10,000$ & $100,000$ & $100$ & $12.90 \, (0.45)$ \\
Least Squares & $50,000$ & $100,000$ & $100$ & $28.45 \, (1.83)$ \\
Least Squares & $10,000$ & $500,000$ & $100$ & $33.00 \, (1.86)$ \\
Least Squares & $10,000$ & $500,000$ & $500$ & $43.00 \, (0.54)$ \\
\midrule
Hinge Loss & $10,000$ & $100,000$ & $100$ & $37.26\, (0.14)$ \\
Hinge Loss & $50,000$ & $100,000$ & $100$ & $160.73\, (0.28)$ \\
Hinge Loss & $10,000$ & $500,000$ & $100$ & $157.09\, (1.18) $ \\
Hinge Loss & $10,000$ & $500,000$ & $500$ & $59.74\, (0.08)$ \\
\bottomrule
\end{tabular}
\end{table}


\subsection{Extension to cardinality-penalized formulation}
Our proposed approach naturally extends to cardinality-penalized estimators as well, where the $0$-{\blue pseudo}norm is added as a penalization term instead of an explicit constraint. Let us consider the $\ell_2$-regularized optimization problem
\begin{equation}
\label{eqn:reg_bic}
\min_{w \in \mathbb{R}^p} \: \sum_{i=1}^n \ell(y_i, w^\top   x_i) + \dfrac{1}{2 \gamma} \Vert w \Vert_2^2 + \lambda \Vert w \Vert_0,
\end{equation}
which corresponds to the estimator \eqref{eqn:bic} in the unregularized limit $\gamma \rightarrow +\infty$. Similarly introducing a binary variable $s$ encoding the support of $w$, we get that the previous problem \eqref{eqn:reg_bic} is equivalent to 
\begin{equation}
\min_{s\in \{0,1\}^p} \: \max_{\alpha \in \mathbb{R}^n} \: f(\alpha, s) + \lambda \textbf{e}^\top   s.
\end{equation}
The new saddle point function $s \mapsto f(\alpha, s) + \lambda \textbf{e}^\top   s$ is still linear in $s$ and concave in $\alpha$. As before, its boolean relaxation
\begin{equation}
\label{eqn:relax_bic}
\min_{s\in [0,1]^p} \: \max_{\alpha \in \mathbb{R}^n} \: f(\alpha, s) + \lambda \textbf{e}^\top   s
\end{equation}
is tight if the minimizer of $f(\alpha, s) + \lambda \textbf{e}^\top   s$ with respect to $s$ is unique. More precisely, we prove an almost verbatim analogue of Theorem \ref{tightness}.
\begin{thm}
The boolean relaxation \eqref{eqn:relax_bic} is tight if there exists a saddle point $(\bar \alpha, \bar s)$ such that the vector $(\lambda - \tfrac{\gamma}{2} \bar \alpha^\top   X_j X_j^\top   \bar \alpha)_{j=1,...,p}$ has non-zero entries. 
\end{thm}

\begin{proof}
The saddle-point problem \eqref{eqn:relax_bic} is a continuous convex/concave minimax problem and Slater's condition is satisfied so strong duality holds. Therefore, any saddle-point $(\bar \alpha, \bar s)$ must satisfy
$$\bar \alpha \in \arg \max_{\alpha \in \mathbb{R}^n} f(\alpha, \bar s) + \lambda \textbf{e}^\top   s, ~~~
\bar s  \in \arg \min_{s \in \in [0,1]^p} \,f(\bar \alpha, s)+ \lambda \textbf{e}^\top   s.$$
If there exists $\bar \alpha$ such that $(\lambda - \tfrac{\gamma}{2} \bar \alpha^\top   X_j X_j^\top   \bar \alpha)_{j=1,...,p}$ has non-zero entries, then there is a unique $\bar s \in \arg \min_{s \in \in [0,1]^p} \,f(\bar \alpha, s)+ \lambda \textbf{e}^\top   s$. In particular, this minimizer is binary and the relaxation is tight.
\end{proof}
This theoretical result suggests that the Lagrangian relaxation \eqref{eqn:relax_bic} can provide a good approximation for the combinatorial problem \eqref{eqn:reg_bic}  in many cases. The same sub-gradient strategy as the one described in Algorithm \ref{subgradient} can be used to solve \eqref{eqn:relax_bic}, with a slightly different partial minimization step: Now, minimizing $f(\alpha, s) + \lambda s^\top   \textbf{e}$ with respect $s \in [0,1]^p$ for a fixed $\alpha$ boils down to computing the components of $(\lambda - \gamma /2 \alpha^\top   X_j X_j^\top   \alpha)_{j=1,...,p}$ and selecting the negative ones, which requires $O\left(n p \right)$ operations. This strategy is also implemented in the package \verb|SubsetSelection|. 

\section{Numerical experiments for regression - Supplementary material}
\label{sec:regression.supp}
In this section, we provide additional material for the simulations conducted in Section \ref{sec:regression} on regression examples. 
\subsection{Synthetic data satisfying mutual incoherence condition}
\label{sec:regression.supp.time}
We first consider the case where the design matrix $X \sim \mathcal{N}(0,\Sigma)$, with $\Sigma$ a Toeplitz matrix. In particular, $\Sigma$ satisfies the so-called mutual incoherence condition required by Lasso estimators to be statistically consistent. In this setting, we provide more details about the computational time comparison between algorithms for sparse regression presented in Sections \ref{sec:regression}. 

\subsubsection{Impact of the hyper-parameters $k$ and $\gamma$}
The discrete convex optimization formulation \eqref{eqn:reg_sparse} and its Boolean relaxation \eqref{eqn:relax} involve two hyper-parameters: the ridge penalty $\gamma$ and the sparsity level $k$. 

Intuition suggests that computational time would increase with $\gamma$. Indeed, when $\gamma \rightarrow 0$, $w^\star = 0$ is an obvious optimal solution, while for $\gamma \rightarrow +\infty$ the problem can become ill-conditioned. We generate $10$ problems with $p=5,000$, $k_{true}=50$, $SNR=1$ and $\rho = .5$ and various sample sizes $n$, fix $k=k_{true}$ and report absolute computational time as $n \times \gamma$ increases in Figure \ref{fig:RegTime.gamma} (p. \pageref{fig:RegTime.gamma}).  For small values of $\gamma$, both methods terminate extremely fast - within $10$ seconds for CIO and in less than $2$ seconds for SS. As $\gamma$ increases, computational time sharply increases. For CIO, we capped computational time to $600$ seconds. For SS, we limited the number of iterations to $T_{max}=200$.
\begin{figure*}
\centering
\begin{subfigure}[t]{.45\linewidth}
	\centering
	\includegraphics[width=\linewidth]{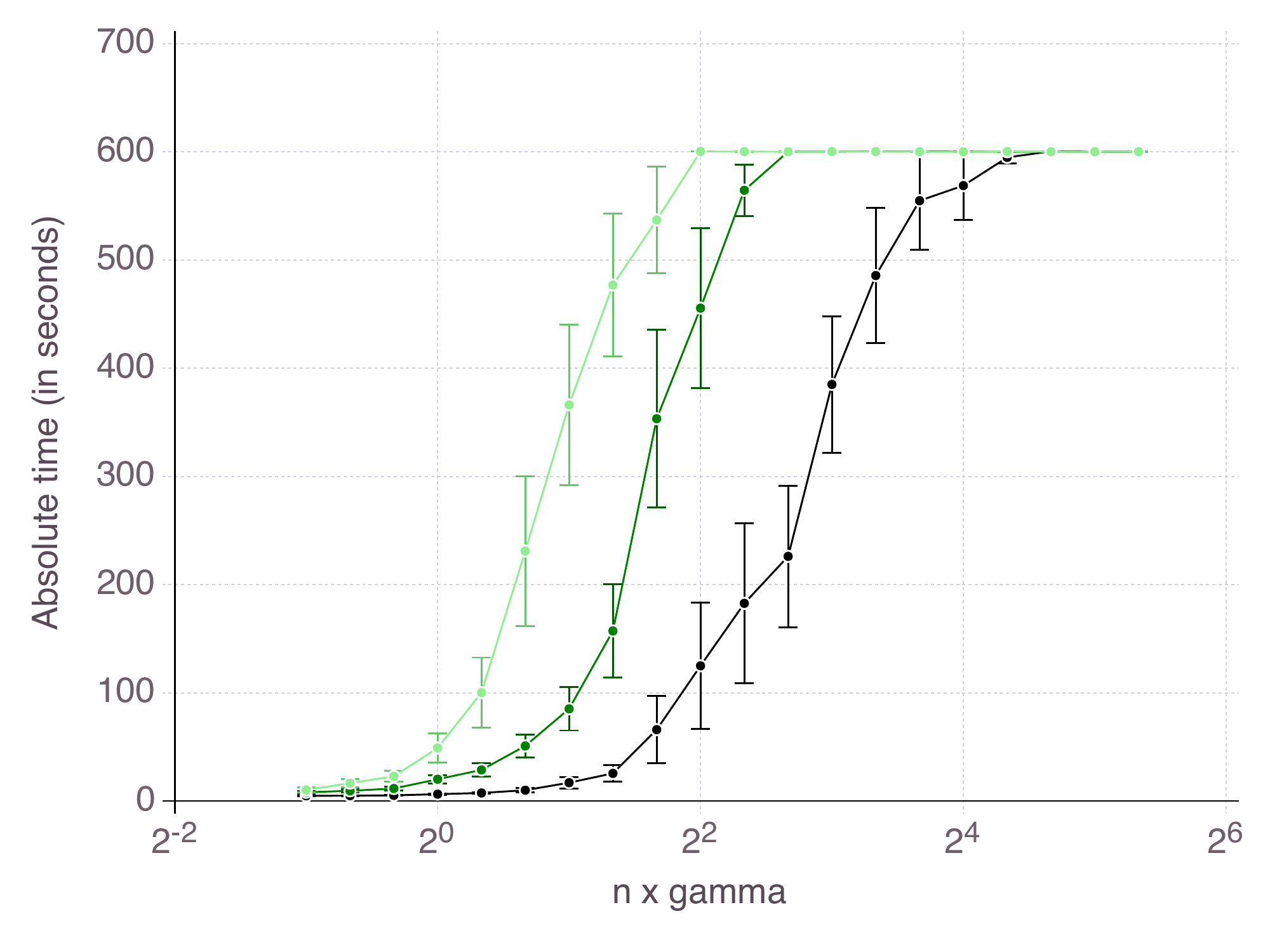}
	\caption{CIO}
\end{subfigure} %
~
\begin{subfigure}[t]{.45\linewidth}
	\centering
	\includegraphics[width=\linewidth]{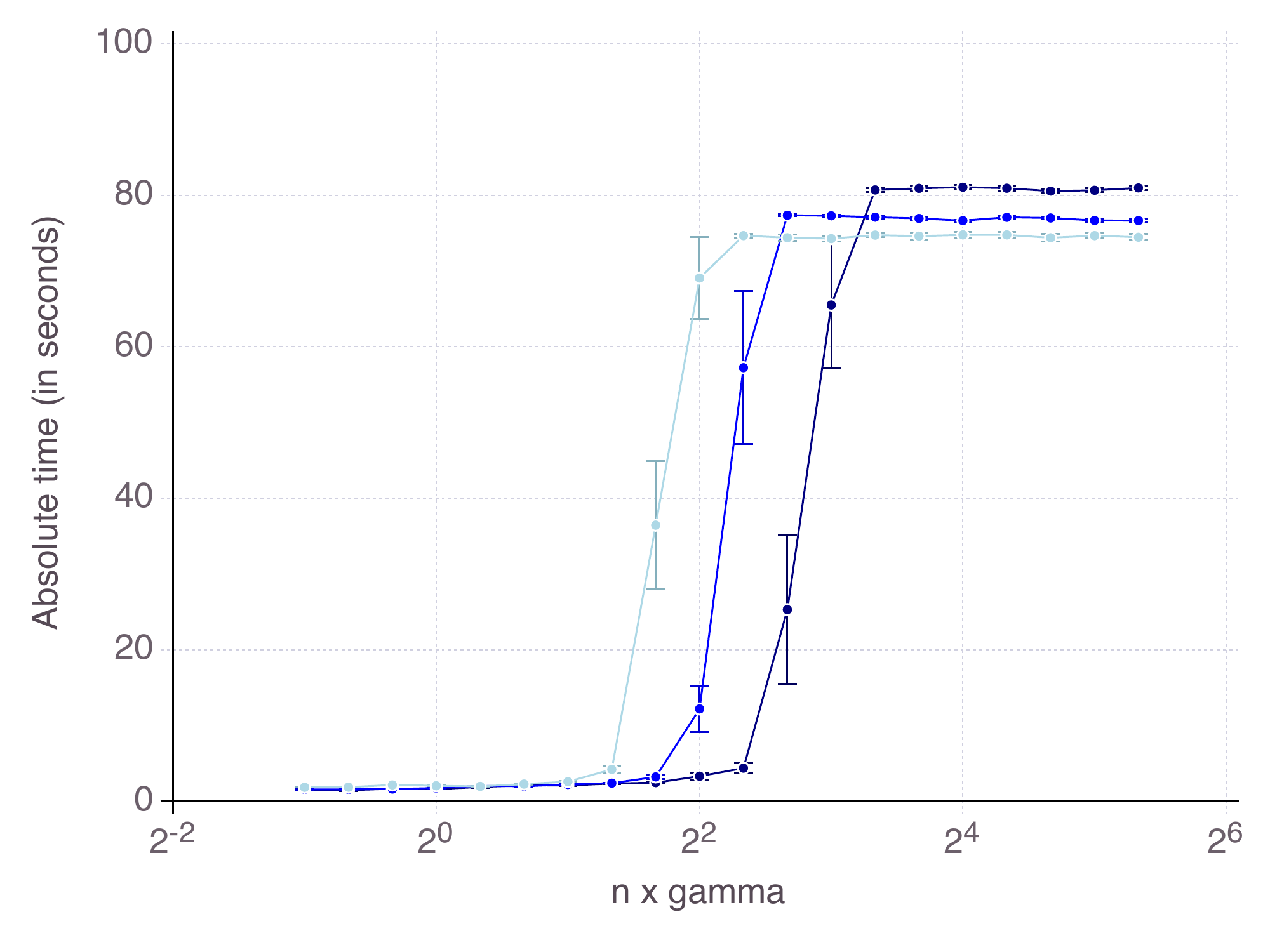}
	\caption{SS}
\end{subfigure} %

\caption{Absolute computational time as $n \times \gamma$ increases, for CIO (left panel), SS (right panel), with OLS loss. We fixed $p=5,000$, $k_{true}=50$, $SNR=1$, $\rho = .5$ and $k=k_{true}$, and averaged results over $10$ data sets.  We report results for $n=500$, $n=1,000$ and $n=2,000$ (from light to dark)}
\label{fig:RegTime.gamma}
\end{figure*}
Regarding the sparsity $k$, the size of the feasible space $\{ s \in \{0,1\}^p \: : \: s^\top   \textbf{e} \leqslant k \}$ grows as $p^k$. Empirically, we observe (Figure \ref{fig:RegTime.k} p. \pageref{fig:RegTime.k}) that computational time increases at most polynomially with $k$.
\begin{figure*}
\centering
\begin{subfigure}[t]{.45\linewidth}
	\centering
	\includegraphics[width=\linewidth]{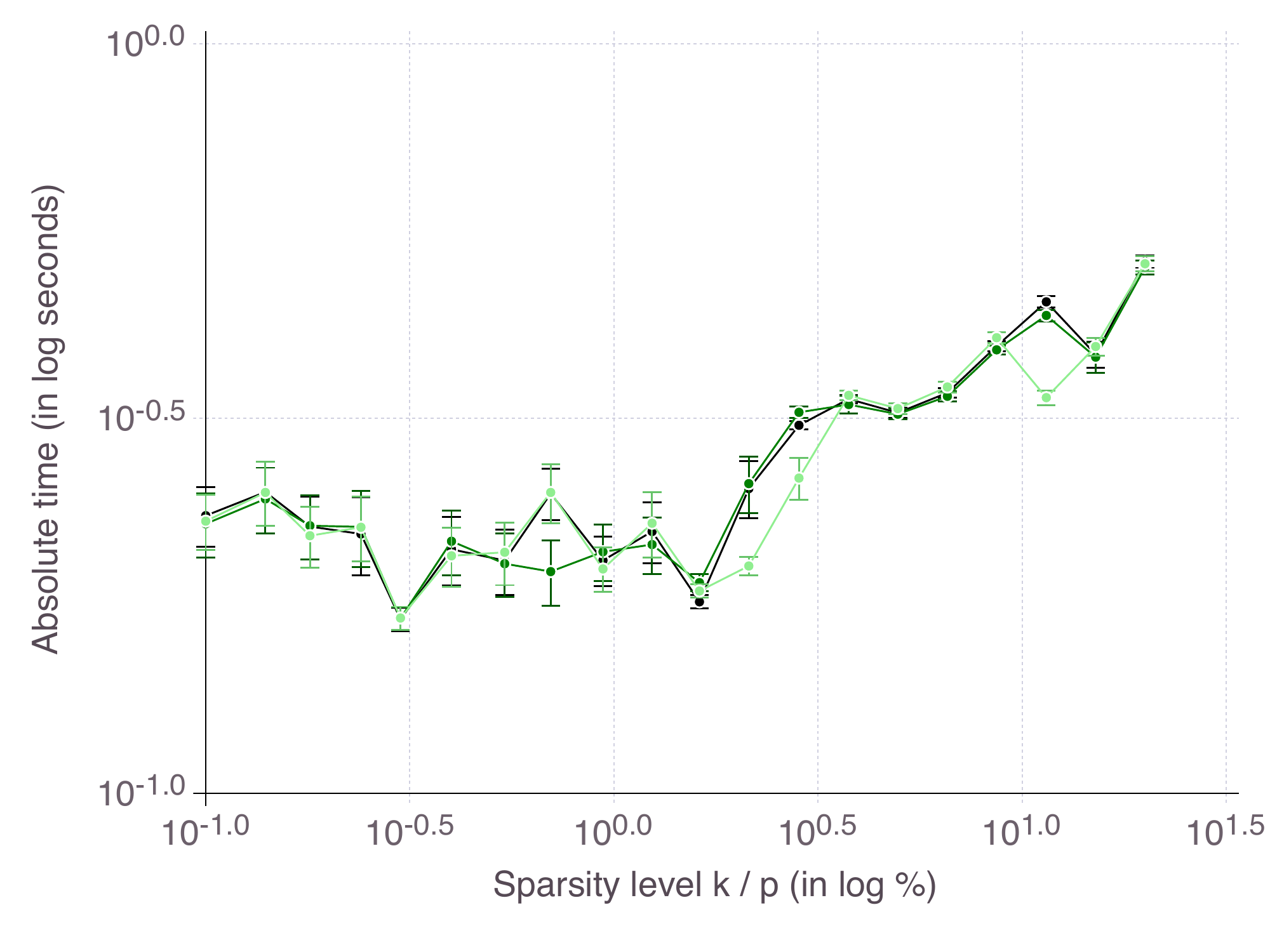}
	\caption{CIO}
\end{subfigure} %
~
\begin{subfigure}[t]{.45\linewidth}
	\centering
	\includegraphics[width=\linewidth]{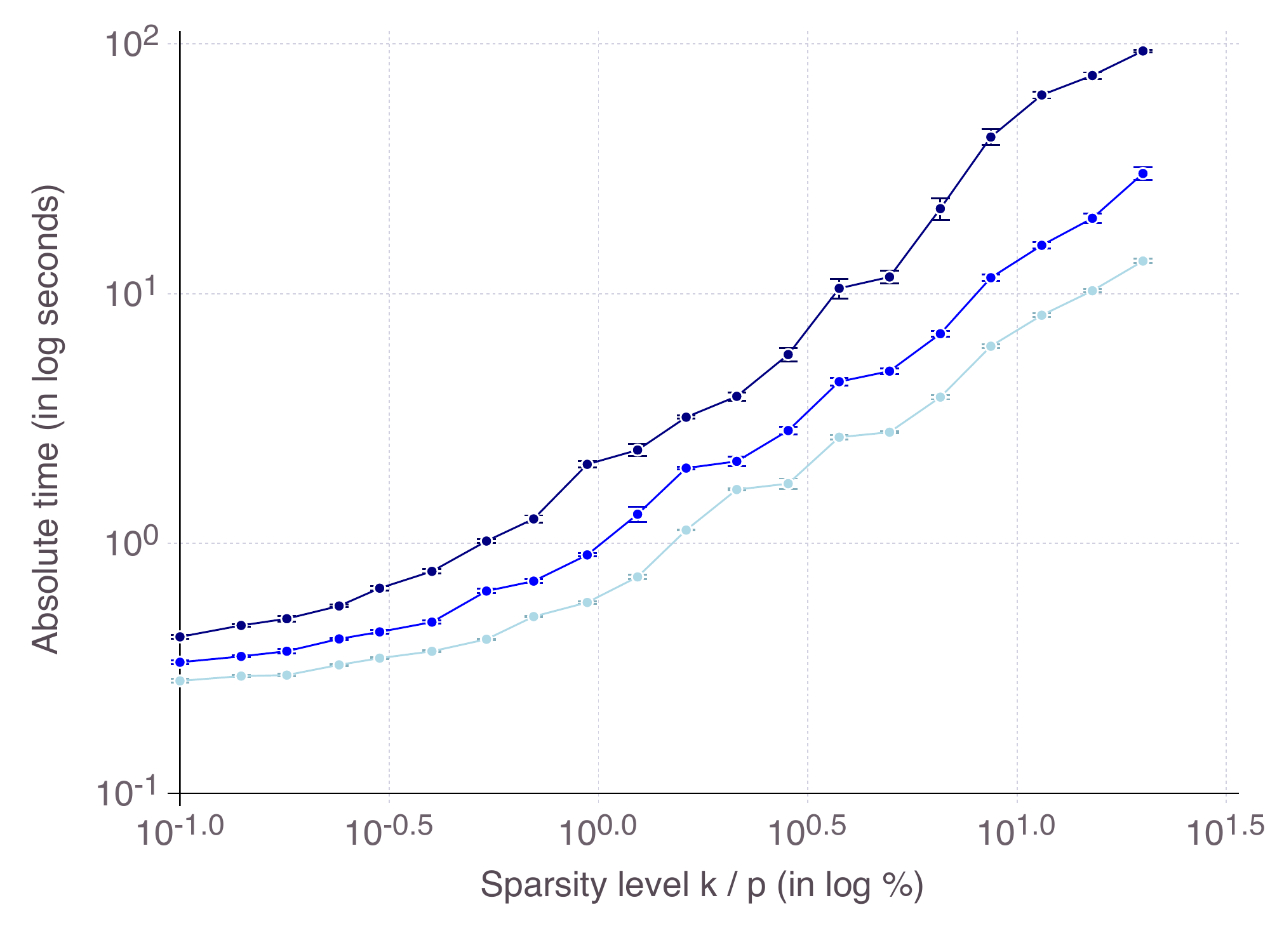}
	\caption{SS}
\end{subfigure} %

\caption{Absolute computational time as $k$ increases from $5$ to $1,000$, for CIO (left panel), SS (right panel), with OLS loss. We fixed $p=5,000$, $k_{true}=50$, $SNR=1$, $\rho = .5$ and $n=1,000$, and averaged results over $10$ data sets. We report results for $\gamma = 2^i \gamma_0$ with $i=0,2,4$ (from light to dark) and $\gamma_0 = \dfrac{p}{n \, k_{true} \, \max_i \|x_i\|^2}$.}
\label{fig:RegTime.k}
\end{figure*}

\subsubsection{Impact of the signal-to-noise ratio, sample size $n$ and problem size $p$}
As more signal becomes available, the feature selection problem should become easier and computational time should decrease. Indeed, in Figure \ref{fig:RegTime.snr} (p. \pageref{fig:RegTime.snr}), we observe that low $SNR$ generally increases computational time for all methods (left panel). The correlation parameter $\rho$ (right panel), however, does not seem to have a strong impact on computational time. In our opinion, with $SNR=1$, the effect of correlation on computational time is second order compared to the impact of noise. 
\begin{figure*}
\centering
\begin{subfigure}[t]{.45\linewidth}
	\centering
	\includegraphics[width=\linewidth]{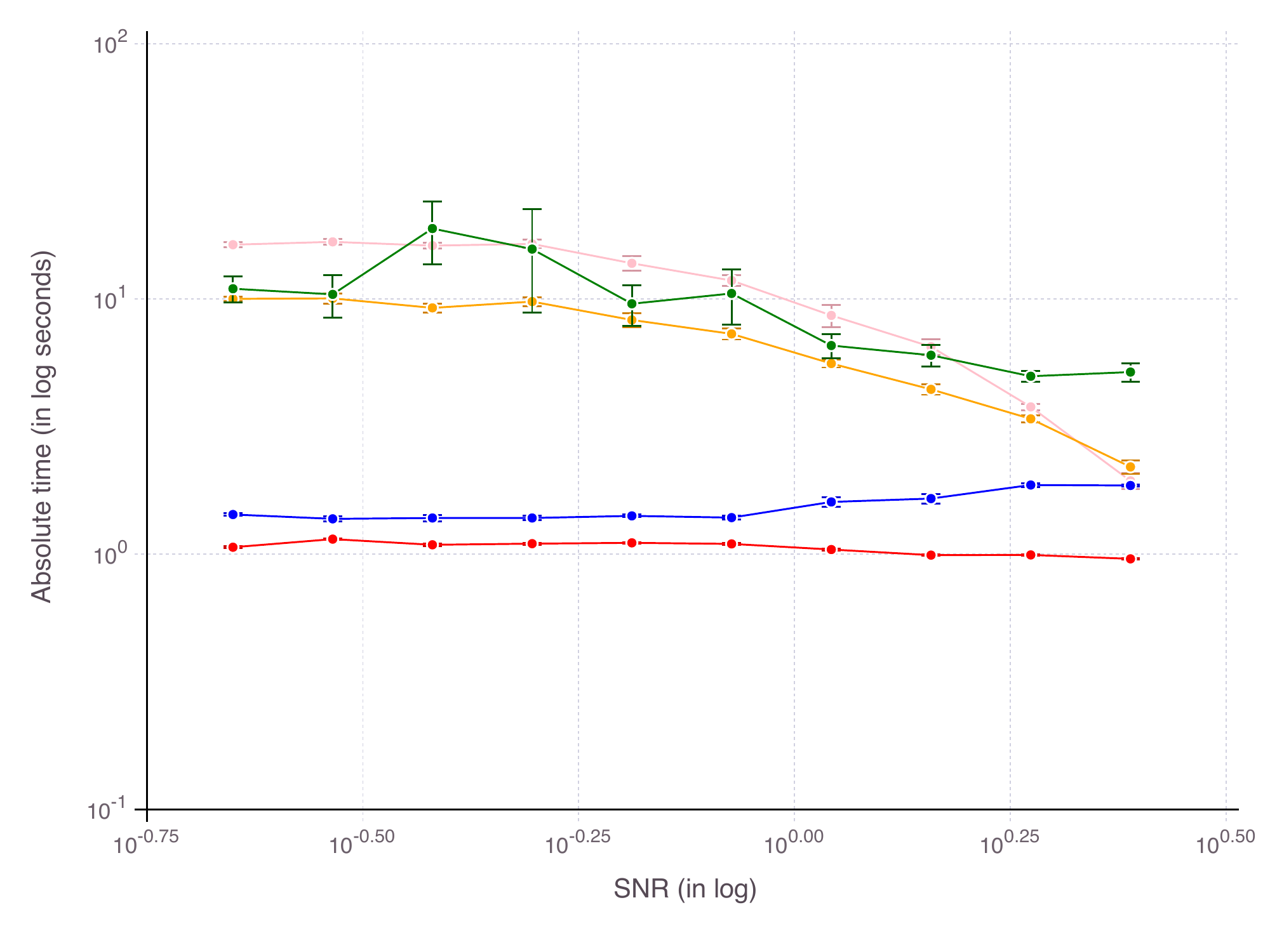}
	\caption{$SNR$, for $\rho = .5$ }
\end{subfigure} %
~
\begin{subfigure}[t]{.45\linewidth}
	\centering
	\includegraphics[width=\linewidth]{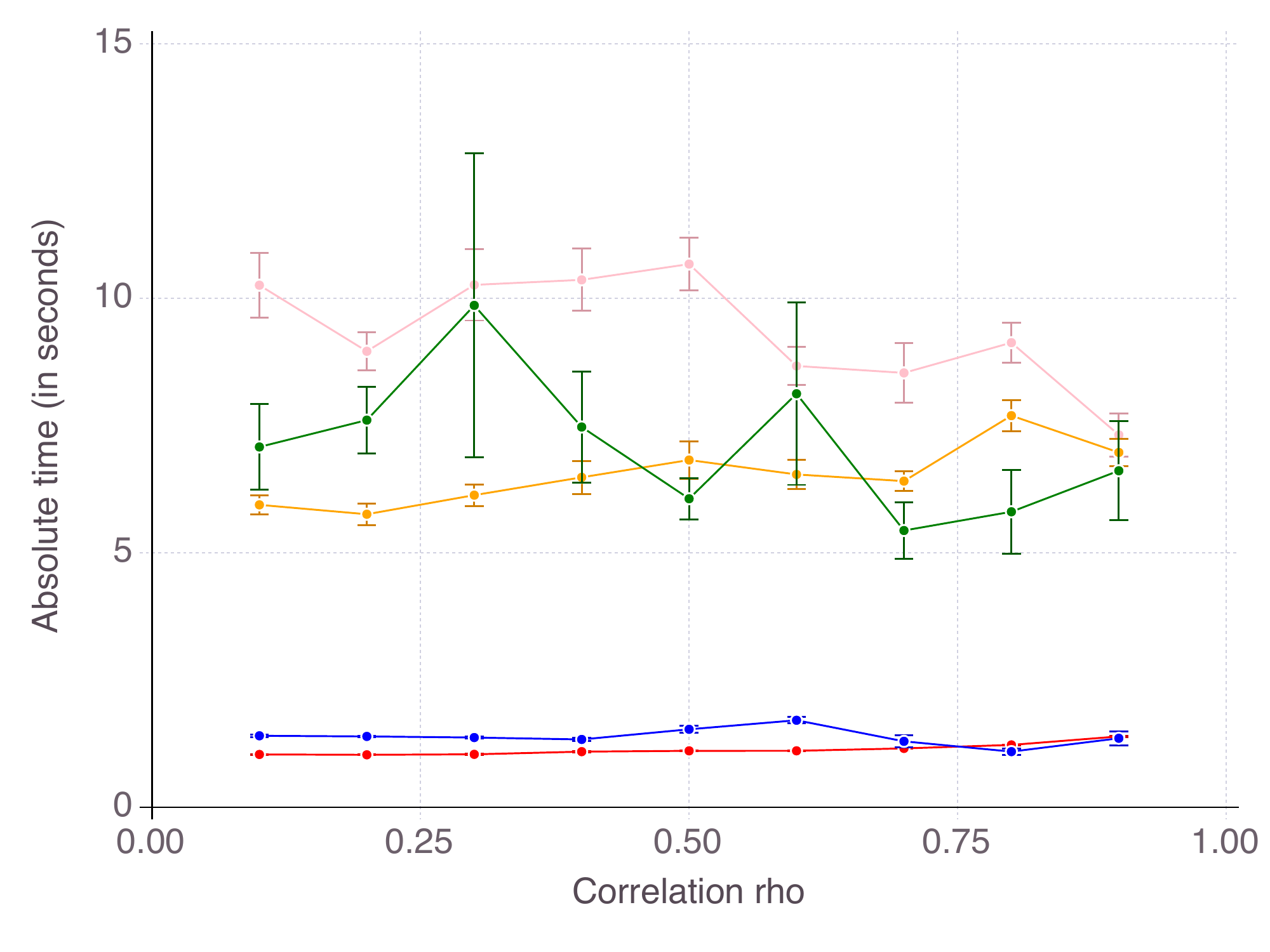}
	\caption{Correlation $\rho$, for $SNR=1$}
\end{subfigure} %
\caption{Absolute computational time as signal-to-noise or correlation increases, for CIO (in green), SS (in blue with $T_{max}=200$), ENet (in red), MCP (in orange), SCAD (in pink) with OLS loss. We fixed $p=5,000$, $k_{true}=50$, and $n=1,000$, and averaged results over $10$ data sets. We report results of CIO and SS with $k=k_{true}$ and $\gamma = \dfrac{p}{2 n \, k_{true} \, \max_i \|x_i\|^2}$.}
\label{fig:RegTime.snr}
\end{figure*}

Figure \ref{fig:RegTime.np} (p. \pageref{fig:RegTime.np}) represents computational for increasing $p$, $n/p$ being fixed and for increasing $n/p$, $p$ being fixed. As shown, all methods scale similarly with $p$ (almost linearly), while CIO and SS are less sensitive to $n /p$ than their competitors. 
\begin{figure*}
\centering
\begin{subfigure}[t]{.45\linewidth}
	\centering
	\includegraphics[width=\linewidth]{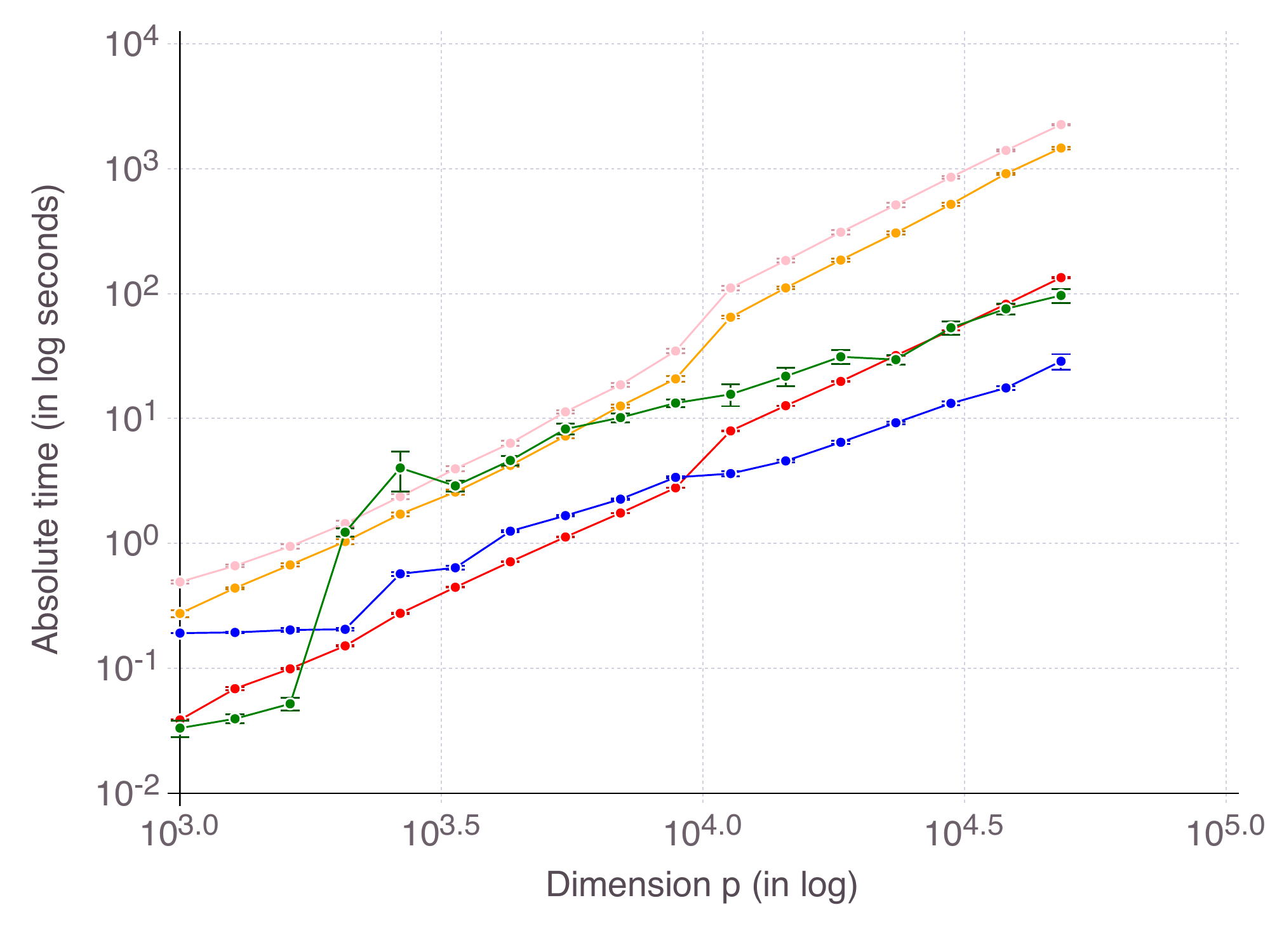}
	\caption{$p$, for $n / p = 1 / 5$ }
\end{subfigure} %
~
\begin{subfigure}[t]{.45\linewidth}
	\centering
	\includegraphics[width=\linewidth]{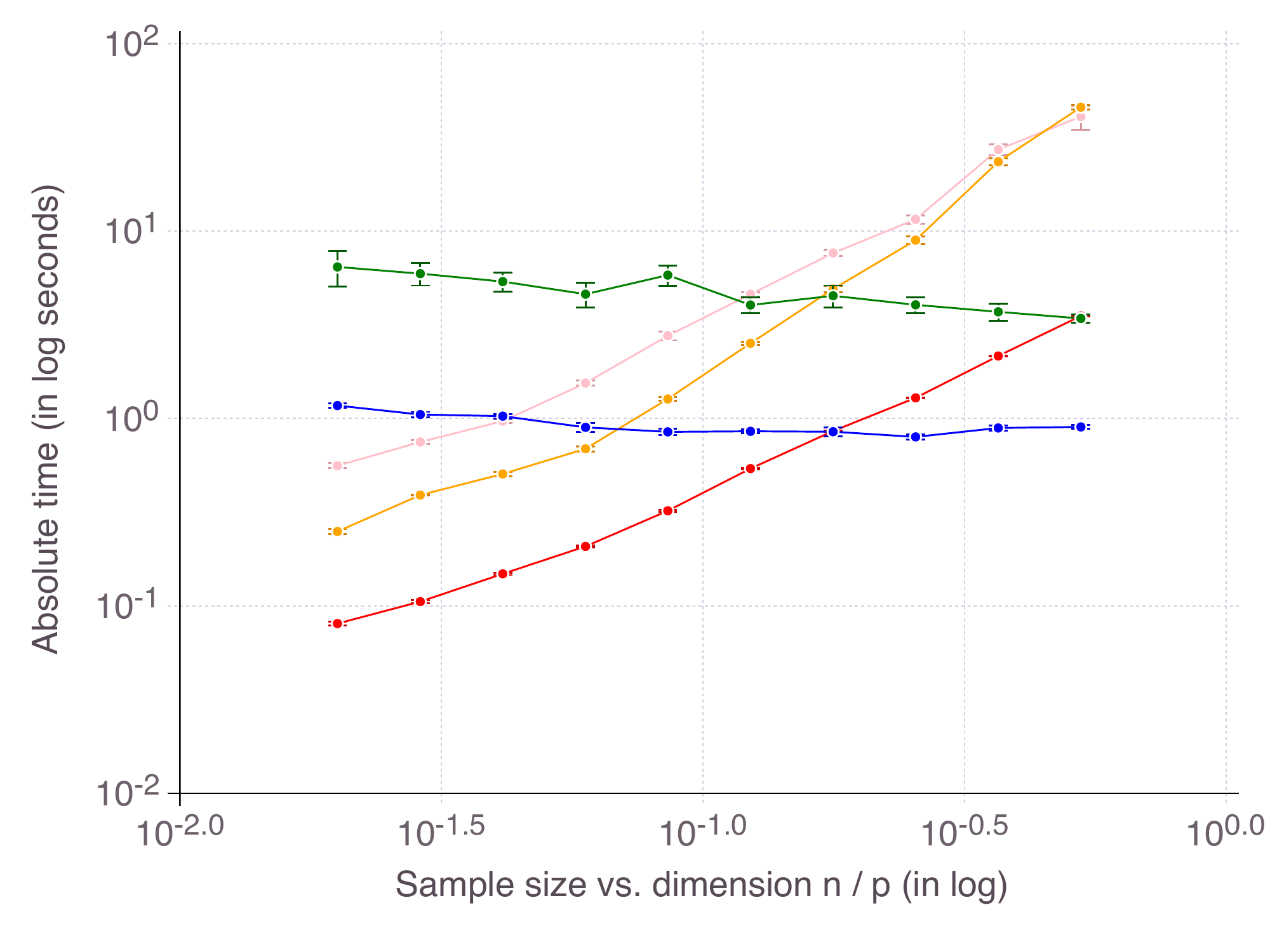}
	\caption{Sample size $n / p$, for $p=5,000$}
\end{subfigure} %
\caption{Absolute computational time as dimension $p$ or sample size $n$ increases, for CIO (in green), SS (in blue with $T_{max}=200$), ENet (in red), MCP (in orange), SCAD (in pink) with OLS loss. We fixed $p=5,000$, $k_{true}=50$ and averaged results over $10$ data sets. We report results of CIO and SS with $k=k_{true}$ and $\gamma = \dfrac{p}{2 n \, k_{true} \, \max_i \|x_i\|^2}$.}
\label{fig:RegTime.np}
\end{figure*}

\subsection{Synthetic data \emph{not} satisfying mutual incoherence condition}
\label{sec:regression.supp.nomic}
We now consider a covariance matrix $\Sigma$, which does not satisfy mutual incoherence, as proved in \citep{loh2017support}. 

\subsubsection{Feature selection with a given support size} 
\label{sec:regression.supp.nomic.fix}
Figure \ref{fig:RegHardFixTime} on page \pageref{fig:RegHardFixTime} reports relative compared to \verb|glmnet| (left panel) and absolute (right panel) computational time { in log scale}. As for the case where mutual incoherence is satisfied, all methods terminates within a $10$-$100$ factor with respect to \verb|glmnet|. 
\begin{figure*}
\centering
\begin{subfigure}[t]{\linewidth}
	\centering
	\includegraphics[width=.45\linewidth]{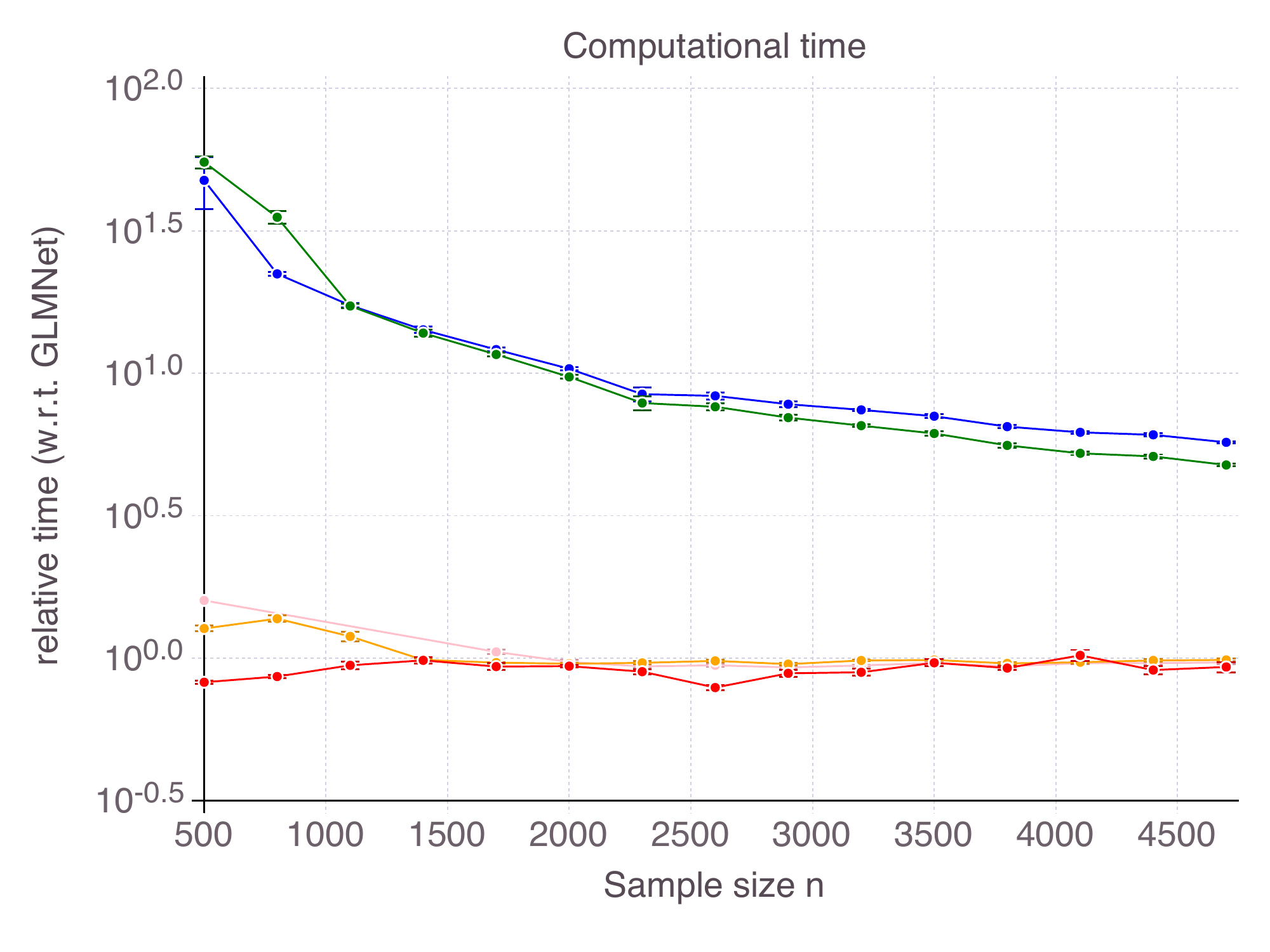}
	\includegraphics[width=.45\linewidth]{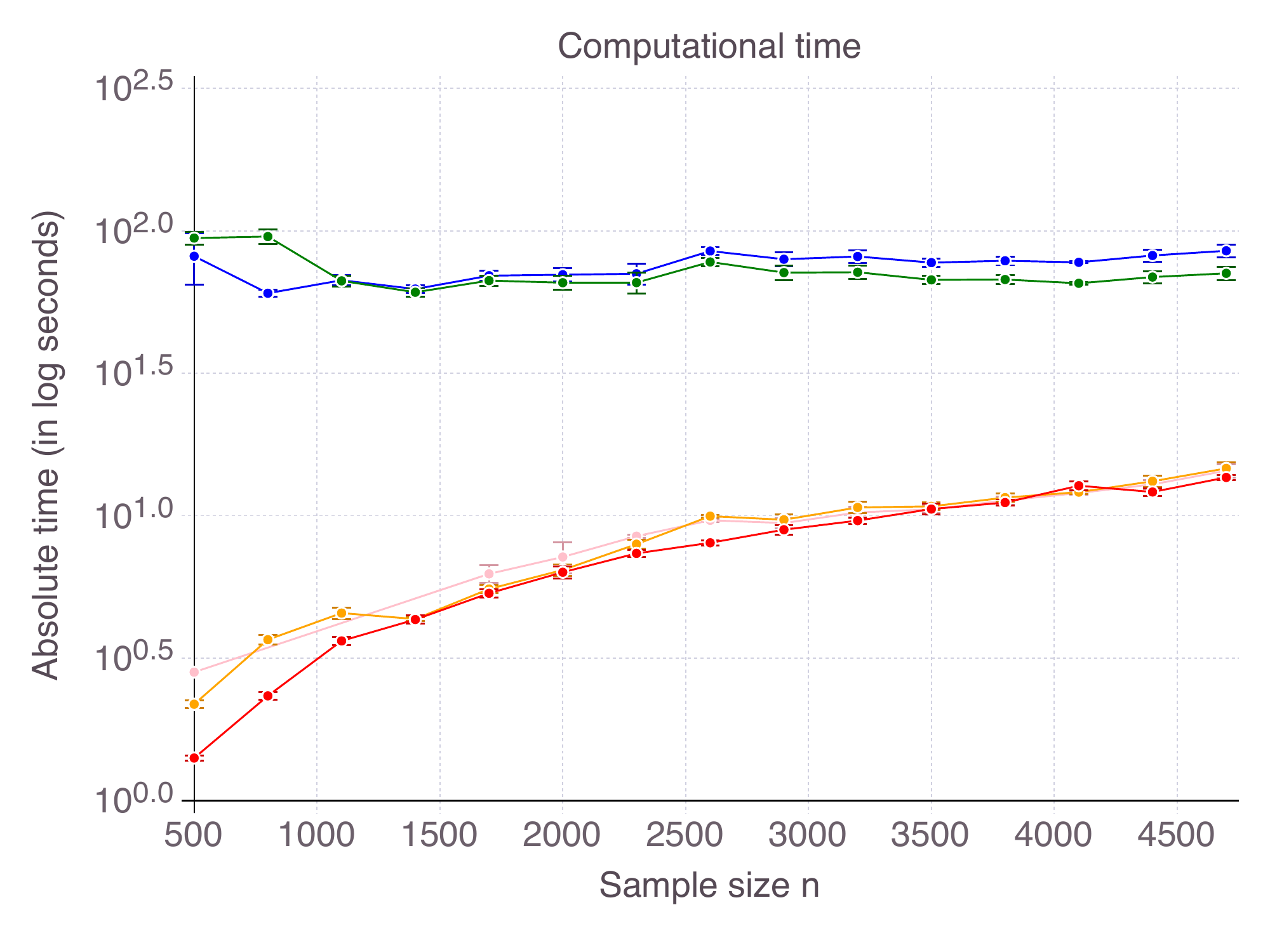}
	\caption{Low noise}
\end{subfigure} %

\begin{subfigure}[t]{\linewidth}
	\centering
	\includegraphics[width=.45\linewidth]{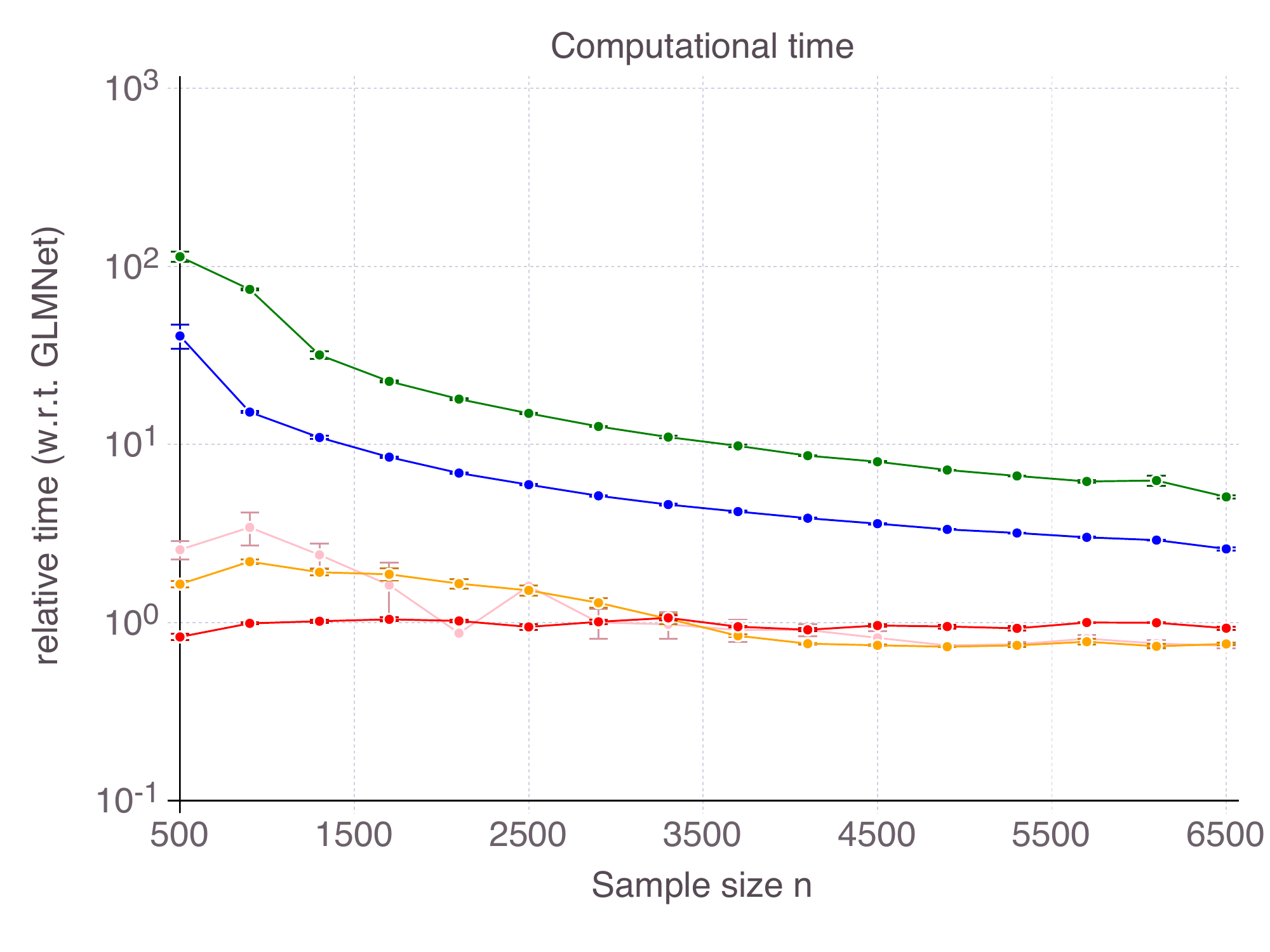}
	\includegraphics[width=.45\linewidth]{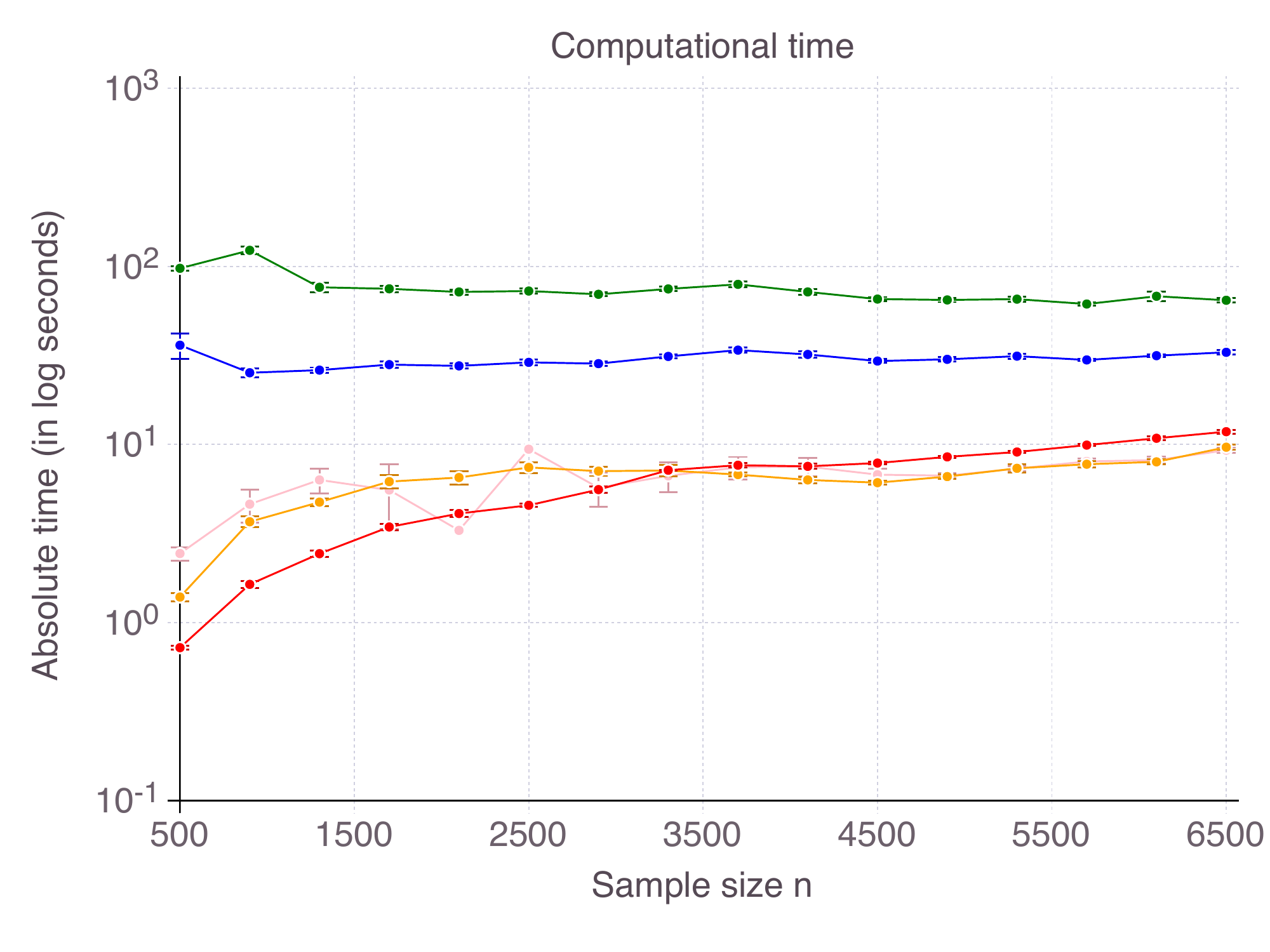}
	\caption{Medium noise}
\end{subfigure} %

\begin{subfigure}[t]{\linewidth}
	\centering
	\includegraphics[width=.45\linewidth]{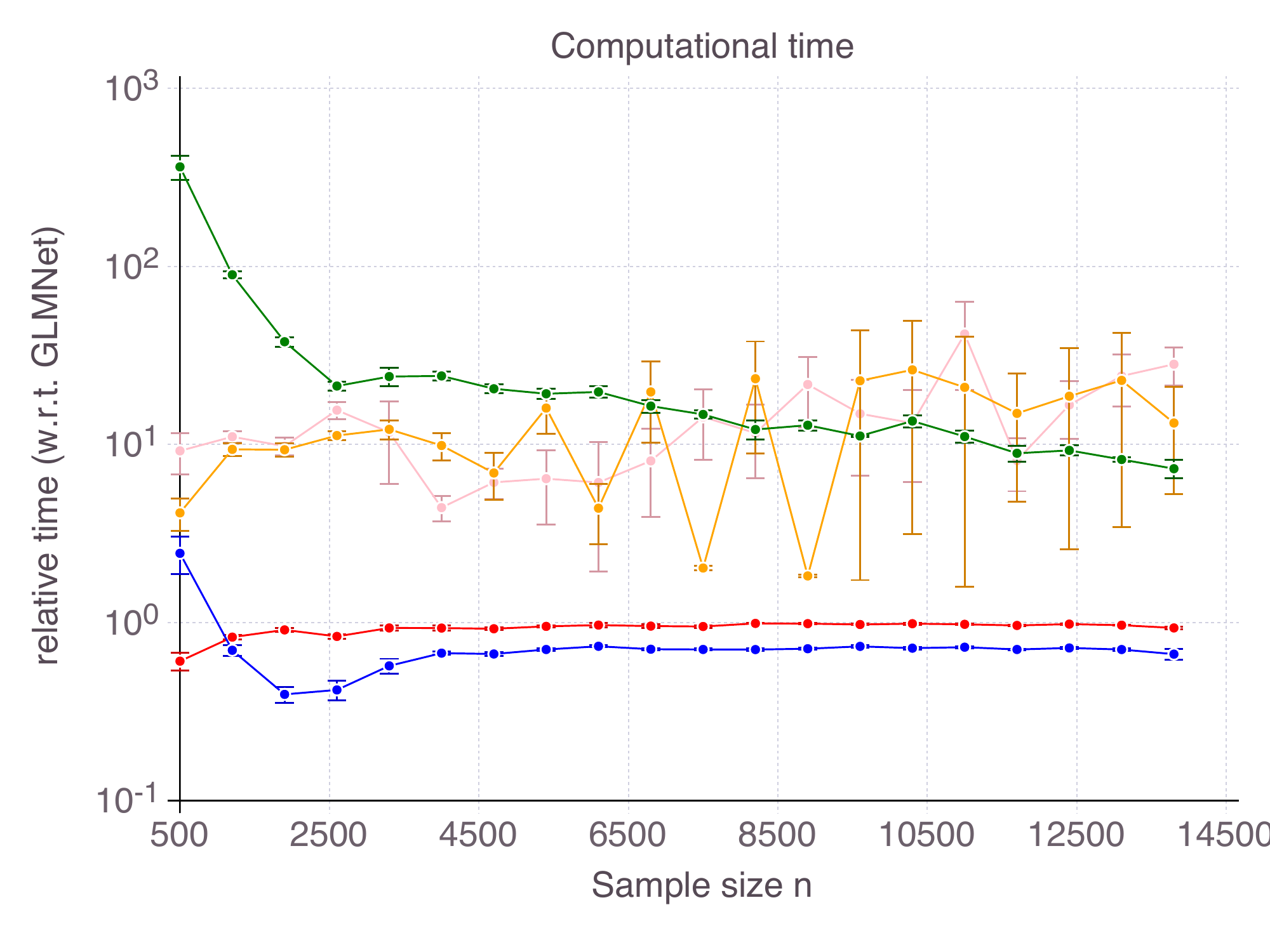}
	\includegraphics[width=.45\linewidth]{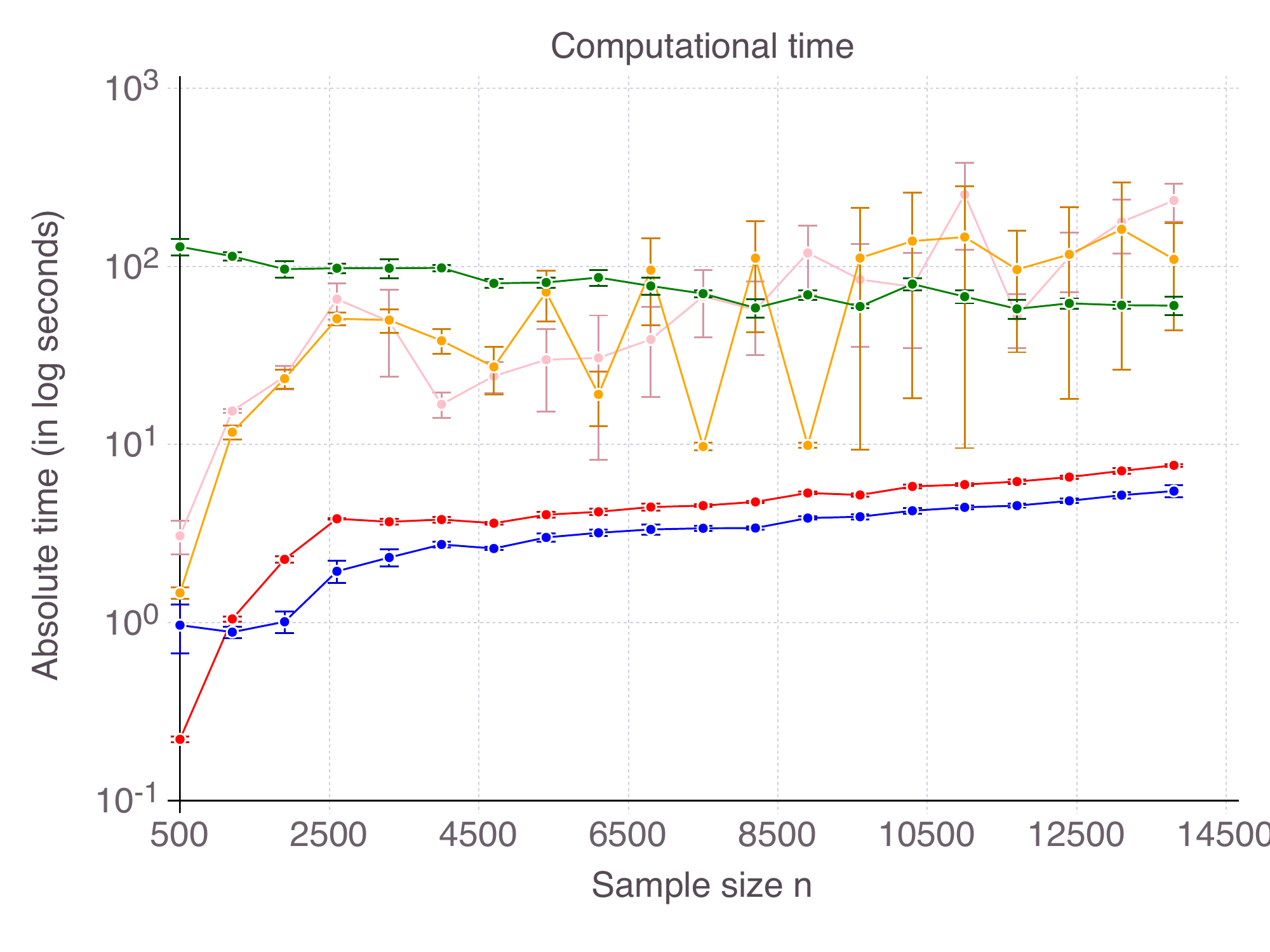}
	\caption{High noise}
\end{subfigure} %
\caption{Relative (left panel) and absolute (right panel) computational timeas $n$ increases, for CIO (in green), SS (in blue with $T_{max}=200$), ENet (in red), MCP (in orange), SCAD (in pink) with OLS loss. We average results over $10$ data sets.}
\label{fig:RegHardFixTime}
\end{figure*}

\subsubsection{Feature selection with cross-validated support size} 
\label{sec:regression.supp.nomic.cv}
We compare all methods when $k_{true}$ is no longer given and needs to be cross-validated from the data itself. 
Figure \ref{fig:RegHardCV} on page \pageref{fig:RegHardCV} reports the results of the cross-validation procedure for { increasing} $n$. In terms of accuracy (left panel), all four methods are relatively equivalent and demonstrate a clear convergence: $A\rightarrow 1$ as $n \rightarrow \infty$. On false detection rate however (right panel), behaviors vary among methods. Cardinality-constrained estimators achieve the lowest false detection rate ($0-30\%$), followed by MCP ($10-60\%$), SCAD ($20-70\%$) and then ENet (c.$80\%$). In case of ENet, this behavior was expected, for $\ell_1$-estimators are provably inconsistent, so that $FDR$ must be positive when $A=1$. 
\begin{figure*}
\centering
\begin{subfigure}[h]{\linewidth}
	\centering
	\includegraphics[width=.45\linewidth]{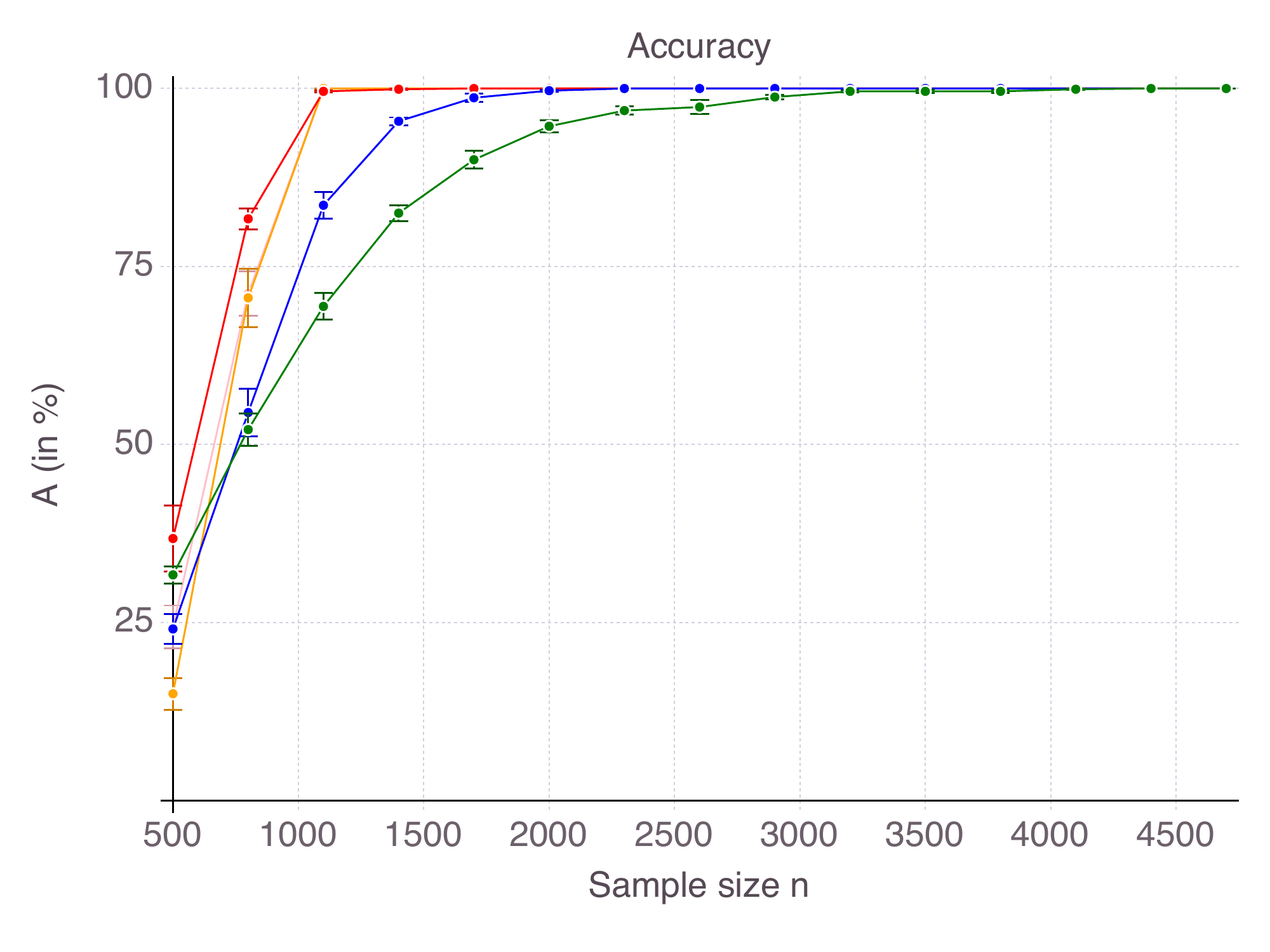}
	\includegraphics[width=.45\linewidth]{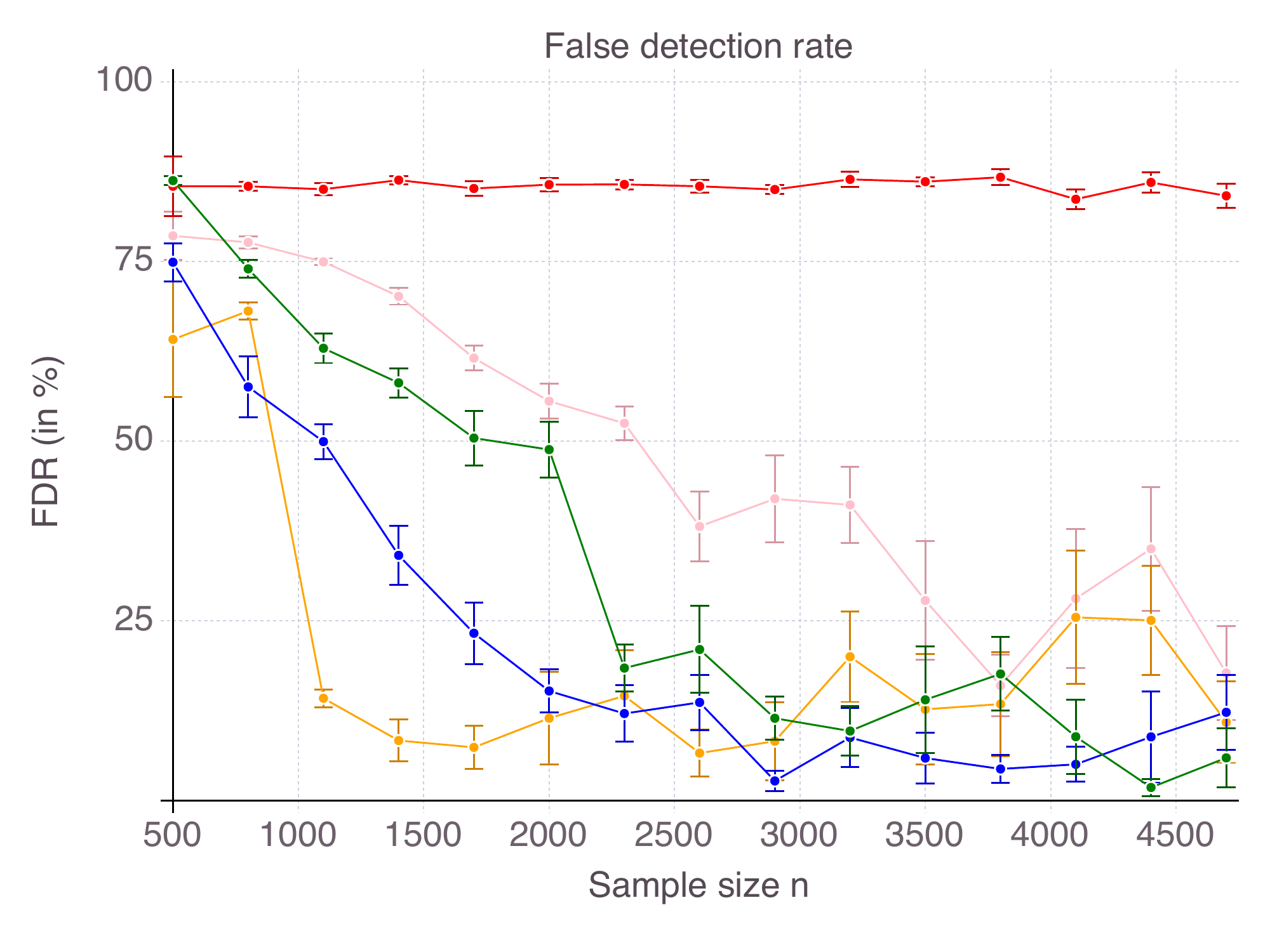}
	\caption{Low noise}
\end{subfigure} %

\begin{subfigure}[h]{\linewidth}
	\centering
	\includegraphics[width=.45\linewidth]{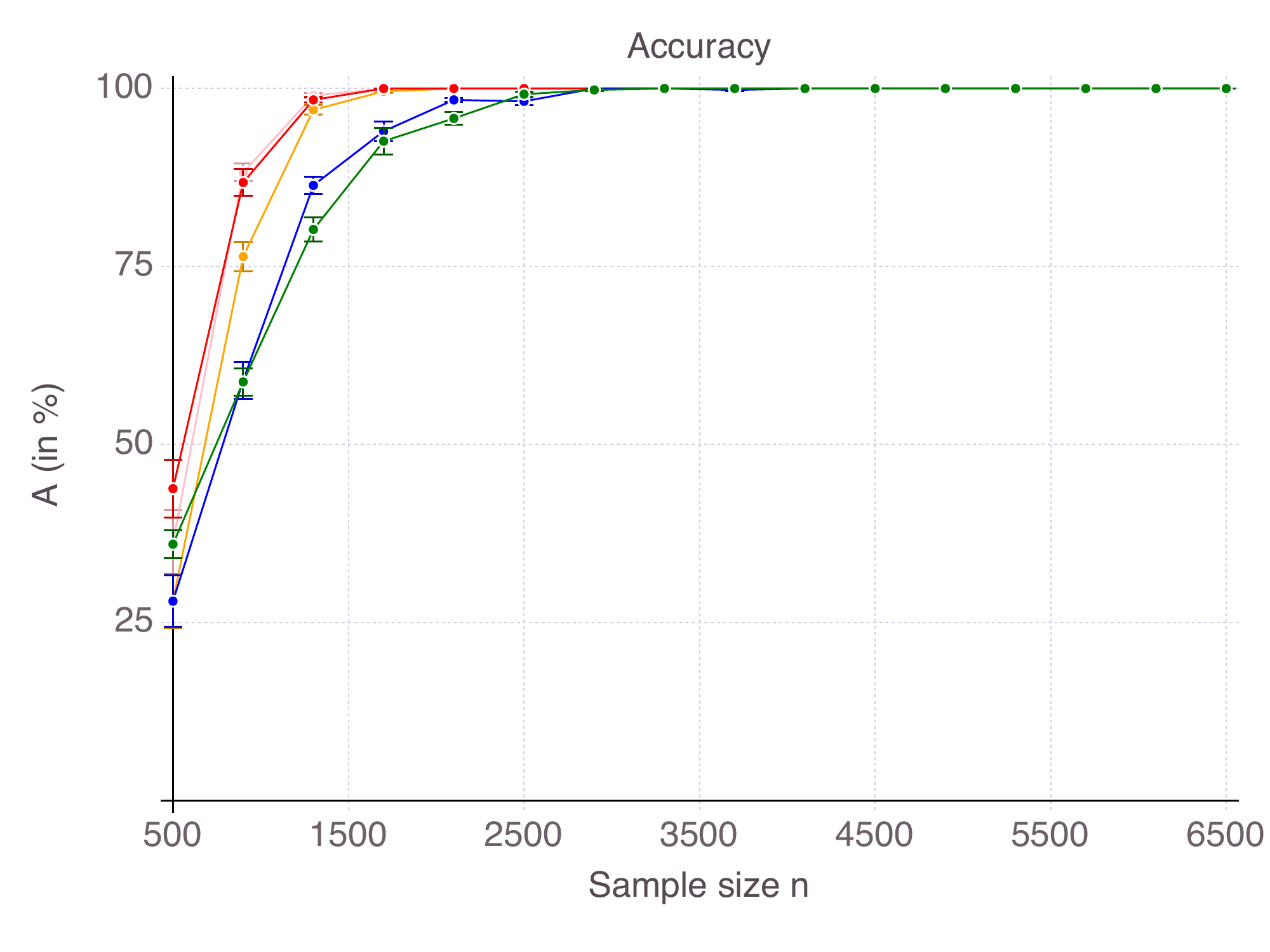}
	\includegraphics[width=.45\linewidth]{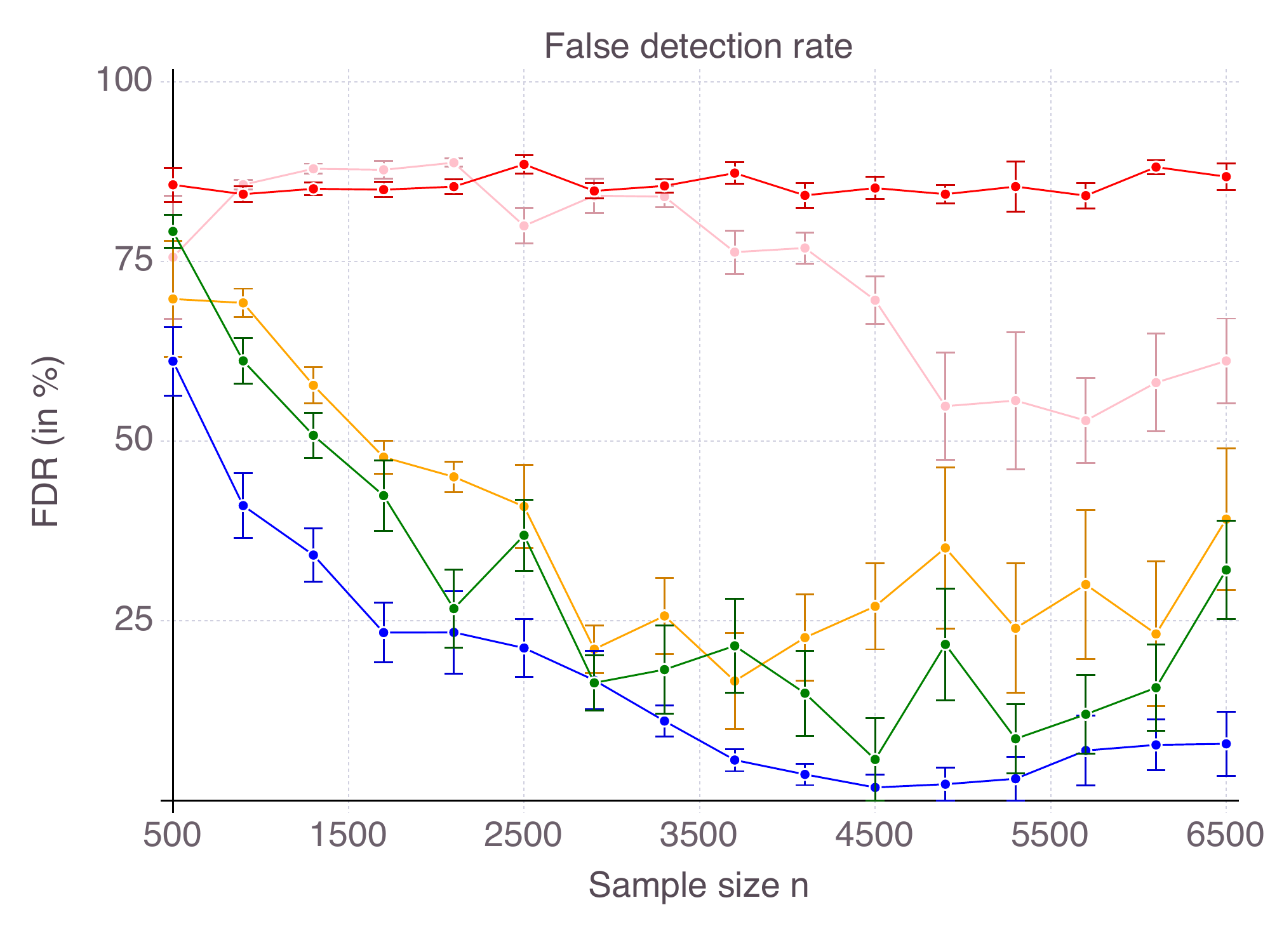}
	\caption{Medium noise}
\end{subfigure} %

\begin{subfigure}[h]{\linewidth}
	\centering
	\includegraphics[width=.45\linewidth]{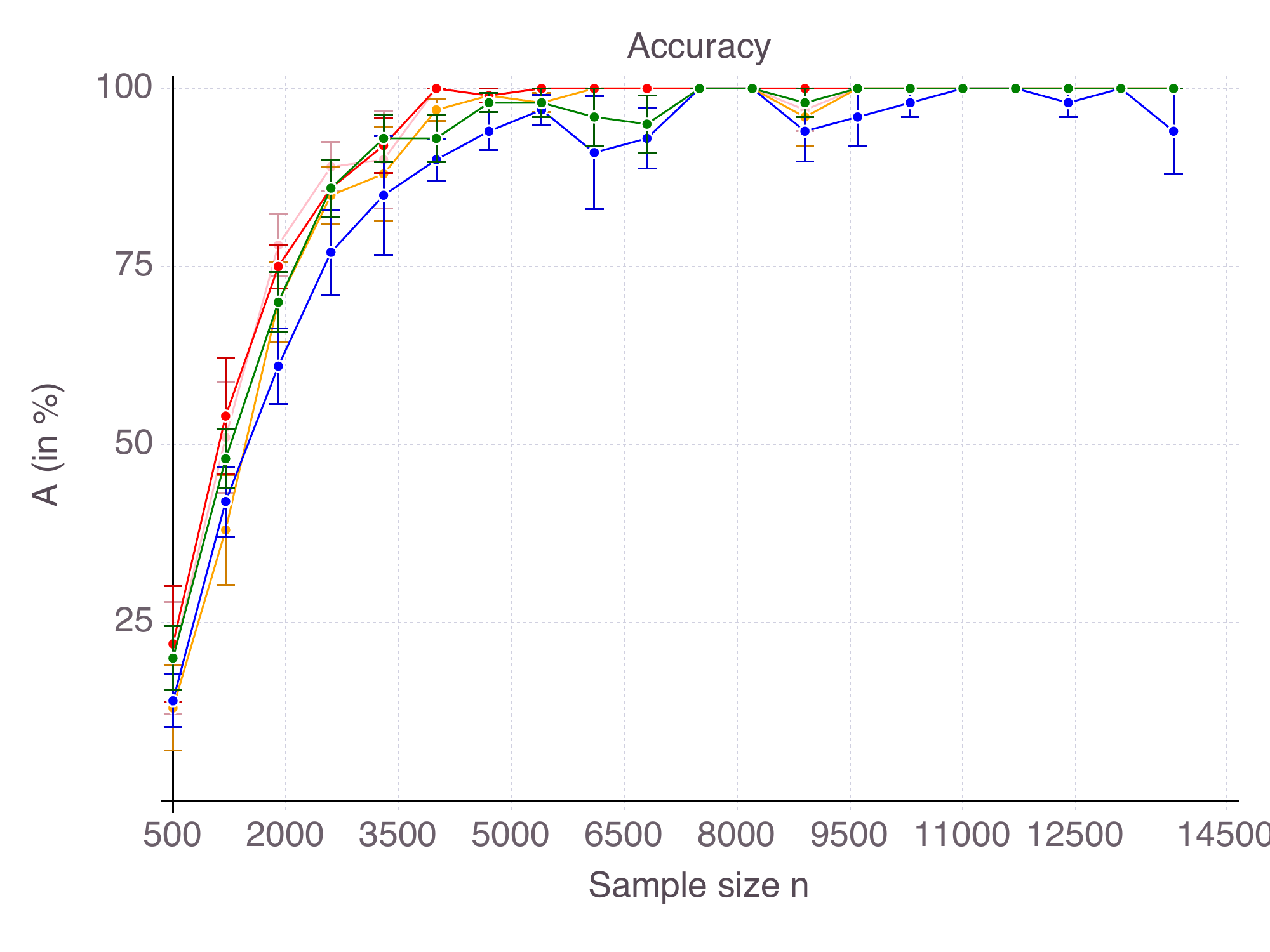}
	\includegraphics[width=.45\linewidth]{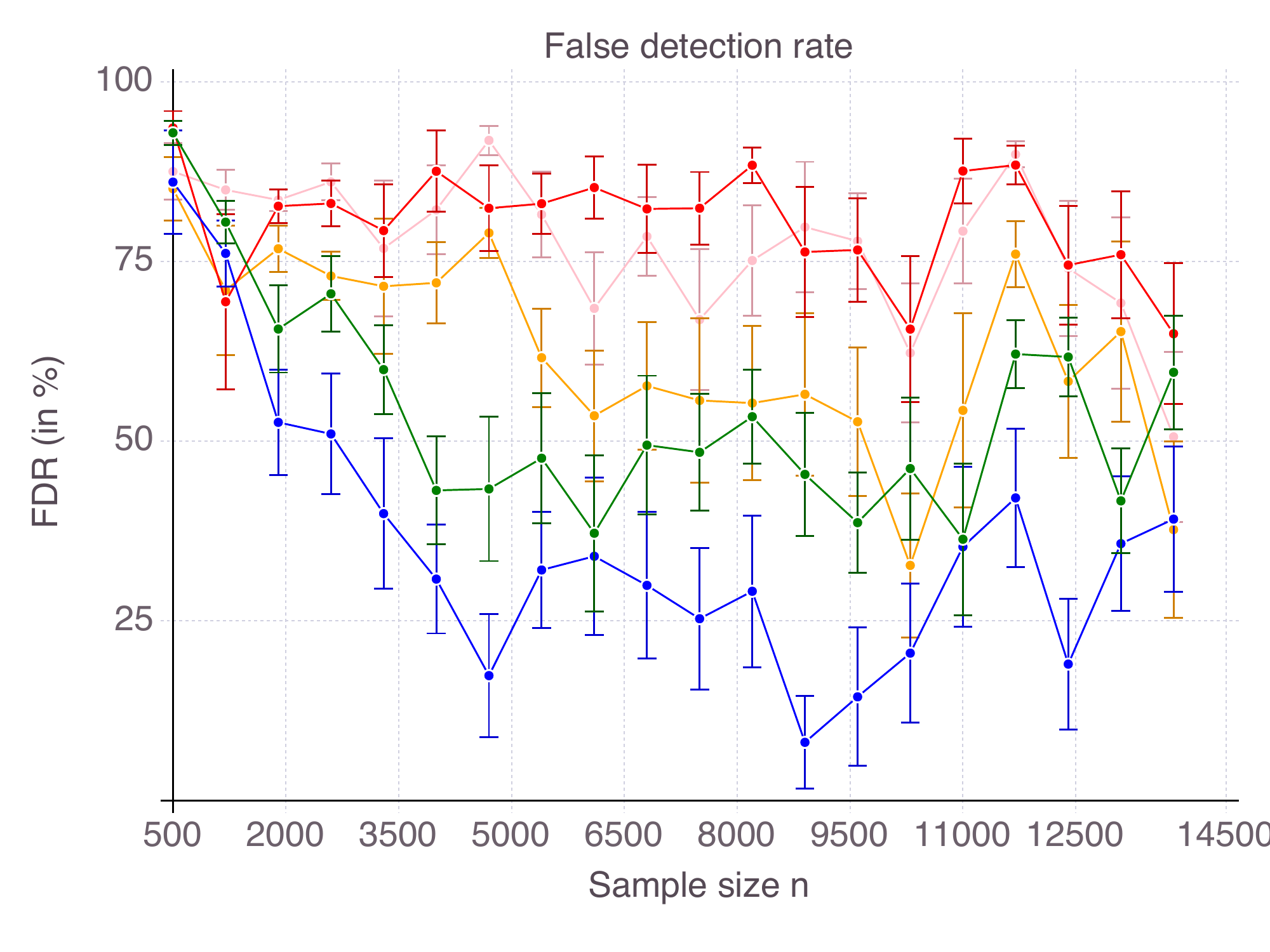}
	\caption{High noise}
\end{subfigure} %
~
\caption{Accuracy $A$ (left panel) and false detection rate $FDR$ (right panel) as $n$ increases, for the CIO (in green), SS (in blue with $T_{max}=150$), ENet (in red), MCP (in orange), SCAD (in pink) with OLS loss. We average results over $10$ data sets.}
\label{fig:RegHardCV}
\end{figure*}

\subsection{Real-world design matrix $X$}
\label{sec:regression.supp.real}
In this section, we consider a real-world design matrix $X$ and generated synthetic noisy signals $Y$ for 10 levels of noise. Figure \ref{fig:RegCancer2} (p. \pageref{fig:RegCancer2}) represents the out-of-sample $MSE$ of all five methods as $SNR$ increases. As mentioned, the difference in $MSE$ between methods is far less acute than the difference observed in terms of accuracy and false detection. 
\begin{figure*}
\centering
	\includegraphics[width=.5\linewidth]{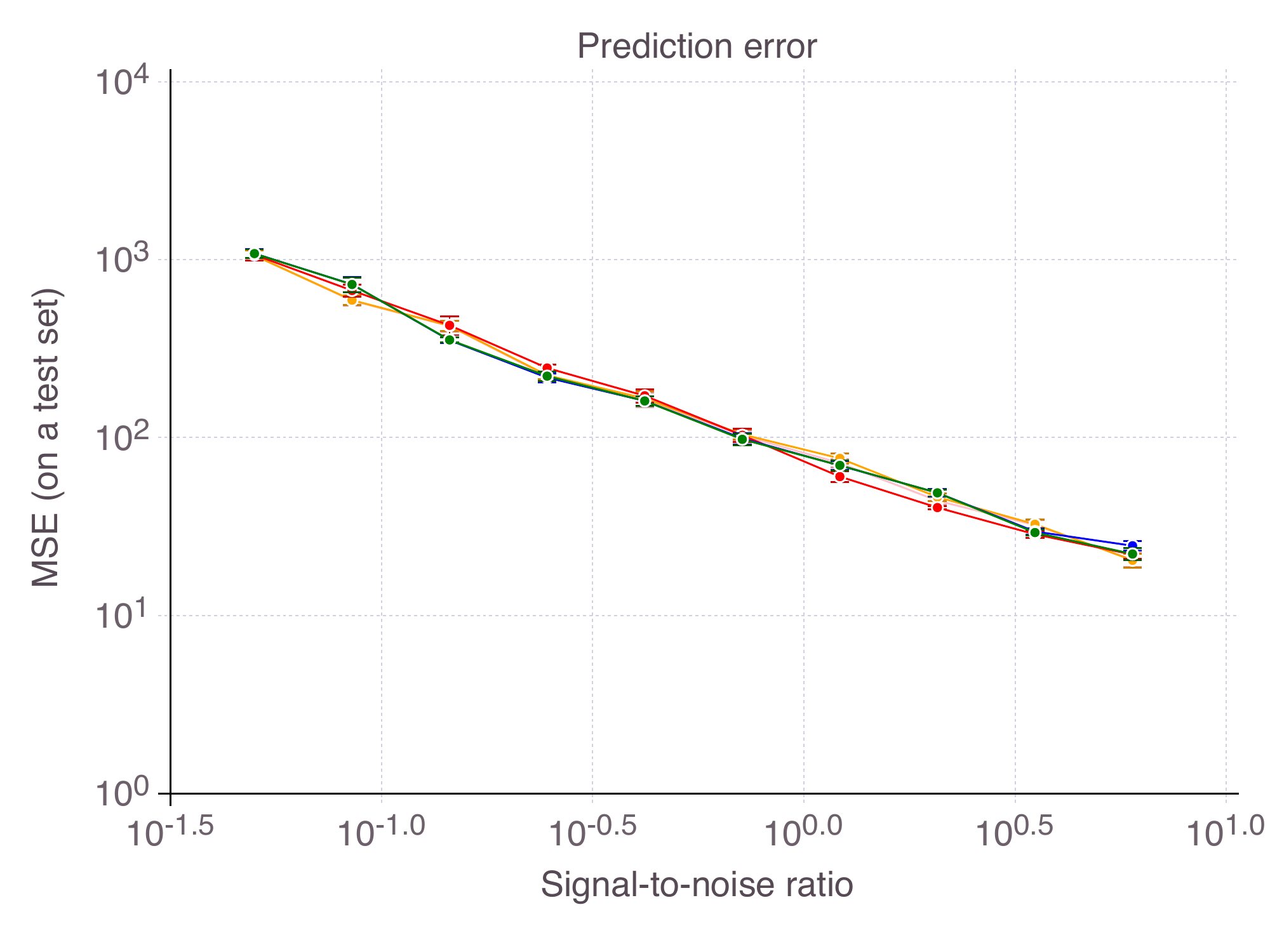}
\caption{Out-of-sample $MSE$ as $SNR$ increases, for the CIO (in green), SS (in blue with $T_{max}=150$), ENet (in red), MCP (in orange), SCAD (in pink) with OLS loss. We average results over $10$ data sets with $SNR=0.05,...,6$, $k_{true}=50$.}
\label{fig:RegCancer2}
\end{figure*}

\newpage
\section{Numerical experiments for classification - Supplementary material}
\label{sec:classification.supp}
\subsection{Synthetic data satisfying mutual incoherence condition} 
\subsubsection{Feature selection with a given support size} \label{sec:classification.supp.mic.fix}
We first consider the case when the cardinality $k$ of the support to be returned is given and equal to the true sparsity $k_{true}$ for all methods. 

As shown on Figure \ref{fig:ClassFixTF} (p. \pageref{fig:ClassFixTF}), all methods converge in terms of accuracy. That is their ability to select correct features as measured by $A$ smoothly converges to $1$ with an increasing number of observations $n\rightarrow \infty$. Compared to regression, the difference in accuracy between methods is much narrower. MCP now slightly dominates all methods, including CIO and SS. The suboptimality gap between the discrete optimization method and its Boolean relaxation appears to be much smaller as well and the two methods perform almost identically. 

\begin{figure*}
\centering
\begin{subfigure}[t]{.45\linewidth}
	\centering
	\includegraphics[width=\linewidth]{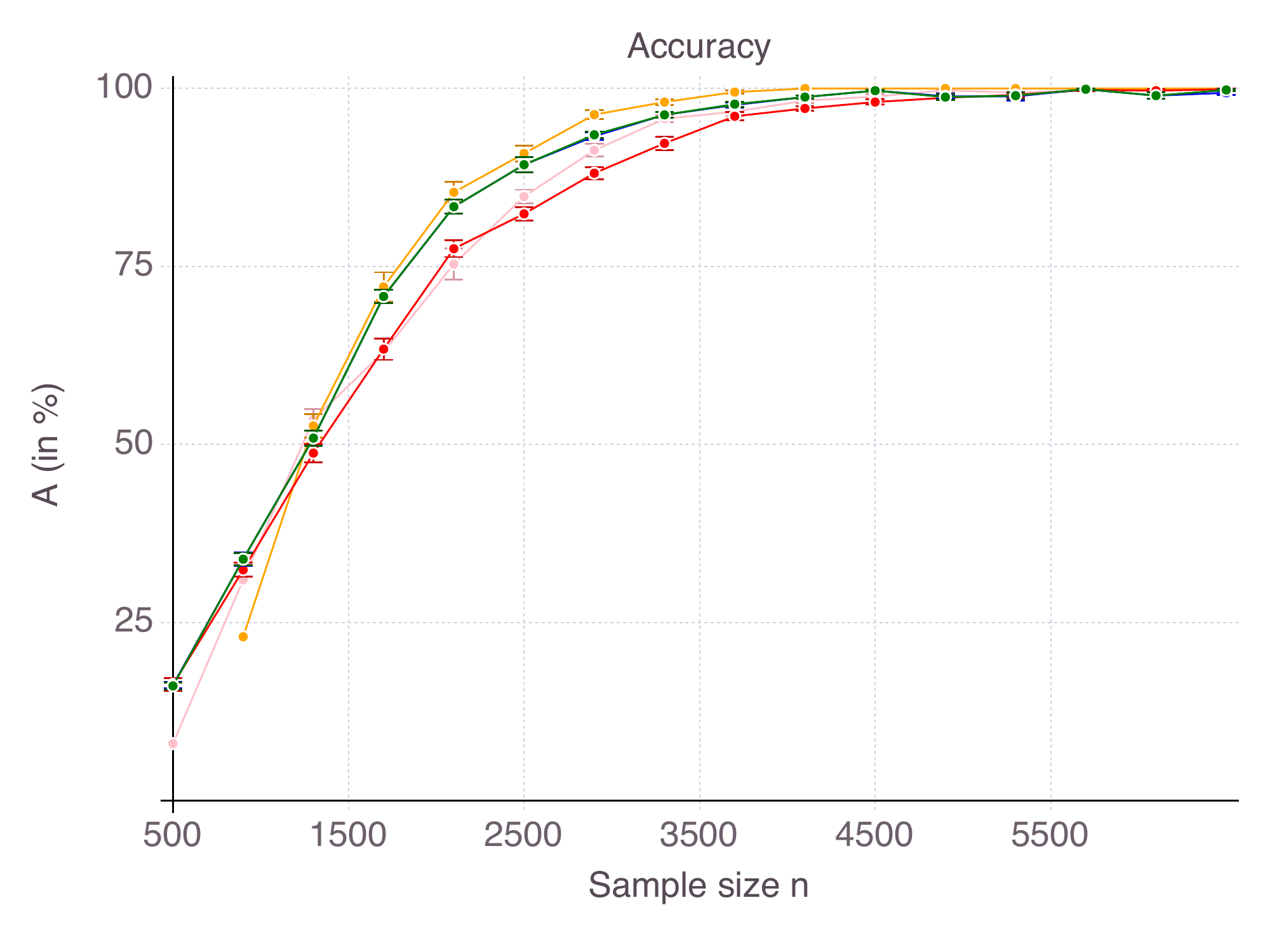}
	\caption{Low noise, low correlation}
\end{subfigure} %
~
\begin{subfigure}[t]{.45\linewidth}
	\centering
	\includegraphics[width=\linewidth]{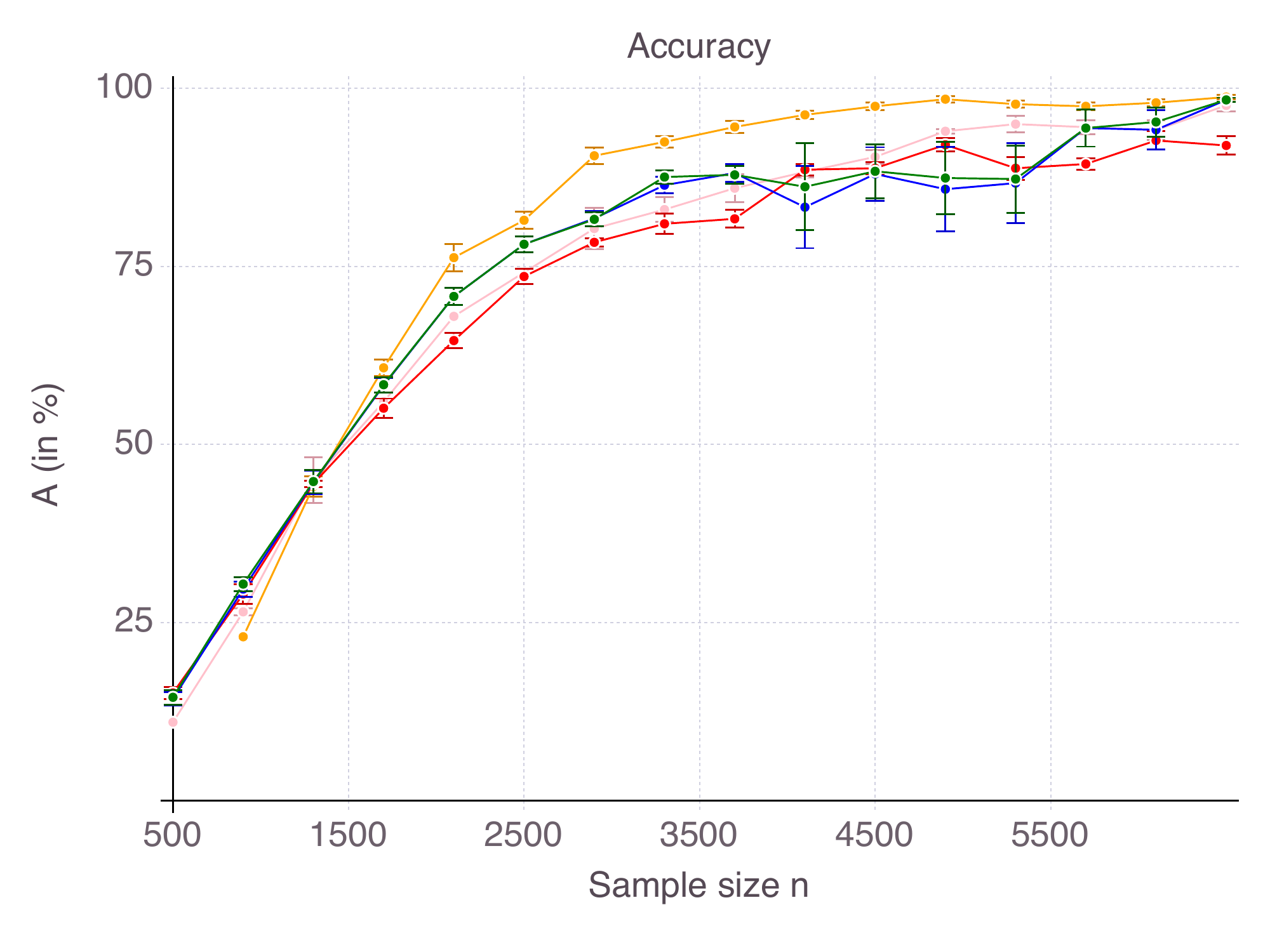}
	\caption{Low noise, high correlation}
\end{subfigure}

\begin{subfigure}[t]{.45\linewidth}
	\centering
	\includegraphics[width=\linewidth]{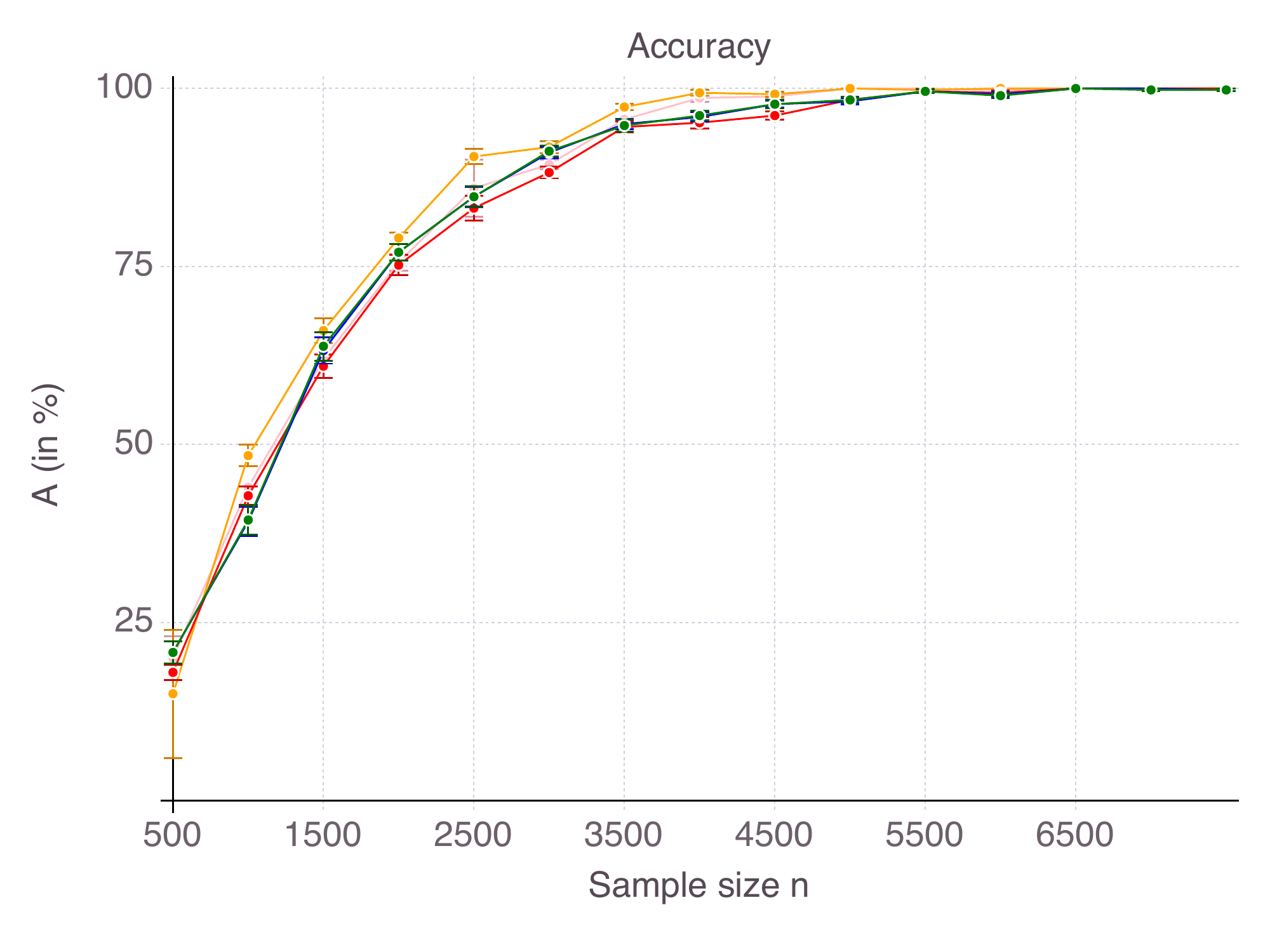}
	\caption{Medium noise, low correlation}
\end{subfigure} %
~
\begin{subfigure}[t]{.45\linewidth}
	\centering
	\includegraphics[width=\linewidth]{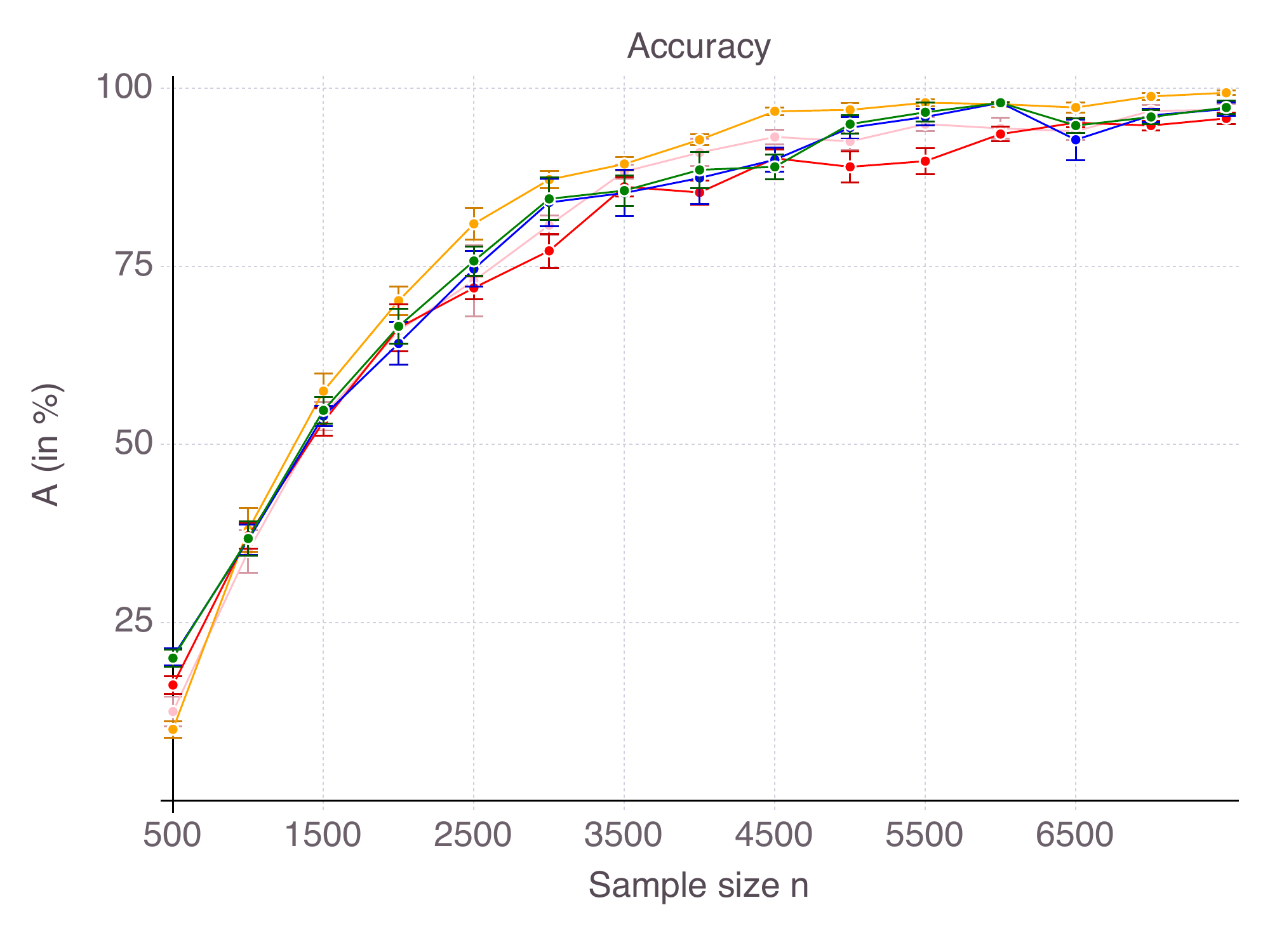}
	\caption{Medium noise, high correlation}
\end{subfigure}

\begin{subfigure}[t]{.45\linewidth}
	\centering
	\includegraphics[width=\linewidth]{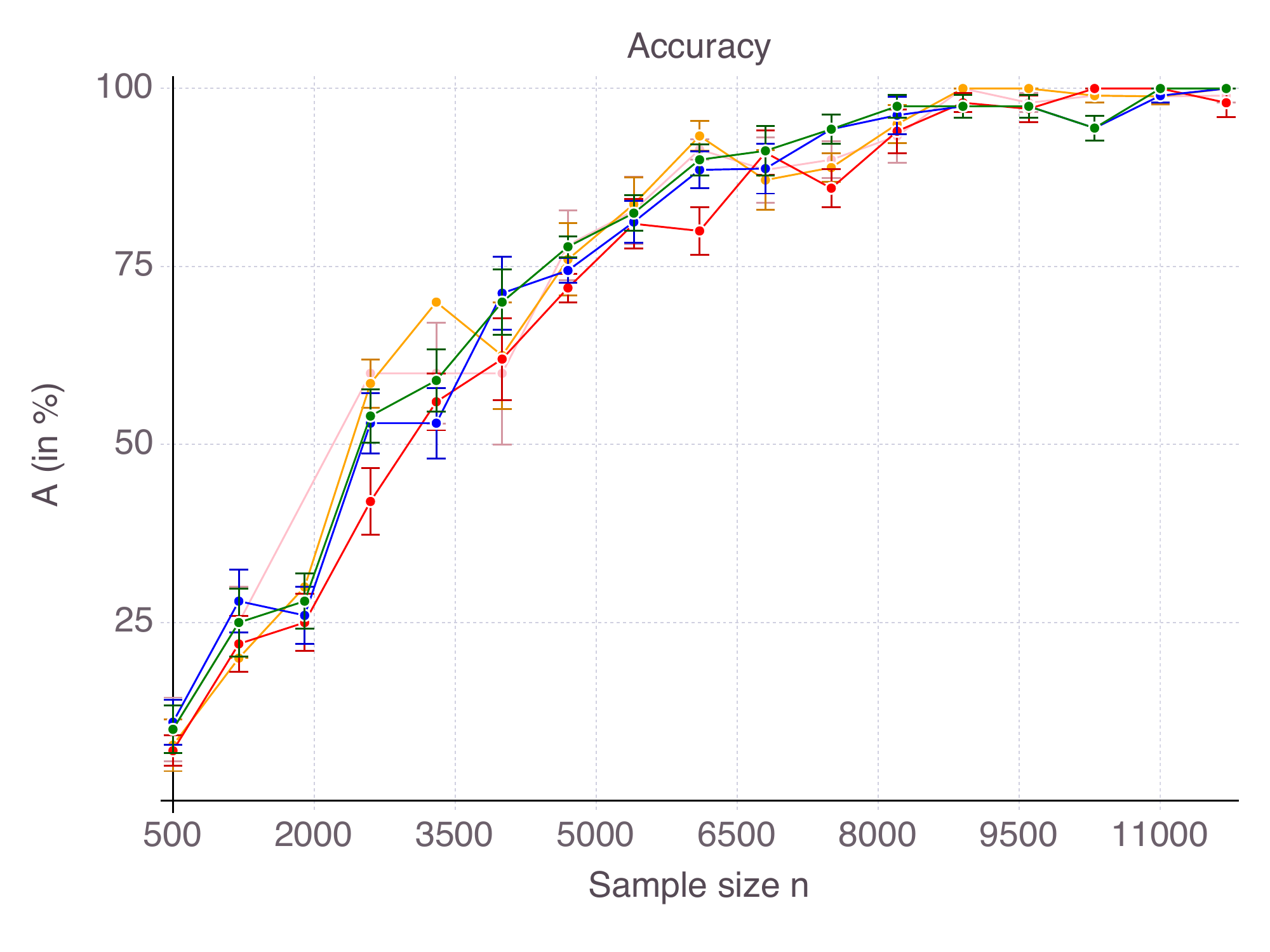}
	\caption{High noise, low correlation}
\end{subfigure} %
~
\begin{subfigure}[t]{.45\linewidth}
	\centering
	\includegraphics[width=\linewidth]{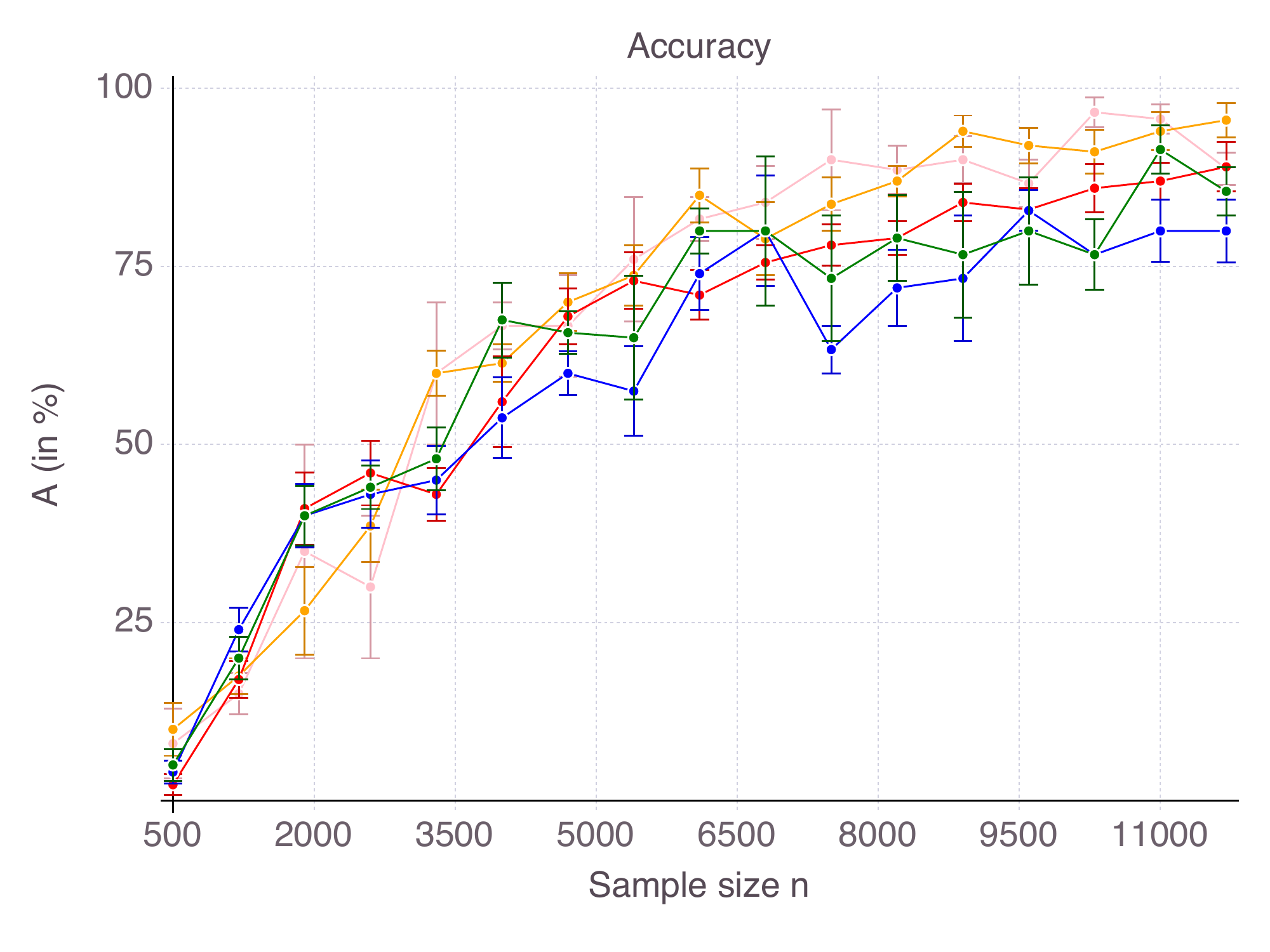}
	\caption{High noise, high correlation}
\end{subfigure}
\caption{Accuracy as $n$ increases, for the CIO (in green), SS (in blue with $T_{max}=200$) with Hinge loss , ENet (in red), MCP (in orange), SCAD (in pink) with logistic loss. We average results over $10$ data sets.}
\label{fig:ClassFixTF}
\end{figure*}

Figure \ref{fig:ClassFixTime} on page \pageref{fig:ClassFixTime} reports relative computational time compared to \verb|glmnet| { in log scale}. It should be kept in mind that we restricted the cutting-plane algorithm to a $180$-second time limit and the sub-gradient algorithm to $T_{max}=200$ iterations. \verb|glmnet| is still the fastest method in general, but it should be emphasized that other methods terminate in times at most two orders of magnitude larger, which is often an affordable price to pay in practice. Combined with results in accuracy from Figure \ref{fig:ClassFixTF}, such an observation speaks in favor of a wider use of cardinality-constrained or non-convex formulations in data analysis practice. As previously mentioned, for the sub-gradient algorithm, using an additional stopping criterion would drastically cut computational time (by a factor 2 at least) but would also deteriorate the quality of the solution significantly for such classification problems. 
\begin{figure*}
\centering
\begin{subfigure}[t]{.45\linewidth}
	\centering
	\includegraphics[width=\linewidth]{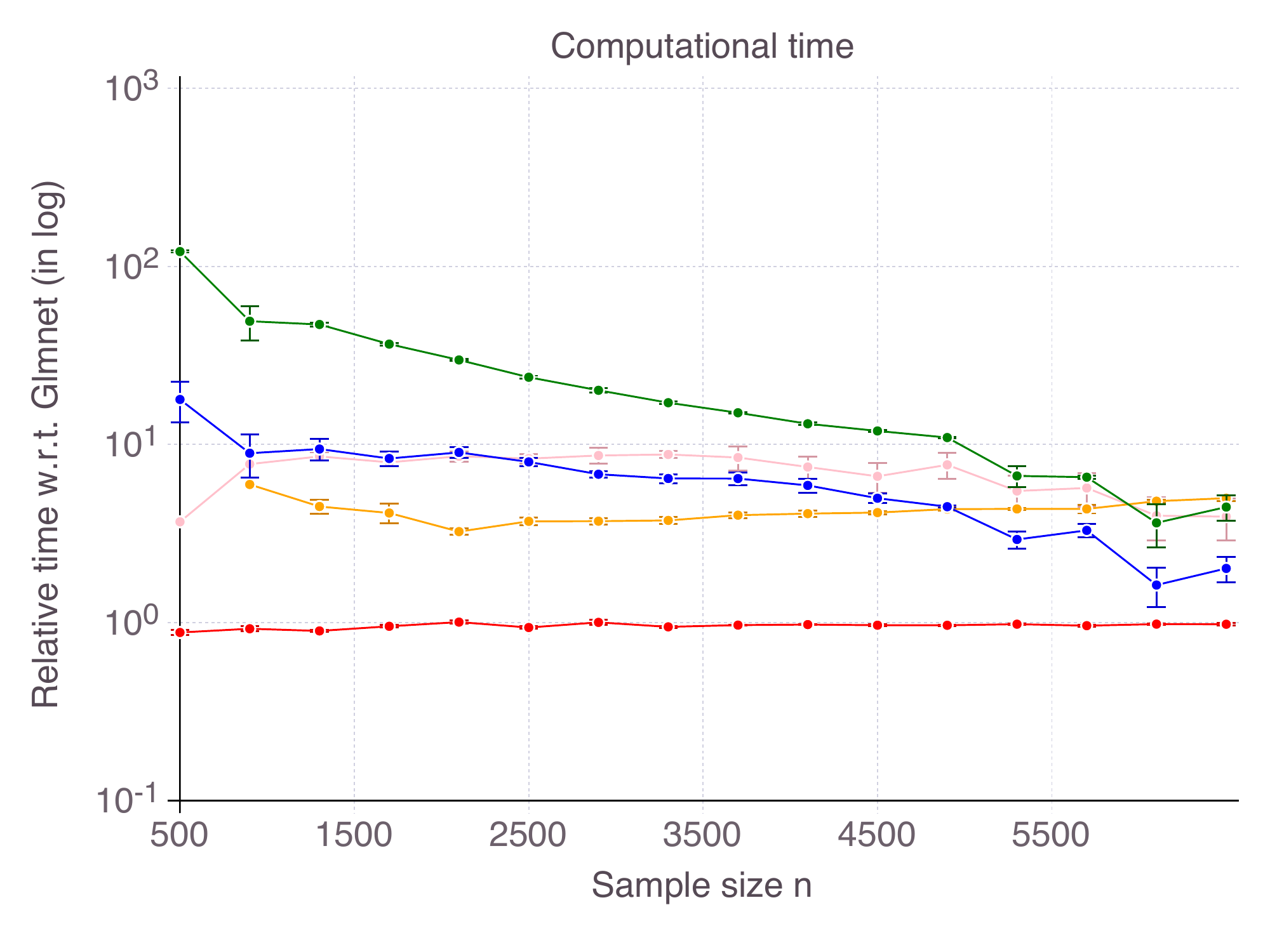}
	\caption{Low noise, low correlation}
\end{subfigure} %
~
\begin{subfigure}[t]{.45\linewidth}
	\centering
	\includegraphics[width=\linewidth]{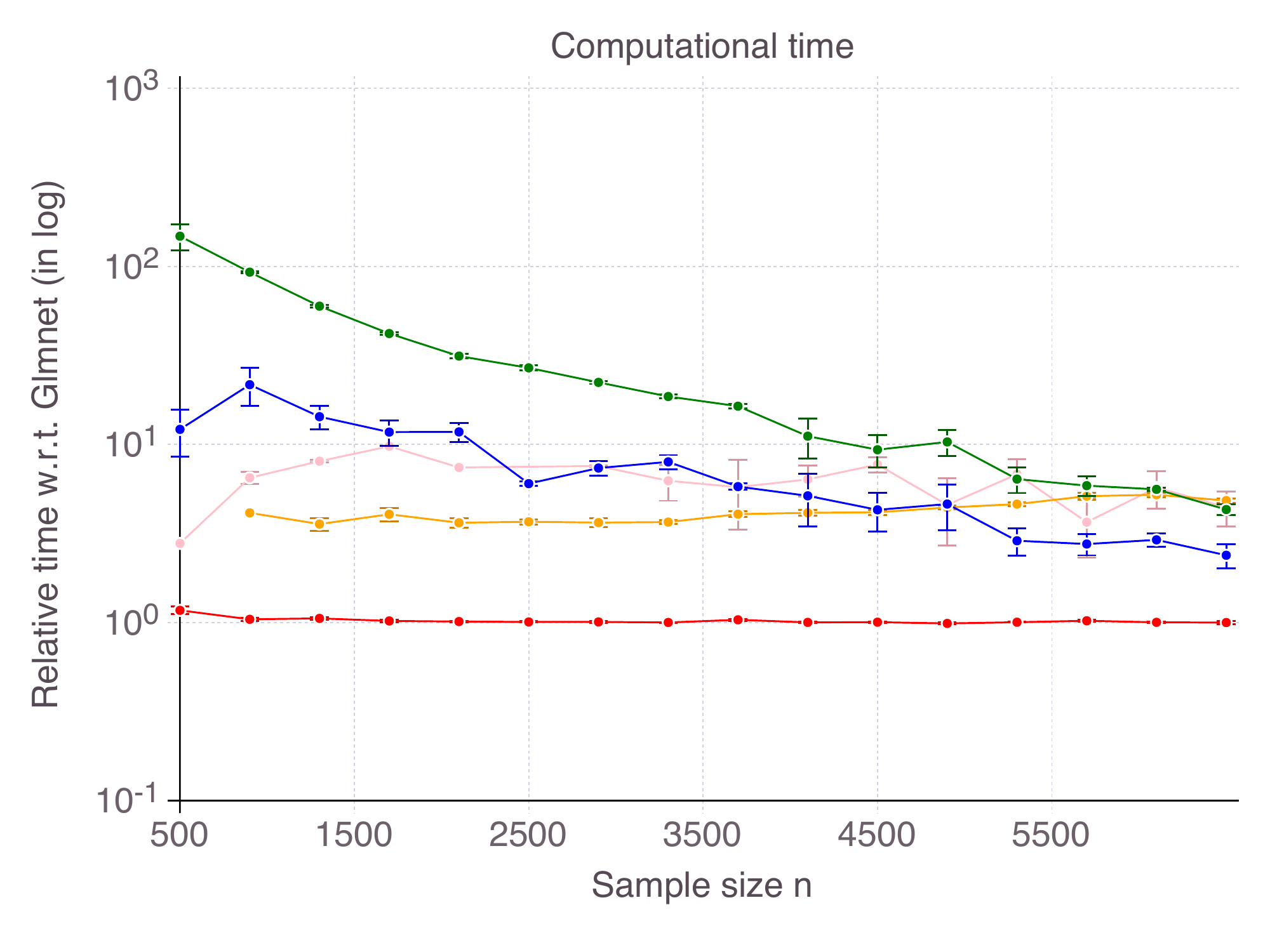}
	\caption{Low noise, high correlation}
\end{subfigure}

\begin{subfigure}[t]{.45\linewidth}
	\centering
	\includegraphics[width=\linewidth]{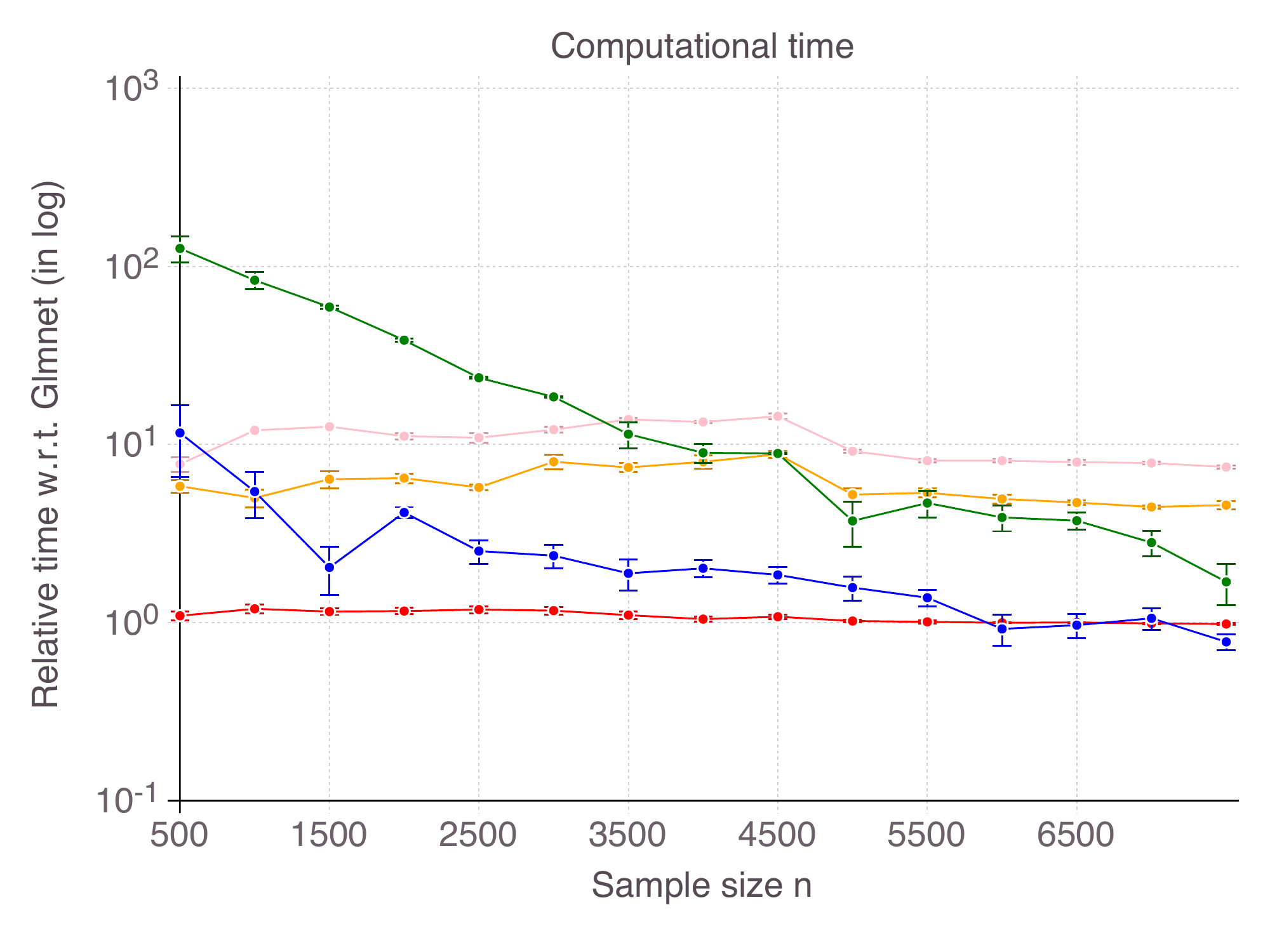}
	\caption{Medium noise, low correlation}
\end{subfigure} %
~
\begin{subfigure}[t]{.45\linewidth}
	\centering
	\includegraphics[width=\linewidth]{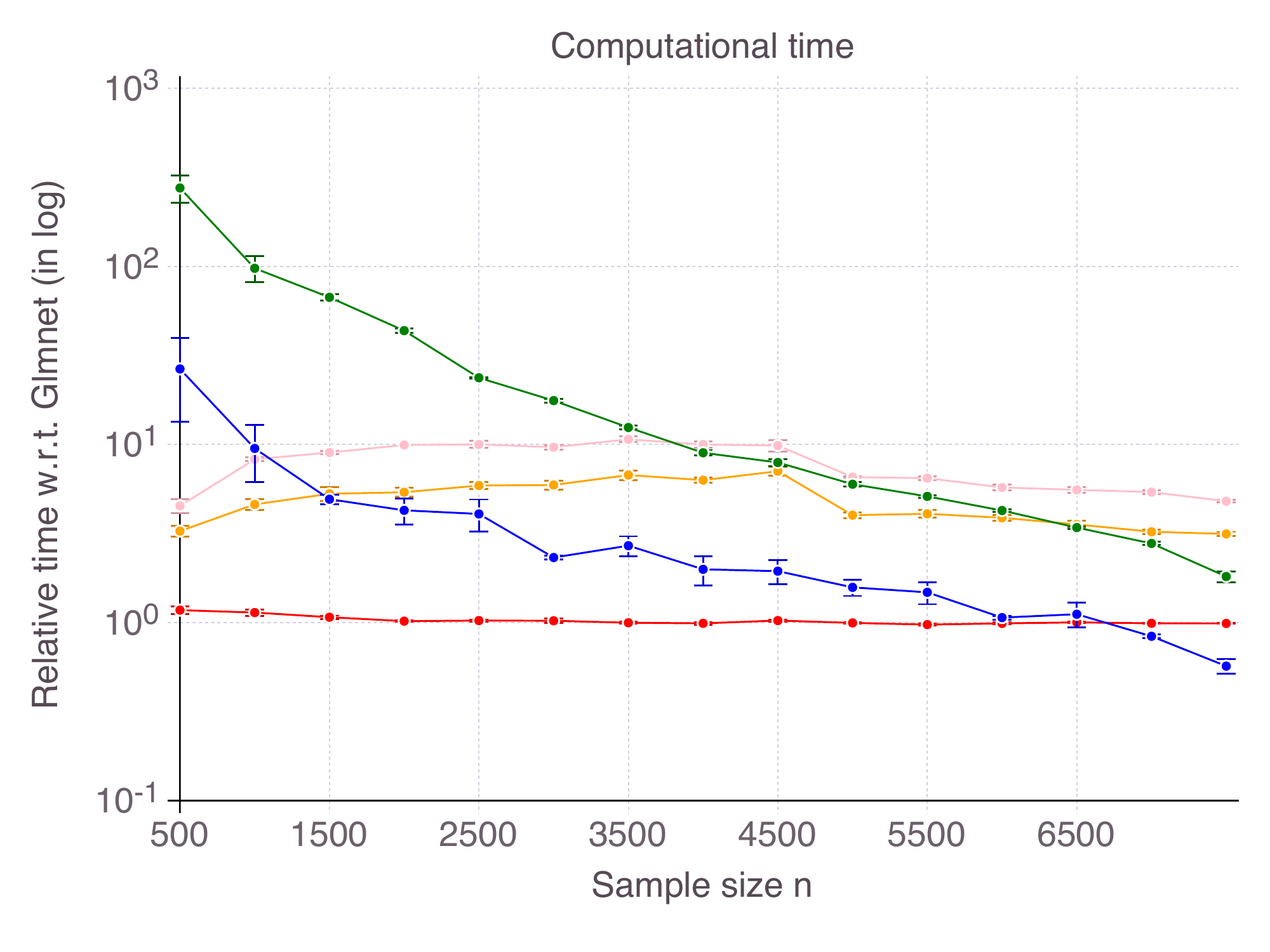}
	\caption{Medium noise, high correlation}
\end{subfigure}

\begin{subfigure}[t]{.45\linewidth}
	\centering
	\includegraphics[width=\linewidth]{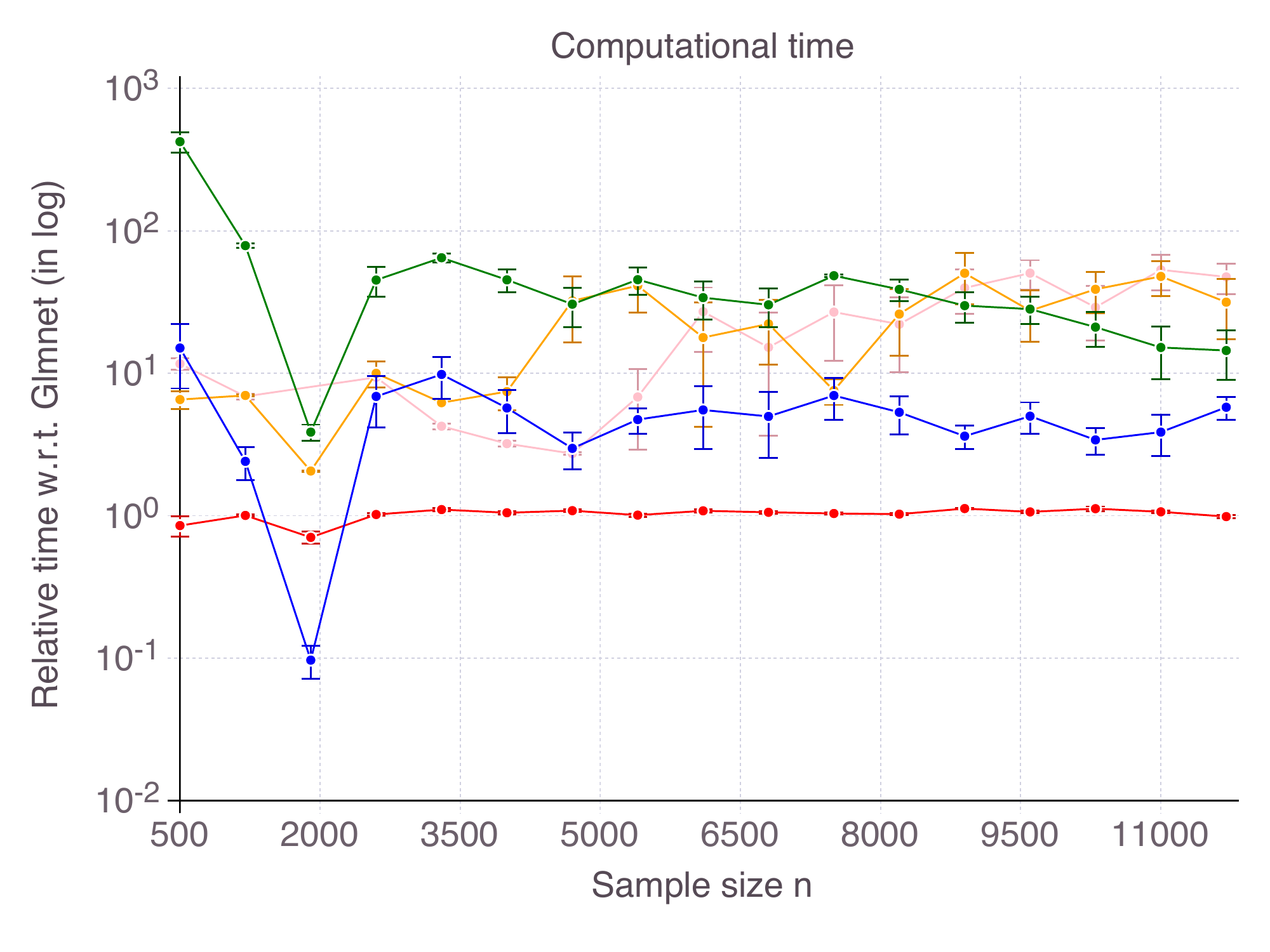}
	\caption{High noise, low correlation}
\end{subfigure} %
~
\begin{subfigure}[t]{.45\linewidth}
	\centering
	\includegraphics[width=\linewidth]{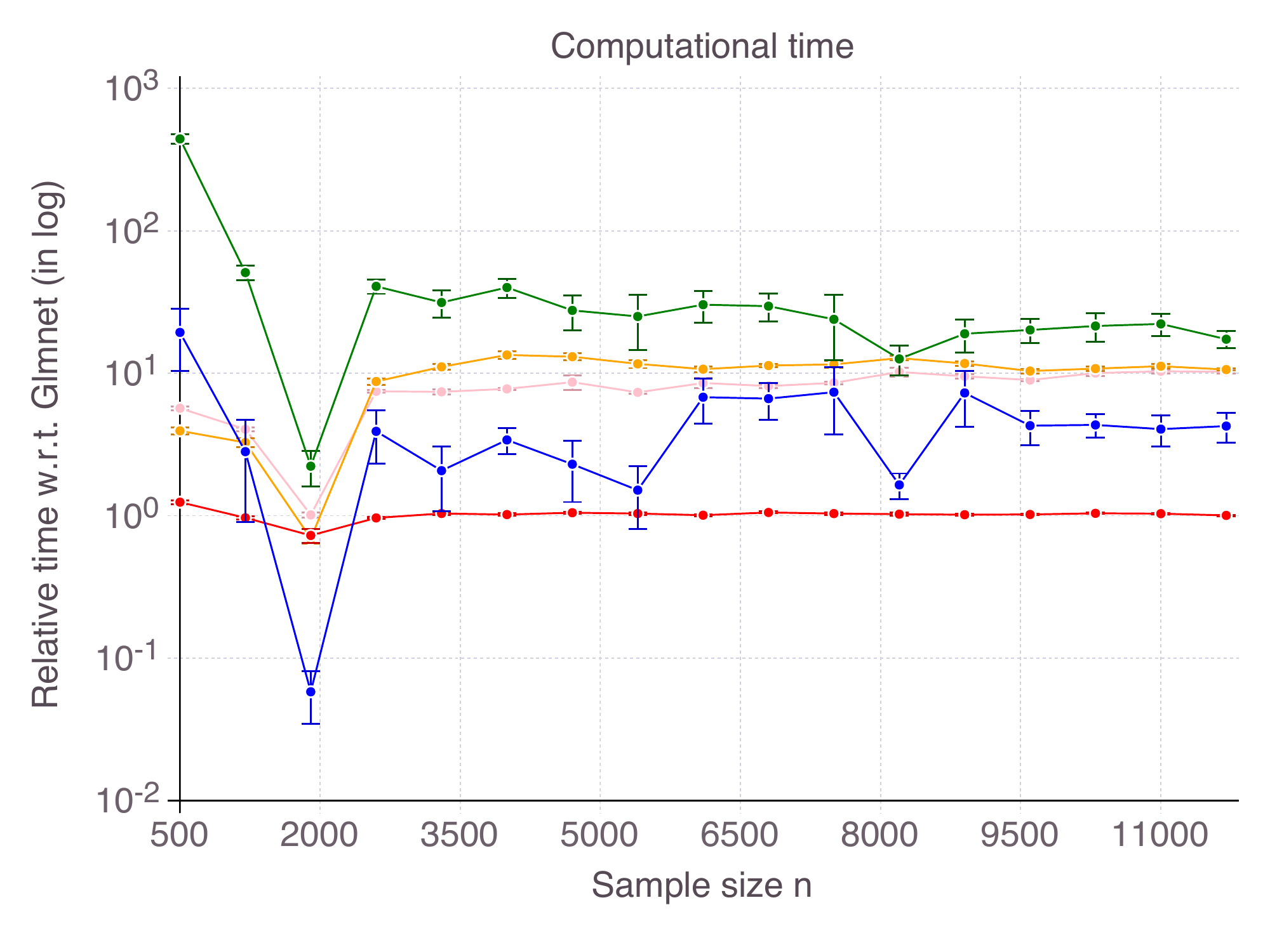}
	\caption{High noise, high correlation}
\end{subfigure}
\caption{Computational time relative to Lasso with glmnet as $n$ increases, for CIO (in green), SS (in blue with $T_{max}=200$) with Hinge loss, ENet (in red), MCP (in orange), SCAD (in pink) with logistic loss. We average results over $10$ data sets.}
\label{fig:ClassFixTime}
\end{figure*}
\begin{figure*}
\centering
\begin{subfigure}[t]{.45\linewidth}
	\centering
	\includegraphics[width=\linewidth]{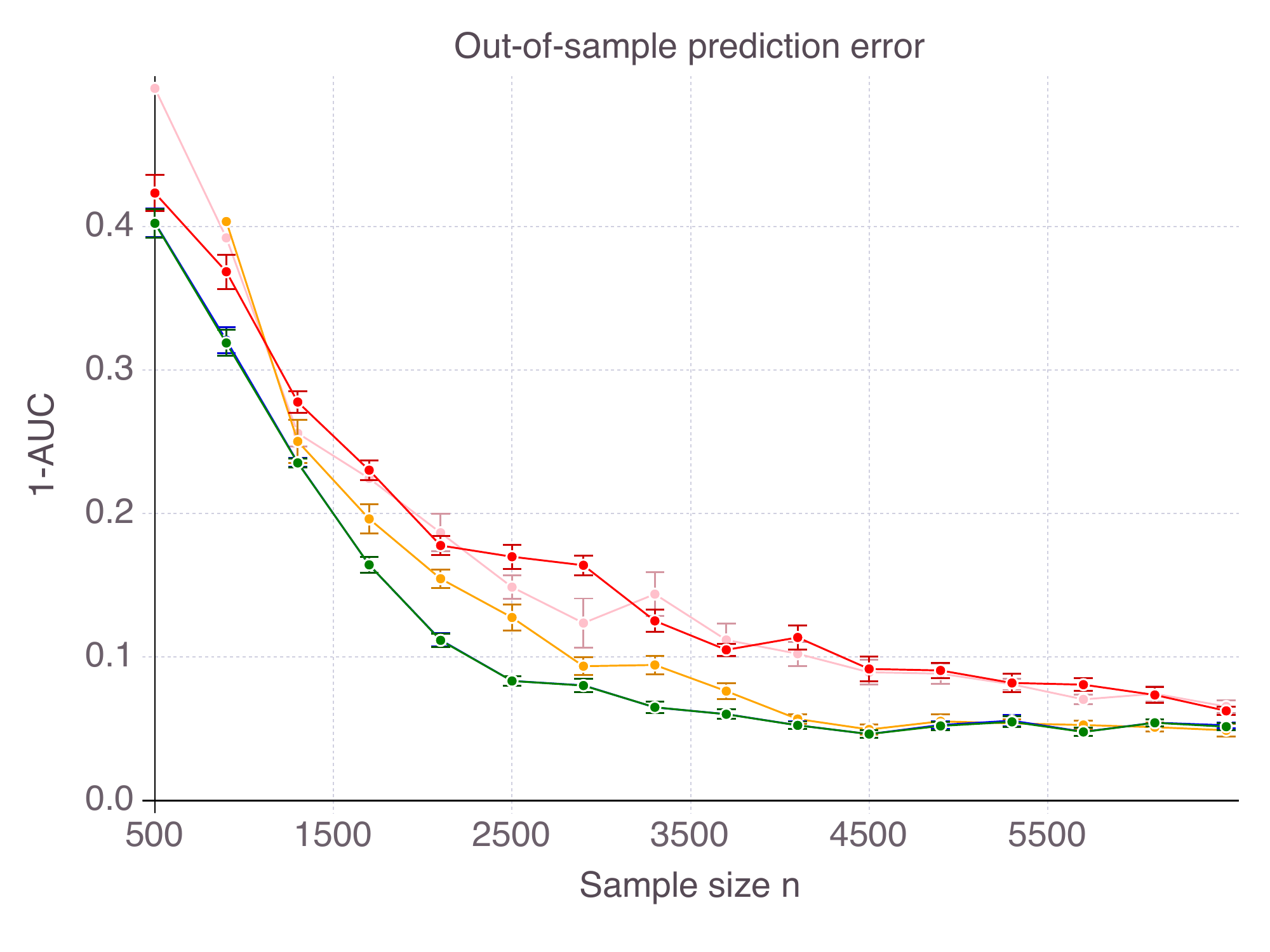}
	\caption{Low noise, low correlation}
\end{subfigure} %
~
\begin{subfigure}[t]{.45\linewidth}
	\centering
	\includegraphics[width=\linewidth]{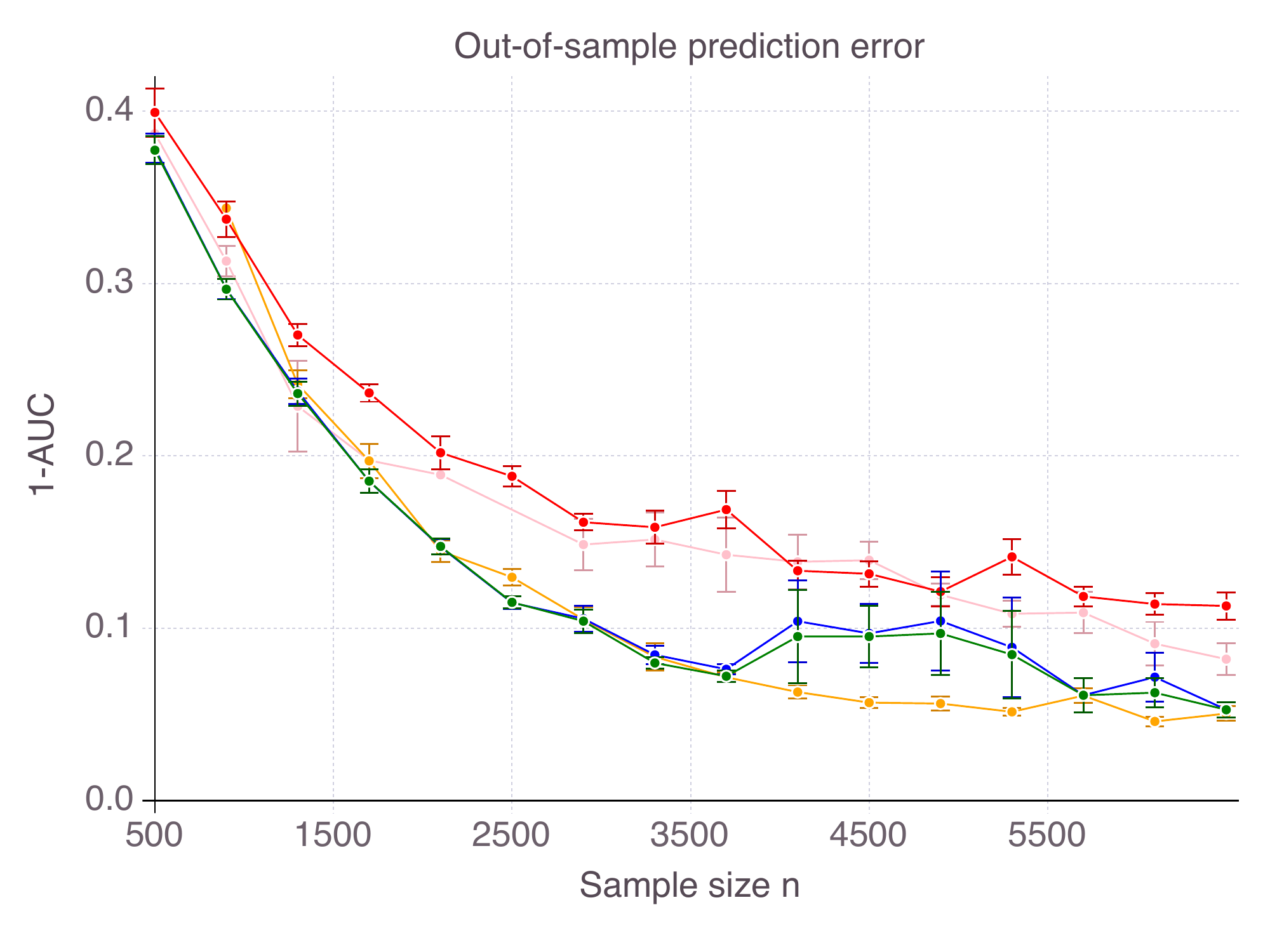}
	\caption{Low noise, high correlation}
\end{subfigure}

\begin{subfigure}[t]{.45\linewidth}
	\centering
	\includegraphics[width=\linewidth]{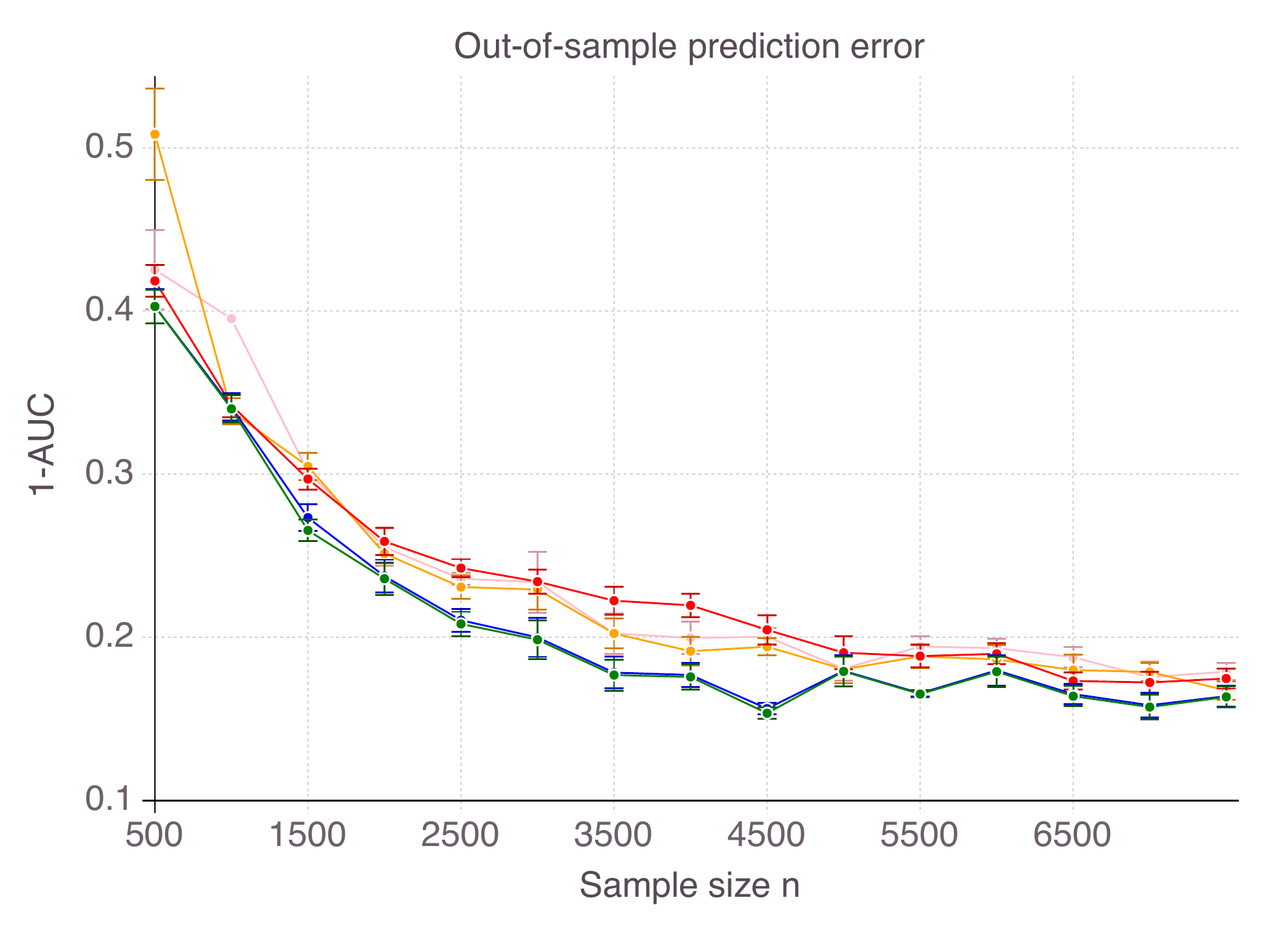}
	\caption{Medium noise, low correlation}
\end{subfigure} %
~
\begin{subfigure}[t]{.45\linewidth}
	\centering
	\includegraphics[width=\linewidth]{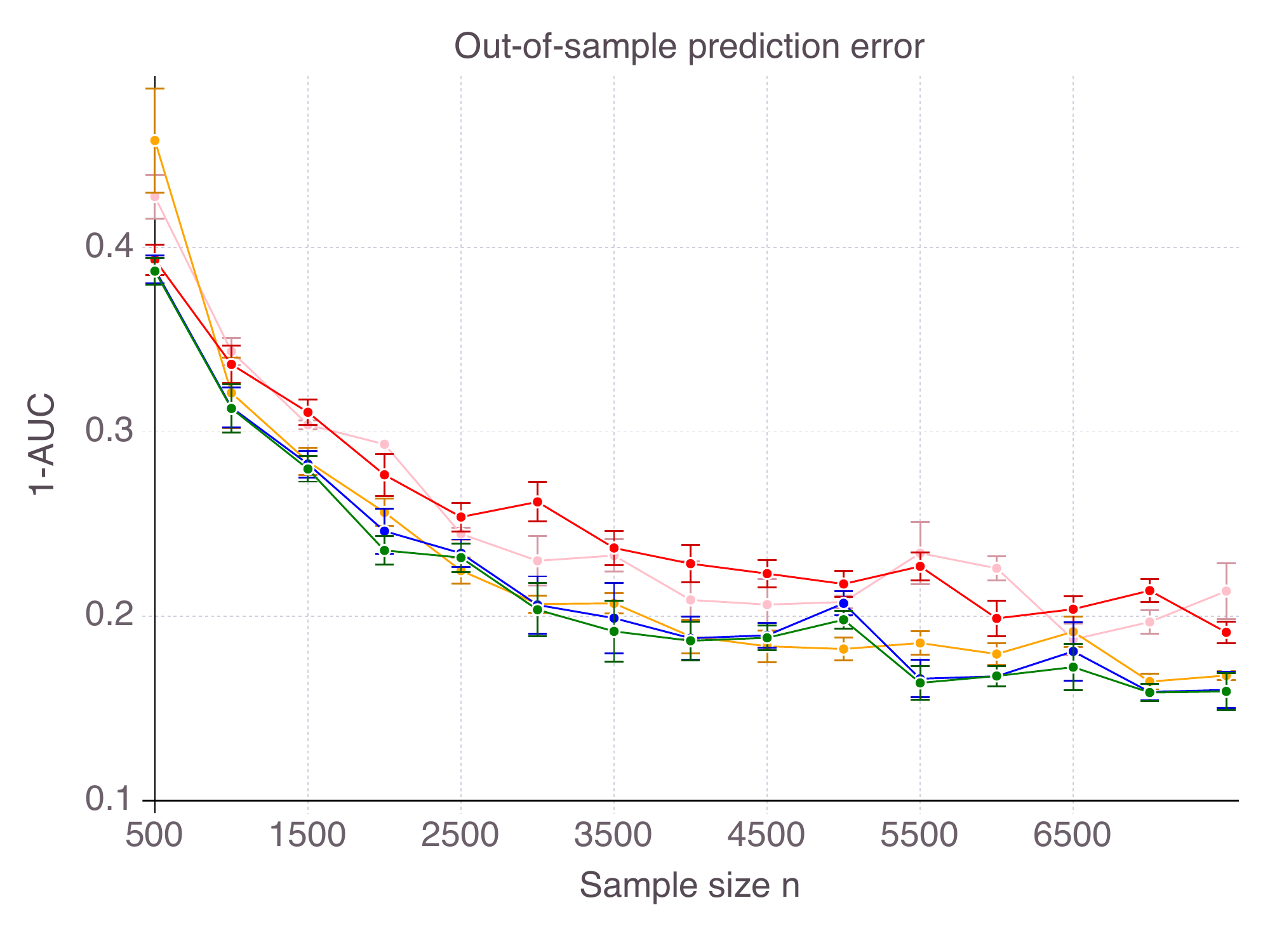}
	\caption{Medium noise, high correlation}
\end{subfigure}

\begin{subfigure}[t]{.45\linewidth}
	\centering
	\includegraphics[width=\linewidth]{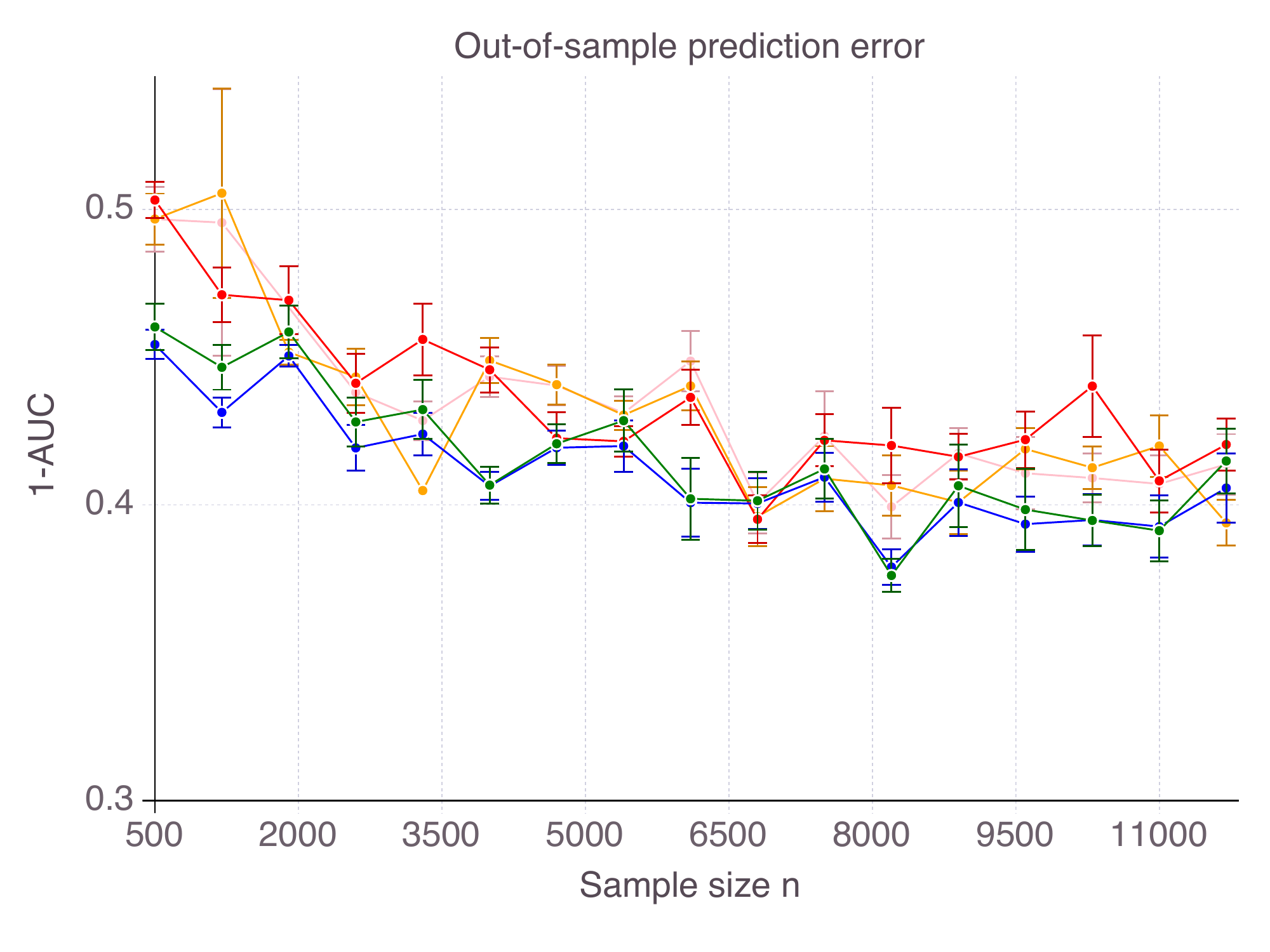}
	\caption{High noise, low correlation}
\end{subfigure} %
~
\begin{subfigure}[t]{.45\linewidth}
	\centering
	\includegraphics[width=\linewidth]{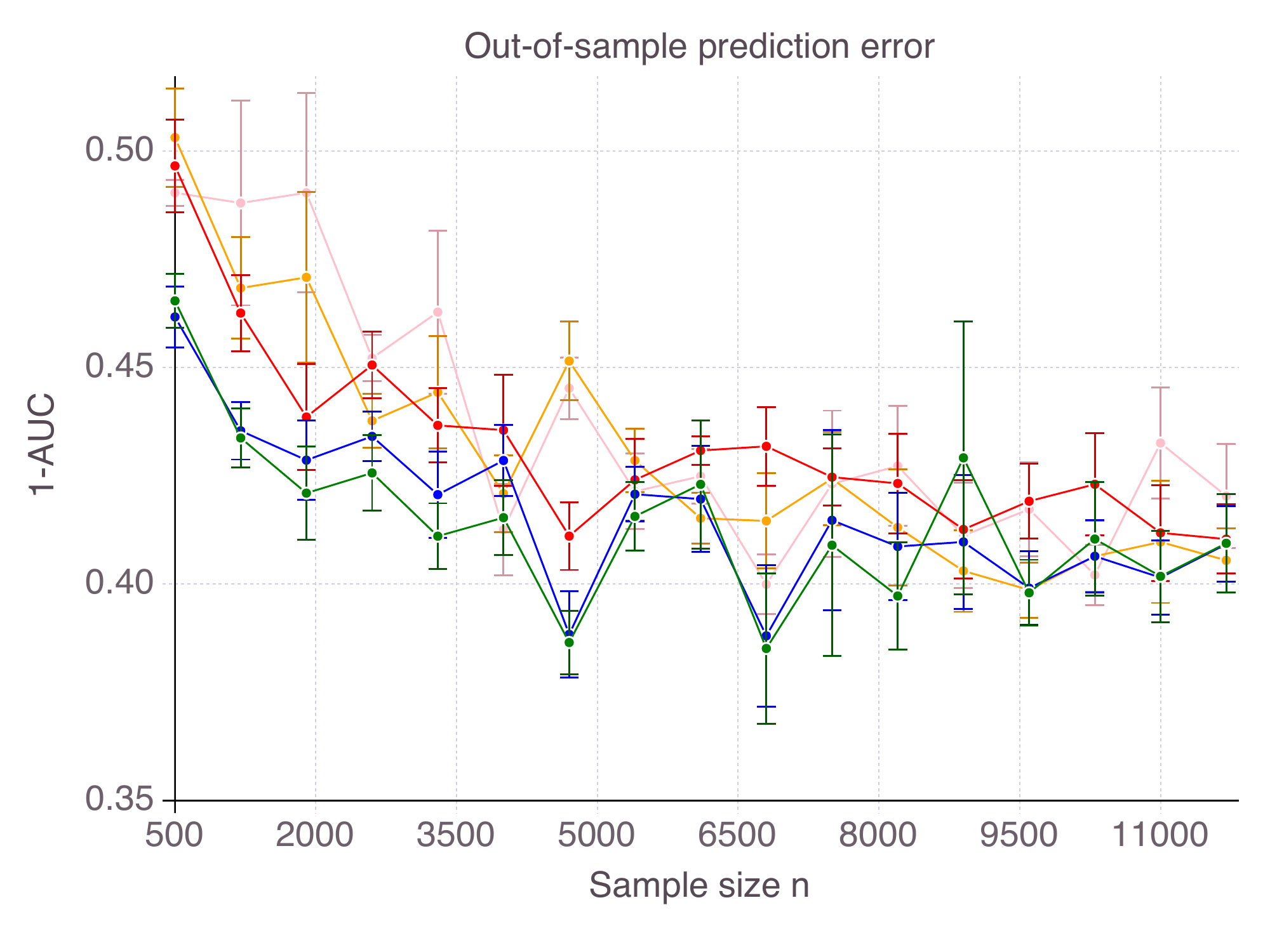}
	\caption{High noise, high correlation}
\end{subfigure}
\caption{Out-of-sample $1-AUC$ as $n$ increases, for the CIO (in green), SS (in blue with $T_{max}=200$) with Hinge loss, ENet (in red), MCP (in orange), SCAD (in pink) with logistic loss. We average results over $10$ data sets.}
\label{fig:ClassFixMSE}
\end{figure*}

Figure \ref{fig:ClassFixMSE} (p. \pageref{fig:ClassFixMSE}) represents the out-of-sample error $1-AUC$ for all five methods, as $n$ increases, for the six noise/correlation settings of interest. There is a clear connection between performance in terms of accuracy and in terms of predictive power, with CIO performing the best. Still, better predictive power does not necessarily imply that the features selected are more accurate. As we have seen for instance, MCP often demonstrates the highest accuracy, yet not the highest $AUC$.

\newpage
\bibliographystyle{plainnat}
\bibliography{biblio}
\end{document}